\documentclass[twoside,12pt,oneside,english,oldfontcommands,memoir]{memoir}
\usepackage[top=3cm, bottom=2.5cm, outer=2.5cm, inner=2.5cm, heightrounded, marginparwidth=2.5cm, marginparsep=-1cm]{geometry}
\usepackage{physics} 
\usepackage{braket}
\usepackage{amsmath} 
\usepackage{mathtools} 
\usepackage{amssymb} 
\usepackage{amsthm} 
\usepackage{mathrsfs} 
\usepackage{stackrel} 
\usepackage{bbm} 
\usepackage[usernames,dvipsnames,svgnames,table]{xcolor}
\usepackage{tabularx}
\usepackage{mdframed} 
\usepackage{tcolorbox} 
\definecolor{mylightblue}{RGB}{243,247,247}

\DeclareMathOperator{\Imag}{Im}
\DeclareMathOperator{\Real}{Re}
\newcommand{\mini}{\scriptscriptstyle}
\renewcommand{\arraystretch}{1.5} 
\usepackage{lscape} 
\usepackage{makecell} 
\usepackage{hhline} 
\usepackage{multirow} 
\usepackage{tikz}
\usepackage{float}
\usepackage[outercaption]{sidecap}
\usepackage{wrapfig}
\usepackage[font=small]{caption}
\usepackage{subcaption}
\usepackage{graphbox}
\usepackage{textcomp}
\usepackage[T1]{fontenc}
\usepackage[utf8]{inputenc}
\usepackage[english]{babel}
\usepackage{hyperref}
\definecolor{myblue}{rgb}{0.0, 0.3, 1}
\hypersetup{colorlinks=true,
linkbordercolor=Mulberry,
urlbordercolor=myblue,
citebordercolor=OliveGreen,
linkcolor=myblue,
filecolor=black,
urlcolor=black,
citecolor=OrangeRed}

\usepackage{url}
\usepackage{csquotes}
\usepackage{comment}
\usepackage{marginnote}
\usepackage{verse}
\usepackage{cite}
\usepackage[export]{adjustbox}
\usetikzlibrary{decorations.pathmorphing}
\numberwithin{equation}{section}
\numberwithin{figure}{section} 
\setsecnumdepth{subsubsection}
\usepackage{minitoc}
\usepackage{bookmark}
\nomtcrule 
\dominitoc[n]
\newtheorem{theoremx}{Theorem}[chapter]

\newtheorem{propositionx}[theoremx]{Proposition}
\theoremstyle{definition}
\newtheorem{definitionx}[theoremx]{Definition}
\newenvironment{definition}
  {\pushQED{\qed}\definitionx}
  {\popQED\enddefinitionx}
  \newenvironment{theorem}
    {\pushQED{\qed}\theoremx}
    {\popQED\endtheoremx}
    \newenvironment{proposition}
      {\pushQED{\qed}\propositionx}
      {\popQED\endpropositionx}
\newtheorem{remarkx}[theoremx]{Remark}
\newenvironment{remark}
{\pushQED{\qed}\remarkx}
  {\popQED\endremarkx}

\newtheorem{examplex}[theoremx]{Example}
\newenvironment{example}
{\pushQED{\qed}\examplex}
  {\popQED\endexamplex}

\newtheorem{counterexamplex}[theoremx]{Counter-example}
\newenvironment{counterexample}
{\pushQED{\qed}\counterexamplex}
  {\popQED\endcounterexamplex}

\usepackage{titlesec}
\titleclass{\part}{top}
\titleformat{\part}[display]
  {\huge\bfseries\centering}{\partname~\thepart}{0pt}{}
\titlespacing*{\part}{0pt}{40pt}{40pt}
\usepackage{environ}
\newcommand\drawingOne{\includegraphics[width=5mm]{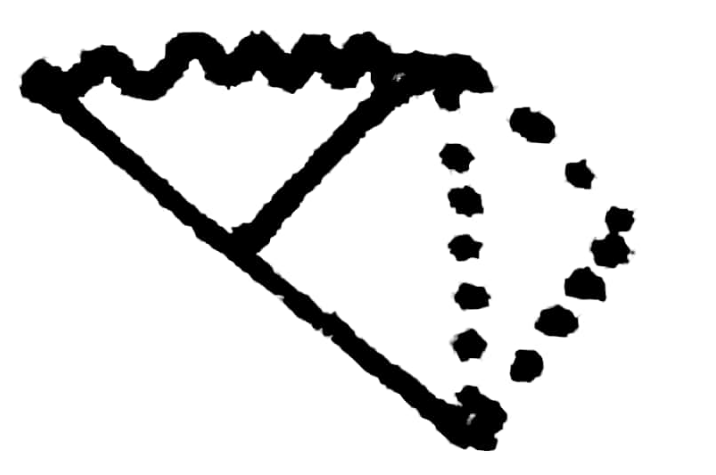}}

\newcommand\drawingTwo{\includegraphics[width=7mm]{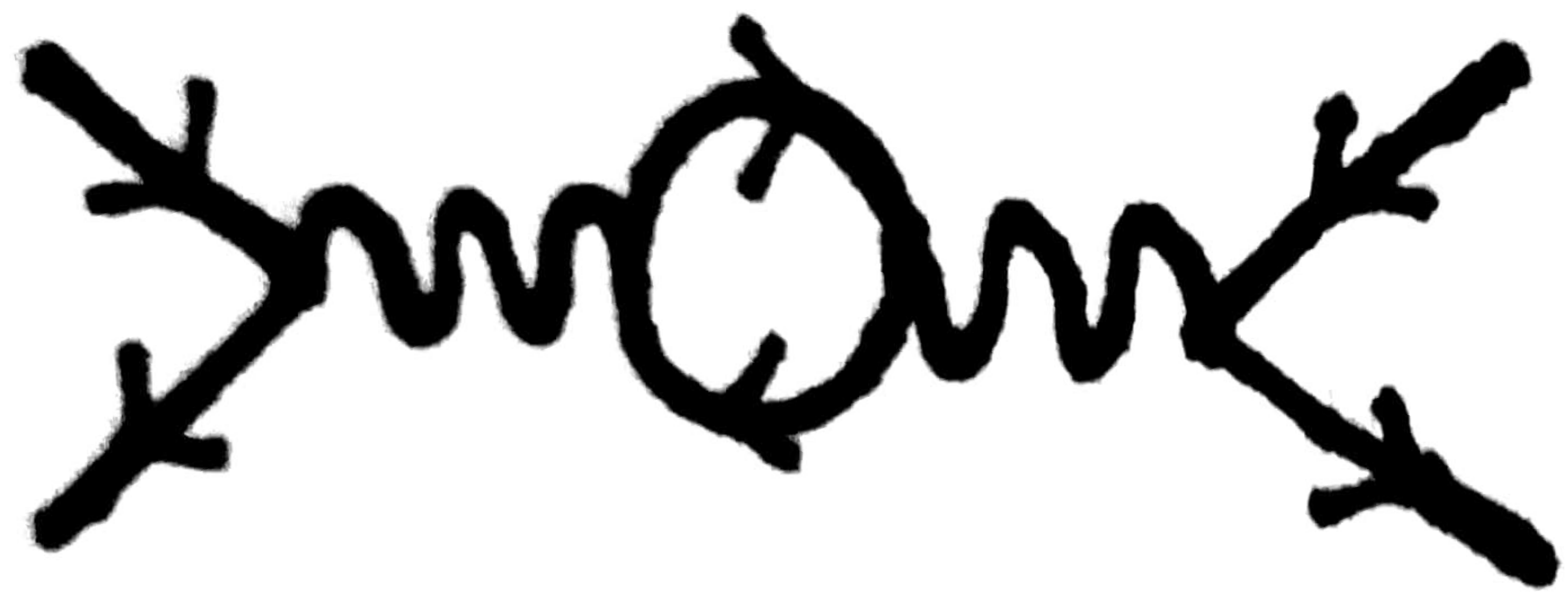}}
\newcommand\drawingThree{\includegraphics[width=2mm]{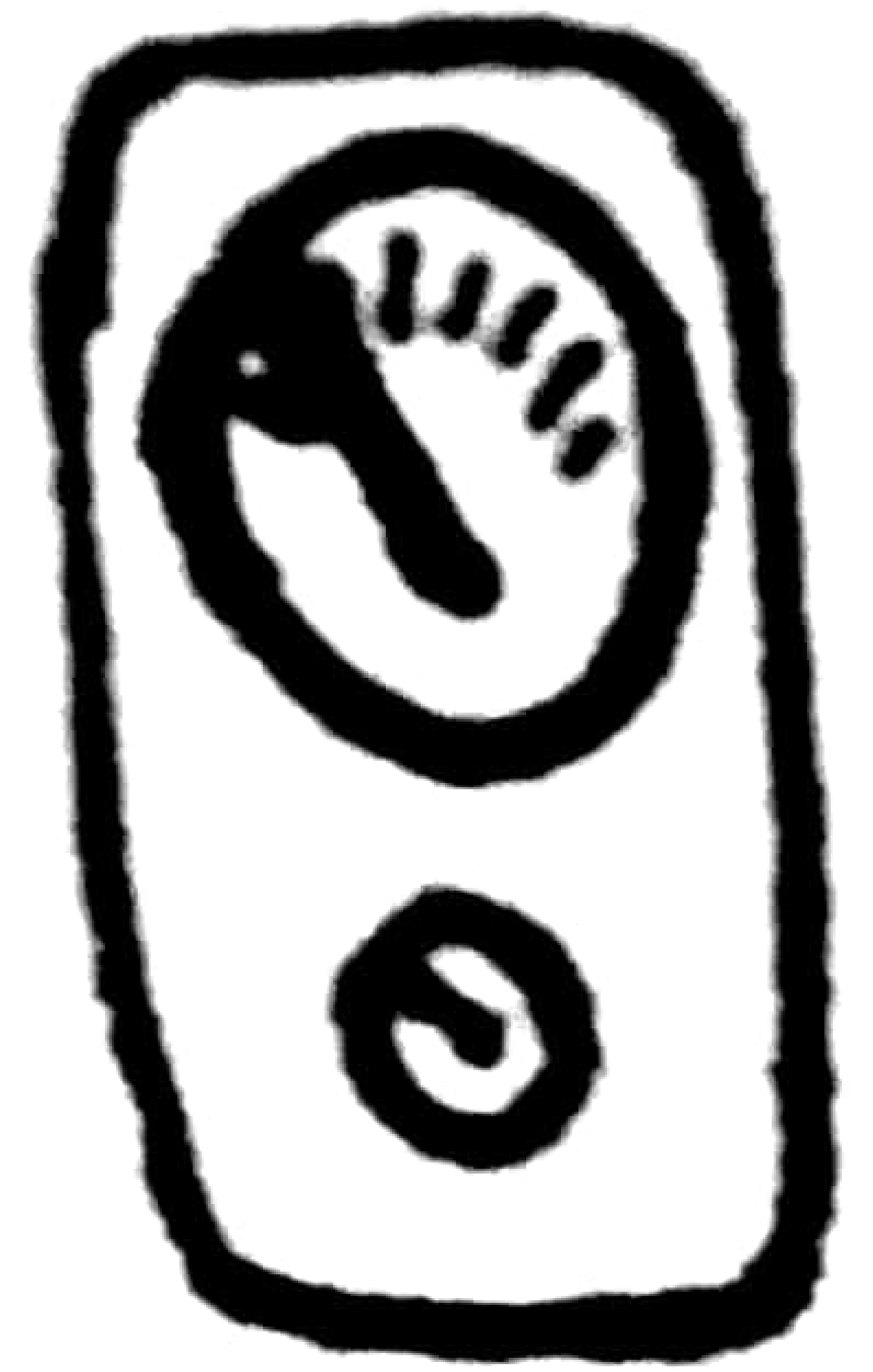}}
\makepagestyle{myPhDpagestyle}
\makeheadrule {myPhDpagestyle}{\textwidth}{1.5\normalrulethickness}
\makeevenhead {myPhDpagestyle}{}{\hfill\bfseries\leftmark} {}
\makeoddhead {myPhDpagestyle}{}{\bfseries\rightmark\hfill}{\thepage}
\makeevenfoot {myPhDpagestyle}{}{\small\thepage} {}
\makeatletter 
\makepsmarks {myPhDpagestyle}{
\nouppercaseheads
\createmark {chapter} {both} {nonumber}{ \@chapapp\ }{ \ }
\createmark {section} {right}{shownumber}{} { \ }
\createmark {subsection} {right}{shownumber}{} { \ }
\createmark {subsubsection}{right}{shownumber}{} { \ }
\createplainmark {toc} {both} {\contentsname}
\createplainmark {lof} {both} {\listfigurename}
\createplainmark {lot} {both} {\listtablename}
\createplainmark {bib} {both} {\bibname}
\createplainmark {index} {both} {\indexname}
\createplainmark {glossary} {both} {\glossaryname}
}
\makeatother
\copypagestyle{myPhDpagestyle1}{myPhDpagestyle}
\makeevenhead {myPhDpagestyle1}{}{\hfill\drawingOne\space\bfseries\leftmark} {}
\copypagestyle{myPhDpagestyle2}{myPhDpagestyle}
\makeevenhead {myPhDpagestyle2}{}{\hfill\drawingTwo\space\bfseries\leftmark} {}
\copypagestyle{myPhDpagestyle3}{myPhDpagestyle}
\makeevenhead {myPhDpagestyle3}{}{\hfill\drawingThree\space\bfseries\leftmark} {}
\makepagestyle{myPhDpagestyleCover}
\makeoddhead {myPhDpagestyleCover}{}{\hrule
}{\thepage}

\makeatother

\makechapterstyle{myPhDchapterstyle}{

}
\chapterstyle{myPhDchapterstyle}
\usepackage{graphicx}
\usepackage{mwe}
\newsavebox{\tempbox}

\begin{document}
\frontmatter
\pagenumbering{roman} 

\pdfbookmark[chapter]{Beginning}{Beginning}
\thispagestyle{empty}
\enlargethispage{3\baselineskip}

\vspace*{-3cm}
\begin{center}
\includegraphics[width=2.8cm]{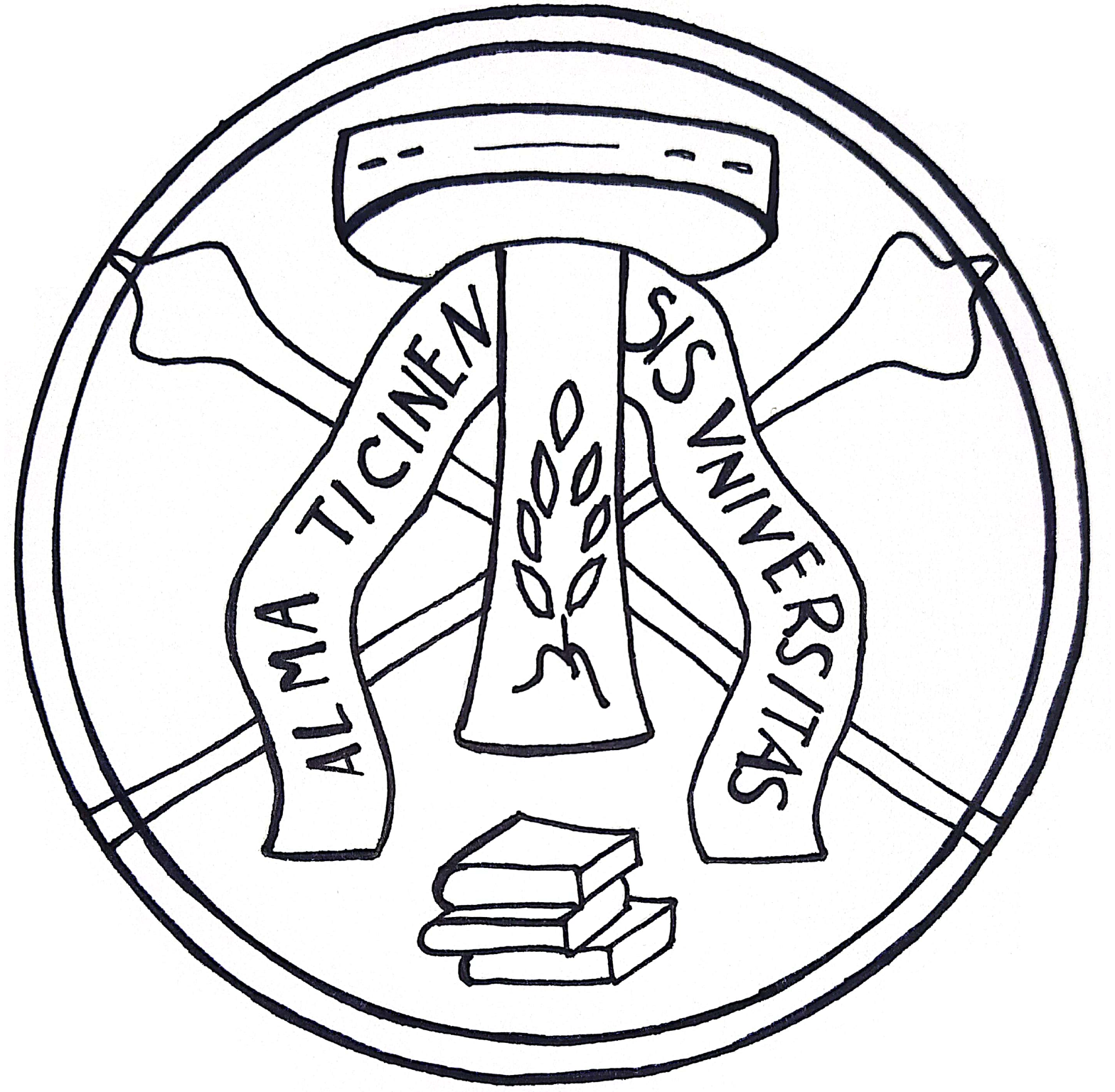}

{
  UNIVERSITÀ DEGLI STUDI DI PAVIA\\
  DOTTORATO DI RICERCA IN FISICA – XXXIV CICLO\\
}

{\LARGE
\vspace{3.7cm}
\bfseries

Probing thermal effects on static spacetimes with Unruh-DeWitt detectors
}

\vspace{6.5cm}
\includegraphics[width=.45\textwidth]{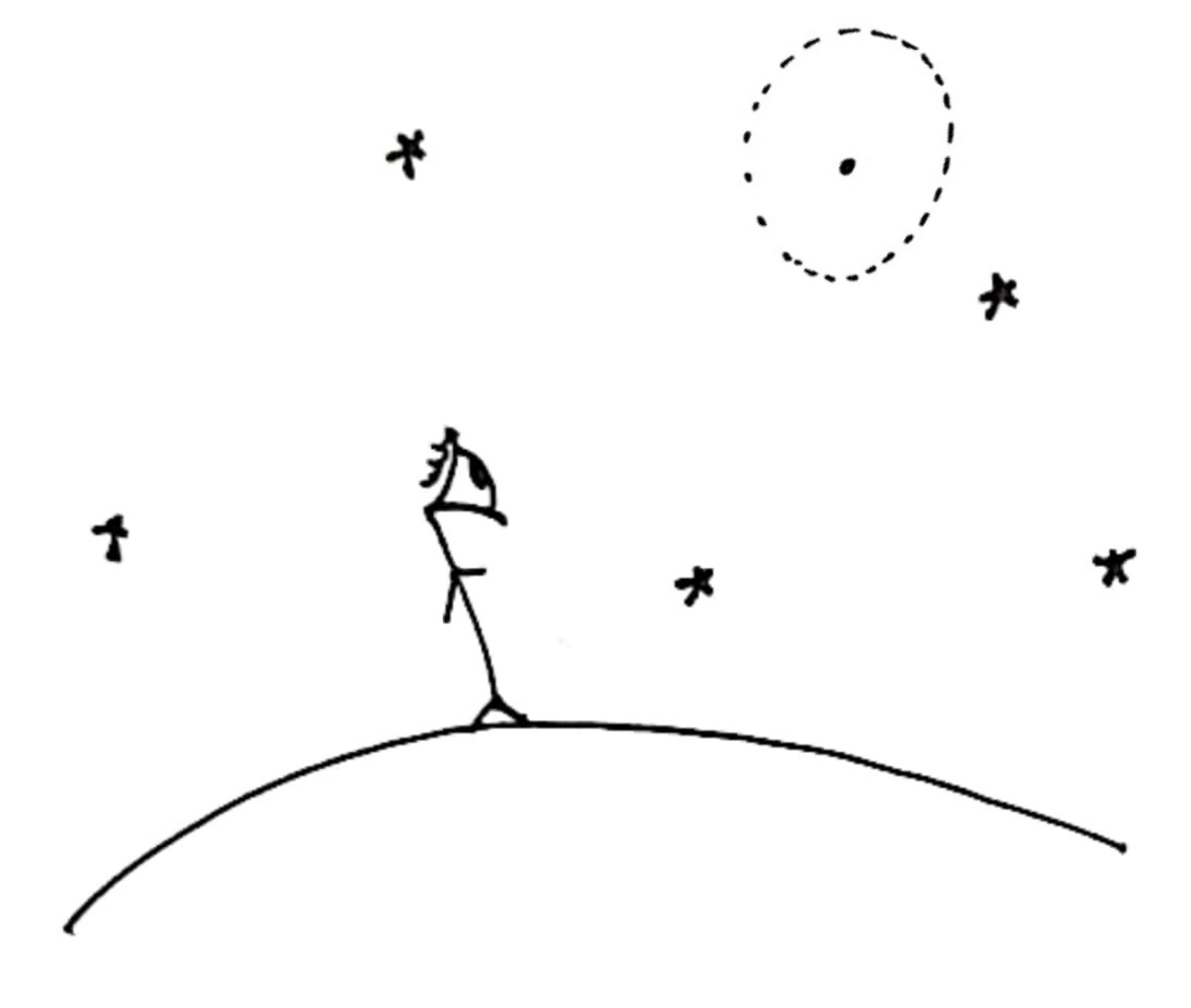}

\vspace*{3.5cm}

Lissa de Souza Campos

\vspace*{.5\baselineskip}
Supervisor: Prof. Claudio Dappiaggi

\end{center}
\clearpage

\vfill\null
\begin{flushleft}
Thesis submitted to the Graduate School of Physics at the University of Pavia \\

in partial fulfillment of the requirements for the degree of\\
DOTTORE DI RICERCA IN FISICA\\
2022
\end{flushleft}
\clearpage
\pdfbookmark[section]{Acknowledgments}{Acknowledgments}
\chapter*{Acknowledgments}

    Not for one second have I questioned my decision of coming to this foreign land and that is because of my advisor Claudio Dappiaggi. For answering all of my questions, for all the discussions, for all the feedbacks, for all the support, for all the opportunities given, for all the mentoring, and, specially, for being a model of what a physicist, a professor, an advisor can be, I am deeply grateful.\\

    I am most grateful to professors J. Barata, A. Piza, J. Pitelli, J. Louko and N. Pinamonti for the discussions and the support in all the application processes... I very much appreciate the feedback and advices from professors J. Louko and R. Mann and I thank them for having thoroughly read the manuscript of this thesis. Part of the results reported here was obtained in collaboration with D. Sina and J. Pitelli; it was a pleasure to work with them. I also appreciate the Italian culture, society and government that, through the University of Pavia, financially supported me. In addition, I acknowledge the ``Istituto Nazionale di Fisica Nucleare'' (INFN) for the financial support.\\

    A heartfelt thanks to my friends: migues, bb (e sua mamis maravilhosa), Lucio, ricardito, robinneke (and kwispeltje!), Shamim and Zé lindo. A special thanks to everyone that joined me for coffee breaks, to the employees of and my friends from UNES who are always so nice and have always had patience to hear my broken Italian, to my coworkers at the physics institute, particularly to Miss Virginie Gallati whose sympathy was incredibly welcomed in the turbulent starting up period, and to the strangers that were randomly kind to me and made some of my days. Finally, I recognize my family's support, without which I wouldn't be here today. \\

\cleardoublepage
\section*{}
\vspace{11cm}
\begin{quote}
\begin{flushright}
  \textit{ad ipsum esse}
\par\end{flushright}
\end{quote}
\hfill\thispagestyle{empty}
\newpage
\pdfbookmark[section]{Poeminha}{Poeminha}
\thispagestyle{empty}
\section*{}
%

\noindent \emph{Lamb shift}\\
\emph{light bend}\\
\emph{so far we've come}\\
\emph{the fire, the wheel}\\
\emph{the enquire, the will}\\
\emph{what will become}\\

%
 \noindent\emph{weak, strong, charges and gravity}\\
\emph{or a unique, so longed, that enlarges reality}\\

\noindent \emph{lo and behold}\\
\emph{quantum fields on smooth manifolds}\\
\emph{a semi-classical search}\\
\emph{not radical, no lurch}\\
\emph{yet much to reveal}\\
\emph{untouched, concealed}\\

%
%
\cleardoublepage
\pdfbookmark[section]{Abstract}{Abstract}
\chapter*{Abstract}

In the lack of a full-fledged theory of quantum gravity, I consider free, scalar, quantum fields on curved spacetimes to gain insight into the interaction between quantum and gravitational phenomena. I employ the Unruh-DeWitt detector approach to probe thermal, quantum effects on static, non-globally hyperbolic spacetimes. In this context, all physical observables depend on the choice of a boundary condition that cannot be singled-out, in general, without resorting to an experiment. Notwithstanding, the framework applied admits a large family of (Robin) boundary conditions and grants us physically-sensible dynamics and two-point functions of local Hadamard form. I discover that the anti-Unruh/Hawking effects are not manifest for thermal states on the BTZ black hole, nor on massless topological black holes of four dimensions. Whilst the physical significance of these statistical effects remains puzzling, my work corroborates their non-trivial relation with the KMS condition and reveals the pivotal influence of the spacetime dimension in their manifestation. On global monopoles, I find that for massless minimally coupled fields the transition rate, the thermal fluctuations and the energy density remain finite at the singularity only for Dirichlet boundary condition. For conformally coupled fields, although the energy density diverges for all boundary conditions, the transition rate and the thermal fluctuations vanish at the monopole; indicating that even if there is infinite energy, no spontaneous emission occur if the quantum field is not fluctuating. Moreover, I explicitly construct two-point functions for ground and thermal states on Lifshitz topological black holes, setting the ground for future explorations in this Lorentz breaking context. I expect my work to bring awareness on the intricate role played by choosing different boundary conditions, to stimulate the debate on thermal effects on black holes and naked singularities, and to promote the power and usefulness of semi-classical analyses.\\

\pagestyle{myPhDpagestyle}

\textbf{Key-words:} Unruh-DeWitt detectors, anti-Unruh effect, anti-Hawking effect, Robin boundary conditions, BTZ black hole, massless topological black holes, global monopoles, Lifshitz topological black holes.
\cleardoublepage
\begin{KeepFromToc}
  \pdfbookmark{\contentsname}{Contents}
  \tableofcontents
  \let\tableofcontents\relax
\end{KeepFromToc}
\mainmatter

\chapter*{Introduction}
\mtcaddchapter[Introduction]
\markboth{Introduction}{Introduction}
  \begin{flushright}
      \textit{\small{---If we squeezed ourselves into grains of sand, and then crushed the entire Milky Way to fit in a family-size pizza box, then each one of us would become a black hole that would dissipate in about ten picoseconds. That's way longer than inflation, but still smaller than the difference of age between your feet and your head due to Earth's gravity.}}\\
    \par\end{flushright}
\pagestyle{myPhDpagestyle}
\vspace{.5cm}

Taking into account quantum field theory and general relativity, I adopt the viewpoint that nature is fundamentally described by quantum fields and curved spacetimes. Quantum fields are the building blocks of the standard model of particles and of their weak, strong and electromagnetic interactions; the geometry of curved spacetimes encodes gravitation and models the large scale structure of the universe. Both theories possess crucial open problems \cite{Fredenhagen2006rv,Coley2018mzr} and outstanding experimental evidences \cite{eddington1919total,Lamb1947zz,aad2013evidence,will2014confrontation}; both theories are still being developed and questioned \cite{tHooft2015sdj}. A solid theory of everything consistently unifying all four fundamental interactions and revealing the underlying quantum gravitational laws of nature is yet to be singled-out and vouched for by reality; in \cite{Esposito2011rx} one can find a list of sixteen different approaches to quantum gravity, including string theory, twistor theory and loop quantum gravity. Notwithstanding, we can theoretically probe aspects of such a unified theory at their semi-classical interface. To this aim, I consider physical phenomena within quantum field theory on curved spacetimes, invoking the algebraic approach for the construction of physically-sensible states, and with focus on the Unruh-DeWitt detector's perspective on quantum, thermal effects on static spacetimes.

Quantum field theory on curved spacetimes has brought to light the fact that black holes emit Hawking radiation \cite{hawking1974black}. Its algebraic approach has clarified the relation between the latter and the existence of Killing horizons \cite{fredenhagen1990derivation,moretti2012state}. It has allowed for generalizations of the singularity theorems of general relativity taking into account quantum effects \cite{Fewster2021mmz}. It has introduced the PCT, spin statistics, Haag and Reeh-Schlieder theorems \cite{Haag1963dh}. It has made clear that quantum field theory does not predict a value, neither right nor wrong, for the cosmological constant \cite{hack2015cosmological}. Markedly, it possesses manifold applications within the Unruh-DeWitt particle detector approach \cite{hu2012relativistic}. We can see quantum field theory on curved spacetimes as an effective theory useful in situations where spacetime curvature is significant, and valid if we stay sufficiently away from the Planck scale. The drawback of considering the underlying background to be generic instead of the standard flat spacetime is the loss of symmetries \cite{wald2009formulation}. The four-dimensional Minkowski spacetime has ten global Killing vector fields, being invariant under four translations, three rotations and three Lorentz boosts. Altogether, Minkowski isometries form the Poincaré group. In turn, in a Hilbert space carrying a unitary representation of the Poincaré group, states are unit rays amongst which there is one, and only one, Poincaré-invariant state: Minkowski vacuum \cite[Pg.28,Pg.55]{Haag1963dh}. Such state is also the unique ground state, containing only positive-frequency modes, with respect to the corresponding (Hamiltonian induced by the) time-translation symmetry. Withal, the symmetries of Minkowski spacetime select a unique vacuum state that plays a central role in the formalization of quantum field theory---Minkowski vacuum is the reference state with respect to which expectation values are computed, disambiguating normal-ordering \cite{birrell1984quantum}. On a generic spacetime there is no isometry group that selects a preferred reference state. In this scenario, one can resort to the algebraic approach to quantum field theory.

The works of Haag and Kastler of over fifty years ago set the ground for a mathematically rigorous approach to quantum field theory in which the principle of locality is incorporated through the notion of algebras of observables, see \cite{Haag1963dh} and the references therein. The main advantages of the algebraic approach is that it allows us to bypass the lack of a preferred vacuum state and the ambiguity of choosing a Hilbert space representation \cite{haagLocalquantumphysics,brunetti2015advances}. Therefore, it is particularly well-suited when considering generic, curved spacetimes. In this context there is a preferred class of states, not necessarily containing a unique distinguished representative, called \textit{Hadamard} and with respect to which expectation values can be defined such that the underlying quantum fluctuations remain finite \cite{WaldQFTCS,fulling1981singularity,fewster2013necessity}. On Minkowski, de Sitter, the exterior Schwarzschild, Friedman-Lemaitre-Robertson-Walker and on any globally hyperbolic spacetime, the algebraic approach provides a straightforward quantization scheme: Hadamard states exist \cite{fulling1981singularity}, the Klein-Gordon equation (and Dirac and Proca equations) yields a well-posed Cauchy problem \cite{WELM,brunetti2015advances}, we can regularize the stress-energy tensor uniquely up to locally covariant geometrical terms \cite{WaldQFTCS} and we can construct the algebra of Wick polynomials \cite{moretti}. On a Kerr black hole, a Gödel universe, asymptotically anti-de Sitter spacetimes, Minkowski with plates, or any non-globally hyperbolic spacetime, such scheme is no longer valid and the existence and uniqueness of fundamental solutions of the wave equation is not guaranteed.

I focus on free, massive, real, scalar, quantum fields arbitrarily coupled to the scalar curvature of the underlying background and whose dynamics is ruled by the Klein-Gordon equation. The spacetimes of interest in this thesis are curved and non-globally hyperbolic, but they are also static and stably-causal. On one hand, their staticness means we can avail of a given global, timelike, irrotational Killing vector field to select a unique ground state, and to identify thermal states. On the other hand, their non-global hyperbolicity entails that establishing of a quantum field theoretical framework---existence of an explicit causal propagator and physically-sensible two-point functions from which expectation values can be obtained---must be performed case-by-case. As it happens, in the scenarios here the Klein-Gordon equation yields a well-posed Cauchy-boundary value problem. The construction of the causal propagator is reduced to the problem of finding self-adjoint extensions of the spatial part of the wave operator \cite{ishibashi2003dynamics}, which in turn requires choosing boundary conditions for the wave functions.

Boundary conditions are ubiquitous in physics. Some amongst them translate into physical principles, some are so natural they go unnoticed \cite{lindsay1929significance}, and others are in one-to-one correspondence with a physical observable \cite{bonneau2001self}. In the context of solving the Klein-Gordon equation as a Cauchy-boundary value problem, each boundary condition specifies a different fundamental solution, hence a distinct causal propagator, an inequivalent dynamics, a specific two-point function, a particular reference state \cite{ishibashi2003dynamics,dappiaggi2019fundamental}. Consequently, expectation values depend on the chosen boundary condition and only an experiment could single one out. In the last years, increasing attention has been given to the significance of choosing different boundary conditions within quantum field theory on non-globally hyperbolic spacetimes. Klein-Gordon and Maxwell fields admitting Robin boundary conditions have been taken into account on $n$-dimensional AdS spacetimes and its patches \cite{ishibashi2004dynamics,warnick2013massive,dappiaggi2016hadamard,dappiaggi2018mode,dappiaggi2018ground,pitelli2019boundary}, on BTZ black holes \cite{garbarz2017scalar,bussola2017ground}, on massless topological black holes \cite{morley2020quantum,morley2021vacuum}, on global monopoles \cite{pitelli2018gmonopole}, and for the study of phenomena such as quasinormal modes \cite{wang2015maxwell,wang2016maxwell,Wang2019qja,Wang2021uix}, superradiance \cite{Ferreira2017tnc,Dappiaggi2017pbe}, and the Casimir effect \cite{Mintz2005sj,Nazari2020vil}.

I consider Klein-Gordon fields on static BTZ black holes, Rindler-AdS$_3$ spacetime, massless hyperbolic black holes, flat, hyperbolic and spherical Lifshitz topological black holes and on global monopoles. I follow the prescription of Ishibashi and Wald \cite{ishibashi2003dynamics} that grants us well-defined dynamics, and I invoke the results of Sahlmann and Verch \cite{sahlmann2000passivity} that guarantee the constructed states are of (local) Hadamard form. Markedly, the framework used gives us physically-sensible two-point functions that admit the large class of Robin boundary conditions. Subsequently, to study thermal, quantum effects, I invoke the Unruh-DeWitt detector approach. Unruh-DeWitt particle detectors \cite{birrell1984quantum,Louko2007mu} have been applied in several different contexts \cite{hu2012relativistic}, as the following list of examples corroborates: the Unruh and Hawking effects \cite{fewster2016waiting,hodgkinson2014static,louko2016unruh,juarez2018quantum,smerlak2013new}, the firewall proposal \cite{louko2014unruh,martin20151}, the Casimir-Polder effect \cite{dkebski2021probing}, cosmological quantum entanglement \cite{martin2012cosmological}, measure theory \cite{ruep2021weakly,grimmer2021tale}, quantum information \cite{lopp2021quantum}, quantum optics \cite{tjoa2021makes, sachs2021unruh}, naked singularities \cite{Davies1987th,Cong2021tnk,deSouzaCampos2021awm}, and quantum gravity \cite{Foo2021fno,faure2020particle}. Here, I model the Unruh-DeWitt detector with a pointlike two-level system, interacting for an infinite proper time with the underlying quantum field, and following a static trajectory. In this setting, the transition rate can be straightforwardly computed from the two-point functions. Then, by studying its behavior with respect to its different parameters, we can infer properties of the underlying quantum field and spacetime.

\subsubsection*{Limitations of the framework}
\label{page Limitations of the framework}
\begin{itemize}

  \item[$\dagger$] The underlying backgrounds are all \textcolor{ProcessBlue}{static}. This means we assume we have a well-defined notion of time and energy, positive frequencies, ground and thermal states. However, both cosmologically and astrophysically, \textcolor{ProcessBlue}{non-static} scenarios are of interest: the universe is expanding and black holes spin. We may assume this is an approximation to a very-slowly rotating system, or we may view this is as a preliminary theoretical exploration to set the path towards generalizing the framework to rotating scenarios. The latter is a standard interpretation and it is regularly justified. For example, on BTZ black holes, the construction of two-point functions has been generalized to the rotating case \cite{bussola2017ground,bussola2018tunnelling}, and also the study on the anti-Hawking effect \cite{robbins2021anti}.
  \item[$\dagger$] I consider \textcolor{Purple}{scalar, bosonic} quantum field theory. \textcolor{Purple}{Fermionic, Dirac, Maxwell} fields can also be considered and it is expected that Unruh-DeWitt detectors are able to distinguish them \cite{hummer2016renormalized,WanMokhtar2018lwi}. Here, Klein-Gordon fields, which do represent particles such as $\eta$, $\eta'$, and Higgs bosons \cite{aad2013evidence,aad2015study}, are viewed as (the simplest) toy-models to employ when formulating quantum field theory on curved spacetimes.
  \item[$\dagger$] The Unruh-DeWitt detector model used admits many generalizations: a \textcolor{red}{pointlike} \textcolor{orange}{two-level} system that follows a  \textcolor{Green}{static} trajectory, has \textcolor{blue}{infinite} time to interact with the quantum field \textcolor{Sepia}{itself} and is formalized within \textcolor{Magenta}{first}-order perturbation theory can be generalized to a \textcolor{red}{spatially extended} \textcolor{orange}{quantum} system that follows \textcolor{Green}{(circular, infalling, geodesic) other} trajectories, has \textcolor{blue}{finite} time to interact with the quantum field\textcolor{Sepia}{'s derivative} and is formalized to \textcolor{Magenta}{higher}-orders in perturbation theory.
\end{itemize}
Each of the generalizations highlighted above would make our analysis more realistic, but more complex. Each lies along different research avenues to follow and constitutes an interesting step to pursue in future work. Aware of its limitations, let us briefly outline the strengths of the framework employed (when compared with other approaches that have the same limitations).

\subsubsection*{Strengths of the framework}
\label{page Strengths of the framework}
\begin{itemize}\setlength\itemsep{0em}
  \item[$\diamond$] It admits general, Robin boundary conditions.
  \item[$\diamond$] It allows for massive arbitrarily coupled fields.
  \item[$\diamond$] From the ground state one directly construct thermal states.
  \item[$\diamond$] Knowing the two-point functions amounts to knowing the transition rates.
  \item[$\diamond$] The expressions for the transition rates are suitable for numerical computations.
\end{itemize}

Within this approximative, simplified, but effective and fruitful framework, this thesis accomplishes two goals: the explicit construction of physically-sensible two-point functions on non-globally hyperbolic spacetimes, and the study of thermal effects by employing Unruh-DeWitt particle detectors. Categorically, we can divide the applications performed in the following three main focuses, but note that the ambiguity of choosing boundary conditions is at the heart of all of them.

\subsubsection*{The three main focuses}
\label{page main focuses}
\begin{enumerate}
    \item \underline{Anti-Unruh and anti-Hawking effects}
    \item[] A static, uniformly-accelerated detector coupled to the ground state on the three-dimensional Minkowski spacetime manifests the (weak) anti-Unruh effect. On a static BTZ black hole, depending on the mass of the black hole and on the boundary condition chosen, it manifests the (weak) anti-Hawking effect. Both these effects regard the fact that the closer the detector is to the horizon, the higher is the local Hawking temperature measured---and yet the transition rate decreases \cite{brenna2016anti,Henderson2019uqo}. By analysing the behavior of the transition rate with respect to the local Hawking temperature for different boundary conditions, for different quantum states and for different black hole masses, I study the anti-Unruh and anti-Hawking effects on static BTZ black holes, on Rindler-AdS$_3$ spacetime, and on massless hyperbolic black holes. These effects are defined in Section \ref{sec: Unruh, Hawking, anti-Unruh, anti-Hawking effects}, and the results obtained are summarized in Sections \ref{sec: On a static BTZ spacetime and on Rindler-AdS3} and \ref{sec: On massless hyperbolic black holes}.
    \item \underline{Lorentz violation}
    \item[]  The existence of a fundamental length scale within quantum gravity leads to the question of whether local Lorentz invariance is a fundamental physical principle verified by all observers at all energy scales \cite{Kostelecky1988zi,GrootNibbelink2004za,Vucetich2005ra,Collins2006bw}. Motivated by this question, I consider Lifshitz topological black holes \cite{Mann2009yx}, whose line elements manifest a scaling behaviour characteristic of Ho\v{r}ava-Lifshitz gravity. 
    Their geometry is outlined in Section \ref{ex: chapter 1 Lifshitz toplogical black holes}, while in Section \ref{sec: On Lifshitz topological black holes}, on these Lorentz violating backgrounds, I detail the construction of ground and thermal states of local Hadamard form admitting Robin boundary conditions at Lifshitz infinity. It is interesting to note that the spherical Lifshitz topological black hole has a naked singularity, and yet no boundary condition is required at such locus.

\newpage
    \item \underline{Naked singularities}
    \item[] On black holes, the existence of a bifurcate Killing horizon defines a global Hawking temperature that, in turn, allows us to identify a specific thermal state with the property of being analytic across the horizon---the Hartle-Hawking state. On naked singularities, there is seemingly no special temperature a priori. Therefore, the fact that physically-sensible thermal states can be constructed on spacetimes containing a naked singularity (as implicitly mentioned in the item above) encourages an exploration on thermal effects on these scenarios. Here, I choose to work on global monopole spacetimes, whose significance arise from the context of grand unified theories \cite{Barriola1989hx,vilenkin1994cosmic}. I generalize the existing quantum field theoretical framework on global monopoles by constructing, for massive and arbitrarily coupled free, scalar fields, the two-point functions of thermal states. In contrast with the spherical Lifshitz topological black hole, where a boundary condition must be chosen at Lifshitz infinity, we find that on global monopoles it is the naked singularity that requires one. In Section \ref{sec: On a global monopole}, I study how the boundary conditions at the singularity affect three quantities of interest: the transition rate of a static Unruh-DeWitt detector coupled to a thermal state, the thermal contributions to the ground state fluctuations, and the thermal contributions to the energy density of the ground state.
\end{enumerate}

\subsubsection*{Thesis structure}

First, in Chapter \ref{chap: The Spacetimes}, I review geometrical concepts such as Killing vector fields, static spacetimes, and non-global hyperbolicity. I define Schwarzschild-like coordinates and I highlight the main geometrical features of the spacetimes that are considered in Chapter \ref{chap: Applications}. In Chapter \ref{chap: Quantum Field Theory on Static Spacetimes} I outline the establishment of a free, scalar, quantum field theory on static spacetimes. I review the algebraic approach, define Hadamard, ground and thermal states, I distill the steps for the construction of physically-sensible two-point functions in Schwarzschild-like coordinates, and I delineate the Unruh-DeWitt detector approach. In Chapter \ref{chap: Applications}, I display the results I have obtained and that were published in \cite{deSouzaCampos2020bnj,deSouzaCampos2020ddx,deSouzaCampos2021awm,deSouzaCampos2021role}. In the Conclusions section, I summarize these results and I ponder on their significance and on the follow-up work they incite.

\cleardoublepage

\chapter{Static Spacetimes }
\label{chap: The Spacetimes}
\minitoc
\pagestyle{myPhDpagestyle1}
\vfill

General relativity is written in the language of Differential Geometry. It unites the notions of space and time in the structure of a Lorentzian manifold: essentially, a spacetime is a set equipped with a topology, a differential structure and causality notions. The spacetimes pertinent to the second part of this thesis, on which I establish a quantum field theoretical framework, possess extra features. They are all static spacetimes with sections of constant sectional curvature that admit Schwarzschild-like coordinates. In this chapter, I recapitulate a few concepts of general relativity necessary to understand their geometries. The first three sections regard Killing vector fields and their consequences. In Section \ref{sec: The definition of static spacetimes}, I give the definition of static spacetimes, which relies on the existence of a particular Killing vector field. In Section \ref{sec: A relation between symmetries and curvature}, I discuss the relation between a spacetime admitting the maximal number of Killing vector fields and the notion of a spacetime of constant sectional curvature. Since a Killing vector field may generate a non-degenerate bifurcate Killing horizon, it may be associated with a geometric temperature, as shown in Section \ref{sec: On the geometric Hawking temperature}. In Section \ref{sec: A consequence of (non-)global hyperbolicity}, I use a few causality notions to relate the concept of global hyperbolicity to that of well-posedness of Cauchy problems of partial differential equations. Throughout this chapter I consider Minkowski spacetime to exemplify the notions introduced, but in Section \ref{sec: Examples chapter 1} I sketch other examples of static spacetimes and I highlight their main geometric features.

\newpage

\section{The definition of static spacetimes}
\label{sec: The definition of static spacetimes}

An $n$-dimensional \textit{spacetime} is a Hausdorff, second-countable, connected, orientable, time-orientable, $n$-dimensional smooth manifold e\-quipped with a Lorentzian metric tensor $g$ and a Levi-Civita connection $\nabla$. \label{page def spacetime} This definition provides a minimal set of mathematical tools. It allow us to, for example, completely characterize the motion of free falling test particles for a given metric \cite[Pg. 70, Pg. 148]{weinberg1972gravitation}. However, a general spacetime lacks many features one can make use of on Minkowski spacetime. Markedly, the enjoyment of geometric symmetries, such as translations, rotations or Lorentz boosts, is not guaranteed. On this account, the class of static spacetimes---which bears one distinct translation symmetry---is of particular relevance.

The precise definition of a static spacetime relies on the notion of a Killing vector field. Hence, let us first elucidate the latter. Suppose $\mathcal{M}$ is a spacetime endowed with a Lorentzian metric tensor $g$ with signature $(-,+,...,+)$. Let $\mathcal{L}_\xi$ be the Lie derivative with respect to a vector field $\xi$ on $\mathcal{M}$, then $\xi$ is a  \textit{Killing vector field } if and only if
\begin{equation}
  \label{eq: definition Killing vector field lie derivative = 0}
  \mathcal{L}_{\xi} g := \nabla_{\mu} \xi_{\nu}+\nabla_{\nu} \xi_{\mu} = 0.
\end{equation}

As detailed in \cite[App.C]{wald2010general}, if Equation \eqref{eq: definition Killing vector field lie derivative = 0} holds for all $p\in\mathcal{M}$, then $g$ is invariant along all integral curves of $\xi$ and the associated flux yields a one-parameter group of isometries. That is, global Killing vector fields identify continuous symmetries of a spacetime.

A spacetime is \textit{stationary} if it admits a timelike Killing vector $\xi$. The corresponding one-parameter group of isometries entails invariance of the metric under translation along the coordinate subordinated to the affine parameter of the underlying integral curves of $\xi$. If, in addition, it holds that $\xi_{[\mu}\nabla_{\nu} \xi_{\lambda]} = 0 $, then $\xi$ is called \textit{irrotational} and the spacetime is said to be \textit{static}. This property, equivalent to $\xi$ being hypersufarce-orthogonal and sometimes referred to as ``integrability condition'', brings about time-reflection symmetry. For details, see e.g. \cite[Pg. 119]{wald2010general}.

\begin{definition}[Static spacetime] \label{def: static spacetime} A spacetime is \textit{static} if it admits a global, non-vanishing, timelike, irrotational Killing vector field.
\end{definition}
\begin{example}[Killing vectors and static spacetime]
\label{eg: 10 Killing in Mink}
  Let us solve the Killing equation \eqref{eq: definition Killing vector field lie derivative = 0} on the four-dimensional Minkowski spacetime with metric tensor $g$ and Cartesian coordinates $(t,x,y,z)$ such that
  \begin{equation*}
    \label{eq: metric mink 4D}
    ds^2=-dt^2+dx^2+dy^2+dz^2.
  \end{equation*}
  For a generic vector
  \begin{equation}
    \label{eq: generic linear vector for ex mink 10 kil}
    \xi_{\mu}=(\xi_{t}(t,x,y,z),\xi_{x}(t,x,y,z),\xi_{y}(t,x,y,z),\xi_{z}(t,x,y,z) ),
  \end{equation}
     Equation \eqref{eq: definition Killing vector field lie derivative = 0} yields
  \begin{subequations}
    \label{eq: killing eqs mink 4D Cartesian}
  \begin{gather}
      \partial_k\xi_k(t,x,y,z)=0, \text{ for }k\in\{t,x,y,z\},\\
      \partial_j\xi_t(t,x,y,z)+\partial_t\xi_j(t,x,y,z)=0 ,  \text{ for } j\in\{x,y,z\},\\
      \partial_i\xi_j(t,x,y,z)+\partial_j\xi_i(t,x,y,z)=0 ,  \text{ for } i\neq j \text{ and } i,j\in\{x,y,z\}.
  \end{gather}
\end{subequations}
Assuming that each component is a linear function of each coordinate, for $i\in\{0,1,2,3,4\}$, $k\in\{t,x,y,z\} $ and constants $c_{k_i}$, it holds $$ \xi_k(t,x,y,z) = c_{k_0} + c_{k_1}t+ c_{k_2}x+ c_{k_3}y+ c_{k_4}z.$$
Substituting this ansatz in Equation \eqref{eq: killing eqs mink 4D Cartesian} restricts the possible values of the $20$ constants $c_{k_i}$. We find that the Killing vector \eqref{eq: generic linear vector for ex mink 10 kil} with components \eqref{eq: killing eqs mink 4D Cartesian} depends on $10$ arbitrary constants and it is of the form
  \begin{align*}
    \xi =   &(c_{t_0} + c_{t_2}x + c_{t_3}y + c_{t_4}z)\partial_t \nonumber\\
            +&(c_{x_0} - c_{t_2}t + c_{x_3}y + c_{x_4}z)\partial_x\nonumber\\
            +&(c_{y_0} - c_{t_3}t - c_{x_3}x + c_{y_4}z)\partial_y \nonumber\\
            +&(c_{z_0} - c_{t_4}t - c_{x_4}x - c_{y_4}y)\partial_z .
  \end{align*}
Hence, there are $10$ linearly independent Killing vectors:
  \begin{gather*}
    \text{Translations: }(1,0,0,0),\, (0,1,0,0),\, (0,0,1,0),\, (0,0,0,1);\\
    \text{Rotations: }  (0,-y,x,0),\, (0,-z,0,x),\, (0,0,-z,y);\\
    \text{Boosts: } (x,-t,0,0),\, (y,0,-t,0),\,(z,0,0,-t).
  \end{gather*}
  Each Killing vector is associated with an isometry of the Poincar\'{e} group, as indicated above. In particular, the Killing vector $\xi_{\mu}=(1,0,0,0)$, identified with $\partial_t$, yields a timelike, irrotational Killing vector field, i.e. the four-dimensional Minkowski spacetime is static.
\end{example}
\begin{counterexample}[Non-static spacetime]
  \label{ceg: kerr is not static}
 A four-dimensional Kerr black hole with mass $M>0$ and angular momentum $M a$ is characterized, in Boyer-Lindquist coordinates $(t, r, \theta, \varphi)$, by the line element
      \begin{equation*}\label{eq: kerr metric} ds^2 = -dt^2 + \frac{2Mr}{\rho(r,\theta)^2}(a\sin(\theta)^2 d\varphi -dt)^2 + \rho(r,\theta)^2 \left( \frac{dr^2}{\triangle(r)} + d\theta^2\right) + (r^2+a^2)d\varphi^2,
      \end{equation*}
where $\rho(r,\theta):=r^2 + a^2 \cos(\theta)^2$ and $\triangle(r):=r^2 -2Mr +a^2$. Since the metric coefficients $g_{\mu\nu}$ are independent of $t$, it is easy to see that $\partial_t$ is a Killing vector. Noting that that the line element above is invariant under a time-reflection transformation $t\mapsto -t$ if and only if $a=0$, we conclude that $\partial_t$ is not irrotational. With extra effort, one can show that in fact there is no irrotational Killing vector field on this spacetime if $a\neq 0$. Therefore, it is not static \cite[Pg. 297]{wald2010general}. For details on this spacetime, see \cite{boyer1967maximal} and \cite[Sec.5.6]{hawking1973large}.
\end{counterexample}
Apart from not necessarily having Killing vector fields or symmetries, a general spacetime might not admit a global coordinate chart or foliations. However, if a spacetime admits a global, non-vanishing, timelike, irrotational Killing vector field, these features are guaranteed, see e.g. \cite{gutierrez2003splitting} and the references therein. Explicitly, a static spacetime with a timelike Killing vector field $\xi$ admits a foliation $\mathcal{M}\cong \{t\}_{t\in\mathbb{R}}\times\Sigma$ of spacelike hypersurfaces with normal direction $\xi$. Accordingly, it admits a coordinate system $(t,\underbar{x})$, where $t\in\mathbb{R}$ and $\underbar{x}\in\Sigma$, such that its line element can be written as
  \begin{equation}
  \label{eq: metric static spacetimes}
    ds^2=-f(\underbar{x})dt^2+h_{ij}(\underbar{x})d\underbar{x}^id\underbar{x}^j,
  \end{equation}
where $f$ and $h_{ij}$ are smooth functions independent of $t$, $f$ is positive and $h_{ij}$ identifies a Riemannian metric $h$ on $\Sigma$. In these coordinates, we have $\xi=\partial_t$, and $t$ is called \textit{Killing parameter}.
\vspace{1cm}
   \begin{figure}[H]
     \centering
      \includegraphics[width=.3\textwidth]{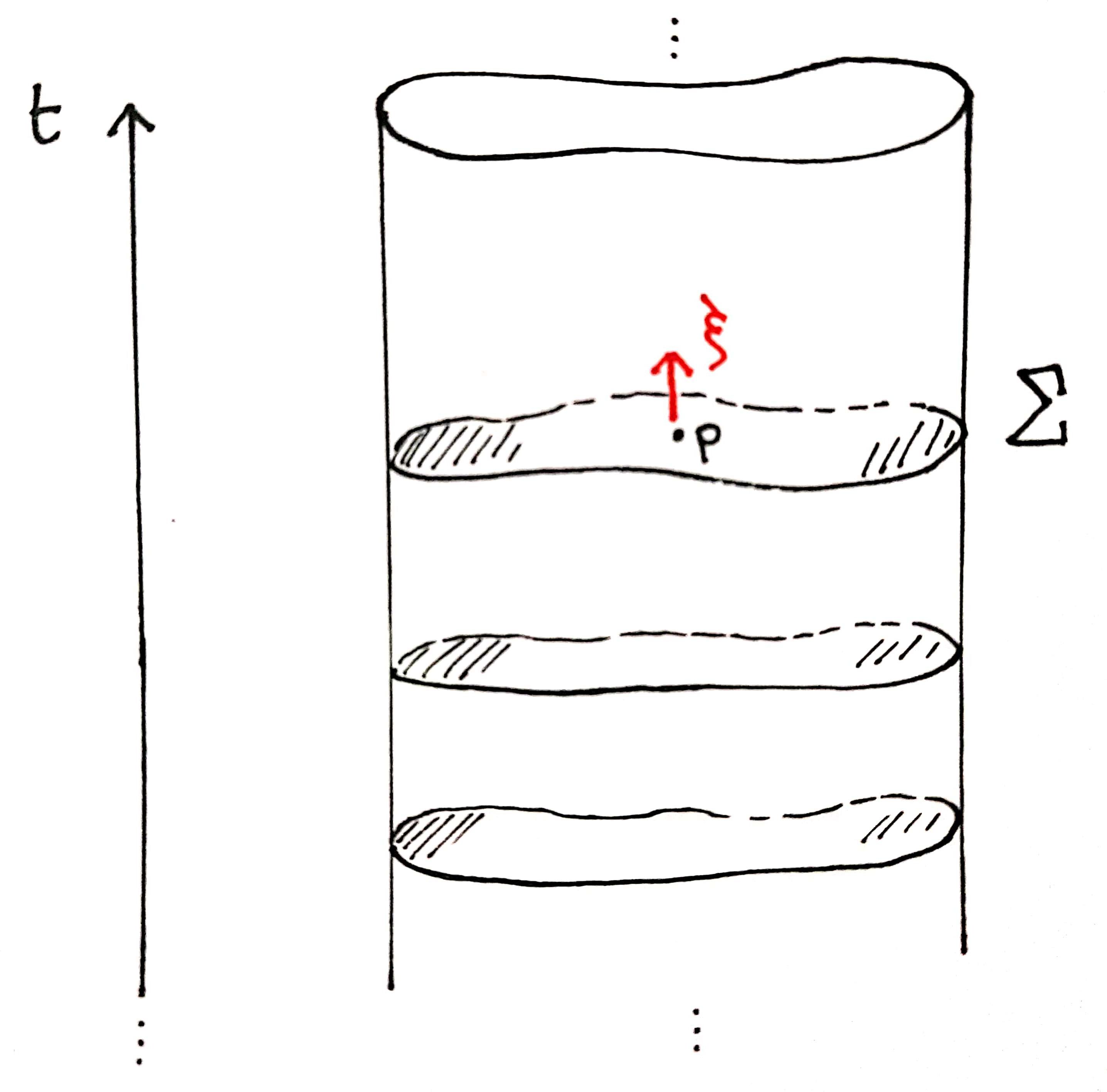}
      \caption{A foliation of a static spacetime. Each $\{t\}\times\Sigma$ is a constant time slice.}
     \label{fig: foliation static spacetime}
       \end{figure}
A physical consequence of the presence of Killing vector fields, since they yield continuous symmetries of the spacetime, is the existence of conserved quantities. Directly from the Killing Equation \eqref{eq: definition Killing vector field lie derivative = 0} follows that along a geodesic with tangent vector $u^{\mu}$ the quantity $\xi_{\mu}u^{\mu}$ is constant. This is a useful property that allows us to consider the gravitational redshift effect on stationary spacetimes and to define a local Hawking temperature, as shown in Section \ref{sec: On the geometric Hawking temperature}. Furthermore, in the case of a static spacetime, the presence of a timelike, irrotational Killing vector field allows us to employ Fourier analysis in the time direction.
In the coordinates of Equation \eqref{eq: metric static spacetimes}, since the volume element $\sqrt{|g|}d\underbar{x}$ is independent of $t$, we can define the Fourier transform $\widehat{\zeta}$ with respect to $t$ of a function $\zeta:\mathcal{M}\rightarrow\mathbb{C}$ analogously to the standard definition on $\mathbb{R}$:
\begin{equation*}
  \label{eq:fouriertransform2}
    \widehat{\zeta}(\omega,\underbar{x})=\frac{1}{2\pi}\int_\mathbb{R} e^{-i\omega t}\zeta(t,\underbar{x})dt.
\end{equation*}
It follows that a quantum field theoretical framework on a static spacetime benefits from well-defined notions of time translation invariance, Hamiltonian, energy, positive-frequencies, ground state, and momentum representation---see \cite{wald2009formulation} for a discussion on the difficulties that emerge when formulating quantum field theory on general curved spacetimes. On this account, and given that general spacetimes lack such structure, one can see static spacetimes as the simplest generalization of Minkowski spacetime.

\section{A relation between symmetries and curvature}
\label{sec: A relation between symmetries and curvature}

In the last section, it was shown that a static spacetime possesses one Killing vector field that yields a time-translation symmetry. With Minkowski spacetime in mind, as in Example \ref{eq: killing eqs mink 4D Cartesian}, we know that a general spacetime could be invariant under $n$ translations, and under $\frac{n(n-1)}{2}$ rotations. 
One may wonder what other continuous symmetries could an $n$-dimensional spacetime have. Yet, as it turns out, translations and rotations account for all possible continuous symmetries---an $n$-dimensional spacetime can have at most $n + \frac{n(n-1)}{2}=\frac{n(n+1)}{2}$ linearly independent Killing vector fields \cite[Pg.  378]{weinberg1972gravitation}. In this case, the spacetime is said to be \textit{maximally symmetric} 
and it necessarily has constant sectional curvature. Since some spacetimes considered in this thesis fall in this class, and all of them have sections of constant sectional curvature, here I discuss maximally symmetric spacetimes, constant sectional curvature and the connection between these concepts.

Let us start by defining sectional curvature. On a spacetime $\mathcal{M}$, let $R_{\mu\nu\alpha\beta}$ and $g_{\mu\nu}$ represent, respectively, the components of the Riemann and of the metric tensors in a given coordinate system. At each $p\in\mathcal{M}$, the \textit{sectional curvature} $s_p: T_p\mathcal{M}\times T_p\mathcal{M}\rightarrow \mathbb{R}$ takes two tangent vectors $\partial_{u}$ and $\partial_{v}$ at $p$ and returns the Gaussian curvature of the two-dimensional surface in $\mathcal{M} $ that has the plane $\sigma_p\subset T_p\mathcal{M} $ generated by $\partial_{u}$ and $\partial_{v}$ as its tangent plane at $p$ and that is obtained as the image of $\sigma_p$ under the exponential map at $p$. Equivalently,
\begin{equation}
  \label{eq:definition sectional curvature}
  s_p(\partial_{u},\partial_{v}):=\frac{R_{uvvu}}{g_{uu}g_{vv}-g_{uv}^2}.
\end{equation}
\begin{definition}[Constant sectional curvature] \label{def: constant sectional curvature} If the sectional curvature $s_p(\partial_{u},\partial_{v})$ is independent of $\partial_{u}$, $\partial_{v}$ and $p$, then we say $\mathcal{M}$ has constant sectional curvature.
\end{definition}
As shown in \cite[P.13, Ch.2]{weinberg1972gravitation}, on a maximally symmetric spacetime, the Riemann tensor, the Ricci tensor, and the Ricci scalar can be written, for a constant $s$, respectively as
\begin{subequations}
\label{eq: maximally symmetric curvature tensors}
\begin{gather}
  R_{\mu\nu\alpha\beta} = s\, (g_{\mu\alpha}g_{\beta\nu}-g_{\mu\beta}g_{\alpha\nu}),\label{eq: maximally symmetric curvature Riemann tensor}\\
  R_{\mu\nu} = s \, (n-1) g_{\mu\nu}, \\
  \mathbf{R} = s\, n(n-1).
\end{gather}
\end{subequations}
It is easy to see that, if the Riemann tensor is given by expression \eqref{eq: maximally symmetric curvature Riemann tensor}, then the sectional curvature \eqref{eq:definition sectional curvature} equals the constant $s$. On the other hand, as shown in \cite[Eq.(34.199)]{notasBarata}, the sectional curvature completely determines the Riemann tensor. As a consequence, if $s_p(\partial_{u},\partial_{v}) = s$ for all $p\in\mathcal{M}$, then the Riemann tensor, the Ricci tensor, and the Ricci scalar take the form of Equations \eqref{eq: maximally symmetric curvature tensors}. In summary, if a spacetime is maximally symmetric, then it has constant sectional curvature and expressions \eqref{eq: maximally symmetric curvature tensors} hold. Note, however, that the converse does not hold, i.e. a constant sectional curvature spacetime is not necessarily maximally symmetric.
\begin{example}[Maximally symmetric spacetime]
  \label{eg: maximally symmetric Mink} In Example \ref{eg: 10 Killing in Mink}, it was shown that the four-dimensional Minkowski spacetime possesses ten Killing vector fields with components that are linear functions of the Cartesian coordinates. Since a four-dimensional spacetime admits at most $\frac{4(4+1)}{2}=10$ Killing vector fields, we have not ``lost'' any continuous symmetry by assuming linearity and we conclude that the four-dimensional Minkowski spacetime is maximally symmetric. In addition, since the curvature tensors of \eqref{eq: maximally symmetric curvature tensors} vanish, its constant sectional curvature vanishes.
\end{example}
\begin{counterexample}[Non-maximally symmetric spacetime]
\label{ceg: Schwarzschild constant sectional curvature}
Consider Schwarzschild black hole: the static case ($a=0$) of Counter-Example \ref{ceg: kerr is not static}. Computing its curvature tensors, we find that both the Ricci tensor and the Ricci scalar vanish, while the Riemann tensor does not. That is, Equation \eqref{eq: maximally symmetric curvature tensors} is not satisfied, this spacetime is not of constant sectional curvature and hence, it is not maximally symmetric.
\end{counterexample}

A straightforward consequence of an $n$-dimensional spacetime having constant sectional curvature $s$ is that it solves Einstein field equations in vacuum with a cosmological constant $\Lambda$ determined by $s$ and by $n$. To verify such claim, consider a spacetime that satisfies the Einstein equations:
\begin{equation}
  \label{eq: einstein vacuum with lambda}
  R_{\mu\nu} + \left(\Lambda-\frac{R}{2}\right)g_{\mu\nu} = 0.
\end{equation}
 It is easy to see that Equation \eqref{eq: einstein vacuum with lambda} together with Equations \eqref{eq: maximally symmetric curvature tensors} hold true if and only if the cosmological constant is given by
\begin{equation}
  \label{eq: lambda constant sectional curvature}
  \Lambda = \frac{s (n-2)(n-1)}{2}.
\end{equation}
In other words, a spacetime of constant sectional curvature $s$ solves Einstein field equations in vacuum with a cosmological constant (under the sign conventions chosen above) of the same sign of $s$ and given by \eqref{eq: lambda constant sectional curvature}.

 \begin{remark}
   \label{rem: Killing-Hopf}
   The definitions of Killing vector field, maximally symmetric spacetime and of sectional curvature also apply if we consider the metric tensor to be Riemannian instead of Lorentzian. Within Riemannian geometry, the classification of maximally symmetric manifolds is given by the Killing-Hopf theorem \cite{hopf1926clifford,killing1891ueber}. It states that the universal covering $\widetilde{\Sigma}$ of a complete Riemannian manifold of constant sectional curvature (a space form) $\Sigma$ is isometric either to a sphere, an Euclidean space, or a hyperbolic space. In turn, if $\Gamma$ is a discrete group of isometries of $\widetilde{\Sigma}$ acting properly discontinuously, then $\Sigma$ is isometric to the quotient $\widetilde{\Sigma}/\Gamma$. This result is useful even within Lorentzian geometry since there are (Lorentzian) spacetimes with maximally symmetric Riemannian subspaces, as exemplified in Section \ref{sec: Examples chapter 1}. With this in mind, I include Example \ref{eg: Riemannian case}. In addition, it is worth mentioning that Lorentzian geometry admits an analogous classification: maximally symmetric spacetimes are isometric to either Minkowski, de Sitter, anti-de Sitter spacetimes or quotients/covers of them \cite[Ch.11]{jos1967spaces}. More details and references regarding the spacetimes mentioned in the examples above are given in the last section of this chapter, where I consider each of them separately.
 \end{remark}

\begin{example}[Riemannian case] \label{eg: Riemannian case} Consider the $2$-sphere endowed with a Riemannian metric such that
  $$ds^2=d\theta^2 + \sin^2\theta d\varphi^2.$$
In this case we can easily solve the Killing equation \eqref{eq: definition Killing vector field lie derivative = 0} for a general vector of the form $\xi_\mu = (\xi_\theta(\theta,\varphi),\xi_\varphi(\theta,\varphi))$, as shown in the notebook \cite{git_myPhD}. It follows that the $2$-sphere is maximally symmetric since there are $3$ Killing vector fields:
\vspace{.25cm}

\hspace{5cm}$ (\sin \varphi ,+ \sin \theta \cos \theta  \cos\varphi ) , $

\hspace{5cm}$ (\cos\varphi ,-\sin \theta  \cos \theta  \sin \varphi ),  $

\hspace{5cm}$ (0,\sin ^2\theta ). $
\end{example}

\section{On the geometric Hawking temperature}
\label{sec: On the geometric Hawking temperature}

General relativity goes without the notion of temperature. Nevertheless, in the presence of a non-degenerate bifurcate Killing horizon we can always define a ``Hawking temperature''. The interpretation of the Hawking temperature hinges on the scenario considered, but its character as a \textit{temperature} cannot be merely a geometrical, classical, consequence. On one hand, the no-hair theorem entails that black hole solutions of Einstein–Maxwell equations are completely determined by their mass, their charge and their spin \cite[Pg.876]{thorne2000gravitation}. On the other hand, thermal effects on black holes have been well-established within quantum field theory on curved spacetimes. Expressly, Hawking radiation corroborates black hole thermodynamics and the Hawking temperature, though geometrically defined, admits its full meaning at the interface between general relativity and quantum field theory \cite{hawking1975particle,hartle1976path,fredenhagen1990derivation,moretti2012state}. Throughout Chapter \ref{chap: Applications} of this thesis, I invoke such geometrical definition and its relation with thermal states on curved spacetimes, which are detailed in Sections \ref{sec: KMS states} and \ref{sec: Unruh, Hawking, anti-Unruh, anti-Hawking effects} of Chapter \ref{chap: Quantum Field Theory on Static Spacetimes}. With this in mind, in this section I define a bifurcate Killing horizon and I show how it yields a Hawking temperature.

Let $\xi$ be a Killing vector field with Killing parameter $t$ and generating a one-parameter group of isometries $G=\{U_t\}_{t\in\mathbb{R}}$ of a static spacetime. A \textit{Killing horizon generated by $\xi$} is a $G$-invariant null hypersurface where $\xi^2=0$. The union of at least two intersecting Killing horizons generated by $\xi$ is a \textit{bifurcate Killing horizon generated by $\xi$}. The intersection, called a \textit{bifurcation surface}, is a spacelike $(n-2)$-dimensional hypersurface where $\xi=0$, i.e. points in the intersection are fixed under the action of $G$. Not all Killing fields generate horizons, and not all Killing horizons can be deformed into a bifurcate Killing horizon. For details on these Killing properties, I refer to \cite{boyer1969geodesic, wald2010general, frolov2012black}.
\begin{example}[Bifurcate Killing horizon] \label{eg: bifurcate Killing horizon}
  Consider the Killing vector fields of the four-dimensional Minkowski spacetime computed in Example \ref{eg: 10 Killing in Mink}. The Killing vector field $\partial_t$ is timelike everywhere, but the Killing vector field of a Lorentz boost, say $\xi=x\partial_t-t\partial_x$, satisfies $\xi^2=t^2 - x^2$. Hence, $\xi$ is timelike for $x^2>t^2$, null at $x=\pm t$, and spacelike otherwise. The surface where $\xi^2=0$ is in fact a bifurcate Killing horizon given by the union of the horizons at $x=t$ and $x=-t$, as illustrated in Figure \ref{fig: bifurcate Killing in Minkowski}.
     \begin{figure}[H]
       \centering
        \includegraphics[width=.5\textwidth]{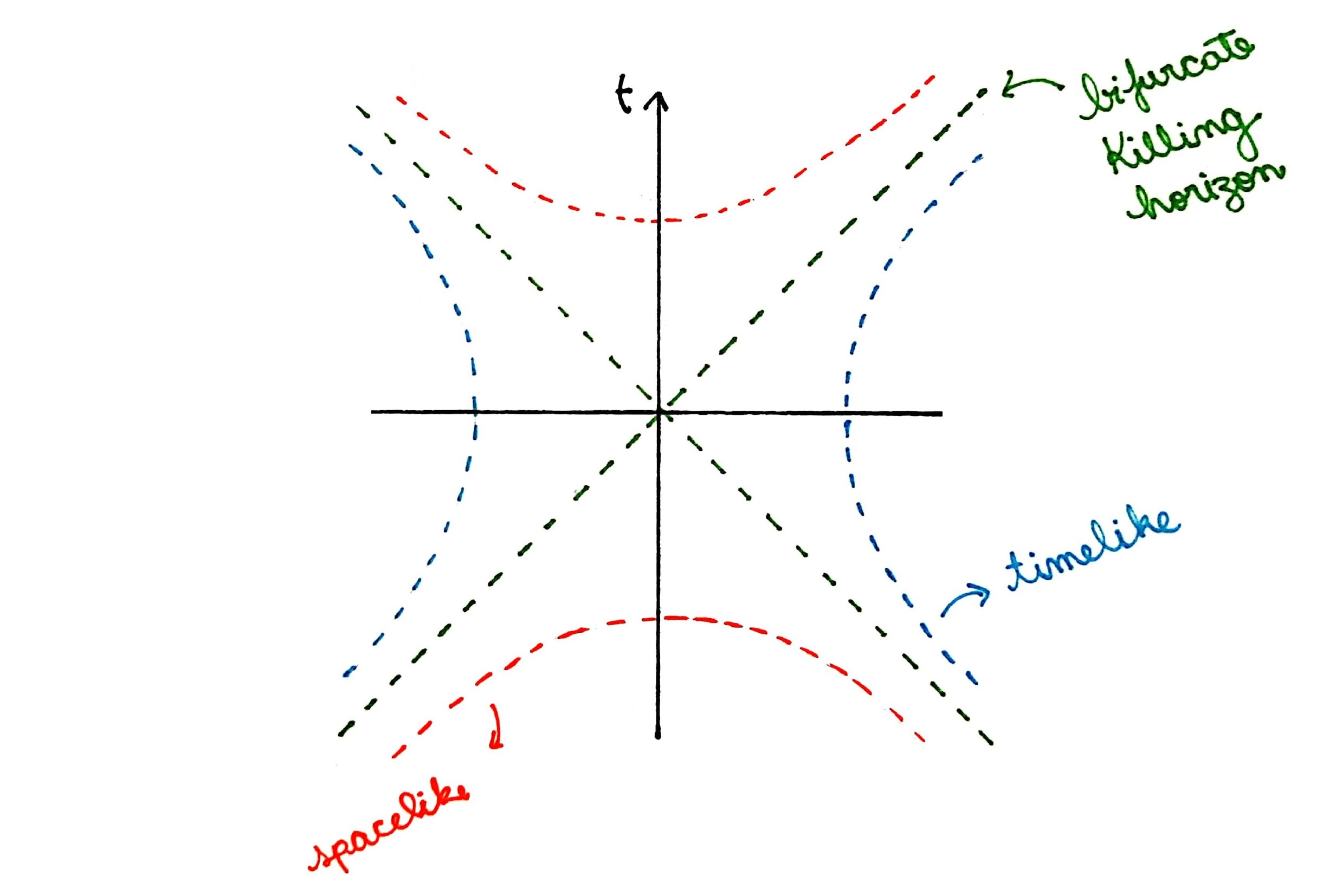}\hspace{.5cm}%
        \caption{The bifurcate Killing horizon generated by $\xi=x\partial_t-t\partial_x$ and its orbits on Minkowski.}
       \label{fig: bifurcate Killing in Minkowski}
         \end{figure}
         \vspace{-.5cm}
\end{example}
On a Killing horizon generated by $\xi$, consider a generator that is a null geodesic with affine parametrization $\lambda$ and tangent vector $\chi$. Since $\xi$ is null, we can write $\xi=\alpha\chi$ with $\alpha=\frac{\partial\lambda}{\partial t}$ such that $\lambda = e^{\kappa t}$ with $\kappa =\frac{\partial \ln \alpha}{\partial t}$ and
\begin{equation}
  \label{eq:surface gravity measure diff}
  \xi^\mu \frac{\partial}{\partial x^\mu} =\frac{\partial}{\partial t}= \kappa \lambda\frac{\partial}{\partial \lambda}.
\end{equation}
The quantity $\kappa$ is called \textit{surface gravity}. In general, it is constant along each null orbit, but not constant on the horizon. If $\kappa$ does not vanish on any orbit, we say the horizon is \textit{non-degenerate}. In addition, for practical computations, a useful form of expressing $\kappa$ is
\begin{equation}
  \label{eq: surface gravity def}
  \kappa = \sqrt{-\frac{1}{2} \nabla^\mu \xi^\nu\nabla_\mu \xi_\nu}.
\end{equation}
Comments on the interpretation of the surface gravity are in due order. By Equation \eqref{eq:surface gravity measure diff}, we can see the surface gravity as a measure of the difference between a Killing and an affine parametrization of the orbits of $\xi$ at the horizon. Note that if $\kappa\neq 0$, then it depends on the normalization of $\xi$. On a static asymptotically flat spacetime, $\kappa$ can also be written as the acceleration that an observer at asymptotic infinity would need to apply on an object at the horizon to keep it static. In this case, there is a preferred normalization: the one in which the Killing vector has unit norm at asymptotic infinity. However, this interpretation does not generalize to non-static or non-asymptotically flat scenarios and, in general, there is no preferred normalization. Notwithstanding, on a non-degenerate bifurcate Killing horizon, one can prove that the surface gravity is constant, not only along the orbits, but on the whole horizon. In analogy with thermodynamics, we associate surface gravity with the notion of temperature and we refer to this statement as ``the zero-th law''. In the following, I clarify such an association. For a discussion on the comments above, see e.g. \cite[Sec.6.3.3]{frolov2012black}.
\begin{definition}[Global Hawking temperature] \label{def: global hawking temperature}
  On a spacetime that has a non-degenerate bifurcate Killing horizon with surface gravity $\kappa_h$, the \textit{global Hawking temperature} is
    \begin{equation}
      \label{eq: def global hawking temperature}
    T_{gH} = \frac{\kappa_h}{2\pi}.
    \end{equation}
\end{definition}
In the seminal work of Hawking in the 70's \cite{hawking1974black}, he showed that black holes emit thermal radiation at $T_{gH}$. More precisely, the global Hawking temperature is the temperature of the radiation detected by an observer far away from a black hole, i.e. at future null infinity. However, a photon emitted by a black hole gets redshifted when traveling away from the horizon. To grasp the effect of the gravitational redshift, let us consider, as discussed in \cite[Sec.6.3]{wald2010general}, a photon propagating from spatial position $p_2$ to $p_1$ on a stationary black hole spacetime with timelike Killing vector field $\xi=\partial_t$ and such that $|g_{tt}||_{\mini{p_1}}=1$, as illustrated in Figure \ref{fig: redshift on stationary local T}.
\begin{figure}[H]
  \centering
  \includegraphics[width=.5\textwidth]{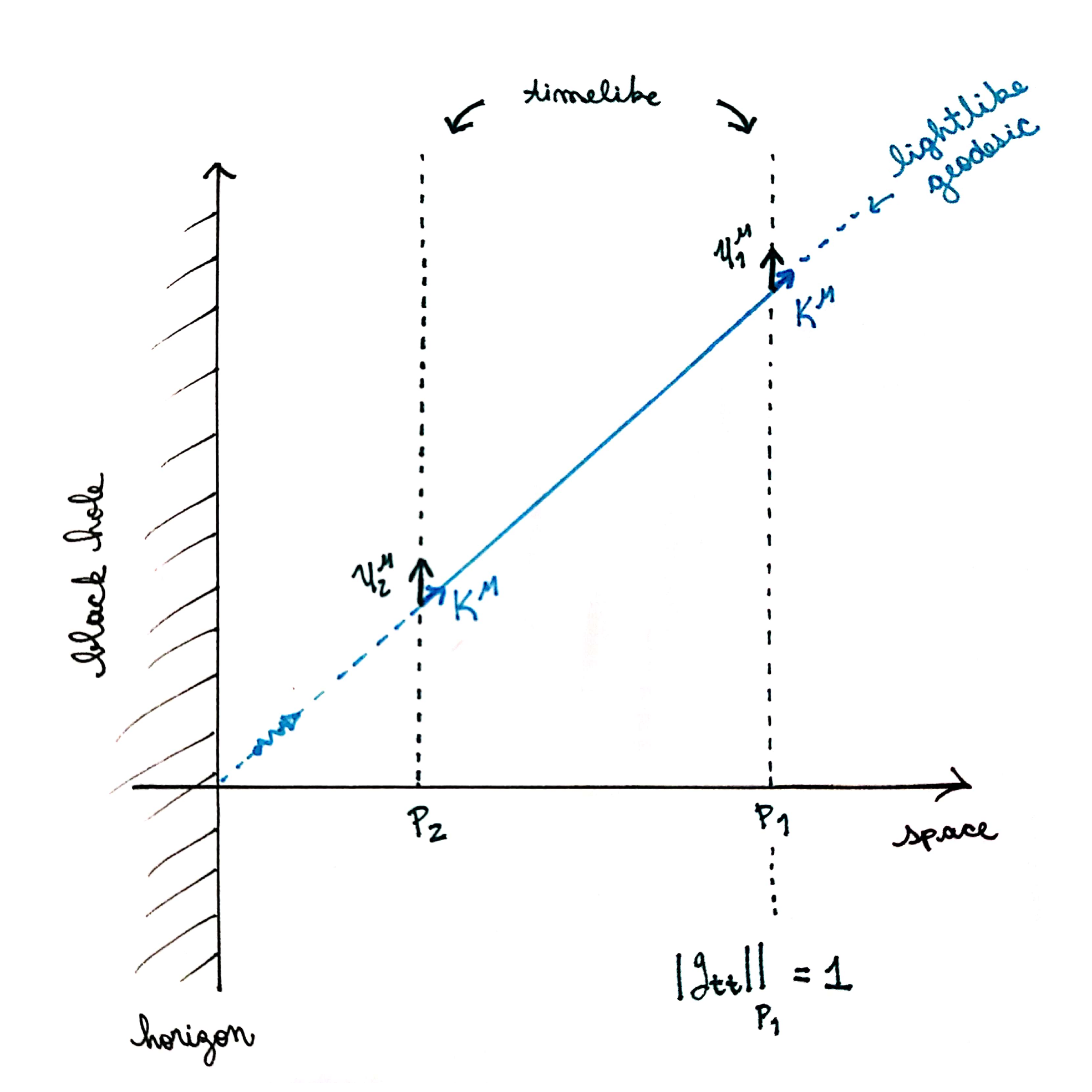}
   \caption{A photon outgoing from the horizon of a stationary black hole.}
  \label{fig: redshift on stationary local T}
  \end{figure}
\noindent Suppose there are two observers following timelike curves at spatial positions $p_1$ and $p_2$ with unit tangent vectors $u_1^{\mu}=\frac{\xi^{\mu}}{\sqrt{|g_{tt}|}}\Big|_{\mini{p_1}}$ and $u_2^{\mu}=\frac{\xi^{\mu}}{\sqrt{|g_{tt}|}}\Big|_{\mini{p_2}}$, respectively.
If $k^{\mu}$ is the tangent vector along the null geodesic followed by the photon, then the frequency measured by each observer is $\omega_j = -k_{\mu}u_j^{\mu}|_{p_j}$, for $j\in\{1,2\}$ \cite[Sec.4.2]{wald2010general}. Since the quantity $k_{\mu}\xi^{\mu}$ is conserved along the geodesic followed by the photon, as mentioned in the end of Section \ref{sec: The definition of static spacetimes}, it follows that
\begin{equation*}
    \omega_2 = - (k_{\mu}u_2^{\mu})|_{p_2} = - \frac{1}{\sqrt{|g_{tt}|}|_{p_2}}(k_{\mu}\xi^{\mu})|_{p_2} =
    - \frac{1}{\sqrt{|g_{tt}|}|_{p_2}}(k_{\mu}\xi^{\mu})|_{p_1} =   \frac{\sqrt{|g_{tt}|}|_{p_1}}{\sqrt{|g_{tt}|}|_{p_2}} \omega_1.
\end{equation*}
Imagine the observer at $p_1$ measures exact black body radiation at temperature $T_{1}$ and frequency $\omega_1$. The thermal spectrum, characterized by a Planckian distribution, depends on the temperature and the frequency by a factor $\frac{\omega_1}{T_{1}}$, which corresponds to
\begin{align}
  \label{eq: frequency redshifted}
  &\frac{\omega_1}{T_{1}} =\frac{\sqrt{|g_{tt}|}|_{\mini{p_2}}}{\sqrt{|g_{tt}|}|_{\mini{p_1}}}\frac{\omega_2}{ T_{1}}= \frac{      \sqrt{|g_{tt}|} |_{\mini{p_2} }  }{ T_{1}} \omega_2 = \frac{\omega_2}{T_{2}}, \quad \text{ for } \quad T_2 \equiv   \frac{T_{1}}{\sqrt{|g_{tt}|}|_{\mini{p_2}}}.
\end{align}
That is, the observer at $p_2$ measures a temperature $T_2$ that equals $T_1$ corrected by a redshift factor. The connection between the frequency of the photon, the redshift effect and the temperature of the radiation exemplified above calls for a position-dependent temperature. This can be taken into account consistently on stationary spacetimes, where Equation \eqref{eq: frequency redshifted} makes sense, and justifies defining a \textit{local} temperature, as follows.
\begin{definition}[(Local) Hawking temperature] \label{def: (Local) Hawking temperature} On a stationary spacetime with metric tensor $g$, time Killing parameter $t$ and associated bifurcate Killing horizon with surface gravity $\kappa_h$, the \textit{(local) Hawking temperature} is defined as the global Hawking temperature corrected by a redshift factor:
              \begin{equation}
                \label{eq: def local hawking temperature}
                T_H\doteq\frac{\kappa_h}{2\pi \sqrt{|g_{tt}|}}.
              \end{equation}
The local Hawking temperature is also known as \textit{Tolman temperature} \cite{tolman1930weight}.
\end{definition}
 Since the surface gravity depends on the normalization of the timelike Killing vector field, so does the Hawking temperature. On asymptotically flat spacetimes, the preferred choice of normalization of imposing that $\kappa_h\geq0$ and $\xi^2 \rightarrow -1$ at spatial infinity yields a finite Hawking temperature at asymptotic infinity. On an asymptotically AdS/Lifshitz spacetime, however, we have that the Hawking temperature vanishes at asymptotic infinity \cite{brown1994temperature}. In this case, the global Hawking temperature is measured, not at infinity, but where the redshift factor is $1$. The following example illustrates the difference between these two scenarios.
         \begin{example}[Asymptotically flat versus AdS] \label{eg: Schd and Schd-AdS temperatures}
         Consider a spacetime with line element
\begin{equation}\label{eq: woiejfjioiej383rjfjgbkd}
              ds^2=-\left(1 - \frac{2M}{r} + kr^2\right)dt^2+\left(1 - \frac{2M}{r}+ kr^2\right)^{-1}dr^2+r^2d\theta^2+r^2\sin^2\theta d\varphi^2,
\end{equation}
           for $M>0$, $r>r_h>0$, $\theta\in[0,\pi)$ and $\varphi\in[0,2\pi)$. The timelike Killing vector field $\xi=\partial_t$ yields a bifurcate Killing horizon at $r=r_h$. For $k=0$, this is Schwarzschild spacetime and we have $r_h=2M$, while if $k>0$, it coincides with Schwarzschild-AdS spacetime and $r_h$ is the unique real root of $f(r)$. Using Equations \eqref{eq: surface gravity def} and \eqref{eq: def global hawking temperature} (see Equation \eqref{eq: surface gravity schd-like}),

          \enlargethispage{2\baselineskip}
           $$\kappa_h = \left(\frac{M}{r^2}+ k r \right)\Big|_{r=r_h} \quad \text{and}\quad T_H = \frac{T_{gH}}{\sqrt{1 - \frac{2M}{r}+kr^2}}.$$\vspace{-.5cm}
           \begin{figure}[H]
             \centering
             \includegraphics[width=.38\textwidth]{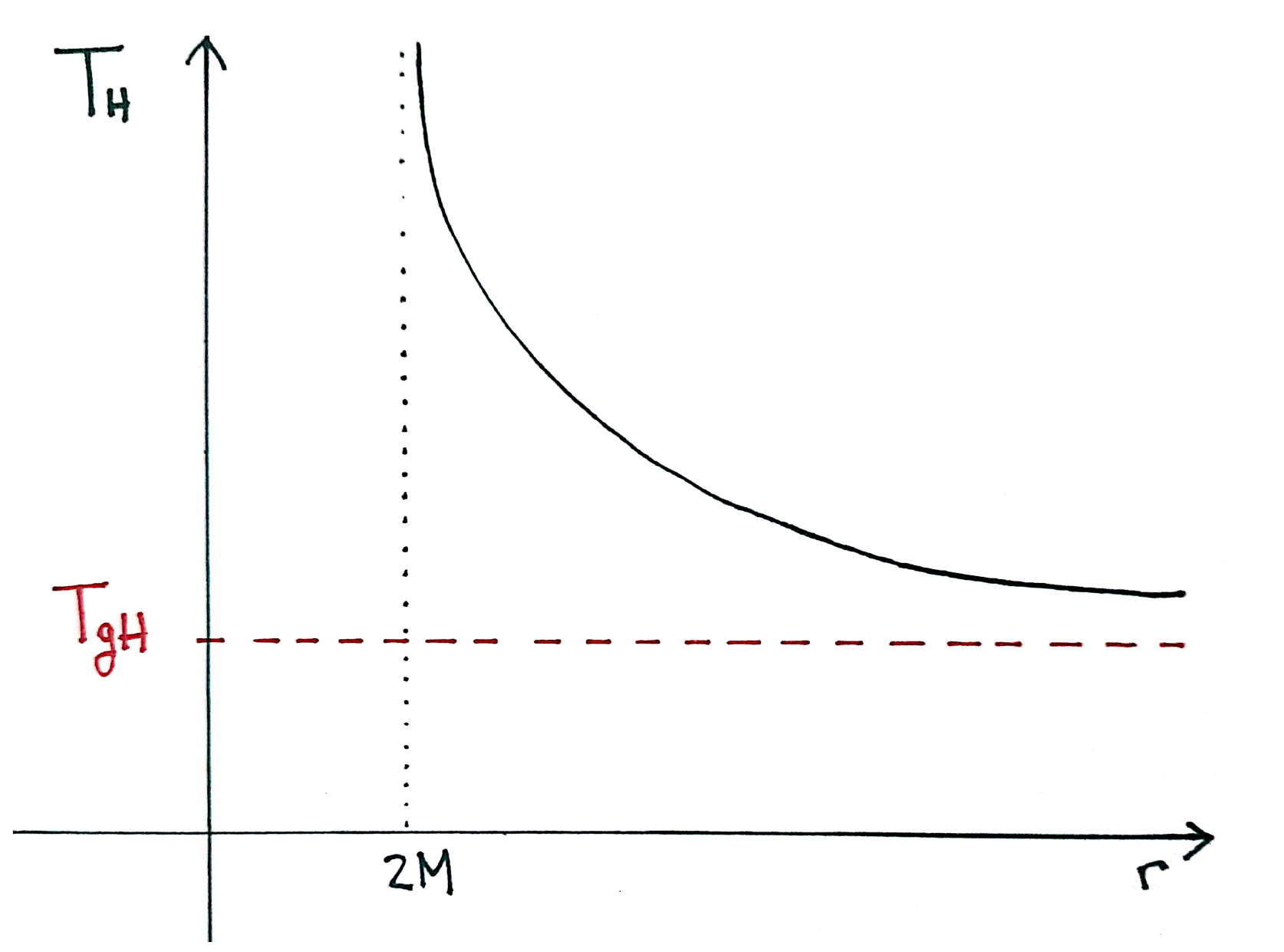}\hspace{.5cm}%
              \includegraphics[width=.38\textwidth]{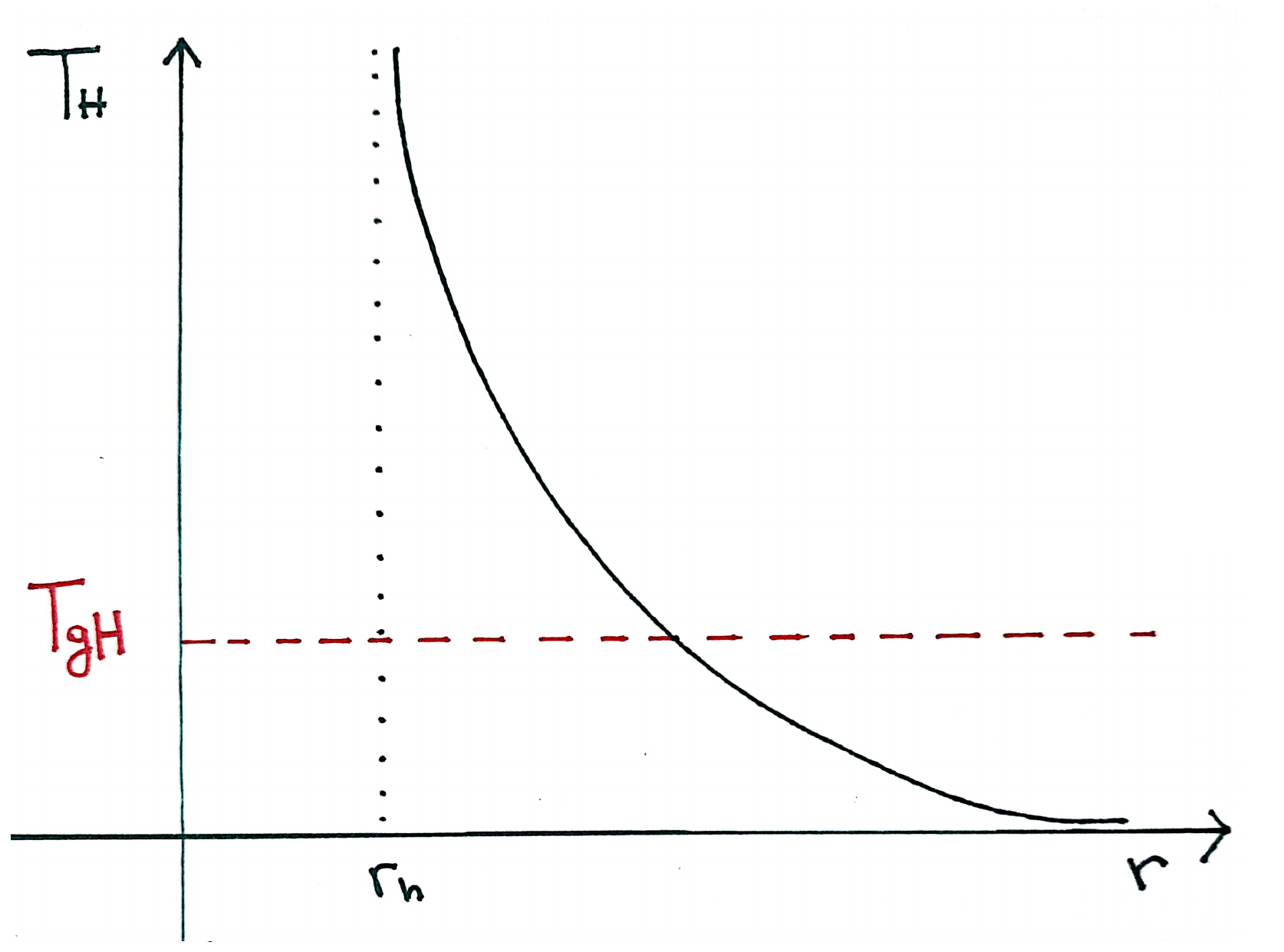}
              \caption{ On the left, on Schwarzschild spacetime; on the right, on Schwarzschild-AdS spacetime.}
               \label{fig: local temperature as function of r asymp flat and ads}
           \end{figure}\vspace{-.7cm}
            \end{example}

\section{A consequence of (non-)global hyperbolicity}
   \label{sec: A consequence of (non-)global hyperbolicity}

   Many problems in physics rely on solving differential equations. Suppose we want to solve a partial differential equation on a static spacetime $\mathcal{M}=\mathbb{R}\times\Sigma$ compatible with initial conditions on a time-slice $\{t_0\}\times\Sigma$. Under certain conditions we can guarantee a solution exists and is unique. In this section, I approach this problem by invoking the concept of global hyperbolicity and a well-known result regarding normally hyperbolic operators. This is particularly relevant for the discussion of the next chapter, where I consider the Klein-Gordon equation on static spacetimes.
   Let us start by defining a globally hyperbolic spacetime. To that end, the following definitions are necessary. On a time-oriented spacetime $\mathcal{M}$, the future Cauchy development $D^+(\mathcal{S})$ of a subset $\mathcal{S}\subset \mathcal{M}$ is the set of points $p\in\mathcal{M}$ for which every past-directed inextendible causal curve through $p$ intersects $\mathcal{S}$. The past Cauchy development $D^-(\mathcal{S})$ is defined analogously, and the \textit{Cauchy development} is given by the union $D(\mathcal{S}) := D^+(\mathcal{S})\cup D^-(\mathcal{S})$. Suppose $\mathcal{S}\subset \mathcal{M}$ is a closed subset for which no pair of points $p,q\in\mathcal{S}$ can be joined by a timelike curve (i.e. $\mathcal{S}$ is achronal). If $D(\mathcal{S})=\mathcal{M}$, then $\mathcal{S}$ is called a \textit{Cauchy surface}.

   \begin{definition}[Globally Hyperbolic Spacetime] \label{def: globally hyperbolic spacetime} A spacetime is \textit{globally hyperbolic} if it contains a Cauchy surface.
   \end{definition}
        An important feature of a globally hyperbolic spacetime is that it admits a foliation given by a one-parameter family of smooth Cauchy surfaces $\mathbb{R}\times\Sigma$ \cite[Thm.4.1.1]{WaldQFTCS}. In turn, there exists a coordinate system $(t,\underbar{x})$ such that its line element reads \cite[Thm.1.1]{bernal2005smoothness}
        \begin{equation*}
        \label{eq: metric globally hyperbolic spacetimes}
          ds^2=-f(t,\underbar{x})dt^2+h_{ij}(t,\underbar{x})d\underbar{x}^id\underbar{x}^j.
        \end{equation*}
        Note that there are two essential differences between this foliation and the one that static spacetimes admit, as per Equation \eqref{eq: metric static spacetimes} and Figure \ref{fig: foliation static spacetime}. First, the spacelike hypersurfaces of the foliation of a static spacetime are not necessarily Cauchy surfaces. Second, there might not be a timelike, irrotational Killing vector field backing the foliation of a globally hyperbolic spacetime. In brief, there are static spacetimes that are not globally hyperbolic, and there are globally hyperbolic spacetimes that are not static.

   \begin{example}[Static and globally hyperbolic] \label{eg: static and globally hyperbolic}
     Equal-time slices of the four-dimensional Minkowski spacetime, as in Example \ref{eg: 10 Killing in Mink}, are Cauchy surfaces. The Cauchy development of the surface $\{t=0\}$ is highlighted in Figure \ref{fig:domainofdependence}. In other words, Minkowski is a globally hyperbolic spacetime.
   \end{example}
   \begin{example}[Non-static and globally hyperbolic] \label{eg: non-static and globally hyperbolic} Consider a Kerr black hole, as in Example \ref{ceg: kerr is not static}. If $r_+$ is the largest real root of $\triangle(r)$, the exterior region $r>r_+$ is a non-static, stationary, globally hyperbolic spacetime.
   \end{example}
   \begin{counterexample}[Static and not globally hyperbolic]
     \label{eg: mink minus a point}
     Minkowski spacetime in spherical coordinates (Minkowski without the spatial origin $r=0$), whose line element in four dimensions is given by Equation \eqref{eq: woiejfjioiej383rjfjgbkd} with $k=M=0$, is not globally hyperbolic.
   \end{counterexample}
   \begin{counterexample}[Static and not globally hyperbolic]
     \label{ceg: timelike conformal boundary non-globally hyperbolic}
    If $\mathcal{M}$ is a static spacetime with a timelike conformal boundary, then a time slice $\{t_0\}\times\Sigma$ is not a Cauchy surface, as illustrated in Figure \ref{fig:domainofdependence}. In addition, it does not admit any Cauchy surface \cite[Pg.133]{hawking1973large}. Hence, a spacetime that is asymptotically AdS/Lifshitz is not globally hyperbolic.
     \end{counterexample}
          \begin{figure}[H]
            \centering
               \includegraphics[width=0.55\textwidth]{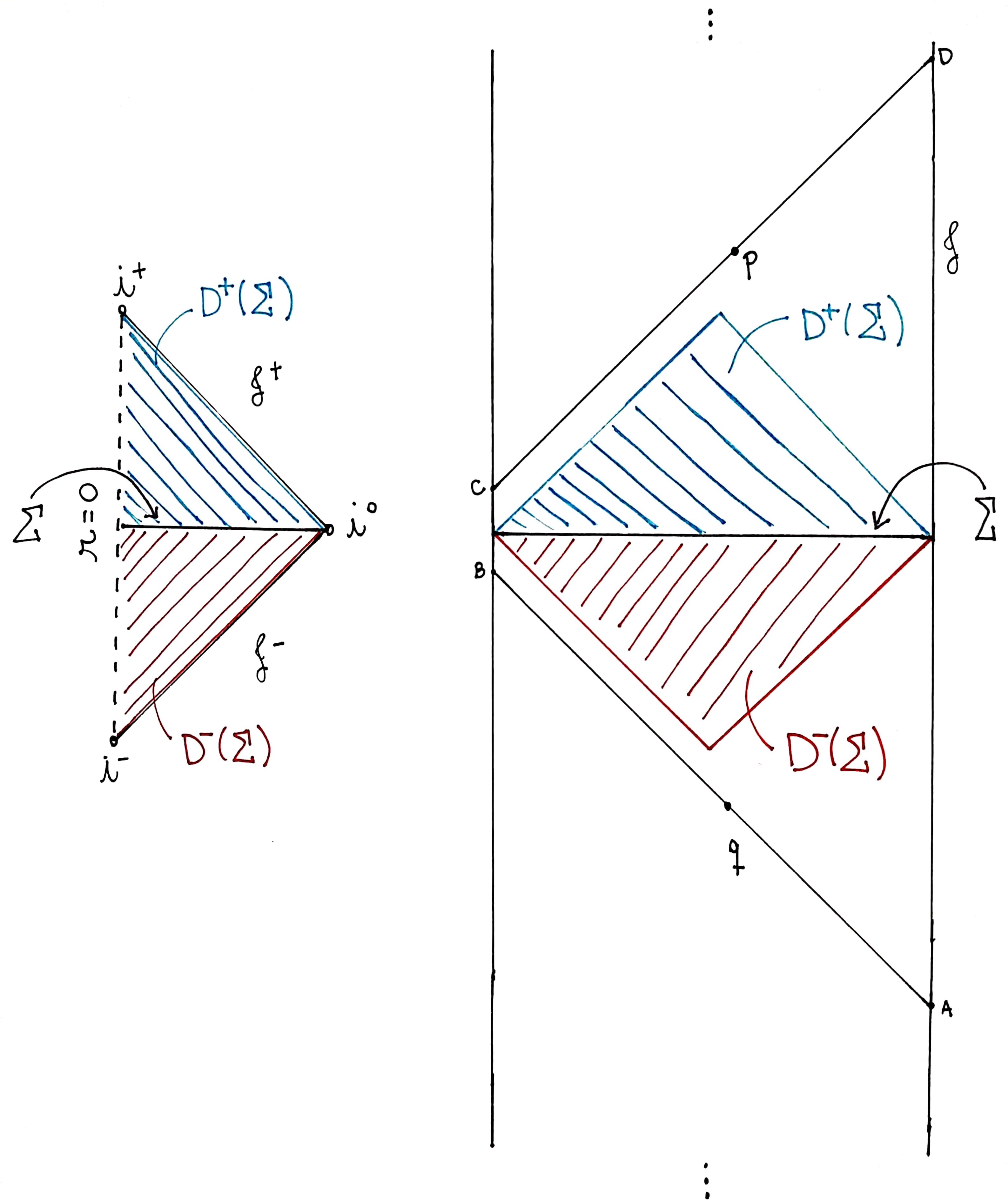}
               \caption{Domain of dependence of a constant time slice $\Sigma$. On the left, $D(\Sigma)$ is the whole Minkowski spacetime. On the right, null geodesics in AdS spacetime from $A$ to $B$ and from $C$ to $D$ do not intersect $\Sigma$ and the points $p$ and $q$ do not lie in $D(\Sigma)$.}
               \label{fig:domainofdependence}
         \end{figure}
Key results concerning equivalent characterizations of global hyperbolicity can be found in \cite{bernal2003smooth,bernal2007globally}, detailed discussions can be found in \cite{WELM,WaldQFTCS}, and explicit considerations relating the properties of being static and of being globally hyperbolic, in \cite[Ch.6]{fulling1989aspects}.

     The notion of global hyperbolicity is connected with the well-posedness of Cauchy problems. A \textit{Cauchy problem} consists of a partial differential equation on a smooth manifold together with Cauchy data given on a smooth submanifold. If a solution exists, is unique and has a continuous dependence on the Cauchy data, then we say the Cauchy problem is \textit{well-posed}. The problem set at the beginning of this section is in fact a concern on the well-posedness of a Cauchy problem, where the submanifold are time-slices and the Cauchy data constitute the set of initial conditions.

     The result connecting global hyperbolicity and Cauchy problems, relevant in the context of free quantum field theory on curved spacetimes, is the following. Let $P$ be a normally hyperbolic scalar operator on a spacetime $\mathcal{M}$ that admits a foliation $\mathbb{R}\times\Sigma$. Let $f\in C_0^\infty (\mathcal{M})$, and let $(t_0,\Psi_0,\dot\Psi_0)$ be the Cauchy data, i.e. $\Psi_0$, $\dot\Psi_0\in C_0^\infty( \{t_0\}\times\Sigma)$. The corresponding Cauchy problem,
     \begin{gather}
       \label{eq:cauchy problem example}
        P \Psi = f,  \quad  \Psi|_{\{t_0\}\times\Sigma} = \Psi_0   \quad \text{ and }\quad  \nabla_\xi\Psi|_{\{t_0\}\times\Sigma} = \dot\Psi_0,
     \end{gather}
     is well-posed if $\mathcal{M}$ is a globally hyperbolic spacetime \cite{WELM,WaldQFTCS,hawking1973large}.

     In the next chapter, I consider the Klein-Gordon equation on static spacetimes that may or may not be globally hyperbolic.  On the globally hyperbolic ones, well-posedness is guaranteed by the above result. On non-global hyperbolic spacetimes, on the other hand, well-posedness is not generally guaranteed. However, on static, non-globally hyperbolic spacetimes that are stably-causal---i.e. possessing a global time function, and consequently without closed timelike curves---
     we can still guarantee well-posedness of Equation \eqref{eq:cauchy problem example} by providing a boundary condition. In other words, the associated \textit{Cauchy-boundary value problem} is well-posed; see Section \ref{sec: Physically-sensible dynamics on static spacetimes}, where a summary of \cite{ishibashi2003dynamics} is given, or \cite[Thm.30]{dappiaggi2019fundamental}. 
\section{Examples}
\label{sec: Examples chapter 1}

As mentioned in the introduction of this chapter, all spacetimes considered in this thesis are static and have sections of constant sectional curvature. Beyond that, their shared degree of symmetry allows us to characterize each one of them using  Schwarzschild-like coordinates. In this section, I define such coordinates, I illustrate some of their properties and I gather examples of some relevant static spacetimes invoking the concepts discussed in this chapter.

\begin{definition}[Schwarzschild-like coordinates]
  \label{def: Schwarzschild-like coordinates}
  \pdfbookmark[subsection]{Schwarzschild-like coordinates}{Schwarzschild-like coordinates}
  Let $\mathcal{M}$ be a static spacetime of the form $\mathcal{M} \cong \mathbb{R}\times \text{I} \times \Sigma_{j}^{n-2}$, where $\text{I}\subset \mathbb{R}$ and $\Sigma_{j}^{n-2}$ are Cauchy-complete, connected, $(n-2)$-dimensional Riemannian hypersurfaces of constant sectional curvature $j$. Schwarzschild-like coordinates $(t,r,\theta,\varphi_1,...,\varphi_{n-3})$ is a coordinate system with respect to which the line element on $\mathcal{M}$ reads
  \begin{subequations}
    \label{eq:metric schwarzschild-like}
	\begin{equation}
		ds^2=-f(r)dt^2+h(r)dr^2+r^2d\Sigma_{j}^{n-2}.
	\end{equation}
  The Killing parameter $t\in\mathbb{R}$ is called the \textit{time coordinate}, $r\in\text{I}\subset\mathbb{R}$ is called the \textit{radial coordinate} and $\underline{\theta}=(\theta,\varphi_1,...,\varphi_{n-3})\in\Sigma_{j}^{n-2}$ are called \textit{angular coordinates}. If $n=4$, we shall denote $\varphi_1\equiv\varphi$. The functions $f(r)$ and $h(r)$ are continuous, strictly positive functions for $r\in\text{I}$. The line element $d\Sigma_{j}^{n-2}$ depends on $j$:
  \begin{equation}
    d\Sigma_{j}^{n-2} \doteq
    \begin{cases}
      d\Xi_{n-2}^2, \text{ the unit metric on the hyperbolic space,} &\text{ for }j<0;\\
      d\Pi_{n-2}^2, \text{ the flat metric,}  &\text{ for }j=0; \\
      d\Omega_{n-2}^2, \text{ the unit metric on the sphere,}  &\text{ for }j>0.
    \end{cases}
  \end{equation}
\end{subequations}
\end{definition}
In the following, let us discuss properties of a spacetime $\mathcal{M}$ that admits, globally, Schwarz\-schild-like coordinates having in mind the concepts introduced in this chapter. First, $\mathcal{M}$ is a static spacetime, as characterized in Section \ref{sec: The definition of static spacetimes}. Second, the sections $\Sigma_{j}^{n-2}$ of constant time and radius are maximally symmetric, as described in Section \ref{sec: A relation between symmetries and curvature}, yet $\mathcal{M}$ itself might not be. In particular, the two properties just mentioned imply that $\mathcal{M}$ has at least $1+\frac{(n-2)(n-1)}{2}$ Killing vector fields.
Also, regarding the notion of global hyperbolicity, as seen in Section \ref{sec: A consequence of (non-)global hyperbolicity} a static spacetime may or may not satisfy it.

Definition \ref{def: Schwarzschild-like coordinates} constitutes a generalization of the Schwarzschild solution, as in Example \ref{eg: Schd and Schd-AdS temperatures}, with two distinctions. First, $f(r)$ is not necessarily equal to $h(r)^{-1}$ and both are arbitrary (continuous, strictly positive) functions. Second, constant time and radius slices are not necessarily spheres. The spacetimes whose associated line element is of the form of Equation \eqref{eq:metric schwarzschild-like} consist instead of ``nested space forms'' $\Sigma_{j}^{n-2}$. Note that, Remark \ref{rem: Killing-Hopf} implies that there are different topological choices for $\Sigma_{j}^{n-2}$. This fact has an interesting consequence in Black Hole Physics as it implies that we can choose different topologies for black hole horizons, see e.g. \cite[Pg.5]{brill1997thermodynamics}
or Section \ref{ex: chapter 1 toplogical black holes}.

  As shown in Section \ref{sec: On the geometric Hawking temperature}, when a Killing vector field generates a non-degenerate bifurcate Killing horizon, the surface gravity is constant and the Hawking temperature is well-defined. Suppose that, in Schwarzschild-like coordinates, $\partial_t$ generates such a horizon at $r=r_h$. The timelike Killing vector is $\xi^{\mu} = (1, 0, 0, 0)$, thus $\xi_{\mu}=(-f(r), 0, 0, 0)$, and the only non-vanishing components of $\nabla^{t}\xi^{r}$ and $\nabla_{t}\xi_{r}$ are
  $$\nabla^{t}\xi^{r} = -\frac{f'(r)}{2 f(r) h(r)} = - \nabla^{r}\xi^{t} \quad \text{ and }\quad
    \nabla_{t}\xi_{r} = + \frac{f'(r)}{2} = - \nabla_{r}\xi_{t}.
  $$
  In view of Equations \eqref{eq: surface gravity def} and \eqref{eq: def local hawking temperature}, it follows that
   \begin{equation}
     \label{eq: surface gravity schd-like}
     \kappa_h=\frac{|f'(r)|}{2 \sqrt{f(r) h(r)}}\Bigg|_{r=r_h} \quad \text{ and }\quad  T_H=\frac{\kappa_h}{2\pi\sqrt{ f(r)}}.
   \end{equation}
   Moreover, in the case $h(r)=f(r)^{-1}$, the surface gravity simplifies to $\kappa_h=\frac{|f'(r)|}{2}\Big|_{r=r_h}$.

In the following sections I highlight the main features of some static spacetimes that admit Schwarz\-schild-like coordinates \eqref{eq:metric schwarzschild-like}. Most examples considered here solve Einstein equations in vacuum (but not all!). Therefore, before getting to them, let us check what restrictions Equation \eqref{eq: einstein vacuum with lambda}, impose on the functions $f(r)$ and $h(r)$ as in Equation \eqref{eq:metric schwarzschild-like}. For simplicity, let us consider only the cases $n=4$ and $n=3$, in this order. Defining the auxiliary function
\begin{align*}
  \label{eq: function J n=4}
    &J(\theta) := \begin{cases}
              \sinh(\theta),& j=-1,\\
              \theta,       & j=0,\\
              \sin(\theta), & j=+1,
            \end{cases}
\end{align*}
the line element \eqref{eq:metric schwarzschild-like} for $n=4$ can be written as
\begin{equation*}
  ds^2=-f(r)dt^2+h(r)dr^2+r^2d\theta^2 + r^2J(\theta)^2d\varphi^2.
\end{equation*}
Then, Einstein field equations \eqref{eq: einstein vacuum with lambda} reduce to
\begin{align*}
  &  r h'(r)-h(r) + \left(j-\Lambda  r^2\right)h(r)^2=0, \\
  & r f'(r)+f(r)  -\left(j-\Lambda  r^2\right)h(r)f(r)=0.
\end{align*}
These are satisfied if and only if
\begin{subequations}
\label{eq: wpioerfjwopi0pwr98}
\begin{align}
  &f(r) = \left( j  +  \frac{c}{r} - \frac{\Lambda}{3} r^2 \right), \\
  &h(r) = f(r)^{-1},
\end{align}
\end{subequations}
where $c$ is an arbitrary integration constant. Note that $r=0$ is a curvature singularity if and only if $c\neq 0$, since the corresponding Kretschmann scalar is given by
\begin{equation}
  \label{eq: Kretschmann 4D}
  \mathbf{K}_{4D} = \frac{12 c^2}{r^6}+\frac{8 \Lambda ^2}{3}.
\end{equation}
In addition, if we consider $\Lambda=0$ and spherical sections $\Sigma_{+1,2}$, we recover the Schwarzschild solution. As it happens, Birkhoff's theorem states that any spherically symmetric asymp\-to\-tic\-ally-flat solution of Einstein field equations in vacuum with $\Lambda=0$ must be static, i.e. it must be a Schwarzschild solution (with some real-valued mass) \cite{eisenstaedt1982histoire}.

Performing the same computation as described above in the three-dimensional case,
  \begin{equation*}
    ds^2=-f(r)dt^2+h(r)dr^2+r^2d\theta^2,
  \end{equation*}
yields \vspace{-.5cm}
\begin{subequations}
\label{eq: wpioerfjwopi0pwr9ddd833}
\begin{align}
  &f(r) = \left( c  - \Lambda r^2 \right), \\
  &h(r) = f(r)^{-1}.
\end{align}
\end{subequations}
What is interesting to note in this case is that the Kretschmann scalar is always a constant,
\begin{equation}
  \label{eq: Kretschmann 3D}
  \mathbf{K}_{3D} = 12 \Lambda ^2,
\end{equation}
which means that there are no curvature singularities on three-dimensional spacetimes covered by Schwarzschild-like chart coordinates within general relativity with vanishing stress-energy tensor. As a matter of fact, this simple example alludes to peculiar features of three-dimensional gravity. Namely, one can show that the Riemman tensor is completely determined by the local content of ``matter plus $\Lambda$'', and that the Weyl tensor vanishes identically. To read about the significance of these features, see e.g. \cite{collas1977general,giddings1984einstein}.

Next in order, we consider $n$-dimensional spacetimes endowed with Schwarzschild-like coordinates. For the ones that solve Einstein's equations, we have that Equations \eqref{eq: wpioerfjwopi0pwr98} and \eqref{eq: Kretschmann 4D} hold for $n=4$, while Equations \eqref{eq: wpioerfjwopi0pwr9ddd833} and \eqref{eq: Kretschmann 3D} hold for $n=3$. One should keep in mind, however, that Schwarzschild-like coordinates do not necessarily cover the maximal analytic extension of the examples given. Also, taking into account the discussions above, the following sections regarding particular spacetimes consist rather of a compactified set of results. Elaborate discussions on exact solutions of Einstein equations, with proofs and beautiful illustrations, can be found in \cite{hawking1973large, wald2010general,morris2012exact,frolov2012black}. In addition, I have summarized the curvature data, such as the Riemann tensor, the Ricci tensor, the Ricci scalar, the Kretschmann scalar and the Christoffel symbols of these examples in a Mathematica notebook available at \cite{git_myPhD}.

\subsection{Minkowski}
\label{ex: Chapter 1 Minkowski}
The $n$-dimensional Minkowski spacetime is a static spacetime that solves Einstein field equations in vacuum with vanishing cosmological constant. Its line element can be written in Cartesian coordinates as
\begin{equation*}
  \label{eq: metric mink nD Cartesian}
  ds^2=-dt^2+dx_1^2+ ... + dx_{n-1}^2,
\end{equation*}
with $(t,x_1,...,x_{n-1})\in\mathbb{R}^n$. It is maximally symmetric, admitting $\frac{n(n+1)}{2}$ Killing vector fields. Amongst these, $\partial_t$ is the global, non-vanishing, timelike, irrotational Killing vector field that yields the property of being static. Note that $\partial_t$ does not generate horizons, but Lorentz boosts generate degenerate bifurcate Killing horizons, as shown in Figure \ref{fig: bifurcate Killing in Minkowski}. Although there is not a preferred normalization for the surface gravity associated to boosts, by choosing a particular uniformly accelerated worldline, Equation \eqref{eq: def global hawking temperature} corresponds to the Unruh temperature. Moreover, Minkowski is a globally hyperbolic spacetime and Cauchy surfaces are highlighted in its conformal diagram in Figure \ref{fig: Minkowski conformal diagram}.

One interesting fact is that Minkowski spacetime does not admit \textit{global} Schwarzschild-like coordinates. In Schwarzschild-like coordinates, which are simply spherical coordinates, its line element reads
\begin{equation}
\label{eq:Minkowski metric spherical}
  ds^2=-dt^2+dr^2+r^2d\Sigma_{1,n-2}^2,
\end{equation}
where $t\in\mathbb{R}$, $r>0$, $\theta\in[0,2\pi)$, $\varphi_1,...,\varphi_{n-3}\in[0,\pi)$. In fact, it covers Minkowski spacetime minus the origin $r=0$. The spacetime covered by Equation \eqref{eq:Minkowski metric spherical} is not globally hyperbolic. Yet, depending on the analysis performed, we can use spherical coordinates to characterize Minkowski spacetime if we can fix, by hand, the ambiguity at $r=0$, as shown in Section \ref{subsec: On Minkowski spacetime without a plate}. The hypersurfaces $\Sigma_{1,n-2}$ are $(n-2)$-dimensional spheres with positive, constant sectional curvature. The Christoffel symbols are non-vanishing, but its curvature tensors and its sectional curvature vanish.

\begin{figure}[H]
  \centering
    \centering
     \includegraphics[width=0.6\textwidth]{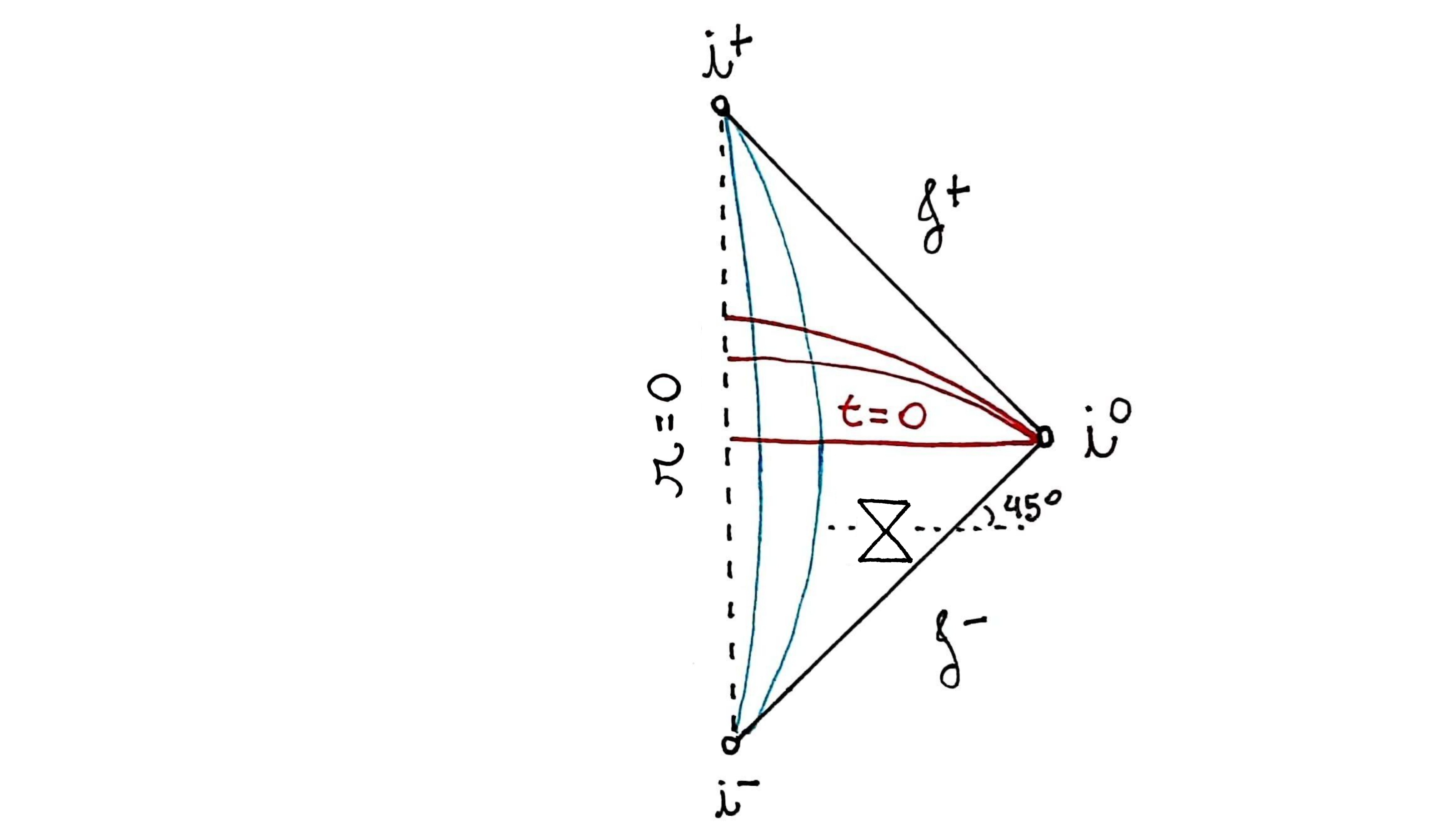}
     \caption{Conformal diagram of Minkowski spacetime. The blue lines connecting $i^-$ to $i^+$ are timelike, and the red lines connecting $r=0$ with spatial infinity are Cauchy surfaces. Lightlike radial curves are straight lines with a slope of $45^{\circ}$ degrees with respect to the $r=0$ line.}
     \label{fig: Minkowski conformal diagram}
\end{figure}

\subsection{de Sitter}
\label{ex: Chapter 1 de Sitter}

De Sitter (dS) spacetime is a maximally symmetric solution of Einstein field equations in vacuum with positive cosmological constant $\Lambda = \frac{(n-2)(n-1)}{2 L^2}$. The quantity $L>0$ determines a length scale and is called \textit{de Sitter radius}. In Schwarzschild-like coordinates, its line element reads
\begin{equation}
\label{eq:de Sitter metric static}
  ds^2=- \left(1 - \frac{r^2}{L^2}\right) dt^2 + \left(1 - \frac{r^2}{L^2} \right) ^{-1}dr^2+r^2d\Sigma_{+1,n-2}^2,
\end{equation}
where $t\in\mathbb{R}$, $r\in[0,L)$, $\theta\in[0,2\pi)$, and $\varphi_1,...,\varphi_{n-3}\in[0,\pi)$. Note that $r=0$ is not a curvature singularity, since $\mathbf{K}_{dS}\propto \Lambda$ is constant. 
These coordinates cover only a part of the maximal analytic extension shown in Figure \ref{fig: dS conformal diagram}. Also, the timelike Killing vector $\partial_t$ generates a bifurcate Killing horizon at $r=L$. Using Equation \eqref{eq: surface gravity schd-like},
we obtain $\kappa_h = \frac{1}{L} $, which in turn yields a well-defined Hawking temperature, and a geometric Hawking temperature of $T_{gH} = \frac{1}{2\pi L}$,
see e.g. \cite[Sec.5.4]{birrell1984quantum}.

\begin{figure}[H]
  \centering
  \includegraphics[width=0.6\textwidth]{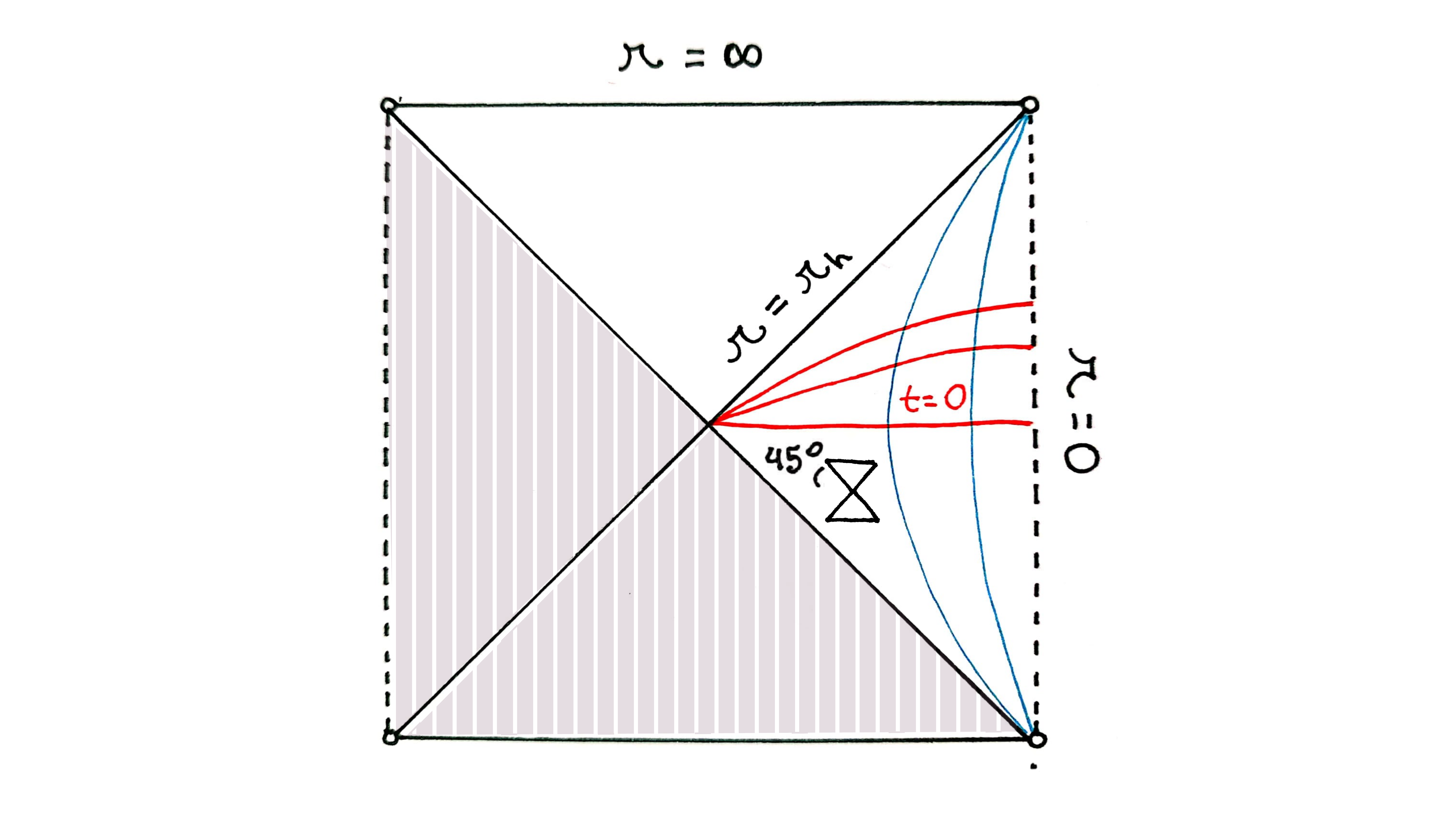}
    \caption{Maximal analytical extension of de Sitter spacetime. The region $r\in[0,L)$ is covered by Equation \eqref{eq:de Sitter metric static}.}
     \label{fig: dS conformal diagram}
\end{figure}

This spacetime is known as a \textit{cosmological solution} of Einstein equations and the hypersurface $r=L$ is called \textit{cosmological horizon}. This nomenclature makes reference to the observable universe and its limit of tangibility. For a discussion on the connection between the de Sitter solution and our universe, the reader can check \cite[Pg.124]{hawking1973large} or \cite{mohajan2017brief}. In addition, an elaborate exposition on charts and geodesics on de Sitter spacetime, as well as on the compactification procedure to obtain its Penrose diagram and on how to define it as a hyperboloid in a higher-dimensional Minkowski spacetime can be found in \cite{kim2002classical}.

\subsection{Anti-de Sitter}
\label{ex: Chapter 1 Anti-de Sitter}

The anti-de Sitter (AdS) spacetime is a static, maximally symmetric solution of Einstein field equations in vacuum with negative cosmological constant $\Lambda = -\frac{(n-2)(n-1)}{2 L^2}$. Analogously to de Sitter, the quantity $L>0$ determines a length scale and is called \textit{anti-de Sitter radius}. In Schwarzschild-like coordinates, its line element reads
\enlargethispage{2\baselineskip}
\begin{equation*}
\label{eq:anti de Sitter metric}
  ds^2=- \left(1 + \frac{r^2}{L^2}\right) dt^2 + \left(1 + \frac{r^2}{L^2} \right) ^{-1}dr^2+r^2d\Sigma_{+1,n-2}^2,
\end{equation*}
where $t\in\mathbb{R}$, $r\geq 0$, $\theta\in[0,2\pi)$, and $\varphi_1,...,\varphi_{n-3}\in[0,\pi)$. In contrast with the de Sitter case, $\partial_t$ does not generate a bifurcate Killing horizon. Be that as it may, similarly to the horizon generated by boosts in Minkowski spacetime, uniformly accelerated observers in the Rindler-AdS wedge see a horizon at $r=L$. In Section \ref{sec: On a static BTZ spacetime and on Rindler-AdS3}, we consider the latter and we see that indeed supercritically accelerated observers in AdS measure a temperature \cite{deser1997accelerated,deser1999mapping}. Moreover, AdS is a non-globally hyperbolic spacetime and its causal structure is illustrated in Figure \ref{fig: AdS and patches conformal diagram}. For more details, check \cite[Pg.131]{hawking1973large} or \cite{gibbons2000anti,sokolowski2016bizarre}.
\begin{figure}[H]
  \centering
    \includegraphics[width=.25\textwidth]{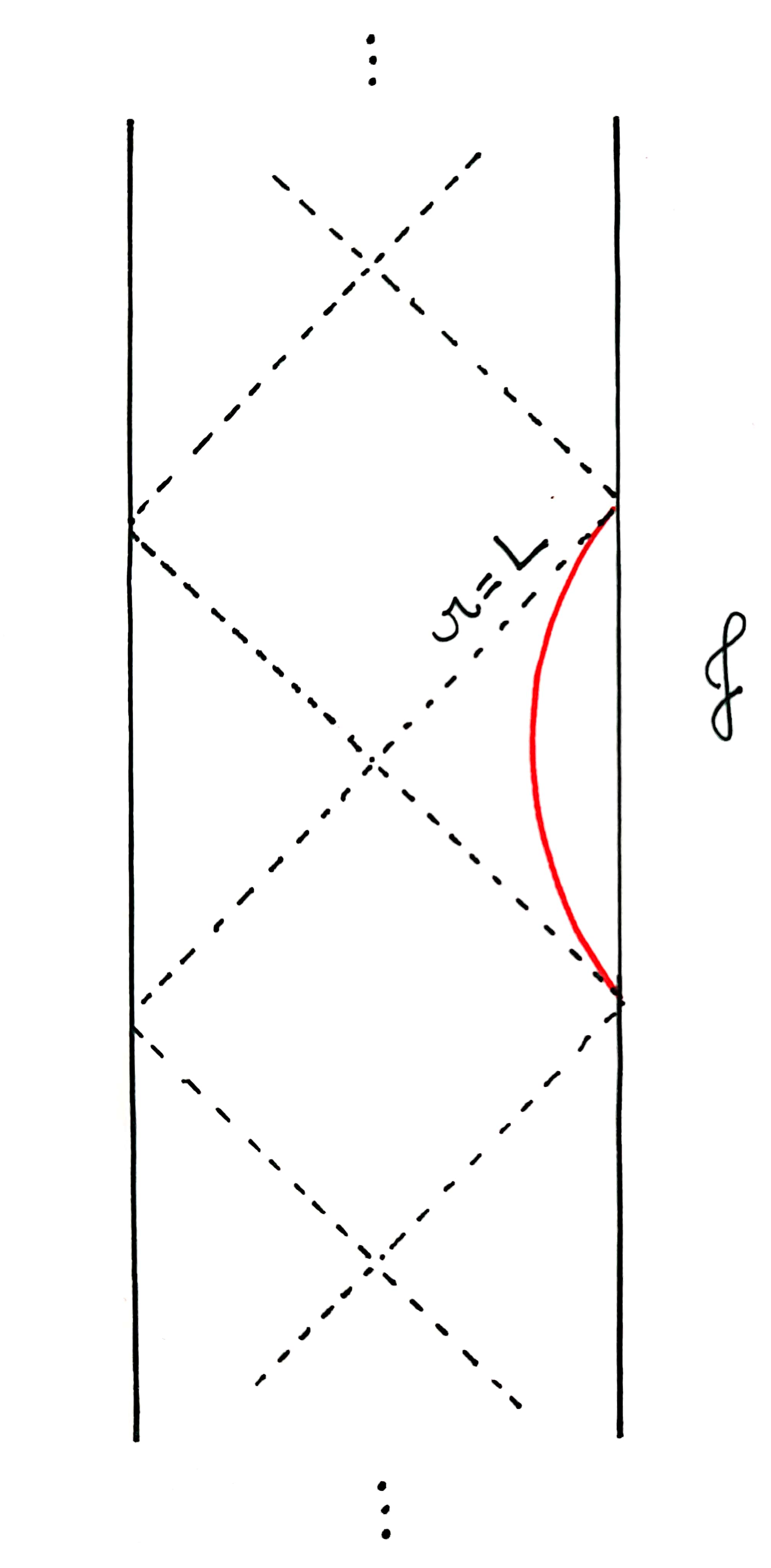}
    \caption{AdS spacetime. The red trajectory represents a supercritically accelerated observer. }
     \label{fig: AdS and patches conformal diagram}
\end{figure}

\subsection{Schwarzschild}
\label{ex: Chapter 1 Schwarzschild}

Schwarzschild spacetime is a static, globally hyperbolic solution of Einstein field equations in vacuum with vanishing cosmological constant. Its line element in Schwarzschild(-like) coordinates reads
\begin{equation*}
\label{eq:Schwarzschild metric}
  ds^2=- \left(1- \frac{2M}{r^{n-3}}\right) dt^2 + \left(1- \frac{2M}{r^{n-3}}\right) ^{-1}dr^2+r^2d\Sigma_{+1,n-2}^2,
\end{equation*}
where $t\in\mathbb{R}$, $r>(2M)^{-(n-3)}$, $\theta\in[0,2\pi)$, $\varphi_1,...,\varphi_{n-3}\in[0,\pi)$. The sections of constant time and radius are $(n-2)$-dimensional spheres, which are maximally symmetric. Hence, beyond time-translational symmetry, it admits $\frac{(n-2)(n-1)}{2}$ Killing vector fields associated with rotations. If $n>3$, then the global, non-vanishing, timelike, irrotational Killing vector field $\partial_t$ generates a non-degenerate bifurcate Killing horizon at $r=(2M)^{-(n-3)}$, with surface gravity
\begin{equation*}
  \kappa_h = (n-3) (2M)^{(n-3)(n-2)}
\end{equation*}
and corresponding local Hawking temperature of $T_H = \frac{\kappa_h}{ 2\pi\sqrt{f(r)} }$. In addition, it is interesting to note that its Kretschmann scalar diverges at $r=0$ but that its Ricci scalar vanishes everywhere. In the four-dimensional case, Equation \eqref{eq: Kretschmann 4D} holds.

The Kruskal-Szekeres coordinates provide its maximal analytical extension \cite[Pg.149]{hawking1973large}, whose conformal diagram is given in Figure \ref{fig: Schd conformal diagram}. Note that Schwarzschild-like coordinates cover only the exterior region given by $r>r_h$.

\begin{figure}[H]
  \centering
    \includegraphics[width=0.6\textwidth]{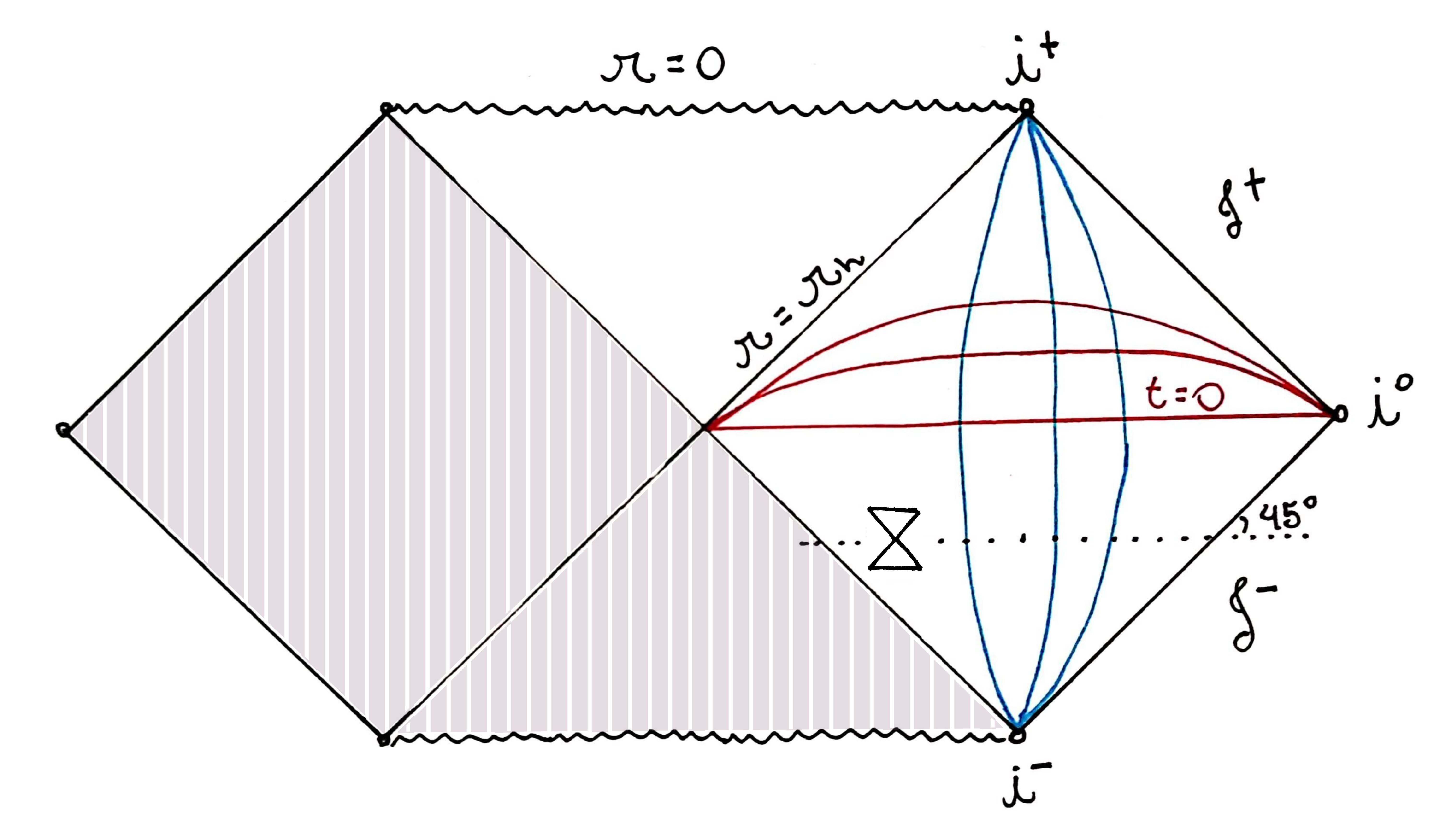}
    \caption{Conformal diagram of the maximal analytic extension of Schwarzschild spacetime. The non-hashed part corresponds to the black hole region (upper triangle) and to the exterior region (diamond). The blue lines connecting $i^-$ to $i^+$ are timelike, and the red lines connecting $r=0$ with spatial infinity are Cauchy surfaces. Lightlike radial curves are straight lines with a slope of $45^{\circ}$ degrees with respect to the $r=0$ line.}
     \label{fig: Schd conformal diagram}
\end{figure}

\subsection{BTZ}
\label{ex: Chapter 1 BTZ}

Gravity in three dimensions is manifestly odd. The absence of curvature singularities as indicated by Equation \eqref{eq: Kretschmann 3D} suggests that there are no ``black holes'' in three dimensions. However, another suitable way to characterize black holes involves, instead of the presence of singularities, the presence of causally disconnected regions in the spacetime \cite[Pg.300]{wald2010general}. In this sense, as shown in \cite{banados1992black,banados1993geometry}, black holes do exist in three dimensions. The Ba\~{n}ados-Teitelboim-Zanelli (BTZ) spacetime is characterized by the line element
\begin{equation}
\label{eq: BTZ metric}
  ds^2= - \left(\frac{r^2}{L^2} - M\right) dt^2 -J dtd\theta+  \left(\frac{r^2}{L^2} - M + \frac{J^2}{4r^2}\right)^{-1}dr^2+r^2d\theta^2,
\end{equation}
where $t\in\mathbb{R}$, $r>0$, $\theta\in[0,2\pi)$. $M$ and $J$ are called, respectively, mass and angular momentum. This spacetime solves Einstein field equations in vacuum with a cosmological constant of $\Lambda=-\frac{1}{L^2}$. It has constant sectional curvature $s=\Lambda$, as per Equation \eqref{eq: lambda constant sectional curvature}, but it is not maximally symmetric.%

When $J=0$, Equation \eqref{eq: BTZ metric} describes the static BTZ black hole with
\begin{equation}
\label{eq: static BTZ metric}
  ds^2= - \left(\frac{r^2}{L^2} - M\right) dt^2 +  \left(\frac{r^2}{L^2} - M\right)^{-1}dr^2+r^2d\theta^2,
\end{equation}
where $t\in\mathbb{R}$, $M>0$, $r>L\sqrt{M}$, $\theta\in[0,2\pi)$, and such that $\theta \cong \theta + 2\pi$. We could also consider $M<0$; in this case, $r=0$ is in fact a naked singularity for $M\neq -1$. In this thesis, however, we only work on the static BTZ spacetime with positive mass, see Section \ref{sec: On a static BTZ spacetime and on Rindler-AdS3}. In this case, $r_h= L\sqrt{M}$ is a non-degenerate bifurcate Killing horizon with surface gravity $\kappa_h=\frac{\sqrt{M}}{L}$ and Hawking temperature $T_{gH}=\frac{\kappa_h}{2\pi}$. Since this solution is asymptotically AdS, it is not globally hyperbolic.
Its conformal diagram is given in Figure \ref{fig: BTZ conformal diagram}. For the construction of a (rotating) BTZ spacetime and other details, beyond the already mentioned references, one can resort to \cite{birmingham2001exact, bussolaThesis}.

\begin{figure}[H]
  \centering
    \includegraphics[width=0.6\textwidth]{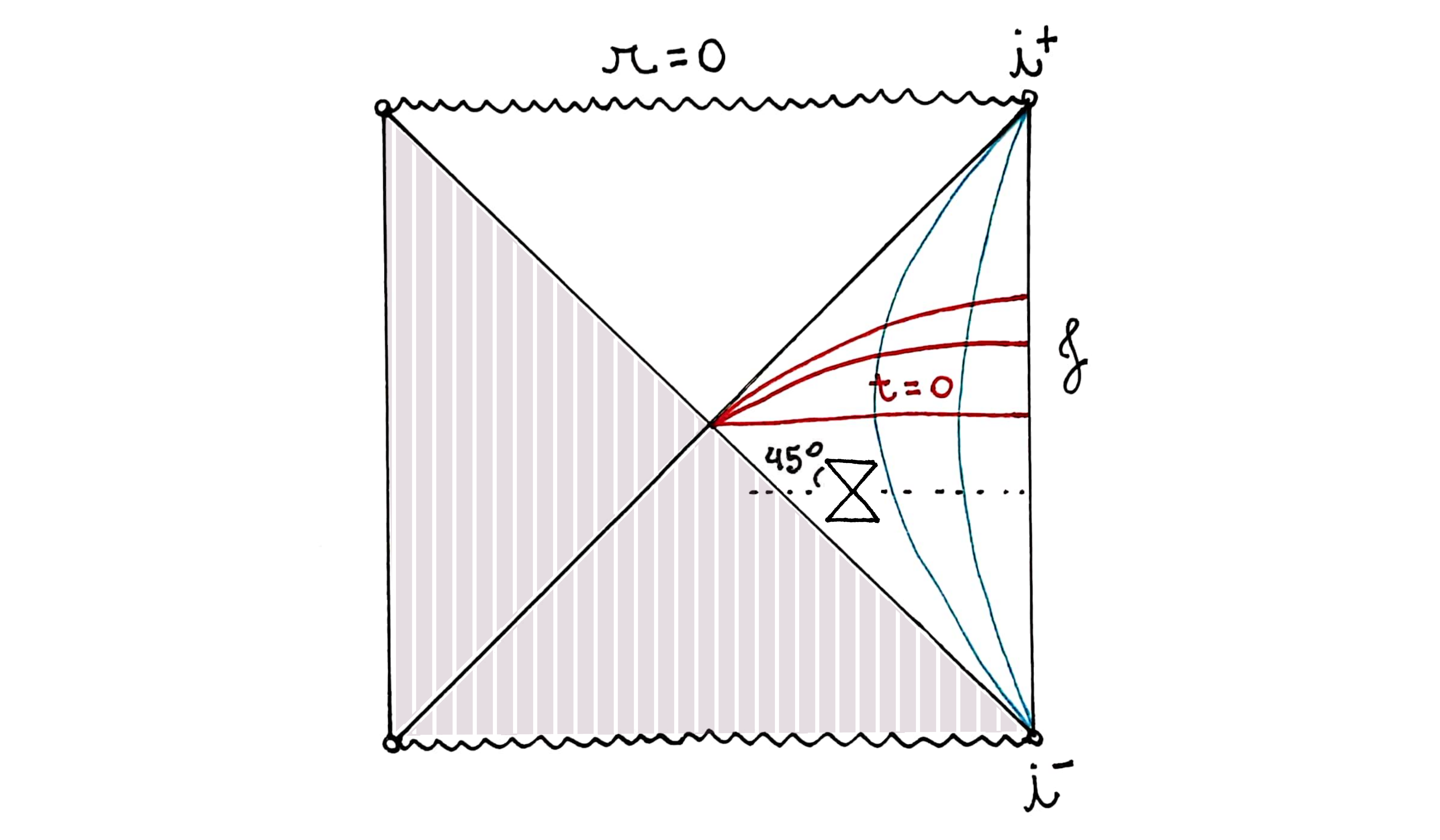}
    \caption{Conformal diagram of the maximal analytic extension of a static BTZ black hole spacetime. The blue lines connecting $i^-$ to $i^+$ are timelike, and the red lines connecting $r=0$ with spatial infinity are Cauchy surfaces. Lightlike radial curves are straight lines with a slope of $45^{\circ}$ degrees with respect to the $r=0$ line.}
     \label{fig: BTZ conformal diagram}
\end{figure}

\subsection{Schwarzschild-AdS}
\label{ex: Chapter 1 Schwarzschild-AdS}
Schwarzschild-AdS spacetime is an asymptotically AdS generalization of the Schwarzschild solution. It is a static solution of Einstein field equations in vacuum with negative cosmological constant, such that its line element in Schwarzschild-like coordinates reads%
\begin{equation}
\label{eq:Schwarzschild-AdS metric}
  ds^2=- \left(1- \frac{2M}{r^{n-3}} + \frac{r^2}{L^2}\right) dt^2 + \left(1- \frac{2M}{r^{n-3}}+ \frac{r^2}{L^2}\right) ^{-1}dr^2+r^2d\Sigma_{+1,n-2}^2,
\end{equation}
where $t\in\mathbb{R}$, $r>r_h$, $\theta\in[0,2\pi)$, $\varphi_1,...,\varphi_{n-3}\in[0,\pi)$. As in the Schwarzschild case, the sections $\Sigma_{+1,n-2}$ are spheres, its sectional curvature is not constant since Equation \eqref{eq: maximally symmetric curvature Riemann tensor} does not hold, and there is a curvature singularity at $r=0$.

Schwarzschild-AdS possesses a bifurcate Killing horizon generated by $\partial_t$ and a timelike conformal boundary. Hence, as AdS spacetime, it is not globally hyperbolic. In addition, from Equation \eqref{eq: surface gravity schd-like}, we obtain
$$\kappa_h = \left(\frac{(n-3)M}{r^{n-2}} + \frac{r}{L^2} \right)\Big|_{r=r_h} \quad \text{and}\quad T_H = \frac{T_{gH}}{\sqrt{1 - \frac{2M}{r^{n-3}}+ \frac{r^2}{L^2}}},$$
which generalizes, to $n$ dimensions, Example \ref{eg: Schd and Schd-AdS temperatures} with $k=-\frac{1}{L^2}$. Its conformal diagram is similar to the one Figure \ref{fig: BTZ conformal diagram}, with the only difference that, for $n\geq4$, the wiggly lines representing $r=0$ bulge inwards, see \cite{Fidkowski2003nf}. For more details on this spacetime, see e.g. \cite{socolovsky2017schwarzschild} or \cite[Sec.9.2.]{emparan2008black}.

There is a counterpart to Schwarzschild-AdS with positive curvature, called Schwarzschild-dS spacetime. It is a static solution of Einstein field equations in vacuum, but with positive cosmological constant. Its line element is given by Equation \eqref{eq:Schwarzschild-AdS metric} with the mapping $L^2\mapsto -L^2$. This subtle difference reverberates in drastic changes on the geometrical and causal structure of the spacetime. Such spacetime possesses both a black hole horizon, like Schwarzschild spacetime, and a cosmological horizon, like de Sitter spacetime; these horizons delineate a globally hyperbolic region, see Figure 3 in \cite{castineiras2003interaction}.

\subsection{Topological black holes}
\label{ex: chapter 1 toplogical black holes}

Topological black holes are static, non-globally hyperbolic solutions of Einstein field equations in vacuum with negative cosmological constant. Their line element in Schwarzschild-like coordinates, with mass $M\geq0$, AdS radius $L>0$, and $j\in\{-1,0,+1\}$, reads
\begin{equation}
\label{eq:massive topological black holes metric}
  ds^2=- \left(j- \frac{2M}{r^{n-3}} + \frac{r^2}{L^2}\right) dt^2 + \left(j- \frac{2M}{r^{n-3} }  + \frac{r^2}{L^2}\right) ^{-1}dr^2+r^2d\Sigma_{j}^{n-2},
\end{equation}
where $t\in\mathbb{R}$, $r>r_h\geq0$, $\theta\in[0,2\pi)$, $\varphi_1,...,\varphi_{n-3}\in[0,\pi)$. For $j=+1$, the line element \eqref{eq:massive topological black holes metric} reduces to Schwarzschild-AdS \eqref{eq:Schwarzschild-AdS metric}. For $j\in\{-1,0\}$, there are different possible choices for the topologies of the sections of constant time and radius---hence, these solutions are called ``topological'' black holes. We shall refer to the cases $j=-1$, $j=0$, and $j=+1$ respectively as \textit{hyperbolic}, \textit{flat}, and \textit{spherical} topological black holes.

The functions $f(r)$ and $h(r)$ in Equation \eqref{eq:metric schwarzschild-like} assume the forms as per Equations \eqref{eq: wpioerfjwopi0pwr98} and \eqref{eq: wpioerfjwopi0pwr9ddd833} for $n=4$ and $n=3$, respectively. In addition,
for $M=0$, $j=-1$ and $n=3$, Equation \eqref{eq:massive topological black holes metric} is similar to a unit-mass, static BTZ black hole, see Equation \eqref{eq: static BTZ metric}. Equality holds if we take $\Sigma_{-1,1}$ to be a circle. With this in mind, we can see massless hyperbolic topological black holes as $n$-dimensional generalizations of a unit-mass static BTZ black hole. In these spacetimes, there is a non-degenerate bifurcate Killing horizon at $r_h=L$ with surface gravity $\kappa_h=\frac{1}{L}$.

For the mathematical construction of these solutions, see \cite{mann1997topological}. For details regarding the different compactifications one can perform, the reader can check \cite[Sec.II.B]{brill1997thermodynamics}. For black hole thermodynamics on topological black holes, see e.g. \cite{Birmingham1998nr,birmingham2007stability,Vanzo1997gw}. Moreover, the three and four-dimensional massless hyperbolic solutions are considered in Section \ref{sec: On massless hyperbolic black holes} of Chapter \ref{chap: Applications} of this thesis, within the context of quantum field theory.

\subsection{Lifshitz topological black holes}
\label{ex: chapter 1 Lifshitz toplogical black holes}

The line element of a four-dimensional Lifshitz topological black hole $\mathcal{M}$ can be written in the form of Equation \eqref{eq:metric schwarzschild-like} by
\begin{subequations}
\label{eq: metric Lifshitz black holes no j}
\begin{align}
    ds^2 =- f(r)dt^2 + h(r)dr^2 +  r^2 d\theta^2 + r^2 J_\kappa(\theta)^2d\varphi^2,
\end{align}
with
\begin{align}
    &f(r):=\frac{r^{2}}{L^2} \left(\frac{r^{2}}{L^2} + \frac{\kappa}{2}\right),\\
    &h(r):=\frac{L^{2}}{r^2}f(r)^{-1}.
\end{align}
The parameter $\kappa$ assumes values only in the set $\{-1,0,+1\}$, and $J_\kappa$ is defined by
    \begin{align}
        &J_\kappa(\theta) := \begin{cases}
                  \sinh(\theta),& \kappa=-1,\\
                  \theta,       & \kappa=0,\\
                  \sin(\theta), & \kappa=+1.
                \end{cases}
    \end{align}
\end{subequations}

Note that $\mathcal{M}$ cannot be a solution of Einstein gravity given the fact that a four-dimensional spacetime that admits Schwarzschild-like coordinates is a solution of Einstein field equations in vacuum if and only if Equation \eqref{eq: wpioerfjwopi0pwr98} holds, which is not the case since  $h(r) \neq f(r)^{-1}$. Yet, as shown in \cite{Mann2009yx}, $\mathcal{M}$ is a metric solution of the action
\begin{align*}
  \label{eq: action Lifshitz gravity}
  S =\int &d^4x\sqrt{|g|}\Big( R-2\Lambda -\frac{1}{4}F_{\mu\nu}F^{\mu\nu}-\frac{1}{12}H_{\mu\nu\tau}H^{\mu\nu\tau} -\frac{C}{\sqrt{|g|}}\varepsilon^{\mu\nu\alpha\beta}B_{\mu\nu}F_{\alpha\beta} \Big),
\end{align*}
where the only non-vanishing field strengths components are given by
 \begin{equation*}
   \label{eq: non-vanishing field strengths components}
       F_{rt} = 2 L r \text{  and   } H_{r\theta\varphi}= 2L^2 r J_\kappa.
 \end{equation*}
In the above, $g$ is the metric tensor on $\mathcal{M}$, $\mathbf{R}= -\frac{22}{L^2} - \frac{\kappa}{r^2}$ the Ricci scalar, $L$ a length scale, $\Lambda=-\frac{5}{L^2}$, $\varepsilon^{\mu\nu\alpha\beta}$ the Levi-Civita tensor density, $F_{\mu\nu}=\partial_{[\mu}A_{\nu]}$ and $H_{\mu\nu\tau}=\partial_{[\mu}B_{\nu]}$ Abelian gauge fields, and $C=\frac{2}{L}$ a constant coupling parameter.

For all $\kappa$, $\mathcal{M}$ is a static, asymptotically Lifshitz, non-globally hyperbolic spacetime. However, the behaviour of Equation \eqref{eq: metric Lifshitz black holes no j} for small $r$ changes drastically with $\kappa$. Specifically, for $\kappa=\pm1$, $r=0$ is a curvature singularity that is naked for $\kappa=+1$, while hidden by a horizon at $r=\frac{L}{\sqrt{2}}$ for $\kappa=-1$. For $\kappa=0$, $r=0$ is a coordinate singularity, where neither the Ricci or Kretschmann scalars diverge. Although only the case $\kappa=-1$ possesses a horizon, in analogy with the topological black holes described in the previous section, for $\kappa=0$, $\kappa=-1$ and $\kappa=+1$, we call these spacetimes, respectively, \textit{flat}, \textit{hyperbolic} and \textit{spherical Lifshitz topological black holes}. Withal, note that a flat Lifshitz black hole is in fact a Lifshitz spacetime with critical exponent $z=2$ and with sectional line element written in polar coordinates. This means that the $\kappa=0$ solution also solves other generalized gravity theories; specifically, Einstein-Maxwell-Dilaton and Einstein-Proca gravity theories, as detailed in \cite[Pg. 27]{Hartnoll2016apf}.

\subsection{Global monopoles}
\label{ex: Chapter 1 Global monopoles}

Global monopoles are examples of topological defects (as are cosmic strings), arising from global symmetry breaking within grand unification theories, possibly produced during phase transitions in the early universe. An interested reader may check \cite{Barriola1989hx,vilenkin1994cosmic}. Here, a global monopole is another example of a static, non-globally hyperbolic spacetime that is not a solution of Einstein field equations in vacuum. For $t\in\mathbb{R}$, $r\in(0,\infty)$, $\theta\in[0,\pi)$, $\varphi\in[0,2\pi)$ and  $\alpha\in(0,1)$, its line element reads
\begin{align}
  \label{eq: metric Global monopole}
    ds^2 =- \alpha^2 dt^2 + \frac{1}{\alpha^2}dr^2 + r^2 (d\theta^2 + \sin^2\theta d\varphi^2).
\end{align}
Equation \eqref{eq: metric Global monopole} does solve Einstein field equations with vanishing cosmological constant with a classical energy-momentum tensor whose only non-vanishing components are
\begin{equation*}
    T_{tt} = -\alpha^4 T_{rr} = \frac{\alpha^2 \mathbf{R}}{2}, \quad \text{ where the Ricci scalar is }\quad \mathbf{R} = \frac{2(1-\alpha^2)}{r^2}.
\end{equation*}
There are three noteworthy features of Equation \eqref{eq: metric Global monopole}. First, for $\alpha\rightarrow 1$ we recover Minkowski spacetime in spherical coordinates. Second, although a two-dimensional conformal-diagram would look like Figure \ref{fig: Minkowski conformal diagram}, a global monopole is not asymptotically flat since each point in the diagram would be conformal to, but not coinciding with, a two-dimensional sphere. Another way of understanding this fact is by computing the Newtonian potential that does not vanish as $r\rightarrow\infty$, as shown in \cite{Barriola1989hx}. Last, $r=0$ is a (naked) curvature singularity since the Kretschmann scalar, $\mathbf{K}=\mathbf{R}^2$, diverges there.

 \cleardoublepage
 \chapter[Quantum Field Theory on Static Spacetimes]{Quantum Field Theory on Static Spacetimes}
 \label{chap: Quantum Field Theory on Static Spacetimes}
 \minitoc

 \pagestyle{myPhDpagestyle2}
\vfill
On static spacetimes we can avail of a global, timelike, irrotational Killing vector field to select a unique ground state, and to identify thermal states for a free, scalar, quantum field theory. In this chapter, I employ the algebraic approach to construct these states and guarantee they are well-defined and physically-sensible. First, in Section \ref{sec: A glimpse on the algebraic approach}, I highlight the main ingredients of quantum field theory on general spacetimes. Second, in Section \ref{sec: Physically-sensible dynamics on static spacetimes}, I define physically-sensible dynamics on static, stably-causal, not necessarily globally hyperbolic spacetimes. Then, in Section \ref{sec: Physically-sensible states in Schwarzschild-like coordinates} I specialize to spacetimes that admit Schwarzschild-like coordinates and I give explicit expressions for the two-point functions of physically-sensible states. In Section \ref{sec: Probing quantum states with particle detectors}, I describe, succinctly, how we can probe these quantum states using an Unruh-DeWitt particle detector. The framework described here sets the ground for all applications shown in the next chapter. To illustrate how all the building blocks connect to each other, in Section \ref{sec: Examples chapter 2}, I apply it on Minkowski spacetime endowed with spherical coordinates.
\newpage
 \section{A glimpse on the algebraic approach}
 \label{sec: A glimpse on the algebraic approach}

With the goal of constructing physically-sensible states in Section \ref{sec: Physically-sensible states in Schwarzschild-like coordinates}, let us review a few key notions of the algebraic description of a free, scalar, quantum field theory. On a spacetime $\mathcal{M}$ equipped with a Lorentzian metric tensor $g$ whose scalar curvature is $\mathbf{R}\in\mathbb{R}$, as defined at Page \pageref{page def spacetime} of Chapter \ref{chap: The Spacetimes}, consider a free, scalar field $\Psi:\mathcal{M}\rightarrow\mathbb{R}$ with mass $m_0\geq 0$. Its dynamics is ruled by the Klein-Gordon equation
 \begin{equation}
   \label{eq: KG}
 P\Psi : = (\Box - m_0^2 - \xi \mathbf{R})\Psi=0,
 \end{equation}
 where $\Box = \frac{1}{\sqrt{|g|}} \partial_\mu \left(\sqrt{|g|} g^{\mu\nu}\partial_\nu \right)$ is the d'Alembertian, while $\xi\in\mathbb{R}$ is a coupling parameter. If $\mathcal{M}$ is a vacuum solution of Einstein's equations with a cosmological constant, hence with constant scalar curvature such that the field-curvature coupling yields a simple shift of mass, we define an \textit{effective mass} by
 \begin{equation}
   \label{eq: eff mass}
   m_{\text{eff}}^2 := m_0^2 + \xi \mathbf{R}.
 \end{equation}

A particularly relevant bisolution of the Klein-Gordon equation is the causal propagator, also known as Pauli-Jordan distribution, or commutator function. It is one of the main ingredients in the implementation of the canonical quantization, which here goes directly in the definition of the (quantum) algebra of observables. Let us properly define both these objects, and clarify how they relate to each other. First, assume $\mathcal{M}$ is a globally hyperbolic spacetime. As mentioned in Section \ref{sec: A consequence of (non-)global hyperbolicity} of Chapter \ref{chap: The Spacetimes}, the Cauchy problem associated with the Klein-Gordon equation on such a spacetime is well-defined, and, as proven in detail in \cite[Ch.3]{WELM}, the following theorem holds.
\begin{theorem}[Advanced and retarded fundamental solutions]
\label{thm:advanced and retarded prop}
  On a globally hyperbolic spacetime $\mathcal{M}$, there exist unique advanced
$E_{+}$ and retarded $E_{-}$ fundamental solutions for the Klein-Gordon operator $P$, i.e.  $E_{\pm}:C_{0}^{\infty}(\mathcal{M},\mathbb{R})\rightarrow C^{\infty}(\mathcal{M},\mathbb{R})$ such that
\begin{equation}
  \label{eq: prop 1 Epm}
  P\circ E_{\pm}=E_{\pm}\circ P=\text{id}_{C_{0}^{\infty}(\mathcal{M},\mathbb{R})}.
\end{equation}
  Moreover, their supports are contained, respectively, in the causal future and past:
  \begin{equation}
    \label{eq: prop 2 Epm}
\text{supp}(E_{\pm}f)\subset J^{\pm}(\text{supp}f)\text{, }\forall f\in C_{0}^{\infty}(\mathcal{M},\mathbb{R})\text{.}
\end{equation}
\end{theorem}
From Theorem \ref{thm:advanced and retarded prop} follows that the difference between the advanced and retarded fundamental solutions $\widetilde{E} = E_+ - E_-$, satisfies the following properties. Equation \eqref{eq: prop 1 Epm} entails
\begin{equation*}
 \label{eq: E+ - E-}
   P\circ \widetilde{E} = \widetilde{E} \circ P  = 0,
 \end{equation*}
while Equation \eqref{eq: prop 2 Epm} yields
 \begin{equation*}
  \label{eq: E+ - E- 2}
   \text{supp}(\widetilde{E}\,f)\subset J^{+}(\text{supp}f)\cup J^{-}(\text{supp}f)\text{.}
 \end{equation*}
Let $d\text{vol}_{\mathcal{M}^N}$ be the volume form induced by $g$ on $\mathcal{M}^N$, for $N\in\mathbb{N}^*$. By linearity and the Schwartz kernel theorem \cite[Pg.128]{hormander2015analysis}, $\widetilde{E}$ identifies a bidistribution $E\in\mathcal{D}'(\mathcal{M}^2)$,
\begin{equation*}
E(f,h) =\int_{\mathcal{M}^2}d\text{vol}_{\mathcal{M}^2}E(x,x')f(x)h(x')\text{ for }f,h\in C_{0}^{\infty}(M,\mathbb{R}),
\end{equation*}
which is skew-symmetric and vanishes if the supports of $f$ and $h$ are causally separated, i.e. in the level of distribution kernel, we have that $E(x,x')=-E(x',x)$ and that $E(x,x')=0$ if $x$ and $x'$ are spacelike separated. Accordingly, $E$ is called the \textit{causal propagator}.

The causal propagator identifies a symplectic form in the space of solutions to the Cauchy problem associated with Equation \eqref{eq: KG}, explicitly given in Equation \eqref{eq:cauchy problem example}, thus generating the classical phase-space of the underlying field theory. Quantization may then be implemented by a canonical procedure, as prescribed in \cite{WELM,benini2013quantum,fewster2015algebraic,hack2015cosmological,moretti}. Here, I choose to omit all details regarding this implementation and I directly define the (quantum) algebra of observables.

\begin{subequations}
 \begin{definition}[algebra of observables]
   \label{def: algebra of observables}
   The \textit{algebra of observables} $\mathcal{A}(\mathcal{M})$, also known as the CCR algebra, is the quotient of the unital $\ast$-algebra generated by the smeared fields $ \{\Psi(f) : f \in C^\infty_0(\mathcal{M}) \}$, with the $\ast$-ideal generated by the following relations: for $f,h\in C^\infty_0(\mathcal{M})$ and $c\in\mathbb{R}$,
         \begin{align}
           \text{Linearity} & \qquad0=\Psi(cf+h)-c\Psi(f)-\Psi(h), \\
           \text{Hermiticity} & \qquad0=\Psi(f)^{\ast}-\Psi(f),\\
           \text{Klein-Gordon} & \qquad0=\Psi((\square - m_0^{2} - \xi \mathbf{R})f),\\
           \text{CCR} & \qquad0=[\Psi(f),\Psi(h)]-iE(f,h)\mathbbm{1}. \label{eq:oweihoijef}
         \end{align}
 \end{definition}
\end{subequations}
\enlargethispage{2\baselineskip}
 The quantization procedure of imposing the canonical commutation relations on the fields, codified by Equation \eqref{eq:oweihoijef}, can be carried out once the causal propagator is built. This, in turn, requires solving a suitable initial value problem, which is identified in the following proposition.

 \begin{proposition}
   \label{prop:ccr} On globally hyperbolic spacetimes, $\mathcal{M}\simeq \mathbb{R}\times\Sigma$ with Cauchy surface $\Sigma$, the initial conditions satisfied by the causal propagator, i.e. for $x=(t,\underbar{x})$, $t\in\mathbb{R}$ and $\underbar{x}\in\Sigma$,
   \begin{subequations}
 				\label{eq:ccrs}
 				\begin{align}
 					&E(x,x')|_{t'=t}  = 0 \text{,}\label{eq:ccr1}\\
 					&\partial_t E(x,x')|_{t=t'}=-\partial_{t'} E(x,x')|_{t'=t} =\delta_{\Sigma}(x,x')\label{eq:ccr2}\text{,}
 				\end{align}
 	\end{subequations}
  yield the canonical commutation relations after the implementation of Equation (\ref{eq:oweihoijef}).
 \end{proposition}
 \begin{proof}
   Let $P$ be the Klein-Gordon operator, as in Equation \eqref{eq: KG}, let $(t_0,\Psi_0,\dot\Psi_0)$ be the corresponding Cauchy data with $\Psi_0$, $\dot\Psi_0\in C_0^\infty( \{t_0\}\times\Sigma)$, $f\in C_0^\infty (\Sigma)$, and $\mathbf{n}$ the future-directed, unit vector field normal to $\Sigma$. On one hand, consider the initial value problem given by Equation \eqref{eq:cauchy problem example}, i.e.
 \begin{equation*}
   P \Psi = f,  \quad  \Psi|_{\{t_0\}\times\Sigma} = \Psi_0   \quad \text{ and }\quad  \nabla_{\mathbf{n}}\Psi|_{\{t_0\}\times\Sigma} = \dot\Psi_0.
 \end{equation*}
 It implies, by \cite[Cor1.2]{dimock1980} and for $f\in C_{0}^{\infty}(\Sigma)$, that
 $$ Ef|_{\Sigma}=0   \text{ and }  \nabla_{\mathbf{n}}Ef|_{\Sigma}=f .$$
At the level of integral kernel, it holds
\begin{equation*}
E(x,x')|_{\Sigma\times\Sigma}=0  \text{ and }  \nabla_{\mathbf{n}}E(x,x')|_{\Sigma\times\Sigma}=\delta_{\Sigma}(x,x')\text{.}
\end{equation*}
On the other hand, the equal-time canonical commutation relations are obtained via pull-back of its covariant counterpart $[\Psi(x),\Psi(x')]=iE(x,x')$ to a generic Cauchy surface:\\

  \hspace{3cm}   $[\Psi(x)|_{\Sigma},\Psi(x')|_{\Sigma}]=i E(x,x')|_{\Sigma\times\Sigma}= 0,$

  \hspace{3cm}   $[\nabla_{\mathbf{n}}\Psi(x)|_{\Sigma},\Psi(x')|_{\Sigma}]=i \nabla_{\mathbf{n}}E(x,x')|_{\Sigma\times\Sigma}= i\delta_{\Sigma}(x,x').$
 \end{proof}
\vspace{.5cm}

\begin{remark}
\label{rem: causal propagator on static non-globally hyperbolic}
 On non-globally hyperbolic spacetimes, the Cauchy problem \eqref{eq:cauchy problem example} is generally ill-defined. However, on the static spacetimes considered in the next chapter it is straightforward to make it well-defined as a Cauchy-boundary value problem such that Theorem \ref{thm:advanced and retarded prop} and Proposition \ref{prop:ccr} still hold. Specifically, assigning a boundary condition at the timelike boundary that engenders the non-global hyperbolicity disambiguates the initial value problem. An elaborate discussion regarding quantum field theory on globally hyperbolic spacetimes with a timelike boundary can be found in \cite{dappiaggi2019fundamental}.
 \end{remark}
To obtain a minimally meaningful physical theory, on the algebra of observables we need a notion of expectation values, and for that we need the notion of states.
\begin{definition}[State and related notions]
  \label{def: state and related notions}
 A \emph{state} $\psi$ on the algebra of observables $\mathcal{A}(\mathcal{M})$, see Definition \ref{def: algebra of observables}, is a linear functional
 \begin{equation*}
     \psi:\mathcal{A} (\mathcal{M})\rightarrow\mathbb{C}
 \end{equation*}
 that is normalized, $\psi(\mathbbm{1})=1$, and positive, $\psi(a^{\ast}a)\geq0,\,\forall a\in\mathcal{A}(\mathcal{M})$. Moreover,
\begin{itemize}
  \item[i)] a state $\psi$ is called \textit{mixed} if it is a convex combination of states: $\psi = \lambda \psi_1 + (1-\lambda)\psi_2$, with $\psi_i\neq \psi$, $i\in\{1,2\}$, and $\lambda<1$;
  \item[ii)] a state is called \textit{pure} if it is not mixed;
 \item[iii)] $\psi(a)$ is called the \emph{expectation value} of the observable $a\in\mathcal{A}(\mathcal{M})$ on the state $\psi$;
   \item[iv)] \label{it: def n point function} for $\Psi(f_i)\in\mathcal{A}(\mathcal{M})$, $f_i\in C_0^\infty(\mathcal{M})$, $i\in\{1,...,n\}$, and $n\in\mathbb{N}^*$, the map
    $$(f_1,...,f_n) \mapsto \psi_n(f_1,...,f_n):=\psi(\Psi(f_1)...\Psi(f_n)) $$
   is called the \textit{$n$-point function} of the state $\psi$.
\end{itemize}\vspace{-.9cm}
 \end{definition}

If we equip $\mathcal{A}(\mathcal{M})$ with a suitable topology, as per \cite[Def. III.1.1.1]{hack2010backreaction}, we can show that the maps $\psi_n$ of item iv) of Definition \ref{def: state and related notions}, are well-defined distributions in $\mathcal{D}'(\mathcal{M}^n)$ such that, by the Schwartz kernel theorem, we can write
\begin{equation*}
  \psi_n(f_1,...,f_n)= \int\limits_{\mathcal{M}^n}d\text{vol}_{\mathcal{M}^n} \psi_n(x_1,...,x_n)f(x_1)...f(x_n).
\end{equation*}
For technical details regarding the topological structure, I refer the reader to \cite{brunetti2015advances,hack2010backreaction}. Here, I simply assume that the maps $\psi_n$ are continuous in the usual test function topology on $C_0^\infty(\mathcal{M}^n)$. The relevance of introducing these maps comes from the fact that the expectation values $\psi(a)$, for $a\in\mathcal{A}(\mathcal{M})$, are completely specified by the $n$-point functions of the state $\psi$, see  \cite[Sec.2.2]{moretti}. As an example, the Minkowski Poincaré invariant vacuum state is fully specified by its two-point function \cite[Pg.91]{fulling1989aspects}. This fact, in turn, inspires the notion of Gaussian states.

 \begin{definition}[Gaussian states] \label{def: gaussian state} A state $\psi$ on $\mathcal{A}(\mathcal{M})$ is \emph{Gaussian} if $\forall f_i\in C^\infty_0(\mathcal{M})$, $i\in\{1,...,n\}$, $n\in\mathbb{N}^*$, its $n$-point functions $\psi_n(f_1,...,f_n)$ satisfy
   \begin{subequations}
   \label{eq: wick procedure mimic gaussian state}
   \begin{align}
     & \psi_n(f_1,...,f_n) = 0, \text{ for odd }n,  \\
     & \psi_n(f_1,...,f_n) = \sum\limits_{\mini{\text{partitions}}} \psi_2(f_{i_1},f_{i_2})...\psi_2(f_{i_{n-1}},f_{i_n}), \text{ for even }n,
   \end{align}
 \end{subequations}
   where \enquote{partitions} indicates all permutations of the indexes $i_j$, $j\in\{1,...,n\}$; i.e. over all decompositions of the set $\{1,...,n\}$ into disjoint subsets of two elements.
 \end{definition}
 For a fixed state $\psi$ \label{page: GNS theorem} on $\mathcal{A}(\mathcal{M})$, there is a concrete realization of the algebra $\mathcal{A}(\mathcal{M})$ on a Hilbert space $\mathcal{H}_\psi$, given by the Gelfand, Naimark, Segal (GNS) theorem. The GNS theorem grants us the quadruple $(\mathcal{H}_\psi,\mathcal{D}_\psi,\pi_\psi,\Omega_\psi)$, where $\pi_\psi:\mathcal{A}(\mathcal{M})\rightarrow \mathcal{L}(\mathcal{D}_\psi)$ is a $*$-representation of $\mathcal{A}(\mathcal{M})$ as linear operators with dense domain $\mathcal{D}_\psi\subset \mathcal{H}_\psi$, such that $\pi_\psi(\mathcal{A}(\mathcal{M}))\Omega_\psi = \mathcal{D}_\psi$
 and $\psi(\Psi(f)) = \braket{\Omega_\psi|\pi_\psi(\Psi(f))|\Omega_\psi}$. Particularly for a Gaussian state $\psi$, which is completely determined by its two-point function as per Definition \ref{def: gaussian state}, the GNS representation corresponds to a Fock representation and Equation \eqref{eq: wick procedure mimic gaussian state} codify the standard Wick procedure \cite{brunetti2015advances,wick1950evaluation}. In \cite{moretti}, one can find a complete discussion of topics at the heart of the above comments, such as the precise definition of a $*$-representation \cite[Def.7]{moretti}, a detailed proof of the GNS theorem \cite[Thm.1]{moretti} and a characterization of Gaussian states \cite[Thm.2]{moretti}.

The algebra of observables as per Definition \ref{def: algebra of observables} is the simplest algebra one can consider. First and foremost, it does not contain all physical observables. Let us see two examples to illustrate that.

\begin{example}[Vacuum fluctuations] \label{eg: vacuum fluctuations}
 Let $\mathcal{M}$ be a four-dimensional Minkowski spacetime endowed with Cartesian coordinates, as in Example \ref{eg: 10 Killing in Mink}. The two-point function for the vacuum state of a massless, free, scalar field reads
 \begin{subequations}
   \label{eq: two point mink example and kernel}
 \begin{equation}
   \label{eq: two point mink example uiofhiosfio}
   \psi_2(f,f') = \psi(\Psi(f)\Psi(f')) =  \lim\limits_{\varepsilon\rightarrow 0^+}  \int\limits_{\mathcal{M}^2} d^4xd^4x' \psi_2(x,x') f(x)f'(x')
 \end{equation}
for $f,f'\in C_0^\infty(\mathcal{M})$, and
\begin{equation}
\label{eq: kernel mink example uiofhiosfio}
  \psi_2(x,x') = \frac{1}{4\pi^2} \frac{1}{\sigma_{\varepsilon}(x,x')},
\end{equation}
\end{subequations}
  with $\sigma_{\varepsilon}(x,x'):= (x-x')^2+i\varepsilon (t-t') + \varepsilon^2$. When $x$ and $x'$ are lightlike separated, $(x-x')^2=0$. This singular behavior is resolved by the $\varepsilon$-regularization already introduced in Equation \eqref{eq: two point mink example and kernel}: as a tempered distribution on $C_0^\infty(\mathcal{M})\times C_0^\infty(\mathcal{M})$, $\psi_2(f,f')$ is well-defined and $\Psi(f)\Psi(f')$ is an element of the algebra of observables. Next, with $ \psi_2(x,x')$ as given above, suppose we want to compute the expectation value of the vacuum fluctuations, denoted $\psi(\Psi^2(f))$. A first attempt is to simply define it by
\begin{equation}
  \label{eq: two point mink example coincidence blow up}
  \psi(\Psi^2(f)) = \lim\limits_{\varepsilon\rightarrow 0^+} \int\limits_{\mathcal{M}^2} d^4xd^4x' \psi_2(x,x') f(x)\delta(x-x').
\end{equation}
However, the right-hand side of Equation \eqref{eq: two point mink example coincidence blow up} is ill-defined since $\psi_2(x,x')f(x)\delta(x-x')$  does not satisfy the H\"{o}rmander criterion for multiplication of distributions \cite{hormander2015analysis}. Two consequences follow: first, to obtain a well-defined $\psi(\Psi^2(f))$ we must implement another regularization procedure; second, $\Psi^2(f)$ is not an element of $\mathcal{A}(\mathcal{M})$.
\end{example}

 \begin{example}[Energy-momentum tensor] \label{eg:Energy-momentum tensor} With the previous example in mind, it follows that the energy momentum tensor $T_{\mu\nu}(f)$, formally given by
   \begin{align}
     \label{eq: stress energy tensor def}
    T_{\mu\nu}(f) =   (1-2\xi) &\nabla_\mu\Psi(f)\nabla_{\nu}\Psi(f) -2\xi\Psi(f)\nabla_\mu\nabla_{\nu}\Psi(f) -\xi G_{\mu\nu}\Psi^2(f) +\nonumber\\ & \quad + g_{\mu\nu}\left[  2\xi \Psi^2(f) + \left(2\xi-\frac{1}{2}  \right) \nabla^\rho\Psi(f)\nabla_{\rho}\Psi(f) -\frac{1}{2}m_0^2\Psi^2(f)\right],
  \end{align}
  is also ill-defined and it is not an element of the CCR algebra $\mathcal{A}(\mathcal{M})$.
 \end{example}

 To extend the algebra of observables and obtain finite quantum fluctuations of all physical observables, we must implement a regularization procedure. Within the Fock space formalism on Minkowski spacetime this is accomplished by normal-ordering, which is tantamount to ``subtracting the vacuum energy'' \cite{birrell1984quantum,LecturesQEI_Fewster,fulling1989aspects}. In the past decades, a long debate regarding how to implement normal-ordering on curved spacetimes engendered the notion of Hadamard states. They play the role of Minkowski vacuum state, but on general spacetimes, having a specific singularity structure that coincides with that of Minkowski vacuum. The idea is that, even though a general spacetime lacks a preferred, reference, no-particle state to disambiguate normal-ordering, in the distinguished class of Hadamard states, we can perform covariant ``divergence-subtractions'' and obtain well-defined quantum fluctuations of all physical observables \cite{fewster2013necessity}. In addition, though we shall not enter into the details here, it is worth mentioning that the introduction of the concept of Hadamard states allows for a completely covariant (perturbative) construction of the algebra of Wick polynomials of an interacting quantum field theory, see e.g. \cite{brunetti2015advances,hollands2001local}. For a proper definition of Hadamard states and an example of a regularized observable, we need some preliminary notions, as follows.

Let $\mathcal{M}$ be a spacetime that admits a well-defined Cauchy time function $t$, such as a static spacetime or a globally hyperbolic one \cite[Prop.4]{bernal2003smooth}. On any geodesically convex neighbourhood $\mathcal{O}\subset\mathcal{M}$, we define the regularized half squared geodesic distance as
  \begin{equation*}
    \label{eq: reg synge function}
    \sigma_{\varepsilon}(x,x') := \sigma(x,x') + 2 i \varepsilon (t(x)-t(x')) + \varepsilon^2,
 \end{equation*}
where $x,x'\in\mathcal{M}$, $\varepsilon>0$, while $\sigma(x,x')$ is the standard half squared geodesic distance (or Synge world function). In addition, for $x,x'\in\mathcal{O}$ we define the modified, regularized Hadamard parametrix by
 \begin{equation}
   \label{eq:Parametrix}
h_{\varepsilon}(x,x'):= \frac{u(x,x')}{[\sigma_{\varepsilon}(x,x')]^{\frac{n}{2}-1} }+ \beta_n v(x,x')\ln\left(\frac{ \sigma_{\varepsilon}(x,x')}{\lambda^2}\right),
 \end{equation}
where $\beta_n:=\delta_{0,\text{ mod}(n,2)}$, and $\lambda>0$ is a reference length scale (see \cite[Pg.123]{hack2010backreaction} for considerations on the dependence on $\lambda$). The functions $u$ and $v$ are called Hadamard coefficients; they are smooth and uniquely determined by imposing that $h_{\varepsilon}(x,x^{\prime})$ satisfies the Klein-Gordon equation both at $x$ and $x^{\prime}$, see \cite{hack2012stress,kay1991theorems,moretti2003comments}. Under the assumptions just given, the following definition makes sense.
 \begin{definition}[Hadamard states] \label{def: Hadamard state, local Hadamard form} A state $\psi$ on $\mathcal{A}(\mathcal{M})$ is called \emph{Hadamard} if its two-point function $\psi_2\in\mathcal{D}'(\mathcal{M}^2)$ satisfies the \textit{local Hadamard condition}, i.e. if for all $x\in\mathcal{M}$ and for all $x'$ lying in any geodesically convex neighbourhood $\mathcal{O}$ of $x$,
 \begin{equation}
   \label{eq:local hadamard form}
 \psi_{2}(x,x^{\prime})= h_{0^+}(x,x') + w(x,x'),
 \end{equation}
  where $w(x,x')$ is a smooth function on $\mathcal{O}\times\mathcal{O}$. Equivalently, when $\psi_2$ satisfies the local Hadamard condition, we say it is of \textit{local Hadamard form}. Accordingly, $w(x,x')$ is the state dependent component of a two-point function of local Hadamard form.
 \end{definition}

\begin{example}[Regularized vacuum fluctuations] \label{eg: Regularized vacuum fluctuations}
  Consider the assumptions and definitions of Example \ref{eg: vacuum fluctuations}. Since the difference of two-point functions of local Hadamard form is a smooth function, then for a Gaussian, Hadamard state $\widetilde{\psi}$ it makes sense to define the regularized fluctuations, dubbed $\widetilde{\psi}(:\Psi^2(f):)$, as
  \begin{equation}
      \widetilde{\psi}(:\Psi^2(f):):= \lim\limits_{\varepsilon\rightarrow 0^+} \int\limits_{\mathcal{M}} d^4x\,   \widetilde{\psi}(:\Psi^2(x):)f(x),
  \end{equation}
  where
  \begin{equation}
    \label{eq: divergence subtraction eifjwojef}
    \widetilde{\psi}(:\Psi^2(x):) :=  \lim\limits_{x'\rightarrow x } \left\{\widetilde{\psi}_2(x,x') - \psi_2(x,x') \right\}.
  \end{equation}
  As mentioned before, the  ``divergence-subtraction'' form of implementing a regularization as in Equation \eqref{eq: divergence subtraction eifjwojef} is tantamount to the usual normal-ordering procedure performed on Minkowski spacetime, as explicitly specified in \cite[Pg.6]{LecturesQEI_Fewster}.
\end{example}

\begin{example}[Regularized energy-momentum tensor] \label{eg: Regularized energy-momentum tensor} As mentioned in Example \ref{eg:Energy-momentum tensor}, the energy-momentum tensor given in Equation \eqref{eq: stress energy tensor def} is ill-defined. However, for a Gaussian, Hadamard state with two-point function $\psi_2$, analogously to Example \ref{eg: Regularized vacuum fluctuations}, the standard regularized energy-momentum tensor
     \begin{equation}
       \label{eq: regularized stress energy tensor def}
         \braket{:T_{\mu\nu}(x):} := \lim\limits_{x'\rightarrow x}\left\{\mathcal{D}_{\mu\nu}(x,x')\left[\psi_2(x,x') - h_{0^+}(x,x')\right]\right\},
     \end{equation}
     with differential operator $\mathcal{D}_{\mu\nu}(x,x')$ given by
     \begin{align}
       \label{eq: differential operator stress energy tensor moretti}
       &  \mathcal{D}_{\mu\nu}(x,x') := (1-2\xi) g_{\nu}\hspace{.5pt}^{\nu'}(x,x')\nabla_\mu\nabla_{\nu'} -2\xi\nabla_\mu\nabla_{\nu} + G_{\mu\nu} +\nonumber\\ & \quad + g_{\mu\nu}\left[  2\xi \Box + \left(2\xi-\frac{1}{2}  \right) g_{\rho}\hspace{.5pt}^{\rho'}(x,x') \nabla^\rho\nabla_{\rho'} -\frac{1}{2}m_0^2\right],
     \end{align}
  is well-defined and finite, see e.g. \cite{decanini2008hadamard,hack2012stress}.
\end{example}

The technique of using a Hadamard parametrix, as per Equation \eqref{eq:Parametrix}, when dealing with linear partial differential equations was first introduced by Hadamard  \cite{Hadamard1923}, widely studied by Riesz \cite{riesz1949integrale}, and eventually brought to the context of quantum field theory, see \cite{kay1991theorems}. In Definition \ref{def: Hadamard state, local Hadamard form}, we introduced the \textit{local} version of the Hadamard condition. However, this notion is more intricate than the brief discussion above suggests. Precisely, as the nomenclature indicates, there is also a \textit{global Hadamard condition}, introduced in \cite{kay1991theorems}. Roughly speaking, to be of global Hadamard form a two-point function must not possess further singularities, other then the lightlike ones of the local Hadamard form, even for arbitrarily spacelike separated points. Almost thirty years ago, Radzikowski \cite{radzikowski1996micro} introduced another associated notion called \textit{microlocal spectrum condition}, which is a local energy positivity condition stated in the context of microlocal analysis and that restricts the wavefront set of the two-point function to coincide with that of Minkowski vacuum state. In addition, he showed that on globally hyperbolic spacetimes the microlocal spectrum condition is equivalent to both the local and global Hadamard conditions. The upshot of his work is that it allowed for the employment of techniques from microlocal analysis in the formalism of quantum field theory, for example, in the construction of explicit examples of Hadamard states on Schwarzschild and on cosmological spacetimes \cite{dappiaggi2011rigorous,dappiaggi2011approximate}.

It has been shown that the Hadamard condition ``propagates'' from one Cauchy surface to another, guaranteeing the existence of global, Gaussian, Hadamard states on any globally hyperbolic spacetime \cite{brunetti2015advances,fulling1981singularity,radzikowski1996micro,sahlmann2000passivity,WaldQFTCS}. However, on non-globally hyperbolic spacetimes the results above-mentioned do not hold in general: the local and global Hadamard conditions and the microlocal spectrum condition are not equivalent. In addition, there is not a unique recipe for the construction of Hadamard states on general spacetimes, see e.g.  \cite{brum2014hadamard,dappiaggi2016constructing,dappiaggi2017hadamard,gerard2014construction,gerard2016construction}. Be that as it may, even though we consider non-globally hyperbolic spacetimes in this thesis, they are all static. In this case, there is a standard prescription for the construction of physically-sensible states \cite{bussola2017ground,dappiaggi2018ground,dappiaggi2016hadamard,dappiaggi2019fundamental,dappiaggi2018mode} such that the local Hadamard condition specifies Hadamard states in every globally hyperbolic subregion of the spacetime, as shown in \cite{sahlmann2000passivity}.

In Section \ref{sec: Physically-sensible states in Schwarzschild-like coordinates}, I show how to obtain ground and thermal states that are both Gaussian and Hadamard on static spacetimes that admit Schwarzschild-like coordinates. Before that, in the next section, I define and describe how to choose a physically-sensible dynamics on static, stably-causal, not necessarily globally hyperbolic spacetimes.

 \section{Physically-sensible dynamics on static spacetimes}
 \label{sec: Physically-sensible dynamics on static spacetimes}
 \enlargethispage{1\baselineskip}
 In this section, I give a succinct account of the results obtained by Ishibashi and Wald in \cite{ishibashi2003dynamics,wald1980dynamics}, where they showed that, even though the dynamics is not unique, on static, stably-causal, non-globally hyperbolic spacetimes, there is a prescription to obtain a physically-sensible dynamics. Specifically, if we impose boundary conditions that yield self-adjoint extensions of the spatial part of the Klein-Gordon operator, then one can associate to the Klein-Gordon equation a well-posed Cauchy-boundary value problem, advanced and retarded fundamental solutions exist, and there is a naturally well-defined, positive and conserved, energy functional for a Klein-Gordon field. First, in Section \ref{subsec: A positive and conserved energy}, I collect their results in a single theorem. Then, in Section \ref{subsec: Obtaining the self-adjoint extensions}, I show how to obtain such self-adjoint extensions on backgrounds that admit Schwarzschild-like coordinates. %

\subsection{A positive and conserved energy}
\label{subsec: A positive and conserved energy}
Let $\mathcal{M}$ be a static, stably-causal, not necessarily globally hyperbolic spacetime with metric tensor $g$ associated with the line element given in Equation \eqref{eq: metric static spacetimes}, global timelike Killing field $\partial_t$, foliation $\{t\}\times \Sigma$ such that each point in $\mathcal{M}$ is labelled $(t,\underbar{x})$ and with $\mathbf{n}$ being the future-directed, unit vector field normal to $\Sigma$. Let $\Psi$ be a free, scalar field with mass $m_0$ and coupled, through the parameter $\xi$, to the scalar curvature $\mathbf{R}$, as per Equation \eqref{eq: KG}. The Klein-Gordon equation reads
\begin{equation}
\label{eq: KG static spacetimes tilde}
  \left(\frac{\partial^2}{\partial t^2} + A \right)\Psi = 0,
\end{equation}
with
\begin{equation*}
\label{eq: KG static spacetimes operator A}
 A := f(\underbar{x}) \left( \frac{1}{\sqrt{|g|}}\partial_i(\sqrt{|g|} h^{ij}(\underbar{x})\partial_j) + m_{0}^2 +\xi \mathbf{R}\right),
\end{equation*}
where the functions $f(\underbar{x})$ and $h^{ij}(\underbar{x})$ are that of Equation \eqref{eq: metric static spacetimes}. Let $d\Sigma$ be the natural volume on $\Sigma$ induced by the metric tensor $g$. The operator $A$ is positive and symmetric on $\mathcal{H}:=L^2(\Sigma,f^{-1}(\underbar{x}) d\Sigma )$ if we take its domain to be $C_0^\infty(\Sigma)$, see \cite{wald1980dynamics}. In this case, $A$ has at least one positive, self-adjoint extension, namely the Friedrichs extension \cite[Pg.193]{zettl2005sturm}. It is worth mentioning that if $\mathcal{M}$ is furthermore globally hyperbolic, then the operator $A$ is essentially self-adjoint on $C_0^\infty(\Sigma)$ \cite{kay1978linear,brum2013vacuum}, i.e. it has a unique self-adjoint extension.

\begin{theorem}[Physically-sensible dynamics]
\label{thm: physically-sensible dynamics} In the above setting, let $A_E$ be any positive, self-adjoint extension of $A$, with domain $\text{\em Dom}(A_E)\subset \mathcal{H} $. For $t\in\mathbb{R}$ and $(\Psi_0,\dot\Psi_0)\in (C_0^\infty(\mathcal{M})\times C_0^\infty(\mathcal{M}))$, define
 \begin{equation}
   \label{eq: sol Psi_t dynamics A_E}
   \Psi_t:=\cos(A_E^{1/2}t)\Psi_0+A_E^{-1/2}\sin(A_E^{1/2}t)\dot\Psi_0\in\mathcal{H}.
 \end{equation}
Let $ K_0 := \text{\em supp}(\Psi_0)\cup\text{ \em supp}(\dot\Psi_0)$ and define the time translation $\widehat{T}_t$ and time reflection $\widehat{R}_t $ operators acting on smooth functions $F:\mathcal{M}\rightarrow \mathcal{M}$ by
 \begin{align}
   &\widehat{T}_t F(s,\underbar{x}) = F(s-t,\underbar{x}), \\
   &\widehat{R}_t F(t,\underbar{x}) = F(-t,\underbar{x}).
 \end{align}
\noindent Then, there exists a unique solution $\Psi\in C^\infty(\mathcal{M})$ such that the following properties hold:
\begin{itemize}
  \item[]\normalfont{i. (Klein-Gordon equation)} $P\Psi = 0 \text{ on } \mathcal{M};$
  \item[]\normalfont{ii. (Agreement within $D(\Sigma)$)} $\Psi|_{\{t\}\times\Sigma} = \Psi_t   \text{ and } \nabla_{\mathbf{n}}\Psi|_{\{t_0\}\times\Sigma} = \partial_t\Psi_t;$
 \item[]\normalfont{iii. (Causality)} $\normalfont\text{supp}(\Psi) \subset J^+(K_0)\cup J^-(K_0);$
 \item[]\normalfont{iv. (Time-translation invariance)} if $\Psi$ is the solution associated with the initial data $(t_0,\Psi_0,\dot\Psi_0)$, then $\widehat{T}_{-t}\Psi$ is the one associated with $(t+t_0,\Psi_t,\dot\Psi_t)$;
 \item[]\normalfont{v. (Time-reflection invariance)} if $\Psi$ is the solution associated with the initial data $(t_0,\Psi_0,\dot\Psi_0)$, then $\widehat{R}_{t}\Psi$ is the one associated with $(-t_0,\Psi_0,-\dot\Psi_0)$.\\
\end{itemize}
Moreover, for solutions $\Psi,\Phi\in C_0^\infty(\mathcal{M})$ with Cauchy data given by, respectively, $(t_0,\Psi_0,\dot\Psi_0)$ and $(t_0,\Phi_0,\dot\Phi_0)$, define the energy functional
  \begin{equation}
  \label{eq: energy functional with A_E two diff sols}
    \mathcal{E}[\Psi,\Phi] := \braket{\dot\Psi_0,\dot\Phi_0}_{L^2} + \braket{\Psi_0,A_E\Phi_0}_{L^2}.
   \end{equation}
Let $\mathcal{W}$ be the space of solutions $\Psi$, given by the prescription above, with initial data in $C_0^\infty(\Sigma)\times C_0^\infty(\Sigma)$. Then, $\mathcal{E}$, as per Equation (\ref{eq: energy functional with A_E two diff sols}) defines an inner product
\begin{equation}
  \label{eq:inner product energy owiejfiowjf}
  \mathcal{E}:\mathcal{V}\times\mathcal{V}\rightarrow \mathbb{R}
\end{equation}
on the vector space $\mathcal{V}$ of solutions to the Klein-Gordon equation that can be written as finite linear combinations of $\widehat{T}_{t}\Phi$ for $\Phi \in \mathcal{W}$. It follows that Equation (\ref{eq: sol Psi_t dynamics A_E}) prescribes the unique dynamics for which it holds that, $\forall \Psi,\Phi\in\mathcal{V}$, and $\forall t\in\mathbb{R}$
\begin{itemize}
  \item[]\normalfont{vi. (Positivity)} $\mathcal{E}[\Psi,\Psi] \geq 0$;
  \item[]\normalfont{vii. (Time-translation invariance)} $\mathcal{E}[\widehat{T}_t \Psi,\widehat{T}_t \Phi] = \mathcal{E}[\Psi, \Phi]$;
  \item[]\normalfont{viii. (Time-reflection invariance)} $\mathcal{E}[\widehat{R}_t \Psi,\widehat{R}_t \Phi] = \mathcal{E}[\Psi, \Phi]$;
  \item[]\normalfont{ix. (Agreement with the globally hyperbolic case)} if $\mathcal{M}$ is globally hyperbolic, then $\mathcal{W}= \mathcal{V}$ and $\mathcal{E}[\Psi,\Phi] = \braket{\dot\Psi_0,\dot\Phi_0}_{L^2} + \braket{\Psi_0,A \Phi_0}_{L^2}$.
  \item[]\normalfont{x. (Compatibility of convergence with respect to the norm induced by $\mathcal{E}$)} if $\{\Psi_n\}$ in $\mathcal{V}$ is a Cauchy sequence (with respect to the Hilbert space norm specified by the inner product of Equation \eqref{eq:inner product energy owiejfiowjf}) such that its initial data $((\Psi_n)_0,(\dot{\Psi}_n)_0)$ converge uniformly on compact subsets of $\Sigma$ to the initial data of $\Psi$, and analogously to all their spatial derivatives, then it converges in norm to $\Psi$: $\lim\limits_{n\rightarrow \infty} \mathcal{E}[\Psi_n-\Psi,\Psi_n-\Psi]=0$.
\end{itemize}
\end{theorem}
We say that a dynamics determined by Equation \eqref{eq: sol Psi_t dynamics A_E} with an energy functional defined by Equation \eqref{eq: energy functional with A_E two diff sols}, such that properties i) to x) hold, is \textit{physically-sensible}. Accordingly, Theorem \ref{thm: physically-sensible dynamics} grants us a prescription for finding physically-sensible dynamics by finding the self-adjoint extensions of the spatial part of the Klein-Gordon operator. When there is more than one positive self-adjoint extension, then there are different, inequivalent, physically-sensible dynamics to be considered. This ambiguity cannot be resolved, in general, without an experiment or further assumptions regarding the physical scenario in consideration. An extra ambiguity emerge when a positive, self-adjoint extension have a discrete zero eigenvalue, i.e. when ``zero-modes'' are present \cite[Ch.8]{rajaraman1982solitons}. Since the latter does not occur in the cases treated in Chapter \ref{chap: Applications}, it is not taken into account here. In addition, non-positive, self-adjoint extensions may also exist, but they are ruled out by an instability argument: generic solutions exhibit exponential growth in time. One can refer to \cite{ishibashi2003dynamics,ishibashi2004dynamics,wald1980dynamics} for a discussion on the comments above and for proofs concerning the theorem itself. Also, as mentioned in the previous section, a discussion on this topic is given in \cite{dappiaggi2019fundamental} for a globally hyperbolic spacetime with timelike boundary.

\subsection{Obtaining the self-adjoint extensions}
\label{subsec: Obtaining the self-adjoint extensions}

Turning to singular Sturm-Liouville theory, in this section we obtain self-adjoint extensions for the spatial part of the Klein-Gordon operator. In the remainder of this chapter, the underlying background $\mathcal{M}$ is taken to be an $n$-dimensional, static, not necessarily globally hyperbolic spacetime that admits Schwarzschild-like coordinates, as defined in Chapter \ref{chap: The Spacetimes}.

We start by writing the Klein-Gordon equation in Schwarzschild-like coordinates $(t,r,\underline{\theta})$, as in Definition \ref{def: Schwarzschild-like coordinates}. Let $\Delta_{j}$ be the $(n-2)$-dimensional Laplacian on the hypersurfaces $\Sigma_{j}^{n-2}$, then the Klein-Gordon equation \eqref{eq: KG} reads
   \begin{equation}
     \label{eq: KG in Schd-like coordinates}
     \left\{-f(r)^{-1}\partial_t^2  + \frac{1}{\sqrt{f(r)h(r)}r^{n-2} } \partial_r (\sqrt{f(r)h(r)}r^{n-2}h^{-1}(r)  \partial_r)+  \frac{1}{r^2}\Delta_{j} -m_0^2-\xi \mathbf{R}\right\}\Psi=0.
   \end{equation}
Let $Y_j(\underline{\theta})$ be the eigenfunctions of $\Delta_{j}$ with eigenvalue $\lambda_j$ \cite{jorgensen1987harmonic}. Note that the specific value of $j$, and hence the cardinality of the set of eigenvalues $\lambda_j$, does not play a role as far as Equation \eqref{eq: KG in Schd-like coordinates} is concerned. In addition, since $\partial_t$ is a Killing vector field, we can attempt to find solutions of Equation \eqref{eq: KG in Schd-like coordinates} by taking the Fourier transform with respect to $t$; equivalently, let us assume its solutions have a harmonic time dependence:
 \begin{equation}
   \label{eq: ansatz KG in mode decomposition}
   \Psi(t,r,\underline{\theta})=e^{-i\omega t}R(r)Y_j(\underline{\theta}).
 \end{equation}
Substituting Equation \eqref{eq: ansatz KG in mode decomposition} into Equation \eqref{eq: KG in Schd-like coordinates} yields the \textit{radial equation}:
\begin{equation}
 \label{eq: the radial equation Shd-like coord}
 R''(r) +   \left[ \frac{1}{2} \left(   \frac{f'(r)}{f(r)} - \frac{h'(r)}{h(r)} \right) + \frac{n-2}{r}  \right] R'(r)+ \left[\left(\frac{\omega ^2}{f(r)}+\frac{\lambda_j }{r^2} -m_0^2 - \xi \mathbf{R}\right) h(r) \right] R(r)= 0 ,
\end{equation}
where the coordinate $r$ is defined in the interval $(a,b)=:I\in\mathbb{R}$, for $ -\infty \leq a < b \leq +\infty$, in accordance with Definition \ref{def: Schwarzschild-like coordinates}. For a fixed $\lambda_j$, Equation \eqref{eq: the radial equation Shd-like coord} specifies a radial operator:
\begin{equation}
 \label{eq: radial operator Shd-like coord}
  A_{\lambda_j} := \frac{d^2}{dr^2} +   \left[ \frac{1}{2} \left(   \frac{f'(r)}{f(r)} - \frac{h'(r)}{h(r)} \right) + \frac{n-2}{r}  \right]  \frac{d}{dr}+ \left[\left(\frac{\omega ^2}{f(r)}+\frac{\lambda_j }{r^2} -m_0^2 - \xi \mathbf{R}\right) h(r) \right].
\end{equation}

Next, we recast the radial equation as a Sturm-Liouville problem. Let
\begin{equation}
  \label{eq: p2 eigenvalues for general SL schd-like}
  p^2:=\begin{cases}
          \omega^2, \text{ if }f(r) \neq 1, \\
          \omega^2 - m_{\text{eff}}^2, \text{ if }f(r) = 1 \text{ and }  \mathbf{R} \text{ is constant}, \\
          \omega^2 - m_0^2, \text{ otherwise}.
    \end{cases}
\end{equation}
\begin{subequations}
  \label{eq: def P, Q S of sturm liouville op}
Equation \eqref{eq: the radial equation Shd-like coord} in Sturm-Liouville form corresponds to the eigenvalue problem
\begin{equation}
 \label{eq: Sturm-Liouville radial eq}
  L_{p^2} R(r) = p^2 R(r),
\end{equation}
with associated Sturm-Liouville operator given by
\begin{equation}
  \label{eq: Sturm-Liouville operator L_p2}
  L_{p^2}:=\frac{1}{S(r)} \left(- \frac{d}{d r} \left(P(r)\frac{d}{d r}\right) + Q(r)\right).
\end{equation}
Since $ L_{p^2} = A_{\lambda_j} + p^2$, we take its domain to coincide with that of $A_{\lambda_j}$. Particularly for $p^2 = \omega^2$, the auxiliary functions are
\begin{align}
  & P(r) := r^{n-2 }\sqrt{\frac{f(r)}{h(r)}} =: \frac{r^{2n-4}}{S(r)},\\
  & Q(r) := r^{n-4}(\lambda_j - r^2 m_{\text{eff}}^2 )\sqrt{f(r)h(r)}.
\end{align}
\end{subequations}
The functions $Q$ and $S$ for the other possibilities of $p^2$ as given in Equation \eqref{eq: p2 eigenvalues for general SL schd-like} can be directly obtained from Equation \eqref{eq: def P, Q S of sturm liouville op}.

Accordingly, we view $A_{\lambda_j}$ as an operator acting on the Hilbert space of square-integrable functions with respect to the measure induced by the function $S$, i.e. $L^2(I, S(r)dr)$ with scalar product and $L^2$-norm given, respectively, by
\begin{equation*}
  \braket{\psi,\phi} :=  \int_a^b \psi(r)\phi(r)^*S(r) dr \quad \text{ and }\quad ||\psi||_{L^2} := \braket{\psi,\psi}
  , \text{ }\quad \forall \psi,\phi\in L^2(I, S(r)dr).
\end{equation*}
If we take $C_0^\infty(I)$ to be the domain of $A_{\lambda_j}$, then it is a symmetric operator. However, in general, $A_{\lambda_j}$ is neither positive, bounded nor self-adjoint. The latter is due to the fact that, if $A_{\lambda_j}^*$ is the adjoint of $A_{\lambda_j}$, then its domain $D(A_{\lambda_j}^*) \subset L^2(I, S(r)dr)$ is, in general, much larger then that of $A_{\lambda_j}$.

In the following, we shall assume $A_{\lambda_j}$ is positive and symmetric. To obtain its self-adjoint extensions, we invoke the works of Weyl \cite{weyl1910gewohnliche} and Von Neumann \cite{neumann1930allgemeine}, as extensively discussed in the literature \cite{reed1975ii,zettl2005sturm}. The idea is that self-adjoint extensions of $A_{\lambda_j}$ relate to square-integrability conditions of its associated solutions at the endpoints $a$ and $b$, which are in correspondence with the admissible boundary conditions at these endpoints. In turn, the boundary conditions can be explicitly stated in terms of the principal and secondary solutions of the Sturm-Liouville problem. To understand these relations, let us introduce Weyl's endpoint classification and the concepts of principal and secondary solutions.

Let $L(J,\mathbb{C})$ be the set of Lebesgue integrable complex valued functions defined almost everywhere on a Lebesgue measurable subset $J$ of $\mathbb{R}$. Weyl's endpoint classification provides the following nomenclature  \cite[Ch.7]{zettl2005sturm}.

\begin{definition}[Weyl's endpoint classification] \label{def: Weyl's endpoint classification} Consider the Sturm-Liouville problem as per Equation \eqref{eq: def P, Q S of sturm liouville op} with $r\in(a,b)$. The left endpoint $a$ is
  \begin{enumerate}
    \item[i)] \textit{regular} if $1/P,Q,S\in L((a,r_0),\mathbb{C})$ for some (and hence any) $r_0\in I$;
    \item[ii)] \textit{singular} if it is not regular;
    \item[iii)] \textit{limit circle} if all solutions of Equation \eqref{eq: def P, Q S of sturm liouville op} are in $L^2((a,r_0),S(r)dr),\, \forall r_0\in I$;
    \item[iv)] \textit{limit point} if it is not limit circle.
  \end{enumerate}
An analogous definition holds for the right endpoint $b$, considering instead $( L((r_0,b),\mathbb{C}))$ in items i) and iii).
\end{definition}

\begin{definition}[Principal and secondary solutions] \label{def: principal and secondary solutions} Let $u$ and $v$ be solutions of the Sturm-Liouville problem as per Equation \eqref{eq: def P, Q S of sturm liouville op} that are non-vanishing for $r\in (a, r_0)\subset I$. Then
  \begin{enumerate}
   \item $u$ is called a \textit{principal solution at $a$} if, for any solution $y$ that is linearly-independent from $u$,
          \begin{equation}
            u(r) \xrightarrow{r\rightarrow a}o(y(r)).
          \end{equation}
    \item $v$ is called a \textit{secondary (or non-principal) solution at $a$} if it is not a principal solution.
  \end{enumerate}
Again, an analogous definition holds for a right endpoint $b$.
\end{definition}

 In all scenarios considered in the next chapter, the radial equation has singular endpoints such that one is limit point and the other is limit circle. This shall be verified case-by-case by computing, asymptotically at each endpoint, the $L^2$-norms of the radial solutions, according to Definition \ref{def: Weyl's endpoint classification}. The next theorem summarizes the necessary results from singular Sturm-Liouville theory that allow us to identify the desired self-adjoint extensions in the case of interest to this thesis, i.e. when one endpoint is limit point and the other is limit circle. For an extensive discussion in this direction, check \cite[Ch.10]{zettl2005sturm}.

\begin{theorem}[Self-adjoint extensions] \label{thm: self-adjoint extensions}
  Consider the Sturm-Liouville problem as per Equation (\ref{eq: def P, Q S of sturm liouville op}), $I=(a,b)$ and $\mathcal{H}:= L^2(I, S(r)dr) $. Suppose $a$ and $b$ are singular endpoints such that $a$ is limit point and $b$ is limit circle. Let $[.,.]$ be the Lagrange sesquilinear form, i.e. for $y,z\in \mathcal{H}$, $r_0\in I$:
  \begin{equation}
    [y,z](r_0) := \lim\limits_{r\rightarrow r_0} \left\{P(r) \left[y(r) \overline{z'(r)} -   y'(r) \overline{z(r) } \right]\right\}.
  \end{equation}
  Let $u$ and $v$ be, respectively, the principal and secondary solutions normalized to $[u,v](r)=1$. In addition, for $y\in \mathcal{H}$, consider boundary conditions of the form
  \begin{align}
      \label{eq: Robin bc with B1 and B2}
          B_1 [y,u](b) + B_2 [y,v](b) = 0,
    \end{align}
    where $B_1$ and $B_2$ are real-valued constants such that $(B_1, B_2) \neq (0,0)$.
    The boundary conditions as per Equation (\ref{eq: Robin bc with B1 and B2}) yields all self-adjoint extensions of the radial operator, as follows. For $A_{\lambda_j}$ as in Equation (\ref{eq: radial operator Shd-like coord}), let $A_{\lambda_j}^*$ be its adjoint with domain $D(A_{\lambda_j}^*)\subset \mathcal{H}$. Then, the operator
    \begin{align*}
            A_{\lambda_j,E} \,y   := A_{\lambda_j}^* y, \quad y\in D(A_{\lambda_j,E}):=  \left\{y\in \mathcal{H}:\, y \text{ satisfies Equation (\ref{eq: Robin bc with B1 and B2})} \right\}
    \end{align*}
    is a self-adjoint extension of $A_{\lambda_j}$.
\end{theorem}

\begin{definition}[(Generalized) Robin boundary conditions] \label{def: generalized Robin boundary conditions} Let $y$, $u$ and $v$ be as in the theorem above. If $y$ satisfies Equation \eqref{eq: Robin bc with B1 and B2}, then we say $y$ satisfies a (generalized) Robin boundary condition parametrized by $B_1$ and $B_2$. Moreover, we can equivalently rewrite it in terms of a real-valued parameter $\gamma$ as
    \begin{align}
        \label{eq: Robin bc with gamma}
            \cos(\gamma)[y,u](b) + \sin(\gamma)[y,v](b) = 0
      \end{align}
and say $y$ satisfies a (generalized) Robin boundary condition parametrized by $\gamma$.
\end{definition}
The nomenclature ``generalized'' is justified by the fact that if the endpoint $b$ is regular, then Equation \eqref{eq: Robin bc with gamma} reduces to the standard form of Robin boundary conditions, as per Equation (7.1.6) in \cite[Pg.410]{greenBook}:
\begin{align}
    \label{eq: standard form of Robin bc}
        y(b) + \beta y'(b) = 0,  \text{ where } \beta=-\frac{\cos(\gamma)u(b) + \sin(\gamma)v(b)}{\cos(\gamma)u'(b) + \sin(\gamma)v'(b)}\in\mathbb{R}.
  \end{align}
If $\beta=0$ in Equation \eqref{eq: standard form of Robin bc}, then $y(b)=0$ and we say $y$ satisfies the Dirichlet boundary condition at $b$. If $\beta\rightarrow\infty$, then $y'(b)=0$ and we say $y$ satisfies Neumann boundary condition at $b$. In addition, if we take $u$ and $v$ such that
\begin{equation*}
\begin{aligned}
u(b) &= 0, \\
u'(b) &= c,
\end{aligned}
\hspace{.5cm}
\begin{aligned}
v(b) &= -c, \\
v'(b) &= 0,
\end{aligned}
\end{equation*}
where $c$ is a constant, then Equation \eqref{eq: standard form of Robin bc} yields $\beta = \tan(\gamma)$. That is, the boundary condition is independent of any other parameters of the problem.
\begin{remark}
A solution of the form $y = \cos(\gamma)u + \sin(\gamma)v$ satisfies Equation \eqref{eq: Robin bc with gamma}. Thus, the parameters $\gamma=0$ and $\gamma=\frac{\pi}{2}$ select, respectively, the principal and the secondary solutions. Taking into account the Definition \ref{def: principal and secondary solutions} and the nomenclature associated to the regular case, we shall refer to $\gamma=0$ and $\gamma=\frac{\pi}{2}$,
respectively,
as (generalized) Dirichlet and Neumann boundary conditions. However, in the singular case there is no preferred secondary solution. For example, if $v$ is a secondary solution, then $v + u$ is also a secondary solution. This means that what is called Neumann boundary condition in the singular case is ambiguous. Nevertheless, the nomenclature is useful.
\end{remark}

When the endpoints of the radial equation consists of one limit point and one limit circle, there is a one-parameter family of (generalized) Robin boundary conditions at the limit circle endpoint, which specify self-adjoint extensions of $A_{\lambda_j}$. Each one of these (acting on $L^2(I, S(r)dr)$), specifies a self-adjoint extension of the spatial part $A$ of the Klein-Gordon operator (acting on $ L^2(\Sigma, f^{-1}(\underbar{x})d\Sigma)$), as in Equation \eqref{eq: KG static spacetimes tilde}. This family of extensions, in turn, corresponds to a one-parameter family of inequivalent physically-sensible dynamics, by Theorem \ref{thm: physically-sensible dynamics}.

It is crucial to note that the class of boundary conditions associated with the self-adjoint extensions of the spatial part of the Klein-Gordon operator, although encompassing infinitely more possibilities than simply choosing Dirichlet boundary condition and the Friedrichs extension, still does not correspond to all possible physically-sensible dynamics one can obtain on a static, stably-causal, not necessarily globally hyperbolic spacetime. It excludes, for example, time dependent boundary conditions such as those of Wentzell type \cite{feller1957generalized,ueno1973wave}. This is not just a technical detail, since the class of Wentzell boundary conditions also describes physical systems, such as acoustic wave equations \cite{gal2003oscillatory}. They are also of interest within quantum field theory and work in the direction of generalizing the approach described here have been investigated in the last years \cite{zahn2018generalized,dappiaggi2018mode,dappiaggi2019fundamental}.

    \section{Physically-sensible states in Schwarzschild-like coordinates}
    \label{sec: Physically-sensible states in Schwarzschild-like coordinates}

    Let us construct the two-point functions of two physically-sensible states for a real, massive, free, scalar quantum field theory: a ground state and a KMS state; both Gaussian, as per Definition \ref{def: gaussian state} and of locally Hadamard form, as per Definition \ref{def: Hadamard state, local Hadamard form}. We start by constructing the causal propagator. Then, in Sections \ref{sec: Ground states} and \ref{sec: KMS states} we define and study each state separately. Again, assume the background spacetime $\mathcal{M}$ is static and admits Schwarzschild-like coordinates.

    Analogously to the ansatz of Equation \eqref{eq: ansatz KG in mode decomposition} for the solutions of the Klein-Gordon equation, and recalling that Schwarzschild-like coordinates are as in Definition \ref{def: Schwarzschild-like coordinates}, time translation invariance and the homogeneity of $\Sigma_{j}^{n-2}$ allow us to write the following ansatz for the integral kernel of the causal propagator
      \begin{equation}
        \label{eq:ansatz causal propagator}
        E(x,x')=\lim_{\varepsilon\rightarrow 0^+} \int_{\sigma(\Delta_{j})} d\eta_j\int_{\mathbb{R}} d\omega \frac{\sin(\omega (t-t'-i\varepsilon))}{\omega}\widetilde{E}(r,r') Y_{j}(\underline{\theta})\overline{Y_{j}(\underline{\theta}')}\text{,}
      \end{equation}
    where $Y_{j}$ are the eigenfunctions of the Laplacian $\Delta_{j}$, as before, while $d\eta_j$ is a suitable measure compatible with its spectrum $\sigma(\Delta_{j})$.
    In the particular case of $j=+1$, the underlying spacetime $\mathcal{M}$ is spherically-symmetric, $Y_{+1}$ are the spherical harmonics and the integral over $\sigma(\Delta_{+1})$ boils down to a discrete sum. In the next chapter, explicit examples involving the different possible values of $j$ are considered in detail. Here, focusing on the radial component $\widetilde{E}(r,r')$, the analysis is valid for all $j$. First, bearing in mind Proposition \ref{prop: causal propagator initial conditions}, let us impose the initial conditions satisfied by the causal propagator and restrict $\widetilde{E}(r,r')$.
    \begin{proposition}
     \label{prop: causal propagator initial conditions}
     A causal propagator with integral kernel given by Equation (\ref{eq:ansatz causal propagator}) satisfies Equation (\ref{eq:ccrs}) if and only if %
     \begin{equation}
          \label{eq:ccr2 implies}
            \int_{\mathbb{R}} d\omega\widetilde{E}(r,r^\prime) = \frac{\delta(r-r^\prime)}{ S(r)}.
        \end{equation}
    \end{proposition}
    \begin{proof}
      Equation \eqref{eq:ccr1} holds true by direct inspection. Equation \eqref{eq:ccr2} gives
     \begin{align}
      \label{eq:adefdfiosdjjd}
        \partial_t E(x,x^\prime)|_{t=t^\prime}=\int_{\sigma(\Delta_{j})} d\eta_j\int_{\mathbb{R}} d\omega \cos(\omega 0^+)\widetilde{E}(r,r')  Y_j(\underline{\theta})\overline{Y_j(\underline{\theta}')}= \frac{\delta(r-r^\prime)\delta(\underline{\theta}- \underline{\theta}')}{S(r)}.
     \end{align}
     In view of the completeness relation of the eigenfunctions of the Laplacian on space forms \cite{limic1966continuous,limic1967eigenfunction,raczka1966discrete},%
     \begin{equation*}
       \label{eq:  completeness relation}
       \int_{\sigma(\Delta_{j})} d\eta_j Y_{j}(\underline{\theta})\overline{Y_{j}(\underline{\theta}')} =  \delta(\underline{\theta}- \underline{\theta}'),
     \end{equation*}
    Equation \eqref{eq:adefdfiosdjjd} formally implies Equation \eqref{eq:ccr2 implies}.
    \end{proof}
     By constructing a Green function of the radial equation and employing spectral techniques, we obtain the radial part $\widetilde{E}(r,r^\prime)$ explicitly, as clarified in the following. Consider the radial equation in Sturm-Liouville form, as per Equation \eqref{eq: def P, Q S of sturm liouville op}, and let $ \mathcal{G}_p(r,r')$ be such that
      \begin{equation}
        \label{eq: radial green function solves Lp2-p2 eq}
        ((L_{p^2}-p^2) \otimes \mathbbm{1})\mathcal{G}_p(r,r') = (\mathbbm{1} \otimes (L_{p^2}-p^2))\mathcal{G}_p(r,r')  =\frac{\delta(r-r')}{S(r)}.
      \end{equation}
    Assume the radial coordinate lies in the interval $I=(a,b)$, where $a$ is limit point and $b$ is limit circle. Let $R_{a}$ and $R_{b}$ denote the most general square-integrable solutions respectively at the left and right endpoints. Define the normalization
      \begin{equation}
       \label{eq: definition normalization of radial green function}
      \mathcal{N}_p := P(r) W_r\left[ R_{a}(r),R_{b}(r)\right].
      \end{equation}
      Then, $ \mathcal{G}_p(r,r')$ can be written as \cite{gerlach2009linear,greenBook}
        \begin{align}
          \label{eq: green function radial equation}
        \mathcal{G}_{p}(r,r') =\frac{1}{ \mathcal{N}_p}\left\{ \Theta(r'-r) R_{a}(r)R_{b}(r') + \Theta(r-r')R_{a}(r')R_{b}(r)\right\}.
        \end{align}

    Let $\mathcal{C}^\infty$ be an asymptotically infinite contour on the $p^2$-complex plane encompassing all the eigenvalues of the Sturm-Liouville problem at hand, i.e. containing all the poles of $\mathcal{G}_p(r,r')$. The spectral resolution of the radial Green function reads, see \cite[Ch.7]{greenBook},
    \begin{equation}
        \label{eq:spectral resolution radial green function}
        \frac{1}{2\pi i}\oint_{\mathcal{C}^\infty} d(p^2) \mathcal{G}_p(r,r') = -\frac{\delta(r-r')}{S(r)}.
    \end{equation}
    By a direct comparison of Equation \eqref{eq:ccr2 implies} with Equation \eqref{eq:spectral resolution radial green function}, we obtain
        \begin{equation}
            \label{eq: identity radial part causal prop}
          \int_{\mathbb{R}}d\omega \widetilde{E}(r,r') =  -\frac{1}{2\pi i} \oint_{\mathcal{C}^\infty} d(p^2) \mathcal{G}_p(r,r').
        \end{equation}
      \begin{remark}[Bound states] \label{rem: bound states}
    Since $\mathcal{G}_{p}(r,r') $ can have poles and branches in the $p^2$-complex plane, the spectrum of the operator $L_{p^2}$ can consist partially of a point spectrum and partially of a continuous spectrum. When this is the case, the contour integral of Equations \eqref{eq:spectral resolution radial green function} and \eqref{eq: identity radial part causal prop} reduces to an integral over the real line plus a sum of residues at the poles of $\mathcal{G}_{p}(r,r') $. Each of these poles, say $\omega_p$, correspond to an exponentially divergent ``bound state'' mode solution of the Klein-Gordon equation, with time component $e^{ + |\Imag(\omega_p)| t}$. It has been shown that amongst the bound states, some modes are related to a superradiance phenomenon \cite{Dappiaggi2017pbe}. However, their physical meaning is not well-understood and lies outside the scope of this thesis. When they are present, we cannot construct the two-point function of the ground state directly from Equation \eqref{eq: identity radial part causal prop}. Notwithstanding, when such modes are absent there exists a tempered distribution $ \widehat{\mathcal{G}}_p(r,r') $ such that the following proposition holds.
  \end{remark}
\begin{proposition}
  \label{prop: radial part of causal propagator}
    Let $ \mathcal{G}_p(r,r') $ be the radial Green function given in Equation (\ref{eq: green function radial equation}) and satisfying Equation (\ref{eq: radial green function solves Lp2-p2 eq}), where $L_{p^2}$ is the associated Sturm-Liouville operator with eigenvalue $p^2$, as per Equation (\ref{eq: def P, Q S of sturm liouville op}). Assume  $L_{p^2}$ has only a continuous spectrum, i.e. assume $ \mathcal{G}_p(r,r') $ has a branch but no poles on the $p^2$-complex plane. Then, there exists a tempered distribution $ \widehat{\mathcal{G}}_p(r,r') $ such that
    \begin{equation*}
    \oint_{\mathcal{C}^\infty} d(p^2) \mathcal{G}_p(r,r')=  \int_{\mathbb{R}}d \omega \widehat{\mathcal{G}}_p(r,r') ,
  \end{equation*}
    and the radial part of the causal propagator, using Equation (\ref{eq: identity radial part causal prop}), is given by
    \begin{equation*}
     \widetilde{E}(r,r') =  -\frac{1}{2\pi i}\widehat{\mathcal{G}}_p(r,r').
   \end{equation*}
\end{proposition}

With the causal propagator in hand, we identify an algebra of observables $\mathcal{A}(\mathcal{M})$ as outlined in Section \ref{sec: A glimpse on the algebraic approach}. In addition, since we focus on static spacetimes, $\mathcal{A}(\mathcal{M})$ may be naturally equipped with a notion of dynamics. %
Precisely, the global timelike Killing vector field on $\mathcal{M}$ generates a continuous, one-parameter group of isometries  $\{\varphi_{\tau} \}_{\tau\in\mathbb{R}}$:
\label{static spacetime, isometries and automorphisms}
  \begin{align*}
    \varphi_{\tau}:\mathcal{M}&\rightarrow\mathcal{M} \\
          (t,\underbar{x})&\mapsto \varphi_{\tau}(t,\underbar{x})=(t+\tau,\underbar{x}).
  \end{align*}
The latter identifies a continuous one-parameter group of $*$-automorphisms of $\mathcal{A}(\mathcal{M})$, say $\{\alpha_{\tau} \}_{\tau\in\mathbb{R}}$, such that
  \begin{align}
    \label{def: alpha_tau}
    \alpha_{\tau}:\mathcal{A}(\mathcal{M})&\rightarrow\mathcal{A}(\mathcal{M}) \nonumber\\
          \Psi(f)&\mapsto \alpha_{\tau}(\Psi(f))=\Psi(f\circ\varphi_{\tau}).
  \end{align}
In addition, it follows that $\{\alpha_{\tau} \}_{\tau\in\mathbb{R}}$ individuates the class of $\alpha_{\tau}$-invariant states on $\mathcal{A}(\mathcal{M})$. To wit, a state $\psi$ on $\mathcal{A}(\mathcal{M})$ is said to be $\alpha_{\tau}$-invariant if
\begin{equation}
  \label{eq: alpha_tau invariant state}
 \psi(\alpha_\tau (a)) = \psi(a)\text{, \quad }\forall a\in\mathcal{A}(\mathcal{M}).
\end{equation}
With the ingredients above, the causal propagator, an algebra of observables, and a notion of dynamics singled out by $\{\alpha_{\tau} \}_{\tau\in\mathbb{R}}$, we can define and construct ground and thermal states, as shown in the following two sections.

\subsection{Ground states}
    \label{sec: Ground states}

Consider an algebra of observables $\mathcal{A}(\mathcal{M})$ on a static spacetime $\mathcal{M}$ and a continuous, one-parameter group of $*$-automorphisms $\{\alpha_{\tau} \}_{\tau\in\mathbb{R}}$, as defined in Equation \eqref{def: alpha_tau}.

\begin{definition}[Ground state]
  \label{def: ground state} A state $\psi$ on $\mathcal{A}(\mathcal{M})$ is a \textit{ground state} if, for $a,b\in\mathcal{A}(\mathcal{M})$, its two-point function $\psi_2$ satisfies
  \begin{itemize}
    \item[i)] $\psi_2(a,\alpha_\tau(b))$ is a bounded function of $\tau$;
    \item[ii)] $\int_{\mathbb{R}} \hat{f}(\tau)\psi_2(a,\alpha_\tau(b))d\tau=0,
        \,\forall f\in C^\infty_0((-\infty,0))$, where $\hat{f}$ is the Fourier transform of $f$.
  \end{itemize}
\end{definition}

    From now on, assume $\mathcal{M}$ admits Schwarzschild-like coordinates such that each of its points is labelled $x=(t,r,\underline{\theta})$. Let $P$ be the Klein-Gordon operator on $\mathcal{M}$, and let  $\psi_2\in\mathcal{D}^\prime(\mathcal{M}^2)$ be a positive bisolution:
         \begin{align*}
           &(P\otimes\mathbb{I})\psi_2=(\mathbb{I}\otimes P)\psi_2=0,   \\
           &\psi_2(f,f)\geq0,\, \forall f\in C_0^\infty (\mathcal{M}).
         \end{align*}
    By imposing compatibility of $\psi_2$ with the isometries of the spacetime, and introducing an $\varepsilon$-regularization to guarantee it is distributionally well-defined, we make the following ansatz:
            \begin{equation}
              \label{eq:ansatz 2 point ground state}
              \psi_2(x,x')=\lim_{\varepsilon\rightarrow 0^+}\int_{\sigma(\Delta_{j})} d\eta_j\int_{\mathbb{R}} d\omega\Theta(\omega) e^{-i\omega (t-t'-i\varepsilon)} \widetilde{\psi}_{2}(r,r') Y_{j}(\underline{\theta})\overline{Y_{j}(\underline{\theta}')}.
            \end{equation}
    Since the antisymmetric part of the two-point function is the causal propagator,
    \begin{equation*}
      iE(x,x')=\psi_2(x,x')-\psi_2(x',x)\text{ for }x,x'\in \mathcal{M},
    \end{equation*}
    it follows that $ \widetilde{\psi}_{2}(r,r')= \widetilde{\psi}_{2}(r',r)$ and
           \begin{equation}
           \label{eq:radial part two point function}
             \widetilde{\psi}_2(r,r^\prime) = -\frac{1}{\omega}\widetilde{E}(r,r^\prime) ,
         \end{equation}
   see Proposition \ref{prop: causal propagator initial conditions}. When the radial Green function has no poles,  $\widetilde{E}(r,r^\prime) $ is obtained explicitly by Proposition \ref{prop: radial part of causal propagator}, and thence so is $\widetilde{\psi}_2(r,r^\prime)$ by Equation \eqref{eq:radial part two point function}:
   \begin{equation}
   \label{eq:radial part two point function 2}
     \widetilde{\psi}_2(r,r^\prime) = \frac{1}{2\pi i \omega}\widehat{\mathcal{G}}_p(r,r') .
   \end{equation}
   By direct inspection, it follows that a state $\psi$ on $\mathcal{A}(\mathcal{M})$ determined by the two-point function \eqref{eq:ansatz 2 point ground state} is, by construction, a ground state as in Definition \ref{def: ground state}. Condition i) is trivially satisfied due to the harmonic time-dependence of $\psi_2$. In turn, the latter together with the fact that the action of $\alpha_\tau$ corresponds to a time-translation of $\tau$, with the convolution theorem and noting that $\psi_2$ has support only on positive Fourier frequencies $\omega$, Condition ii) in Definition \ref{def: ground state} holds true. Therefore, as shown in \cite{sahlmann2000passivity}, the state $\psi$ satisfies the Hadamard condition in each globally hyperbolic subregion of $\mathcal{M}$. Accordingly, we say that the two-point function $\psi_2$ is of local Hadamard form. All in all, we can state the following theorem.
\begin{theorem}[Physically-sensible ground state]
\label{thm: 2 point ground state schd coord} Let
\begin{itemize}
  \item[i)] $\mathcal{M} $ be a static spacetime that admits Schwarzschild-like coordinates, with points in $\mathcal{M}$ labelled $x=(t,r,\underline{\theta})$, as per Definitions \ref{def: static spacetime} and \ref{def: Schwarzschild-like coordinates};
\item[ii)] $\mathcal{A}(\mathcal{M})$ be the algebra of observables, as in Definition \ref{def: algebra of observables};
\item[iii)] $\sigma(\Delta_{j})$, $d\eta_j$ and $Y_{j}$ be as defined for the causal propagator in Equation (\ref{eq:ansatz causal propagator}).
\end{itemize}
 Then, a Gaussian state $\psi$ on $\mathcal{A}(\mathcal{M})$ with a two-point function $\psi_2(f,f^\prime)$ whose integral kernel reads
   \begin{align}
   \label{eq: def 2 point ground state schd coord}
   \psi_2(x,x^\prime)= &  \lim_{\varepsilon\rightarrow 0^+} \int_{\sigma(\Delta_{j})}d\eta_j\int_{\mathbb{R}}d\omega \Theta(\omega) e^{-i\omega (t - t'- i\varepsilon)} \widetilde{\psi}_{2}(r,r')Y_{j}(\underline{\theta})\overline{Y_{j}(\underline{\theta}')},
   \end{align}
is a \textit{physically-sensible ground state}: it is the unique ground state, as per Definition \ref{def: ground state}, it satisfies the canonical commutation relations, as in Proposition \ref{prop: causal propagator initial conditions}, and it is of local Hadamard form, as in Definition \ref{def: Hadamard state, local Hadamard form}.
\end{theorem}
   In the next section, I give the definition of KMS states and an explicit expression for their two-point functions.

    \subsection{Thermal states}
    \label{sec: KMS states}

       Within quantum field theory on general spacetimes, states in thermal equilibrium are characterized by the Kubo, Martin and Schwinger (KMS) condition, first introduced in the 50's \cite{kubo1957statistical,martin1959theory}. Since then, the KMS condition has been stated in many different forms, some equivalent and some more general than others. Here, I do not give a historical account, but rather a brief argument illustrating that a property satisfied by Gibbs states inspires, and is generalized by, the KMS condition. Subsequently, I write down an explicit expression for the two-point function of KMS states on the spacetimes of interest. Mathematically rigorous treatment of KMS states, with elaborate discussions regarding their physical meaning, can be found in  \cite{BragaDeGoesEVasconcellos2019dmn,bratteli2012operator,ricardito2018,haag1967equilibrium,hack2010backreaction}.

 Statistical Mechanics gives us three canonical descriptions to study thermodynamical systems of non-interacting particles: the microcanonical, the canonical and the Gibbs (grand) canonical ensembles. The Gibbs canonical ensemble is best suited to describe an open system in a thermal bath, i.e. at fixed temperature, with neither particle number nor total energy fixed. Since the notion of particle, hence of particle number, is not uniquely defined within quantum field theory, we rely on the Gibbs canonical ensemble to introduce the notion of thermality.

Consider a quantum system in a bounded region $\Omega$ of a static spacetime $\mathcal{M}$ whose dynamics is ruled by a Hamiltonian $H$, whose states are elements of a Hilbert space $\mathcal{H}$, and whose observables are elements of an algebra $\mathfrak{A}$ consisting of bounded operators on $\mathcal{H}$. Let $\{e_k\}$ be an orthonormal basis of $\mathcal{H}$, and let $\text{Tr}$ denote the trace over $\mathcal{H}$, i.e. $\text{Tr}(a) = \sum\limits_k \braket{a\,e_k|e_k}$ for $a\in\mathfrak{A}$. Assuming $e^{i\tau H}$ is bounded and linear, it follows that it is also of trace class, see \cite[Pg. 200]{haagLocalquantumphysics}, and that the quantum mechanical time-evolution may be implemented by $e^{i\tau H}a e^{-i\tau H}=:\alpha_\tau(a)$. Within the Gibbs canonical ensemble, the expectation value of the observable $a\in\mathfrak{A}$ on a Gibbs state $\psi_\Omega\in\mathcal{H}$ at fixed inverse-temperature $\beta>0$ is of the form
         \begin{equation}
           \label{eq:def Gibbs state}
           \psi_\Omega(a)=\frac{\text{Tr}(e^{-\beta H}a)}{\text{Tr}(e^{-\beta H})}.
         \end{equation}
        Assuming $\alpha_\tau $ has a suitable analytic extension to complex times, the cyclic property of the trace implies $ \text{Tr}(e^{-\beta H}\alpha_\tau (a) b) =\text{Tr}(e^{-\beta H}b \alpha_{\tau +i\beta}(a))$.  The latter property together with the definition of Gibbs states as per Equation \eqref{eq:def Gibbs state} yields, for $a,b\in\mathfrak{A}$
         \begin{equation}
         \label{eq:prop Gibbs state}
           \psi_\Omega( \alpha_\tau (a) b)=\psi_\Omega(b \alpha_{\tau +i\beta}(a)).
         \end{equation}

Amongst the plethora of complications that emerge when considering infinite degrees of freedom, one is that neither $H$ nor $e^{-\beta H}$, nor observables in general, are necessarily of trace class---Equation \eqref{eq:def Gibbs state} is generally ill-defined. Notwithstanding, in 1967, Haag, Hugenholtz, and Winnink \cite{haag1967equilibrium} showed that, under reasonable conditions, the thermodynamic limit  $$\lim_{\Omega\rightarrow\infty}\psi_\Omega(a)=\psi(a)$$ exists, and that, when it does, $\psi(a)$ satisfies the property of Gibbs states given in Equation \eqref{eq:prop Gibbs state}. Accordingly, for unbounded systems, the characterization of thermal states is given by a generalization of the property in Equation \eqref{eq:prop Gibbs state}. The precise definition follows.

     \begin{definition}[KMS state]
     \label{def: KMS state *alg}
      Let $\mathcal{A}(\mathcal{M})$ be the algebra of observables, as in Definition \ref{def: algebra of observables}, on a static spacetime $\mathcal{M}$, with a continuous, one-parameter group of $*$-automorphisms $\{\alpha_{\tau} \}_{\tau\in\mathbb{R}}$, as defined in Equation \eqref{def: alpha_tau}. Consider an $\alpha_\tau $-invariant state $\psi$ on $\mathcal{A}(\mathcal{M})$, as in Equation \eqref{eq: alpha_tau invariant state}. With $a,b\in \mathcal{A}(\mathcal{M})$, define the functions
     \begin{equation*}
        F_{ab}(\tau) := \psi(b\alpha_\tau(a)) \quad\text{ and }\quad G_{ab}(\tau) := \psi(\alpha_\tau(a)b) .
     \end{equation*}
     $\psi$ is called a \textit{KMS state with respect to $\alpha_\tau$ at inverse-temperature $\beta>0$} if it satisfies the \textit{KMS condition}, i.e. if $F_{ab}$ and $G_{ab}$ extend to analytic functions on the strips $\{z\in\mathbb{C}:\Imag(z) \in (-\beta,0)\cup(0,\beta)\}$, which are continuous at the boundaries $\{z\in\mathbb{C}:\Imag(z) \in \{-\beta,0,\beta\}$, and that satisfy
     \begin{equation*}
       F_{ab}(\tau +i\beta) = G_{ab}(\tau ).
     \end{equation*}
     \end{definition}
     Physically, temperature relates to energy, energy relates to time; hence, it is not unreasonable that the notion of thermality introduced by the KMS condition is linked to a notion of time. Specifically, by Definition \ref{def: KMS state *alg}, a temperature can be assigned to a quantum state on a general spacetime when there is a time-evolution that yields an automorphism on the algebra of observables. This is always the case when the underlying spacetime is static, as described on Page \pageref{static spacetime, isometries and automorphisms}. In this case, it is straightforward to construct KMS states, as given by the following theorem and proved, in details, in \cite[Pg. 138-142]{hack2010backreaction}.

       \begin{theorem}[Physically-sensible KMS state]
         \label{thm: 2 point KMS state schd coord} With the assumptions and notation of Theorem \ref{thm: 2 point ground state schd coord}, a two-point function $\psi_2(f,f^\prime)$ with integral kernel
         \begin{align}
           \label{eq:def 2 point KMS state Schd coord}
         \psi_2(x,x^\prime)=\lim_{\varepsilon\rightarrow 0^+}\int_{\sigma(\Delta_{j})}d\eta_j \int_{\mathbb{R}} d\omega \Theta(\omega) \left[ \frac{e^{-i\omega (t-t'-i\varepsilon)}}{1 - e^{-\beta\omega}}+ \frac{e^{+i \omega (t-t'+i\varepsilon)}}{{e^{\beta\omega}-1}} \right] \widetilde{\psi}_{2}(r,r') Y_j(\underline{\theta})\overline{Y_j(\underline{\theta}')}.
         \end{align}
        identifies a unique physically-sensible thermal state $\psi_\beta$ at inverse-temperature $\beta$ with respect to the dynamics induced by time translations: it is invariant under all symmetries of the underlying spacetime, it satisfies the canonical commutation relations, is of local Hadamard form, and satisfies the KMS condition.
     \end{theorem}

    \section{Probing quantum states with particle detectors}
    \label{sec: Probing quantum states with particle detectors}
     	An Unruh-DeWitt detector is a spatially localized quantum system existing on a spacetime and interacting with a quantum field. It has been introduced by Unruh \cite{unruh1976notes} and DeWitt \cite{dewitt1979quantum} within quantum field theory on curved spacetimes, and it has been since then applied in several contexts, theoretically \cite{cong2021quantum,faure2020particle,toussaint2021detecting} and experimentally \cite{biermann2020unruh,fewster2016waiting,Good2020hav,vsoda2021acceleration} motivated. In the next chapter, I apply the particle detector approach described in this section in the study of two physical phenomena. One consists of anti-correlation effects on a suitable class of asymptotically AdS spacetimes, and the other concerns thermal effects on a naked singularity spacetime. Here, I describe the theoretical model chosen, as considered in \cite{birrell1984quantum}, and I obtain an explicit expression for the transition rate in Schwarzschild-like coordinates by using the two-point functions constructed in the previous section.

       Take the detector to be the simplest quantum model of an ``atom'': a pointlike two-level system with energy gap $\Omega>0$, characterized by a Hilbert space $\mathcal{H}_D$, and by a Hamiltonian $H_D$. Let $\{\ket{0_D},\ket{\Omega_D}\}$ be an orthonormal basis such that $H_D\ket{0_D}=0$ and $H_D\ket{\Omega_D}=\Omega\ket{\Omega_D}$. Suppose the detector is following a smooth timelike trajectory $x$ parametrized by its proper time $\tau$, on a static spacetime $\mathcal{M}$. Let $\Psi:\mathcal{M}\rightarrow \mathbb{R}$ be a Klein-Gordon field initially in a Gaussian, Hadamard state $\psi_i$ acting on an algebra of observables $\mathcal{A}(\mathcal{M})$ as in Definition \ref{def: algebra of observables}. Recall that $\psi_i$ identifies a Hilbert space $\mathcal{H}_{\psi_i}$ by the GNS theorem, as described on Page
       \pageref{page: GNS theorem}.

       Assume the detector couples to the quantum field through the interaction Hamiltonian
     		\begin{equation}
     		\label{eq:Hint}
     			H_{int}(\tau)=c\chi(\tau)\Psi(x(\tau))\otimes\mu(\tau)\text{,}
     		\end{equation}
     	where $c\in\mathbb{R}$ is a small coupling constant, $\chi\in C_0^\infty(\mathbb{R})$ is a switching function and
      \begin{equation}
        \label{eq: monopole-moment operator}
        \mu(\tau) : = \ket{\Omega_D}\bra{0_D}e^{i\Omega \tau} + \ket{0_D}\bra{\Omega_D}e^{-i\Omega \tau}
      \end{equation}
       is an operator acting on $\mathcal{H}_D$. The interaction Hamiltonian given in Equation \eqref{eq:Hint} is said to be of \textit{monopole-type} and $\mu$ is called a \textit{monopole-moment} operator. This terminology descends from recognizing this model as a ``zero-th order'' simplification of a light-matter interaction, such as that of an electron bound in an atom, and taking into account that it is customary to implement multipole expansions in ``moments'' when dealing with localized charges, see e.g. \cite[Sec.II]{pozas2016entanglement}.
 The total Hilbert space is $\mathcal{H}_{\psi_i}\otimes\mathcal{H}_D$ and the total Hamiltonian reads
      \begin{equation*}
        H=H_{\psi_i}\otimes \mathbb{I}_{\mathcal{H}_{D}}+\mathbb{I}_{\mathcal{H}_{\psi_i}}\otimes H_D+H_{int}.
      \end{equation*}
       The probability of a system, whose initial state at time $\tau_i$ is $\ket{\psi_i}\equiv\ket{\psi_i,0_D}$, to be found in the final state $\ket{\psi_f}\equiv\ket{\psi_f,\Omega_D}$ at time $\tau_f$, for $\psi_f\in\mathcal{H}_{\psi_i}$, can be computed as follows. In the interaction picture, we have that $\ket{\psi_f}=U(\tau_f,\tau_i)\ket{\psi_i}$, where $U$ is given by the Dyson series. Up to first order in perturbation theory,
     		\begin{equation*}
     			U(\tau_f,\tau_i)=\mathbb{I}-i\int_{\tau_i}^{\tau_f}d\tau H_{int}(\tau),
     		\end{equation*}
     	and the amplitude of a transition $\ket{\psi_i}\rightarrow\ket{\psi_f}$ is
     		\begin{equation}
          \label{eq: amplitude 1}
     			\mathscr{M}=-i\int_{\tau_i}^{\tau_f}d\tau \bra{\psi_f} H_{int}(\tau)\ket{\psi_i}.
     		\end{equation}
   Equation \eqref{eq: amplitude 1} together with Equations \eqref{eq:Hint} and \eqref{eq: monopole-moment operator} yields
     		\begin{equation*}
     		  \mathscr{M}=-ic\bra{\Omega_D}\mu(0)\ket{0_D}\int_{\tau_i}^{\tau_f}d\tau e^{i\Omega\tau}\chi(\tau)\bra{\psi_f}\Psi(x(\tau))\ket{\psi_i}.
     		\end{equation*}
     	Therefore, the probability of a transition $\ket{\psi_i}\rightarrow\ket{\psi_f}$ is
     		\begin{equation}
          \label{eq: prob excitation}
     			|\mathscr{M}|^2=c^2|\bra{\Omega_D}\mu(0)\ket{0_D}|^2\left|\int_{\tau_i}^{\tau_f}d\tau e^{i\Omega\tau}\chi(\tau)\bra{\psi_f}\Psi(x(\tau))\ket{\psi_i}\right|^2.
     		\end{equation}

     	Note that if we take $\Omega<0$ instead, then Equation \eqref{eq: prob excitation} corresponds to the probability of a transition $\ket{\psi_i,\Omega_D}\rightarrow \ket{\psi_f,0_D}$, i.e. the probability of the detector undergoing a de-excitation. The first term of Equation \eqref{eq: prob excitation}, $c^2|\bra{\Omega_D}\mu(0)\ket{0_D}|^2$, depends on the internal details of the detector. The second term is called \textit{response function}, and we shall denote it $\mathcal{F}$. Let us rewrite it in a more convenient form. First, by expanding the expression, we find
     \begin{align*}
       \mathcal{F} &=\int_{\tau_i}^{\tau_f}\int_{\tau_i}^{\tau_f}d\tau d\tau' e^{-i\Omega(\tau-\tau')}\chi(\tau)\chi(\tau')\bra{\psi_i}\Psi^\ast(x(\tau))\ket{\psi_f}\bra{\psi_f}\Psi(x(\tau'))\ket{\psi_i}.
     \end{align*}
     	Recalling that we are considering real-valued fields, and summing over all possible final states $\psi_f$, completeness entails
     		\begin{align*}
     			\mathcal{F} &= \int_{\tau_i}^{\tau_f}\int_{\tau_i}^{\tau_f}d\tau d\tau' e^{-i\Omega(\tau-\tau')}\chi(\tau)\chi(\tau')\bra{\psi_i}\Psi(x(\tau))\Psi(x(\tau'))\ket{\psi_i}.
     		\end{align*}
     	The term $\bra{\psi_i}\Psi(x(\tau))\Psi(x(\tau'))\ket{\psi_i}$ is the pullback of the two-point function $\psi_2(x,x')$, of the state $\psi_i$, to the detector's trajectory. For convenience, we shall denote it simply by $\psi_2(\tau,\tau')$. To explicitly compute the response function, and to obtain the probability of the detector undergoing an excitation, we are left with the problem of choosing a switching function. To bypass transient effects and the subtleties that emerge in this regard \cite{louko2006often,satz2007then,schlicht2004considerations}, we assume that the detector has always been, and will always be, switched on. In infinite interaction time limit the response function reads
      \begin{align}
        \label{eq:woperfkwopefp}
        \mathcal{F} &= \int_{-\infty}^{\infty}\int_{-\infty}^{\infty}d\tau d\tau' e^{-i\Omega(\tau-\tau')}   \psi_2(\tau,\tau').
      \end{align}

      The final simplification we make is that of considering the detector to follow a static trajectory. As exemplified in the next chapter, this simplified model is still useful to study physical phenomena. On one hand, due to the time-translation invariance and the infinite interaction time, the response function given by Equation \eqref{eq:woperfkwopefp} is divergent. On the other hand, taking into account that along stationary trajectories $\psi_2(\tau,\tau')$ depends on $\tau$ and $\tau'$ solely through their difference $\tau-\tau'=:s$ and writting $\psi_2(\tau,\tau')=:\psi_2(s)$, we can drop one of the two integrals in Equation \eqref{eq:woperfkwopefp} to define a finite quantity: the \textit{transition rate},
        \begin{align}
         \label{eq:transitionjustlike33}
         \dot{\mathcal{F}} &= \int_\mathbb{R}ds e^{-i\Omega s} \psi_2(s).
        \end{align}
        Note that the right-hand side of Equation \eqref{eq:transitionjustlike33} coincides with the Fourier transform with respect to $s$ of the pullback $\psi_2(s)$ of the two-point function, evaluated at $\Omega$.

        \subsection{The transition rate in Schwarzschild-like coordinates}

      Suppose an Unruh-DeWitt detector is following a static trajectory on a spacetime $\mathcal{M}$ that admits Schwarzschild-like coordinates, as per Definition \ref{def: Schwarzschild-like coordinates}. Let us write down the transition rate explicitly for the ground and thermal states constructed in the previous section. By Theorems \ref{thm: 2 point ground state schd coord} and \ref{thm: 2 point KMS state schd coord}, the two-point functions of the ground and thermal states are given by Equations \eqref{eq: def 2 point ground state schd coord} and \eqref{eq:def 2 point KMS state Schd coord}, respectively.
      Since the proper time $\tau$ is related to the coordinate time $t$ by
      $\tau=\sqrt{f(r)}\,t$, we have $t - t' = \frac{\tau - \tau'}{\sqrt{f(r)}}=: \frac{s}{\sqrt{f(r)}}$. Thus, for a thermal state, substituting Equation \eqref{eq:def 2 point KMS state Schd coord} into Equation \eqref{eq:transitionjustlike33} and using that
      \begin{equation}
        \label{eq: 2jioejiojiuoe938df}
      \int_\mathbb{R}ds e^{-i\Omega s}e^{\pm i\omega \frac{ s}{\sqrt{f}}}=2\pi\delta\left(\Omega\mp\frac{\omega}{\sqrt{f}} \right),
      \end{equation}
      one obtains
      \begin{align}
                \label{eq: twkemlwkjrpotgjk222tg}
               \dot{\mathcal{F}}_\beta 	= \lim_{\varepsilon\rightarrow 0^+}\int_{\sigma(\Delta_{j})}d\eta_j \int_{\mathbb{R}} d\omega \Theta(\omega) 2\pi \left[ \frac{\delta\left(\Omega+\frac{\omega}{\sqrt{f}} \right)}{1 - e^{-\beta\omega}}+ \frac{\delta\left(\Omega-\frac{\omega}{\sqrt{f}} \right)}{{e^{\beta\omega}-1}} \right] \widetilde{\psi}_{2}(r,r) Y_j(\underline{\theta})\overline{Y_j(\underline{\theta})}.
      \end{align}
      In the expression above, $(r,\underline{\theta})$ gives the fixed spatial position of the detector. By integrating Equation \eqref{eq: twkemlwkjrpotgjk222tg} in $\omega$, and doing an analogous computation for the ground state, the following theorem holds.

      \begin{theorem}[Transition rate for the physically-sensible states]
        \label{thm: Transition rate for the physically-sensible states} With the assumptions and notation of Theorem \ref{thm: 2 point ground state schd coord}, the transition rate, as per Equation (\ref{eq: twkemlwkjrpotgjk222tg}) of a static Unruh-DeWitt detector with energy gap $\Omega$ at fixed spatial position $(r,\underline{\theta})$, interacting for an infinite proper time with the ground state defined in Theorem \ref{thm: 2 point ground state schd coord}, is given by
        \begin{align}
            \label{eq: from thm transition rate ground Schd coord}
        \quad    \dot{\mathcal{F}}_\infty	= 2\pi \Theta( -\Omega )  \int_{\sigma(\Delta_{j})}d\eta_j |Y_j(\underline{\theta})|^2 \widetilde{\psi}_{2}(r,r) \big|_{\omega = -\sqrt{f(r)} \,\Omega}.
        \end{align}
        For a thermal state at inverse-temperature $\beta$ defined in Theorem \ref{thm: 2 point KMS state schd coord}, we have
        \begin{align}
            \label{eq: from thm transition rate KMS Schd coord}
                 \dot{\mathcal{F}}_\beta	= 2\pi \frac{\text{\em sign}(\Omega)}{e^{\beta\sqrt{f(r)}\Omega}-1} \int_{\sigma(\Delta_{j})}d\eta_j |Y_j(\underline{\theta})|^2 \widetilde{\psi}_{2}(r,r) \big|_{\omega =\sqrt{f(r)} |\Omega| }.
        \end{align}
      \end{theorem}
      The transition rate of excitations ($\Omega>0$) for the ground state, given by Equation (\ref{eq: from thm transition rate ground Schd coord}), vanishes identically, while for the thermal state, it does not. Moreover, it is easy to see that the transition rate seen as a function of the energy gap, $\dot{\mathcal{F}}_\beta(\Omega)$, given by Equation (\ref{eq: from thm transition rate KMS Schd coord}), is such that
      \begin{align}
          \label{eq: detailed balance transition rate KMS Schd coord}
              \frac{ \dot{\mathcal{F}}_\beta(\Omega)}{\dot{\mathcal{F}}_\beta(-\Omega)}	=  e^{-\beta\sqrt{f(r)}\Omega}.
      \end{align}
      When Equation \eqref{eq: detailed balance transition rate KMS Schd coord} holds true, we say the detector satisfies the \textit{detailed balance condition at inverse-temperature $T_D^{-1}=\beta\sqrt{f(r)}$}. This condition specifies an equilibrium between the process of excitation and its reverse (de-excitation).

\subsection{Unruh, Hawking, anti-Unruh, and anti-Hawking effects}
\label{sec: Unruh, Hawking, anti-Unruh, anti-Hawking effects}

        Suppose the underlying background also admits a non-degenerate, bifurcate Killing horizon with surface gravity $\kappa_h$. As explained in Section \ref{sec: On the geometric Hawking temperature}, there is a naturally defined global Hawking temperature, as per Definition \ref{def: global hawking temperature}, given by $T_{gH} = \frac{\kappa_h}{2\pi}$. In this context, borrowing the therminology common within black hole physics \cite{candelas1980vacuum,frolov2012black}, we interpret the ground state given by Theorem \ref{thm: 2 point ground state schd coord} as a \textit{Boulware-like state}, and the thermal states given by Theorem \ref{thm: 2 point KMS state schd coord} as \textit{Hartle-Hawking-like states}. In particular, the \textit{Hartle-Hawking state} is the thermal state at temperature $T_{gH}$. Consider a detector coupled to the Hartle-Hawking state such that both Equations \eqref{eq: from thm transition rate KMS Schd coord} and \eqref{eq: detailed balance transition rate KMS Schd coord} hold with $\beta = \frac{1}{T_{gH}}$. Then, the temperature measured by the detector, $T_D =  \frac{T_{gH}}{\sqrt{f(r)}}$ in accordance with Equation \eqref{eq: detailed balance transition rate KMS Schd coord}, coincides with the local Hawking temperature $T_{H}$, as in Definition \ref{def: (Local) Hawking temperature}. In this scenario, we say the detector has \textit{thermalized} at $T_{H}$---and duly noted the \textit{Hawking effect}. Withal, three paramount remarks follow.
        \begin{itemize}
          \item[i)] The Boulware-like and Hartle-Hawking-like states are not, in general, Hadamard at the horizon and the transition rate, accordingly, diverges therein. However, if the temperature of the field coincides with $T_{gH}$---if it is the Hartle-Hawking state---then it is regular at the horizon and the transition rate is well-defined even there \cite{kay1991theorems,moretti2012state}.
          \item[ii)] If we consider a finite interaction time, the results above do not hold in general: the temperature measured by the detector does not match exactly the redshifted temperature of the field. This fact gives rise to relevant physical phenomena, such as the strong anti-Unruh/Hawking effects described in the remainder of this section.
          \item[iii)] As shown in \cite{carballo2019unruh}, a detector can  ``thermalize'' with a non-thermal state. Hence, to properly say that a detector that satisfies the detailed balance condition at the local Hawking temperature has detected Hawking radiation, we should take note of the state to which the detector is coupled. 
        \end{itemize}

      This scenario of when we can also study the transition rate as a function of temperature is of relevant to the applications shown in the next chapter. Consider a detector that, at each fixed spatial position $\underbar{x}$, thermalizes at a temperature $T_{D}(\underbar{x})$, i.e. Equation \eqref{eq: detailed balance transition rate KMS Schd coord} holds true. Since $\dot{\mathcal{F}}$ depends on the trajectory of the detector, we can see it as a function of $\underbar{x}$, hence we can also see it as a function of $T_D$ by looking at the inverse relation $\underbar{x}=\underbar{x}(T_D)$. In this scenario it makes sense to define the anti-correlation effects, as follows.

      \begin{definition}[Anti-correlation effects] \label{def: anti-correlation effects} If
        \begin{equation*}
            \frac{\partial \dot{\mathcal{F}}(T_D)}{\partial T_D} < 0,
        \end{equation*}
      then we say the transition rate manifests an \textit{anti-correlation effect}. If the bifurcate Killing horizon of the underlying background is an observer-dependent acceleration horizon, such as the Rindler horizon on Minkowski or AdS spacetimes, then we also refer to it as the \textit{anti-Unruh effect}. If the underlying spacetime is a black hole, accordingly, we call it \textit{anti-Hawking effect}.
      \end{definition}

      Both the anti-Unruh and anti-Hawking phenomena are negative differential effects that were recently discovered \cite{brenna2016anti,Henderson2019uqo} and possess a weak and a strong formulation. Both formulations concern the behaviour of the transition rate with respect to the observed temperature. Definition \ref{def: anti-correlation effects} above gives their weak versions. The strong anti-correlation effects concern the fact that a detector may approximately thermalize, i.e. for finite but long interaction times
      \begin{align*}
              \frac{ \dot{\mathcal{F}}(\Omega)}{\dot{\mathcal{F}}(-\Omega)}	\sim  e^{-\Omega/T_D},
      \end{align*}
      at a temperature ($T_D$) that decreases as the (redshifted) temperature of the field ($(\beta\sqrt{f(r)})^{-1}$) increases. Since here we only consider infinite interaction times, $T_D \equiv (\beta\sqrt{f(r)})^{-1} $ and only the weak form is of relevance. Thus, we omit the adjective ``weak'' when referring to the anti-Unruh and anti-Hawking effects.

      To understand the need of introducing Definition \ref{def: anti-correlation effects}, let us consider two standard scenarios. First, the following example from statistical physics, and subsequently the Unruh effect on Minkowski spacetime.

      \begin{example}[Bosons and Fermions] \label{eg: bosons and fermions gases} Let $k_B$ be the Boltzmann constant. The average number of particles with energy $E$ of an ideal gas of bosons and fermions are ruled, respectively, by the Bose-Einstein and Fermi-Dirac distributions:

        \begin{subequations}
        	\label{eq:BE FD distributions}
          \begin{tabularx}{.9\textwidth}{Xp{.5cm}X}
            \begin{equation}
                \label{eq:BE distributions}
                  N_{\text{B-E}} \propto  \frac{1}{e^{E/k_B T}-1},
            \end{equation}& &
          \begin{equation}
            \label{eq:FD distributions}
              N_{\text{F-D} }\propto   \frac{1}{e^{E/k_B T}+1}.
          \end{equation}
        	\end{tabularx}
        \end{subequations}

        \noindent Computing the derivatives with respect to the temperature $T$ of Equations \eqref{eq:BE FD distributions}, we find that in both cases $c\in\{\text{B-E}, \text{F-D}\}$, it holds true
        $$ \text{sign}\left( \frac{\partial N_c}{\partial T}\right) = \text{sign}(E)>0.$$
        Since only positive $E$ makes sense in this context, we conclude that for ideal gases of bosons and fermions the number of particles in a given state increases with temperature.
      \end{example}

      Now let us return to quantum field theory and consider the Unruh effect: ``a uniformly accelerated observer sees Minkowski vacuum as a thermal bath''. On one hand, it can be formalized within standard quantum field theory in terms of creation and annihilation operators. In this context, one computes the expectation value of the number operator in Minkowski vacuum as seen by a Rindler observer, and obtains a certain distribution that, in four dimensions, is given by a certain polynomial (the local density of states) multiplied by a Planckian distribution that selects the Unruh temperature \cite{unruh1976notes, takagi1986vacuum}. On the other hand, analogously to the Hawking effect, the Unruh effect can also be stated in the framework of an Unruh-DeWitt detector: on Minkowski spacetime, a detector following a Rindler trajectory with proper acceleration $a$, and interacting with Minkowski vacuum for an infinite proper time, will observe the Unruh effect by thermalizing at the Unruh temperature $T_U=\frac{a}{2\pi}$. In this scenario, as shown in \cite{takagi1986vacuum,hodgkinson2012often}, the explicit form of the transition rate, which depends on the number of dimensions, is given by

      \begin{subequations}
        \setlength{\tabcolsep}{4pt}
        \renewcommand{\arraystretch}{.1}
      	\label{eq:transition rate Minkowski 3 to 6}
        \begin{tabularx}{.9\textwidth}{XX}
        \begin{equation}
        		\label{eq:transition rate Minkowski 3}
            	\dot{\mathcal{F}}_{\text{Mink}_3}= \frac{1}{2} \frac{1}{e^{2\pi \Omega/a}+1},
        \end{equation}&
        \begin{equation}
        	\label{eq:transition rate Minkowski 4}
        	\dot{\mathcal{F}}_{\text{Mink}_4}= \frac{1}{2\pi} \frac{\Omega}{e^{2\pi \Omega/a}-1},
        \end{equation}\\
      	\begin{equation}
      			\label{eq:transition rate Minkowski 5}
      			\dot{\mathcal{F}}_{\text{Mink}_5}= \frac{1}{32\pi} \frac{4\Omega^2 + a^2}{e^{2\pi \Omega/a}+1},
      	\end{equation}&
      	\begin{equation}
      		\label{eq:transition rate Minkowski 6}
      		\dot{\mathcal{F}}_{\text{Mink}_6}= \frac{1}{12\pi^2} \frac{\Omega(\Omega^2+a^2)}{e^{2\pi \Omega/a}-1}.
      	\end{equation}
      \end{tabularx}
      \end{subequations}

For all spacetime dimensions, the closer the detector is to the horizon, the higher its proper acceleration, and the higher the temperature measured. However, with respect to $T_U$ the transition rate is not always a monotonically increasing function. The expressions in Equation \eqref{eq:transition rate Minkowski 3 to 6} are plotted in Figure \ref{fig:transition rate Minkowski 3 4 5 6 as function of a}, which makes it clear that for $n\geq 4$ we have that temperature, acceleration and transition rate are all directly proportional, while for $n=3$ together with $\Omega<0$ (de-excitation) we have that $\dot{\mathcal{F}}_{\text{Mink}_3}$ decreases with $a$---that is, the anti-Unruh effect is manifest.

      Guiding our intuition by what happens in standard scenarios such as the one in Example \ref{eg: bosons and fermions gases}, we would expect that the higher the temperature, the higher the energy available, the higher the number of collisions, the higher the number of particles, the higher the transition probabilities for an Unruh-DeWitt detector. That is indeed the case for a detector experiencing the Unruh effect on the four-dimensional Minkowski spacetime (and on five- and six-dimensional Minkowski spacetimes). Therefore, the contrasting behaviour observed in three dimensions gains the name of \textit{anti-}Unruh effect. This spacetime dimension dependence for the manifestation of these anti-correlation effects has been verified on topological black holes in \cite{deSouzaCampos2020ddx,deSouzaCampos2020bnj}, works that I summarize in Sections \ref{sec: On a static BTZ spacetime and on Rindler-AdS3} and \ref{sec: On massless hyperbolic black holes} of the next chapter, and on other spacetimes such as the Bertotti-Robinson solution \cite{conroy2021response}.

      \newpage
\vspace*{\fill}
      \begin{figure}[H]
      \centering
       \includegraphics[width=.45\textwidth]{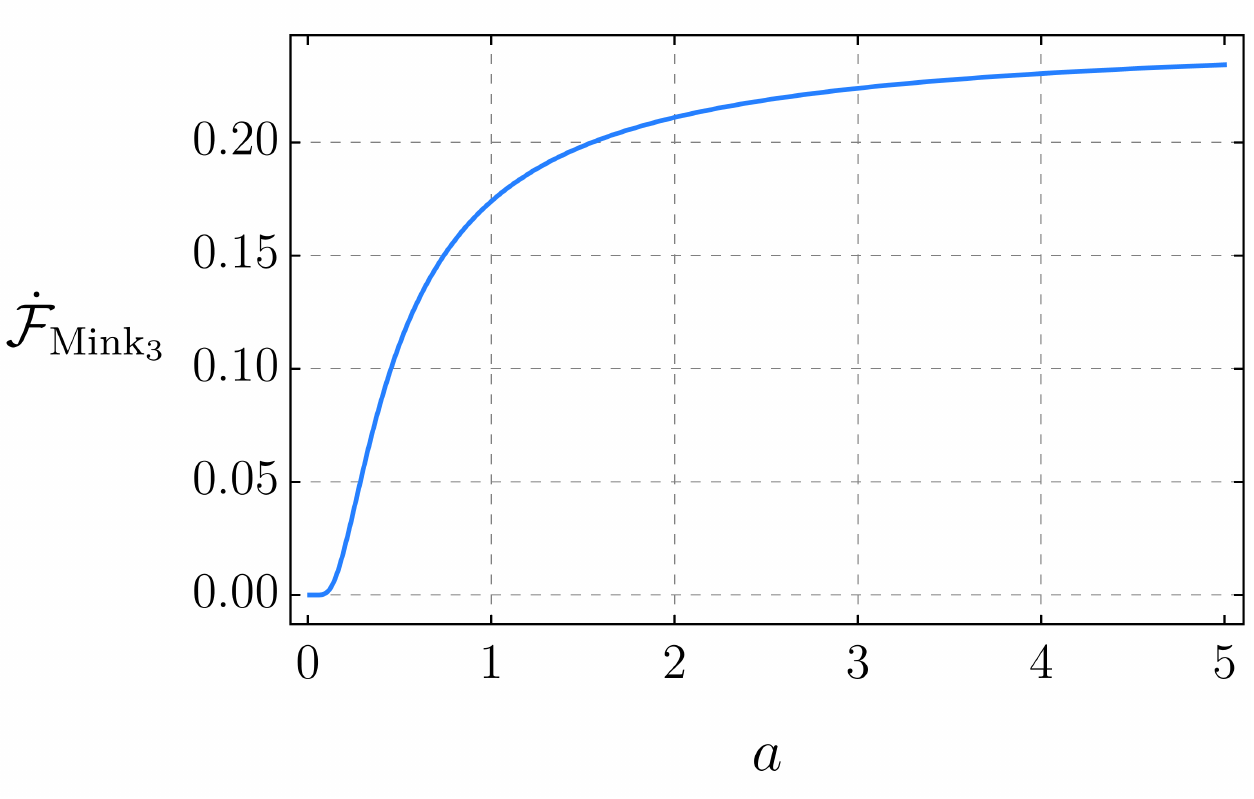}\hspace{.5cm}
        \includegraphics[width=.45\textwidth]{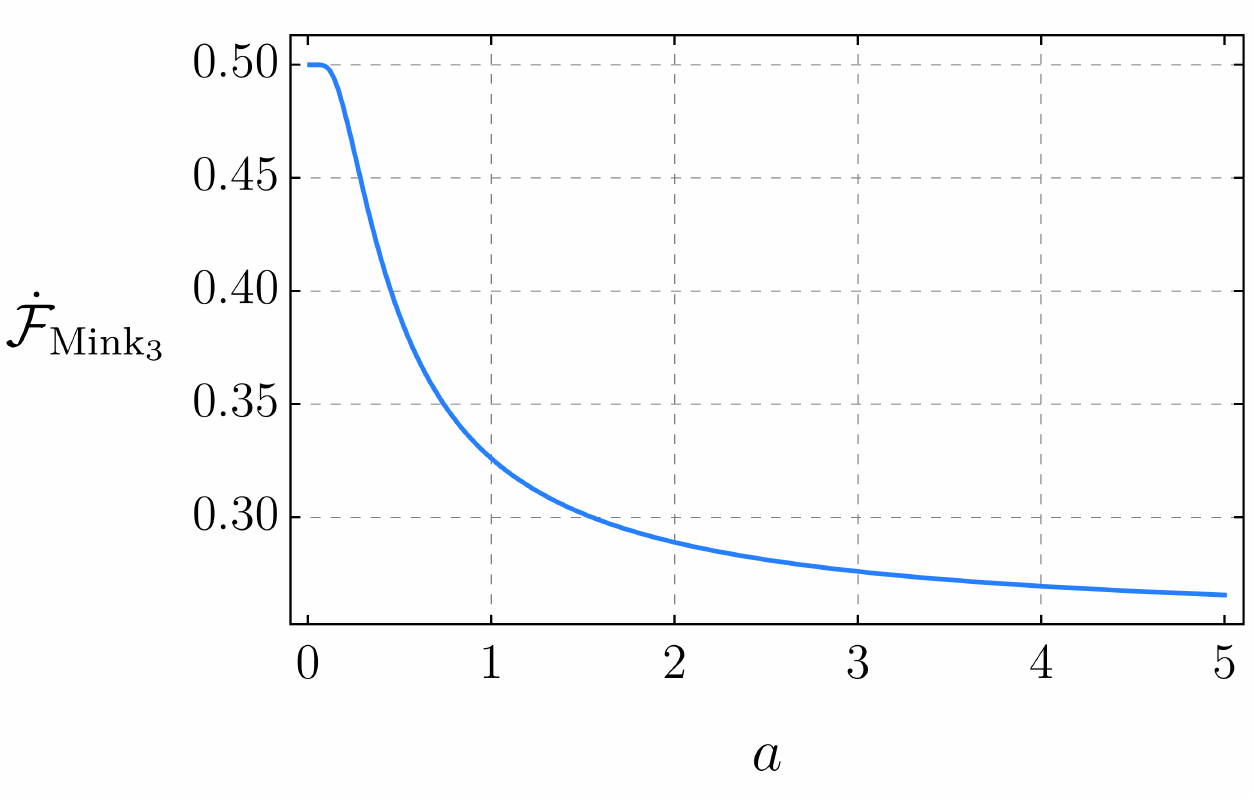}
        \includegraphics[width=.45\textwidth]{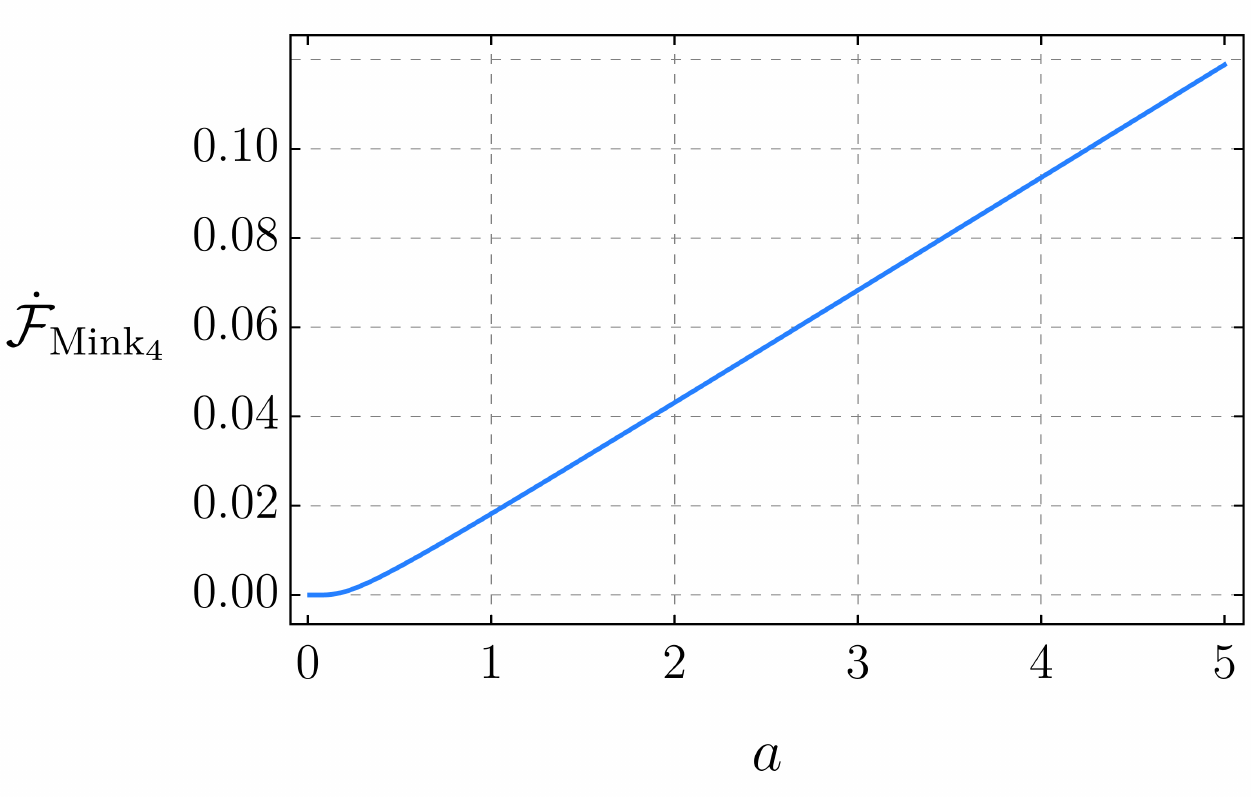}\hspace{.5cm}
      	\includegraphics[width=.45\textwidth]{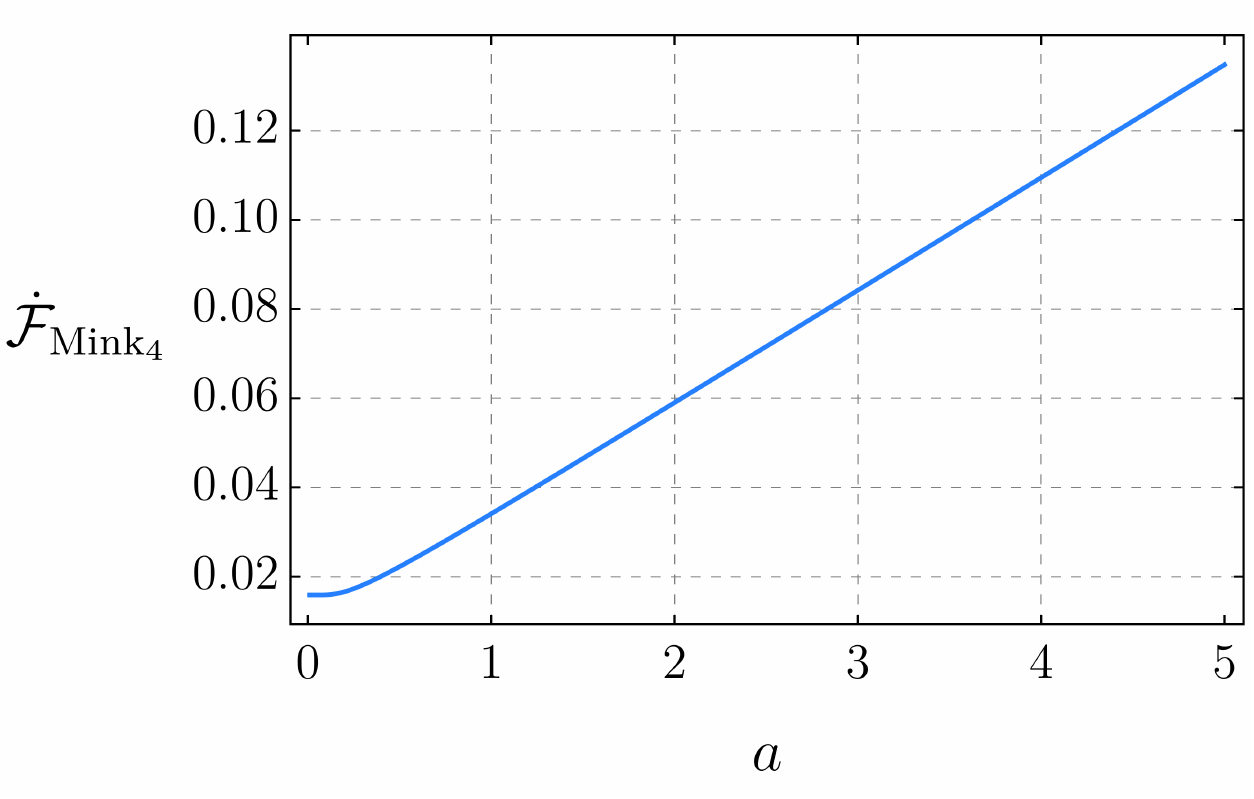}
      	\includegraphics[width=.45\textwidth]{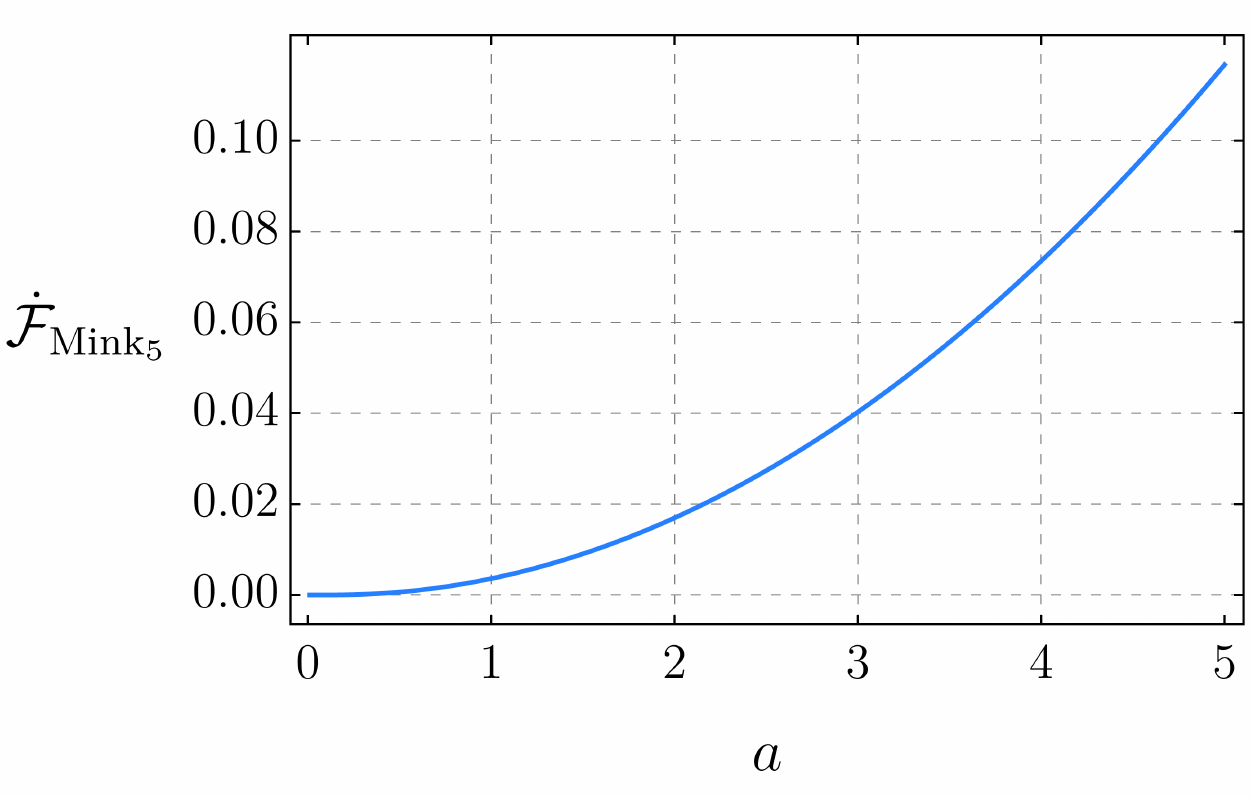}\hspace{.5cm}
      	\includegraphics[width=.45\textwidth]{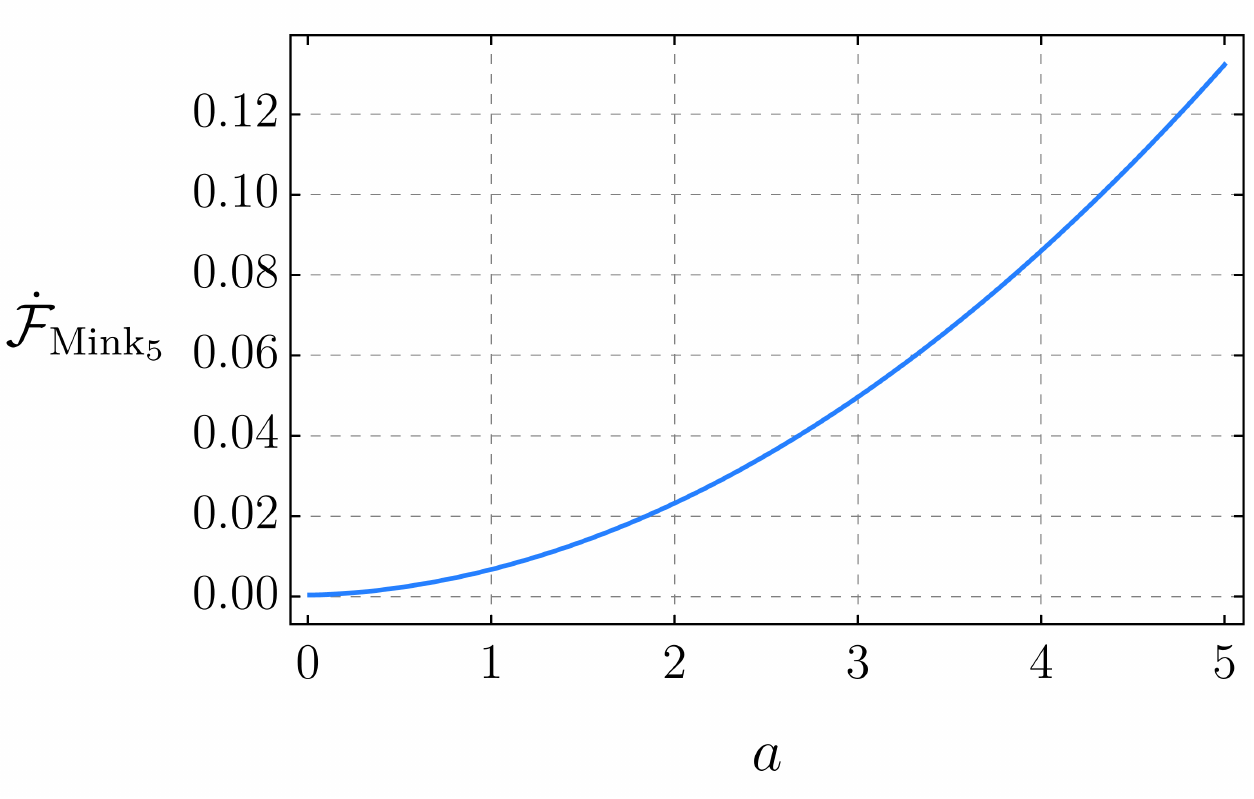}
      	\includegraphics[width=.45\textwidth]{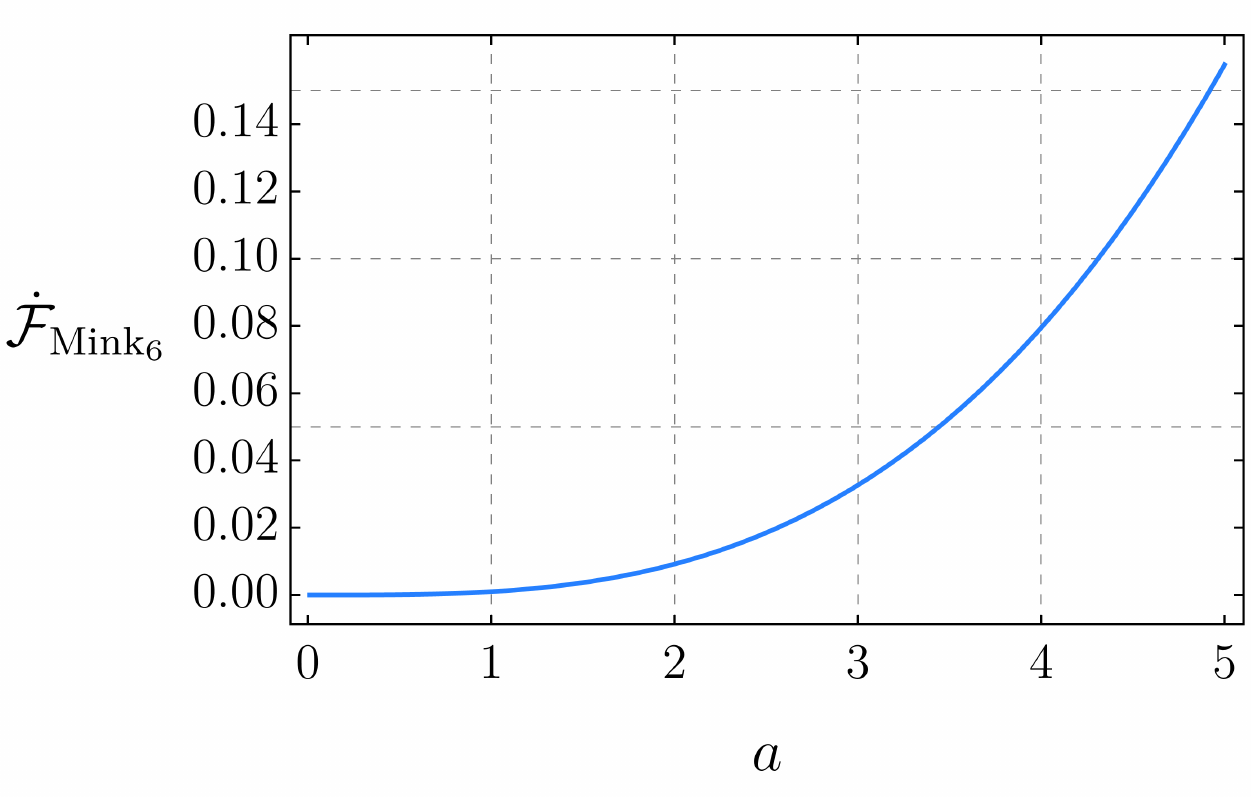}\hspace{.5cm}
      	\includegraphics[width=.45\textwidth]{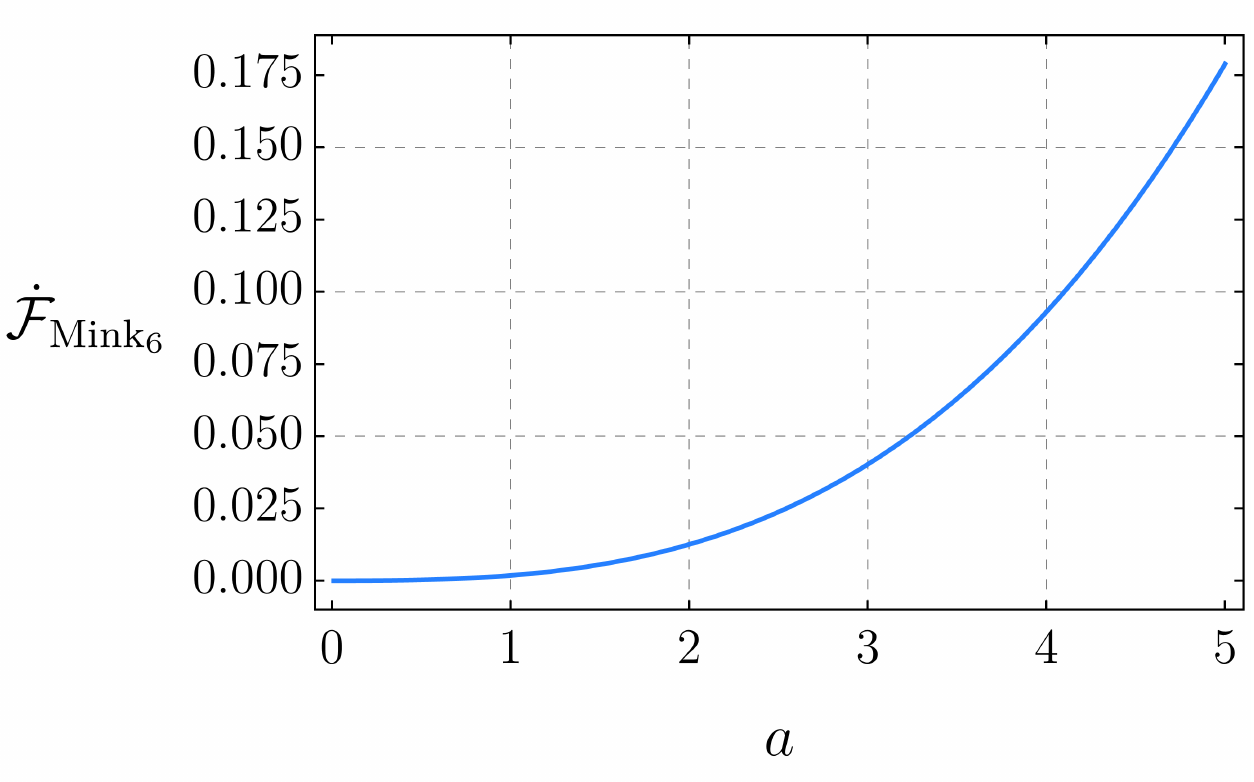}
      \caption{ The transition rate on Minkowski spacetime as a function of the proper acceleration of the detector. On the left, for $\Omega=0.1$; on the right, for  for $\Omega=-0.1$.}
      \label{fig:transition rate Minkowski 3 4 5 6 as function of a}
      \end{figure}
\newpage

      For completeness and for future reference, I include Figure \ref{fig:transition rate Minkowski 3 4 5 6 as a function of the energy gap} that illustrates the behaviour of the transition rates given in Equation \eqref{eq:transition rate Minkowski 3 to 6} as a function of the energy gap. Note that seen as a function of $\Omega$, the transition rate also shows a rather different behavior for $\Omega<0$ in the three-dimensional case.
\vspace{.8cm}
      \begin{figure}[H]
      \centering
        \includegraphics[width=.45\textwidth]{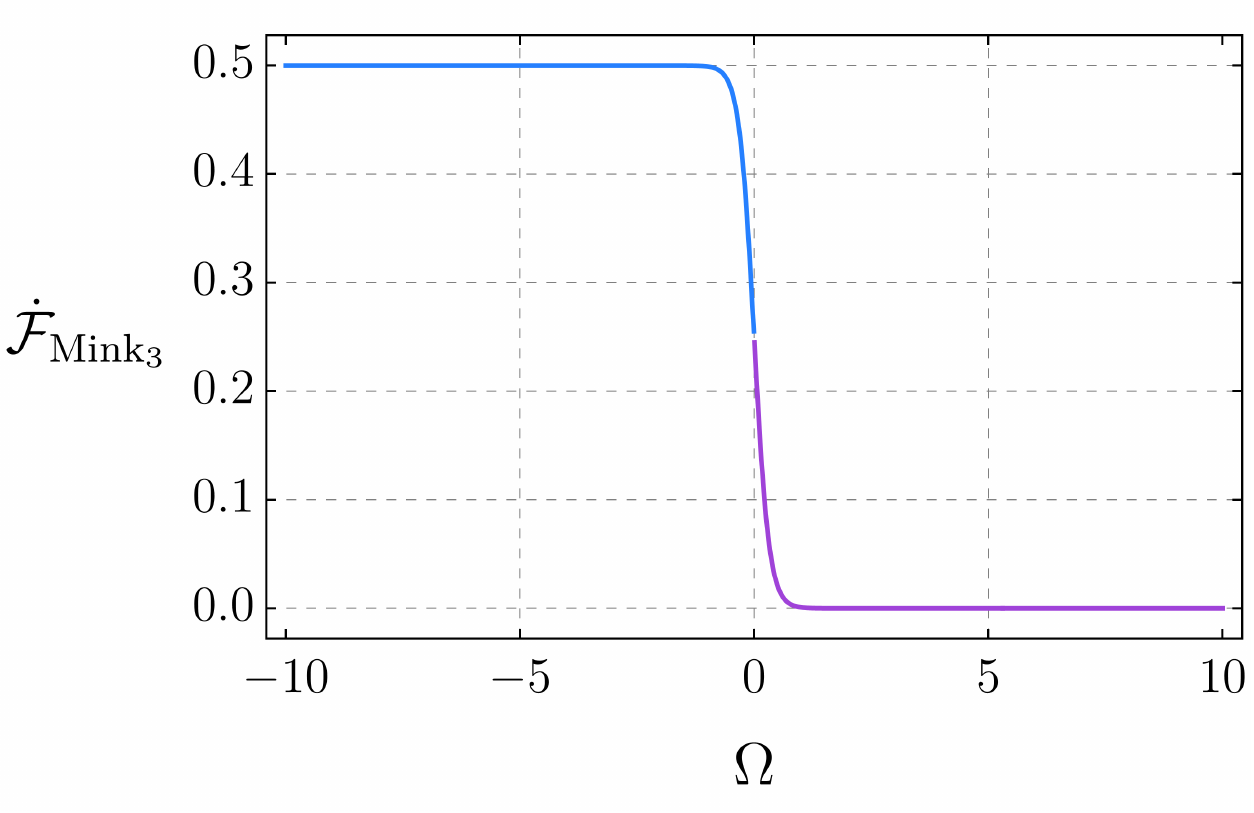}\hspace{.5cm}
          \includegraphics[width=.45\textwidth]{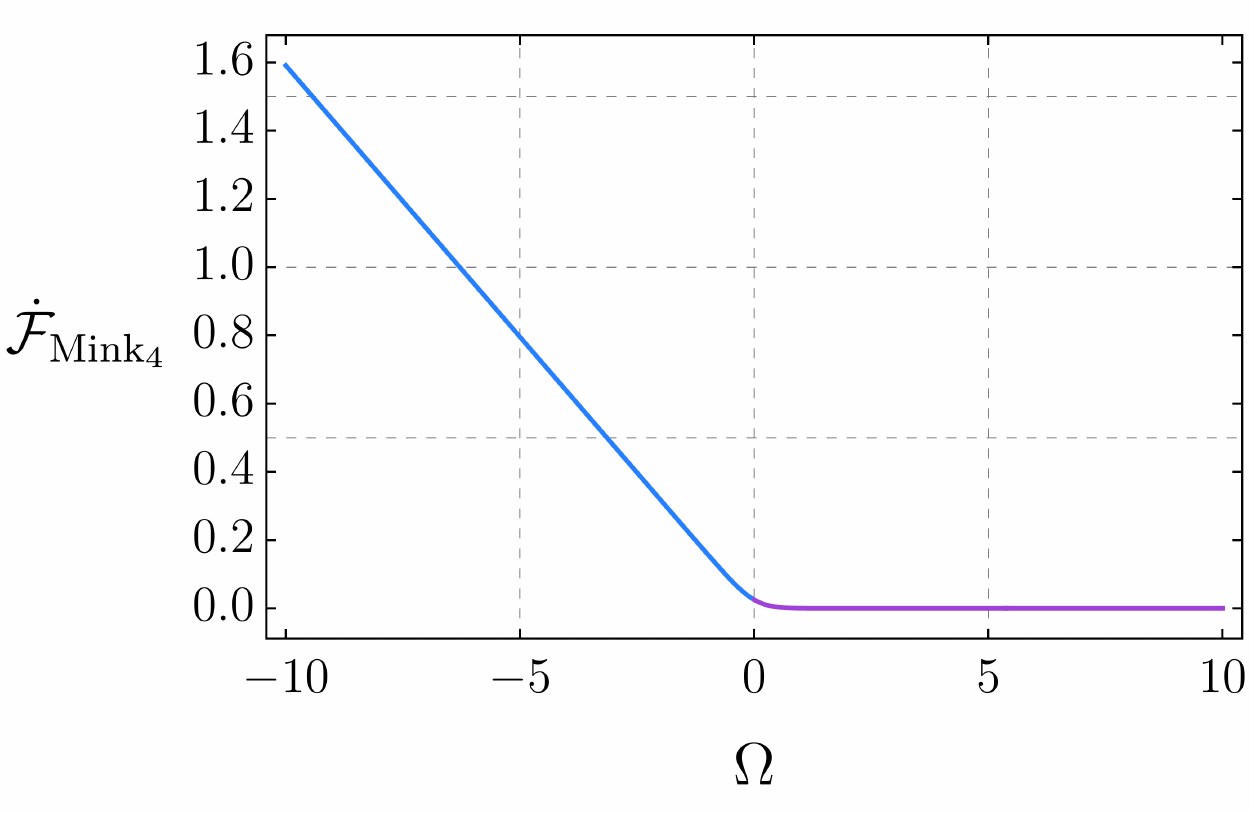}
            \includegraphics[width=.45\textwidth]{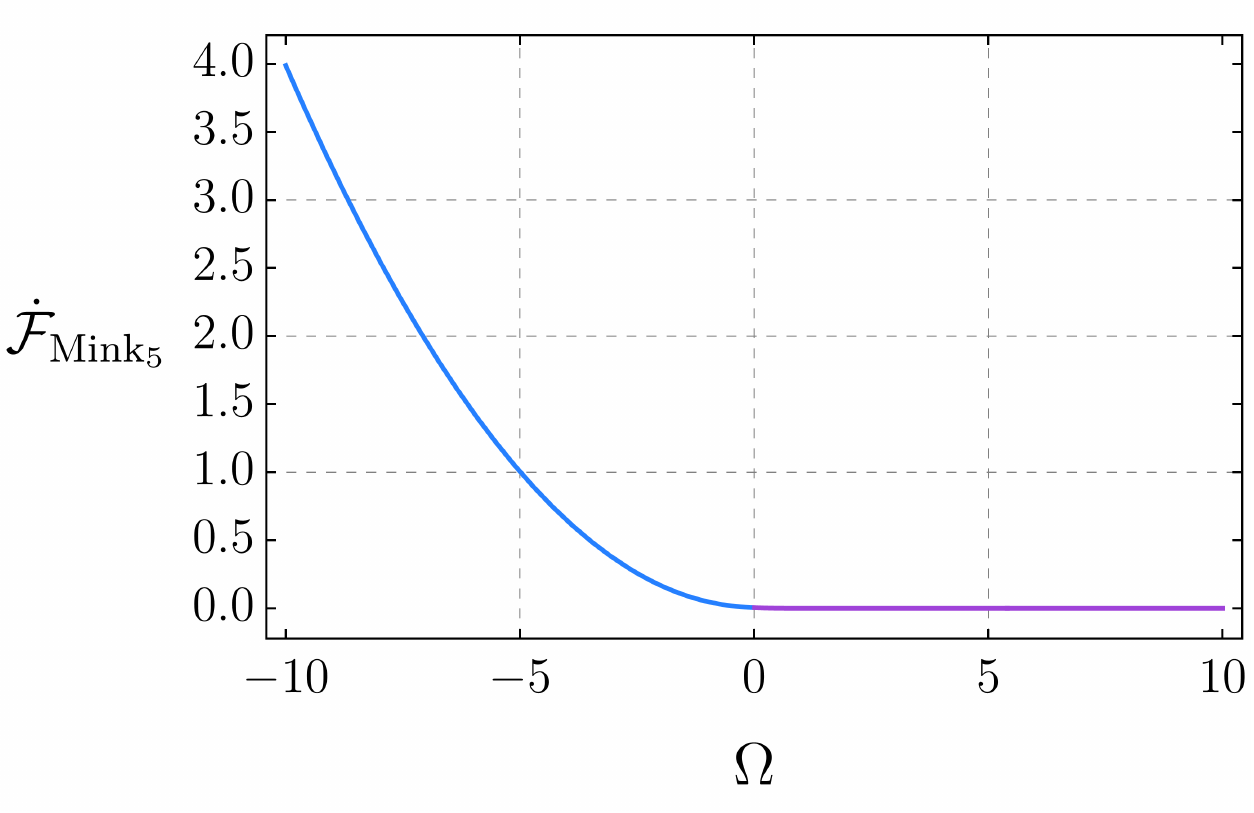}\hspace{.5cm}
              \includegraphics[width=.45\textwidth]{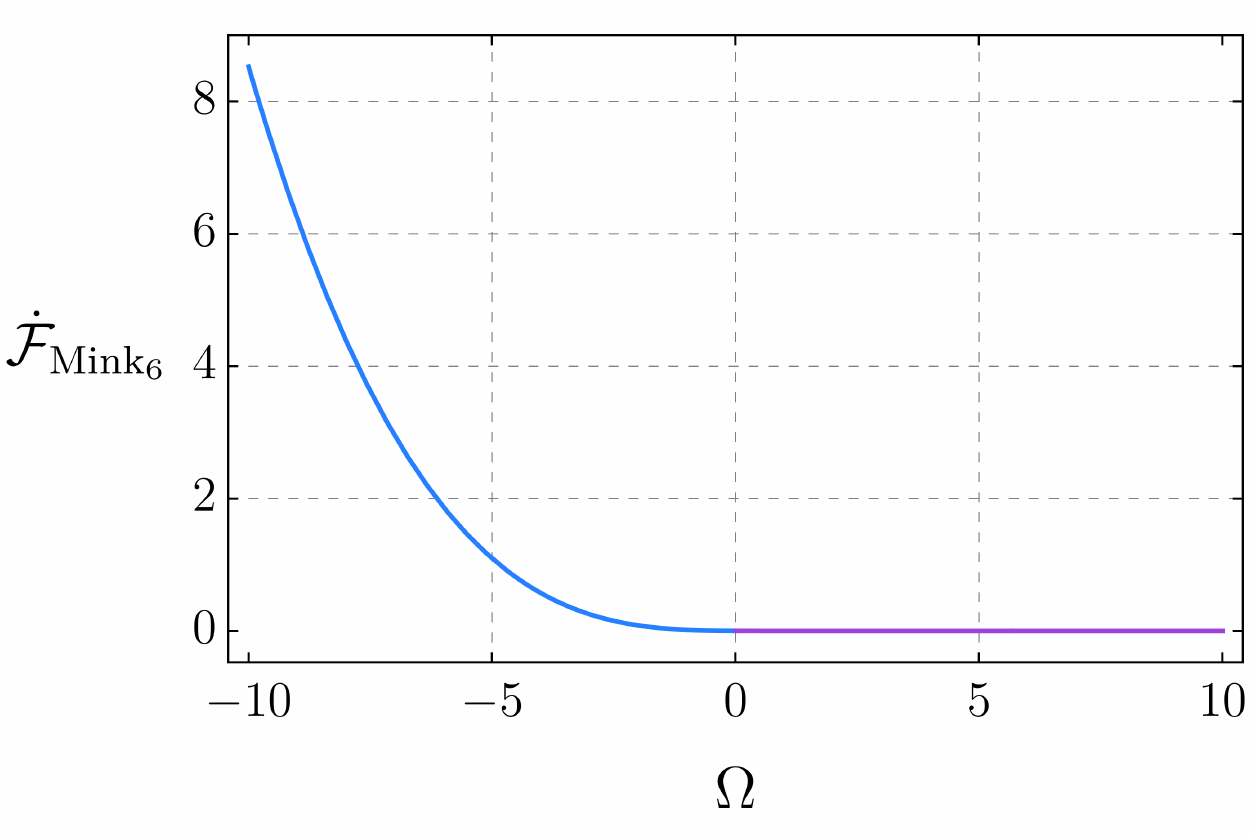}
      \caption{The transition rate on Minkowski spacetime as a function of the energy gap $\Omega$ for a fixed proper acceleration $a=1$ and for spacetime dimensions three to six.}
      \label{fig:transition rate Minkowski 3 4 5 6 as a function of the energy gap}
      \end{figure}
      \vspace{.2cm}


\section{Examples}
\label{sec: Examples chapter 2}

To illustrate how to apply the framework described in this chapter, in Section \ref{sec: On Minkowski spacetime in spherical coordinates} I construct the ground and KMS states on a four-dimensional Minkowski spacetime endowed with spherical coordinates. This is a nice example to consider since not only it constitutes a scenario where abstract expressions can be written in closed form in terms of special functions, but it also admits a one-parameter family of physically-sensible boundary conditions. For completeness, in Section \ref{subsec: On Minkowski spacetime without a plate}, I show that choosing Dirichlet boundary condition at the coordinate singularity at $r=0$ yields the standard results on a four-dimensional Minkowski spacetime.

\subsection{On Minkowski spacetime with a wire}
\label{sec: On Minkowski spacetime in spherical coordinates}

Consider a scalar field $\Psi : \mathcal{M} \to \mathbb{R}$ with mass $m_0$ on a four-dimensional Minkowski spacetime $\mathcal{M}$ endowed with Schwarzschild-like coordinates $(t,r,\theta,\varphi)$. For $t\in\mathbb{R}$, $r\in(0,\infty)$, $\theta\in[0,\pi)$, $\varphi\in[0,2\pi)$, its line element is given by Equation \eqref{eq:Minkowski metric spherical} with $n=4$:
\begin{align}
  \label{eq:Minkowski metric spherical 4D}
    ds^2 =-dt^2 + dr^2 + r^2d\theta^2 + r^2\sin(\theta)^2d\varphi^2.
\end{align}
As mentioned in Section \ref{ex: Chapter 1 Minkowski}, spherical coordinates do not cover the whole Minkowski spacetime. However, the coordinate singularity at $r=0$ can be interpreted as a physical boundary (not a mathematical boundary): the hypersurface $r=0$ is the set of points $(t,0,\theta,\varphi)\in\mathcal{M}$, forming a ``line'', a ``wire'', in $\mathcal{M}$. If we impose Dirichlet boundary condition at $r=0$, in the sense of Definition \ref{def: generalized Robin boundary conditions} with $\gamma=0$ and as done in the next section, we recover the results for the transition rate on Minkowski spacetime (without a wire).

By invoking Theorems \ref{thm: 2 point ground state schd coord} and \ref{thm: 2 point KMS state schd coord}, there is only one step to undertake to obtain the ground and thermal states. More precisely, we need to obtain the spectral resolution of the Green function of the radial equation. With this goal in mind, I first obtain the radial equation and its solutions. Subsequently, I study the admissible boundary conditions and I construct the Green function of the radial equation, obtaining also its spectral resolution. With the two-point functions in hand, I consider, respectively, the thermal contributions to the ground state fluctuations and to the energy density and I compute the transition rate of an Unruh-deWitt particle detector.

\subsubsection*{The radial equation}
\label{The radial equation Mink 4D}

Let $Y_\ell^m(\theta,\varphi)$ be the spherical harmonics on the $2$-sphere with eigenvalues $\lambda_\ell^m:=-\ell(\ell+1)$. The ansatz
\begin{align*}
\Psi(t,r,\theta,\varphi) =  e^{-i\omega t } R(r)Y_\ell^m(\theta,\varphi)
\end{align*}
solves the Klein-Gordon equation \eqref{eq: KG} if and only if the function $R(r)$ solves the radial equation \eqref{eq: the radial equation Shd-like coord}. In the coordinates of Equation \eqref{eq:Minkowski metric spherical 4D}, for $p^2 := \omega^2 - m_0^2$, the radial equation reduces to one of Bessel type:
\begin{equation}
  \label{eq: the radial equation Minkowski 4D}
    R''(r)+ \frac{2}{r}R'(r)+ \left(p^2 + \frac{\lambda_\ell^m}{r^2} \right)R(r)=0.
\end{equation}

\subsubsection*{The radial solutions}
\label{The radial solutions Mink 4D}
The Sturm-Liouville operator $  L_{p^2}$ associated with the radial equation, as in Equation \eqref{eq: the radial equation Minkowski 4D}, reads
\begin{equation*}
  L_{p^2}:=-\frac{1}{r^2} \left(\frac{d}{d r} \left(r^2\frac{d}{d r}\right)+ \lambda_\ell^m \right).
\end{equation*}
Accordingly, let us check the square-integrability of the solutions at each endpoint with respect to the measure $r^2dr$. Bases of solutions of Equation \eqref{eq: the radial equation Minkowski 4D} can be written in terms of the spherical Bessel functions $j_{\nu}(pr)$ and $y_{\nu}(pr)$ with index
\begin{align*}
& \nu := \frac{-1 + \sqrt{1 - 4\lambda_\ell^m}}{2}\geq0.
\end{align*}
 At $r\rightarrow 0$, a suitable basis is
\begin{align*}
&R_{1(0)}(pr) = j_{\nu}(pr),
\quad R_{2(0)}(pr) = p \,y_{\nu}(pr).
\end{align*}
The $L^2$-norms of the solutions above behave asymptotically as
    \begin{align*}
    R_{1(0)}(pr) \overset{r\rightarrow 0}{\sim} r^{\nu} \Rightarrow || R_{1(0)}(pr) ||_{L^2}  \overset{r\rightarrow 0}{\sim} r^{2\nu+3}<\infty \iff \nu>-\frac{3}{2},\\
    R_{2(0)}(pr) \overset{r\rightarrow 0}{\sim} r^{- \nu-1} \Rightarrow || R_{2(0)}(pr) ||_{L^2}  \overset{r\rightarrow 0}{\sim}  r^{-2\nu+1}<\infty \iff \nu<\frac{1}{2},
  \end{align*}
thus $r= 0$ is limit circle for $\ell=0$ (i.e. $\nu=0$), and limit point for $\ell>0$. Defining
\begin{align*}
  \gamma_{\ell}:=
\begin{cases}
  \gamma\in[0,\pi), \quad &\text{ if }\ell=0\\
  0, \quad &\text{ if }\ell>0\\
\end{cases},
\end{align*}
the most general solution that is square-integrable at the singularity can be written as
\begin{align}
    \label{eq: R_gamma Mink}
  R_{\gamma_{\ell}}(pr) = \cos(\gamma_{\ell})R_{1(0)}(pr) - \sin(\gamma_{\ell})R_{2(0)}(pr).
  \end{align}
Markedly, for $\ell=0$, the solution \eqref{eq: R_gamma Mink} satisfies the (generalized) Robin boundary conditions parametrized by $\gamma$, see Equation \eqref{eq: Robin bc with gamma}. For $\ell>0$, it reduces to the principal solution $R_{1(0)}(pr)$, consistent with imposing Dirichlet boundary condition.

At infinity, a suitable basis is given by the Hankel functions
\begin{subequations}
\label{eq: basis radialfunction infty0}
\begin{align}
&R_{1(\infty)}(pr) =h^{(1)}_{\nu}(pr) = j_{\nu}(pr) + i  y_{\nu}(pr)=  R_{1(0)}(pr) + \frac{i}{p} R_{2(0)}(pr) ,\label{eq: basis radialfunction infty}\\
&R_{2(\infty)}(pr) = h^{(2)}_{\nu}(pr) =j_{\nu}(pr) - i  y_{\nu}(pr)= R_{1(0)}(pr) - \frac{i}{p}R_{2(0)}(pr)  .
\end{align}
\end{subequations}
It follows that $r= \infty$ is limit point, since
  \begin{align*}
    R_{1(\infty)}(pr)  \overset{r\rightarrow \infty}{\sim}  e^{+ipr}r^{-1}\Rightarrow || R_{1(\infty)}(pr) ||_{L^2} \overset{r\rightarrow \infty}{\sim} e^{-\Imag(p)r}<\infty \iff \Imag(p)>0,\\
    R_{2(\infty)}(pr) \overset{r\rightarrow \infty}{\sim}  e^{-ipr}r^{-1}\Rightarrow || R_{2(\infty)}(pr) ||_{L^2} \overset{r\rightarrow \infty}{\sim}  e^{+\Imag(p)r}<\infty \iff \Imag(p)<0 .
  \end{align*}
Therefore, the most general square-integrable solution at infinity can be written as
\begin{align*}
   R_{\infty}(pr)&= R_{1(\infty)}(pr)\Theta(\Imag(p)) + R_{2(\infty)}(pr)\Theta(-\Imag(p))\nonumber\\
                       &= \left(R_{1(0)}(pr) + \frac{i}{p} R_{2(0)}(pr)\right)\Theta(\Imag(p)) + \left(R_{1(0)}(pr) - \frac{i}{p} R_{2(0)}(pr)\right)\Theta(-\Imag(p))\nonumber\\
                              &= R_{1(0)}(pr) +  \frac{i\, \text{sign}(\Imag{p})}{p} R_{2(0)}(pr).
 \end{align*}

Note that the Sturm-Liuville operator and the radial solutions defined above satisfy the following useful properties with respect to complex-conjugation
\begin{subequations}
  \label{eq: Schwartz reflections}
\begin{align}
  & \overline{L_{p^2}} = L_{\overline{p^2}} ,\\
  & \overline{R_{2(\infty)}(pr)} = R_{1(\infty)}(\overline{p}r) ,\\
  & \overline{R_{j(0)}(pr)} = R_{j(0)}(\overline{p}r), \,j\in\{1,2\}.
\end{align}
\end{subequations}

\subsubsection*{The radial Green function}
\label{sec: The radial Green function Mink 4D}

With $R_{\gamma_\ell}$ and $R_\infty$ as in the previous section, the Green function of the radial equation, as per Equation \eqref{eq: green function radial equation}, reads
  \begin{align}
    \label{eq: radial green function mink 4d}
  \mathcal{G}_{p}(r,r') =\frac{1}{ \mathcal{N}_p}\left( \Theta(r'-r) R_{\gamma_{\ell}}(r)R_{\infty}(r') + \Theta(r-r')R_{\gamma_{\ell}}(r')R_{\infty}(r)\right).
  \end{align}
Using that $P(r)=r^2$ and that the Wronskian between the spherical Bessel functions is $W_z[j_{\nu}(z),y_{\nu}(z)] = \frac{1}{z^2}$, for the normalization defined in Equation \eqref{eq: definition normalization of radial green function} we find
\begin{align}
  \label{eq: normalization mink 4d}
  \mathcal{N}_p &=  \left\{ \cos(\gamma_{\ell}) \text{sign}(\Imag{p}) \frac{i}{p}  +\sin(\gamma_{\ell})\right\}.
\end{align}
In addition, given Equations \eqref{eq: Schwartz reflections}, it holds
\begin{align}
      \label{eq: property normalization and green conjugate}
    \mathcal{N}_{\overline{p}} &= \overline{\mathcal{N}_{p}} \quad \text{ and }\quad\mathcal{G}_{\overline{p}}=\overline{\mathcal{G}_{p}}.
    \end{align}
Note that
\begin{equation*}
  \mathcal{N}_p=0\iff \ell=0 \text{ and }\tan(\gamma)=-\frac{1}{|\Imag(p)|}\iff \gamma\in\left(\frac{\pi}{2},\pi\right).
\end{equation*}
hence for $\gamma\in\left[0,\frac{\pi}{2}\right]$ the radial Green function has no poles.

\subsubsection*{Spectral resolution of the radial Green function}
\label{sec: Spectral resolution of the radial Green function Mink 4D}

The spectral resolution of the radial Green function $\mathcal{G}_{p}(r,r')$ as in Equation \eqref{eq: radial green function mink 4d} is given by Equation \eqref{eq:spectral resolution radial green function}. For $\gamma\in\left[0,\frac{\pi}{2}\right]$, it simplifies to
\begin{equation*}
    \frac{1}{2\pi i}\oint_{\mathcal{C}^\infty} d(p^2) \mathcal{G}_p(r,r') = -\frac{\delta(r-r')}{S(r)}.
\end{equation*}
Recalling that $\mathcal{G}_p(r,r')$ is defined whenever $\Imag(p)\neq 0 $, $\mathcal{C}^\infty$ can be taken as the infinite radius limit of either a ``pac-man'' contour, say $\mathcal{C}_1$, in the $p^2$-complex plane or of two semi-disks, say $\mathcal{C}_2$, in the $p$-complex plane, since it holds
\begin{equation*}
  \oint_{\mathcal{C}_1} d(p^2)\mathcal{G}_{p}(r,r') = \oint_{\mathcal{C}_2} dp \, p \mathcal{G}_{p}(r,r').
\end{equation*}
Both contours are illustrated in Figure \ref{fig:contour pacman and disks}.
\begin{figure}[H]
\centering
\includegraphics[width=.8\textwidth]{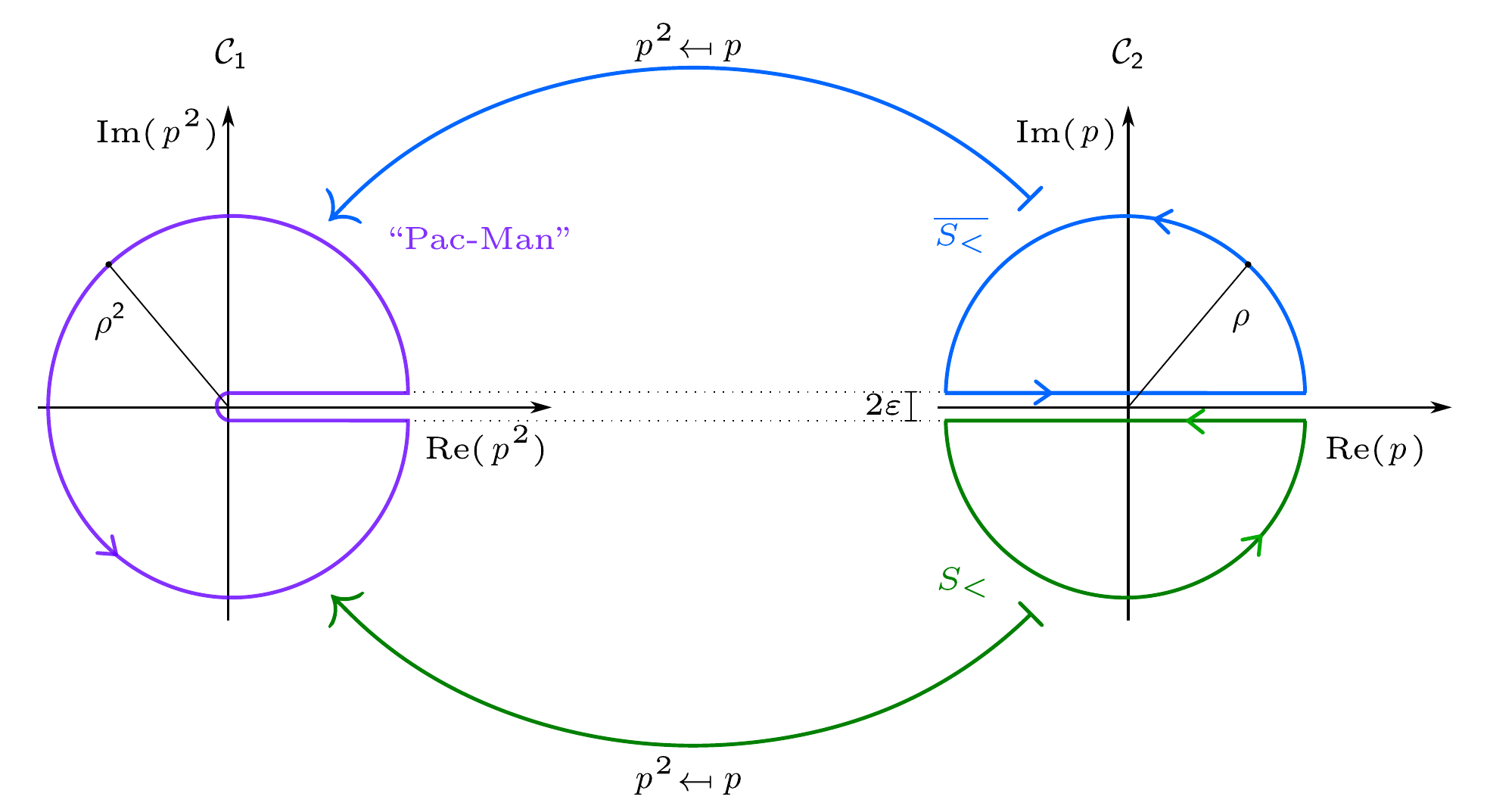}
\caption{Each semi-circle on the right-hand side is mapped to a ``pac-man'' contour.}
\label{fig:contour pacman and disks}
\end{figure}
Let us choose the ``pac-man'' contour and compute the integral
\begin{align*}
  \oint_{\mathcal{C}^\infty} d(p ^2) \mathcal{G}_{p}(r,r')=\lim\limits_{\varepsilon\rightarrow0}\lim\limits_{\rho\rightarrow\infty} \oint\limits_{\mathcal{C}_1} d(p^2)\mathcal{G}_{p}(r,r').
\end{align*}
In view of Equation \eqref{eq: property normalization and green conjugate}, we have
\begin{align}
  \label{eq: conjugate G omega < = G > conjugate omega}
\mathcal{G}_{p}(r,r')|_{\Imag(p)<0}=\overline{\mathcal{G}_{p}(r,r')|_{\Imag(\overline{p})>0}}.
\end{align}
Taking into account Jordan's lemma and Equation \eqref{eq: conjugate G omega < = G > conjugate omega}, we get
\begin{align}
  \label{eq:oirjfoijrf2233}
  \oint_{\mathcal{C}^\infty} d(p ^2) \mathcal{G}_{p}(r,r')&= \lim\limits_{\varepsilon\rightarrow0} \int_{0}^\infty  d(p^2) \left\{\mathcal{G}_{p+i\varepsilon}(r,r') - \overline{\mathcal{G}_{p+i\varepsilon}(r,r')}\right\}.
\end{align}
Consider the right and left propagating terms of $\mathcal{G}_{p}(r,r')$ for $\Imag(p)>0$, respectively:
\begin{align}
  \label{eq: def right and left propag mink 4d}
&\mathcal{G}_{p}^>(r,r') = \frac{1}{ \mathcal{N}_{p}} \Theta(r'-r) R_{\gamma_\ell}(r)R_{1(\infty)}(r'),
&\mathcal{G}_{p}^<(r,r') =\mathcal{G}_{p}(r,r') - \mathcal{G}_{p}^>(r,r').
\end{align}
The fundamental relation that connects the solution $R_{1(\infty)}(r)$ with those at $r=0$ reads
  \begin{equation}
  \label{eq: fundamental relation R1infty and R10 and R20 }
  R_{1(\infty)}(r) = A R_{1(0)}(r) + B R_{2(0)}(r),
  \end{equation}
with coefficients $A=1$ and $B=\frac{i}{p}$, as in Equation \eqref{eq: basis radialfunction infty}. The normalization \eqref{eq: normalization mink 4d} in terms of $A$ and $B$, reads
  \begin{equation}
  \label{eq: normalization A and B mink 4d}
  \mathcal{N}_p= \{B\cos(\gamma_\ell) + A\sin(\gamma_\ell)\}, \text{ for }\Imag(p)>0.
  \end{equation}
By direct substitution of Equations \eqref{eq: def right and left propag mink 4d}, \eqref{eq: fundamental relation R1infty and R10 and R20 } and \eqref{eq: normalization A and B mink 4d} into Equation \eqref{eq:oirjfoijrf2233}, for the right propagating term we find
\begin{align*}
  \oint_{\mathcal{C}^\infty} d(p ^2) \mathcal{G}_{p}^>(r,r')%
= & \Theta(r'-r)\int_{0}^{\infty}  d(p^2)\, \left\{ \frac{\overline{A}B- A\overline{B}}{|B\cos(\gamma_\ell) + A\sin(\gamma_\ell)|^2}  \right\}R_{\gamma_\ell}(r)R_{\gamma_\ell}(r').
\end{align*}
For the left propagating term we find an analogous expression. It follows that
\begin{align}
  \label{eq:this one to compare it}
\oint_{\mathcal{C}^\infty} d(p^2) \, \mathcal{G}_{p}(r,r') =   \int_{\mathbb{R}} d(p^2) \Theta(\omega-m_0) \left\{  \frac{1}{\pi}\frac{p}{\cos(\gamma_\ell)^2+p^2 \sin(\gamma_\ell)^2} R_{\gamma_\ell}(pr) R_{\gamma_\ell}(pr')\right\}.
\end{align}

\subsubsection*{Two-point functions }
\label{sec: Two-point functions Mink 4D}

The two-point functions for the ground and thermal states of a free, scalar, massive Klein-Gordon theory is given by Theorems  \ref{thm: 2 point ground state schd coord} and \ref{thm: 2 point KMS state schd coord}. On the spacetime $\mathcal{M}$ with line element \eqref{eq:Minkowski metric spherical 4D}, let us write their integral kernels as explicitly as possible. Consider first the ground state, whose two-point function has integral kernel
\begin{equation}\label{eq: 2 point ground state mink 4}
 \psi_{2,\infty}(x,x') = \lim_{\varepsilon \to 0^{+}} \sum_{\ell=0}^\infty \sum_{m=-\ell}^\ell \int_{0}^{\infty} d\omega e^{-i\omega(t-t'-i\varepsilon)}\widetilde{\psi}_2(r,r')Y_\ell^m(\theta,\varphi)\overline{Y_\ell^m(\theta',\varphi')}.
\end{equation}
where $Y_\ell^m(\theta,\varphi)$ are the spherical harmonics on the $2$-sphere with eigenvalues $\lambda_\ell^m=-\ell(\ell+1)$. On account of Equations \eqref{eq:radial part two point function 2} and \eqref{eq:this one to compare it}, the radial part $\widetilde{\psi}_2(r,r')$ is given by
\begin{equation}
  \label{eq: radial part psi mink 4d}
    \widetilde{\psi}_2(r,r')=\frac{1}{\pi}\frac{p}{\cos(\gamma_\ell)^2+p^2 \sin(\gamma_\ell)^2}\Theta(\omega-m_0)R_{\gamma_\ell}(pr) R_{\gamma_\ell}(pr'),
\end{equation}
where $p$ is taken as the positive solution to $p^2 = \omega^2-m^2_0$. We can sum with respect to the index $m$ by using the addition formula of the spherical harmonics \cite[Pg. 105]{jorgensen1987harmonic}
 \begin{equation*}%
  \sum_{m=-\ell}^\ell Y_\ell^m(\theta,\varphi)\overline{Y_\ell^m(\theta',\varphi')}= \frac{2\ell+1}{4\pi} P_\ell(\Phi),
\end{equation*}
 where $P_\ell$ is the Legendre function and $\Phi\equiv \cos(\theta)\cos(\theta')+\sin(\theta)\sin(\theta')\cos(\varphi-\varphi')$. Note that $\Phi$ is symmetric under the mappings $\theta\leftrightarrow\theta'$ and $\varphi\leftrightarrow\varphi'$. Then, Equation \eqref{eq: 2 point ground state mink 4} reads
 \begin{equation}\label{eq: 2 point ground state mink 4 summed in m}
  \psi_{2,\infty}(x,x') = \lim_{\varepsilon \to 0^{+}} \sum_{\ell=0}^\infty \int_{0}^{\infty} d\omega\, e^{-i\omega(t-t'-i\varepsilon)}\widetilde{\psi}_2(r,r')\frac{2\ell+1}{4\pi} P_\ell(\Phi).
 \end{equation}

For KMS states at inverse-temperature $\beta$ with respect to the Killing field $\partial_t$, we obtain
\begin{equation}\label{eq: 2 point kms state mink 4 summed in m}
 \psi_{2,\beta}(x,x') = \lim_{\varepsilon \to 0^{+}} \sum_{\ell=0}^\infty \int_{0}^{\infty} d\omega  \left[ \frac{ e^{-i \omega (t-t'-i\varepsilon)}}{1-e^{-\beta\omega}}  +  \frac{e^{+i \omega (t-t'+i\varepsilon)} }{e^{\beta\omega}-1} \right]\widetilde{\psi}_2(r,r')\frac{2\ell+1}{4\pi} P_\ell(\Phi).
\end{equation}

\subsubsection*{The transition rate}
\label{sec: The transition rate Mink 4D}
To compute the transition rate of a static Unruh-DeWitt detector as given by Theorem \ref{thm: Transition rate for the physically-sensible states}, one needs to take into account Equation \eqref{eq: 2jioejiojiuoe938df}, together with the identities $P_\ell(1)=1$ and, as given in \cite[(10.60.12)]{NIST_10},
 \begin{equation*}
   \sum_{\ell=0}^\infty (2\ell+1)  j_{\ell}(pr)^2 =1 .
 \end{equation*}
For the ground state, with two-point function as in Equation \eqref{eq: 2 point ground state mink 4}, we obtain
\begin{align}
  \label{eq: this one which is right above}
\dot{\mathcal{F}}_{\infty} =&  \frac{\Theta(-\Omega-m_0)}{2\pi r^2} \Bigg\{ p r^2 + \frac{\sin(\gamma)^2}{\cos(\gamma)^2+ p^2 \sin(\gamma)^2}\bigg[p\cos(2pr) + \cot(\gamma) \sin(2pr)\bigg]
\Bigg\}.
\end{align}
For the thermal state, with two-point function as in Equation \eqref{eq: 2 point kms state mink 4 summed in m}, Equation \eqref{eq: from thm transition rate KMS Schd coord} together with Equation \eqref{eq: this one which is right above} yields
\begin{equation*}
\dot{\mathcal{F}}_\beta =\frac{\text{sign}(\Omega)}{e^{\Omega\beta}-1}\dot{\mathcal{F}}_{\infty} ( -|\Omega|).
\end{equation*}

\subsection{On Minkowski spacetime (without a wire)}
\label{subsec: On Minkowski spacetime without a plate}

With the assumptions and definitions of the previous section, let us set Dirichlet boundary condition at the singularity $r=0$ and show that it yields the standard, closed form expressions on Minkowski spacetime \cite{birrell1984quantum}.

\subsubsection*{Two-point functions }
\label{sec: Two-point functions Mink 4D Dirichlet}

Define the auxiliary parameter $z :=  \sqrt{r^2 + r'^2 -2 r r' \Phi}$ and consider the identities \cite[(10.60.2)]{NIST_10} together with \cite[(10.23.6)]{NIST_10}):
\begin{equation*}
  \sum_{\ell=0}^\infty (2\ell+1) j_{\ell}(pr) j_{\ell}(pr')  P_\ell (\Phi) = \frac{\sin(pz)}{pz}.
\end{equation*}
Substituting Equation \eqref{eq: radial part psi mink 4d} in Equation \eqref{eq: 2 point kms state mink 4 summed in m} with $\gamma=0$, for a KMS state we have
\begin{equation}\label{eq: G+ k=0 mink thermal 1}
\psi_{2,\beta}(x,x') = \lim_{\varepsilon \to 0^{+}} \int_{m_0}^{\infty} d\omega \left[ \frac{e^{-i \omega (t-t'-i\varepsilon)}}{1-e^{-\beta\omega}}  +  \frac{e^{+i \omega (t-t'+i\varepsilon)} }{e^{\beta\omega}-1}\right] \frac{1}{4\pi^2}
 \frac{\sin(\sqrt{\omega^2-m_0^2} z)}{ z}.
\end{equation}
For the ground state, instead, we have
\begin{equation}\label{eq: G+ k=0 mink 4}
  \psi_{2,\infty}(x,x') = \lim_{\varepsilon \to 0^{+}} \int_{m_0}^{\infty} d\omega e^{-i \omega (t-t'-i\varepsilon)}\frac{1}{4\pi^2}
   \frac{\sin(\sqrt{\omega^2-m_0^2} z)}{ z}.
\end{equation}
Let us focus on the ground state and write the above expression in a closed form. Let $\sigma_\varepsilon:= z^2 -(t-t'-i\varepsilon)^2 $. For $m_0=0$, the integral in Equation \eqref{eq: G+ k=0 mink 4} yields
\begin{equation}\label{eq: G+ k=0 mink 7}
\psi_{2,\infty}(x,x') =  \lim_{\varepsilon \to 0^{+}} \frac{1}{4\pi^2\sigma_\varepsilon}.
\end{equation}
For $m_0>0$, writing it as an integral in $p$ and using expression \cite[3.914.9]{gradshteyn} we obtain
\begin{equation}\label{eq: G+ k=0 mink 9}
\psi_{2,\infty}(x,x') =   \lim_{\varepsilon \to 0^{+}} \frac{m_0}{4\pi^2}  \frac{K_1(m_0 \sqrt{\sigma_\varepsilon})}{\sqrt{\sigma_\varepsilon} },
\end{equation}
where $K_1$ is the modified Bessel function of second kind. Note that taking the limit $m_0\rightarrow 0$ of Equation \eqref{eq: G+ k=0 mink 9} we obtain Equation \eqref{eq: G+ k=0 mink 7}.

\subsubsection*{Thermal fluctuations}
\label{sec: Thermal fluctuations Mink 4D Dirichlet}

To obtain the regularized thermal fluctuations $\psi_\beta(\Psi^2(f))$ for a Gaussian, Hadamard, thermal state $\psi_\beta$ with kernel \eqref{eq: G+ k=0 mink thermal 1}, one could apply the regularization prescribed in Example \ref{eg: Regularized vacuum fluctuations}. However, for simplicity, and since the difference between the two-point functions of two Hadamard states yields a smooth function, per Definition \ref{def: Hadamard state, local Hadamard form}, let us regularize the thermal state using the ground state instead of using the Hadamard parametrix. That is, let us compute the thermal contribution to the ground state fluctuations. We call
\begin{align}\label{eq: G+ k=0 mink thermal diff 1}
 \Delta \psi_{2,\beta}(x,x') :&= \psi_{2,\beta}(x,x') -  \psi_{2,\infty}(x,x')\nonumber \\
       &= \int_{m_0}^{\infty} d\omega \frac{ 2\cos(\omega(t-t'))}{e^{\beta\omega}-1}\frac{1}{4\pi^2}
 \frac{\sin(\sqrt{\omega^2-m_0^2} z)}{ z}.
\end{align}
For the massless case, we can perform the integration in Equation \eqref{eq: G+ k=0 mink thermal diff 1} analytically. Its coincidence point limit yields the thermal contribution to the ground state fluctuations:
\begin{equation*}%
 \Delta \psi_{2,\beta}  := \lim\limits_{x'\rightarrow x}  \Delta \psi_{2,\beta}(x,x')=\frac{1}{12 \beta ^2}.
\end{equation*}

\subsubsection*{Regularized energy-momentum tensor}
\label{sec: Regularized energy-momentum tensor Mink 4D Dirichlet}

Let us compute the thermal contributions to the energy density of the ground state. With Example \ref{eg: Regularized energy-momentum tensor} in mind, consider the energy-momentum tensor of the regularized thermal state determined by Equation \eqref{eq: G+ k=0 mink thermal diff 1}:

\begin{equation*}
  \braket{:T_{\mu\nu}:}_\beta := \lim\limits_{x'\rightarrow x}\left\{\mathcal{D}_{\mu\nu}(x,x')\left[ \Delta \psi_{2,\beta}(x,x') \right]\right\},
\end{equation*}
where the differential operator is the same as that of Equation \eqref{eq: differential operator stress energy tensor moretti}. The thermal contribution to the energy density is simply $ \rho := \braket{:T_{00}:}_\beta $, with
\begin{align*}%
\mathcal{D}_{00}(x,x') = \frac{1}{2 }\left\{\partial_t\partial_{t'} +  \partial_r\partial_{r'} + \frac{1}{r r'}\partial_\theta\partial_{\theta'} + \frac{1}{r r'} \csc(\theta)\csc(\theta')\partial_\phi\partial_{\phi'}\right\}.
\end{align*}
Numerically, it is straightforward to verify that%
\begin{equation*}
     \rho = \frac{\pi ^2 }{30 \beta ^4}.
\end{equation*}
\subsubsection*{Transition rate}
\label{sec: Transition rate Mink 4D Dirichlet}
For the transition rate, setting $\gamma=0$ in Section \ref{sec: The transition rate Mink 4D} yields the following standard results on Minkowski spacetime: for the vacuum state
\begin{align*}
\dot{\mathcal{F}}_{\infty}=\frac{\sqrt{\Omega^2-m_0^2}\Theta(-\Omega-m_0)}{2\pi},
\end{align*}
and for a thermal state at inverse-temperature $\beta$
\begin{align*}
\dot{\mathcal{F}}_\beta=\frac{\text{sign}(\Omega)}{e^{\Omega\beta}-1}\frac{\sqrt{\Omega^2-m_0^2}\Theta(|\Omega|-m_0)}{2\pi}.
\end{align*}

\cleardoublepage
\chapter{Applications}
\label{chap: Applications}
\minitoc
\pagestyle{myPhDpagestyle3}
\vfill


In this chapter I employ the quantum field theoretical framework outlined in Chapter \ref{chap: Quantum Field Theory on Static Spacetimes} on static spacetimes described in Chapter \ref{chap: The Spacetimes}. Let $\mathcal{M}$ be an $n$-dimensional, static, stably-causal, non-globally hyperbolic spacetime equipped with a Lorentzian metric tensor $g$, with scalar curvature $\mathbf{R}\in\mathbb{R}$, and endowed with Schwarzschild-like coordinates $(t,r,\underline{\theta})$, $\underline{\theta}\equiv(\theta,\varphi_1,...,\varphi_{n-3})$, as per Definition \ref{def: Schwarzschild-like coordinates}.
On $\mathcal{M}$, consider a free, scalar Klein-Gordon field $\Psi:\mathcal{M}\rightarrow\mathbb{R}$ with mass $m_0\geq 0$ (or effective mass $m_{\text{eff}}^2$, as per Equation \eqref{eq: eff mass}), coupled to curvature through the parameter $\xi\in\mathbb{R}$ and subjected to Equation \eqref{eq: KG}. Throughout this chapter we take $\mathcal{M}$ to be either a static BTZ black hole, Rindler-AdS$_3$, a massless flat, hyperbolic or spherical topological black hole, a flat, hyperbolic or spherical Lifshitz topological black hole, or a global monopole. In each scenario, the procedure followed for the construction of physically-sensible states is the same as the one exemplified for Minkowski spacetime in Section \ref{sec: Examples chapter 2}. Notwithstanding, different thermal phenomena are studied on different spacetimes, and they are discussed accordingly within each section. The results obtained and summarized in this chapter have been published in \cite{deSouzaCampos2020ddx,deSouzaCampos2020bnj,deSouzaCampos2021awm,deSouzaCampos2021role}.

\newpage

\section{On a static BTZ and Rindler-AdS$_3$}
\label{sec: On a static BTZ spacetime and on Rindler-AdS3}

In this section I consider Klein-Gordon fields on a static BTZ black hole and on Rindler-AdS$_3$. First, I outline the construction of physically-sensible states, following the same steps as those on Minkowski spacetime taken in Section \ref{sec: Examples chapter 2}. I remark that such analysis has been performed for the more general and intricate scenario of rotating BTZ black holes in \cite{bussola2017ground,bussola2018tunnelling}. Here, in this regard, I merely summarize their results imposing vanishing angular momentum. Then, with the ground and thermal states in hand, in Section \ref{subsec: The transition rate on a static BTZ spacetime and on Rindler-AdS3}, I employ the Unruh-DeWitt detector approach to study quantum effects: the recently disclosed and still-puzzling anti-Unruh and anti-Hawking effects, as described in Section \ref{sec: Unruh, Hawking, anti-Unruh, anti-Hawking effects}. These anti-correlation effects were discovered and studied on BTZ black holes for the ground state with Dirichlet, transparent and Neumann boundary conditions in \cite{Henderson2019uqo}.
Notwithstanding, the results exposed here have been published in \cite{deSouzaCampos2020ddx} and generalize this previous work by admitting more general boundary conditions, by also taking into account thermal states and by including Rindler-AdS$_3$ in the analysis.

Geometric features of BTZ black holes are outlined in Section \ref{ex: Chapter 1 BTZ}. For $f(r):=\frac{r^2-r_h^2}{L^2}$, the line element for the static scenario in Schwarzschild-like coordinates reads
  \begin{equation}
    \label{eq: metric on BTZ chapter 3}
      ds^2=-f(r)dt^2+f(r)^{-1}dr^2+r^2d\theta^2,
  \end{equation}
where $t\in\mathbb{R}$, $r>r_h>0$, $\theta\in[0,2\pi)$, and $M\equiv r_h^2/L^2$ is the black hole mass. The Killing vector field $\partial_t$ generates a bifurcate Killing horizon at $r=r_h$ with associated local Hawking temperature
  \begin{equation}
    \label{eq:local hawking temperature btz}
      T_H := \frac{r_h}{2\pi L\sqrt{r^2-r_h^2}}.
  \end{equation}
The unwrapping of the angular coordinate yields the universal covering of BTZ: on Rindler-AdS$_3$, we have $\theta\in\mathbb{R}$. Since both spacetimes are described by Equation \eqref{eq: metric on BTZ chapter 3}, we take both into account concomitantly, bearing in mind the different behaviour with respect to the angular coordinate.

\subsection{Ground and thermal states}
\label{subsec: Ground and thermal states BTZ spacetime and on Rindler-AdS3}

\subsubsection*{The radial equation}
\label{subsec: The radial equation on a static BTZ spacetime and on Rindler-AdS3}

By mode-decomposition, as per Equation \eqref{eq: ansatz KG in mode decomposition}, we assume
\begin{equation*}
    \Psi(t,r,\theta)=e^{-i\omega t} R(r)Y_\ell(\theta),
\end{equation*}
where $Y_\ell(\theta) = e^{i\ell\theta}$ are the eigenfunctions of the one-dimensional Laplacian $\Delta_1$. Note that on BTZ we have $\ell\in \mathbb{Z}$, while on Rindler-AdS$_3$, $\ell\in\mathbb{R}$. Taking into account the line element given in Equation \eqref{eq: metric on BTZ chapter 3}, the Klein-Gordon Equation \eqref{eq: KG in Schd-like coordinates} yields the radial Equation \eqref{eq: the radial equation Shd-like coord}. Applying the coordinate change
\begin{equation}
  \label{eq: coord transformation BTZ}
r\mapsto z= \frac{r^2-r_h^2}{r^2}\in(0,1),
\end{equation}
and using the Frobenius method, we find that the radial solutions behave asymptotically as
\begin{align}
  \label{eq: radial sol asymp btz rindlerads}
    R(z)\sim \begin{cases}
            z^{\alpha_\pm},   &\text{ as }z\rightarrow 0 \quad (r\rightarrow r_h),\\
          (1-z)^{\beta_\pm}, & \text{ as }z\rightarrow 1\quad (r\rightarrow \infty);
    \end{cases}
\end{align}
with exponents
\begin{equation*}
  \alpha_\pm := \pm \frac{i L^2 \omega }{2 r_h}  \quad \text{ and }\quad \beta_\pm := \frac{1}{2}\left(1 \pm \nu \right),
\end{equation*}
and auxiliary parameter
\begin{equation*}
\nu:=\sqrt{1+ L^2 m_{\text{eff}}^2}.
\end{equation*}
Consequently, we obtain that the ansatz
      \begin{equation*}
        R(z) =z^{\alpha_+} (1-z)^{\beta_+} \zeta(z)
      \end{equation*}
solves the radial equation if and only if $\zeta(z)$ solves the hypergeometric equation
    \begin{equation}\label{eq:hypergeom eq BTZ}
    \left\{z(1-z)\frac{d^2}{dz^2}+(c-(a+b-1)z))\frac{d}{dz}-ab \right\}\zeta(z)=0,
    \end{equation}
such that, for $\Upsilon_\ell:= i\frac{\ell}{ 2r_h}$, the parameters are given by
  		\begin{align*}
  			a &:= \alpha_+ + \beta_+ + \Upsilon_\ell	, \\
  			b &:=  \alpha_+ + \beta_+ -	\Upsilon_\ell, \\
  			c &:= 1 + 2\alpha_+. 
  		\end{align*}

\subsubsection*{The radial solutions}
\label{subsec: The radial solutions on a static BTZ spacetime and on Rindler-AdS3}

The integrality conditions of the hypergeometric parameters $a$, $b$ and $c$ given in the previous section, as well as the square-integrability analyses regarding the radial solutions, and the selection of suitable bases of solutions at the endpoints are carefully described in \cite{bussola2017ground,bussola2018tunnelling}. In addition, in Section \ref{subsec: The radial solutions on massless hyperbolic black holes} I go into detail of the analysis for the $n$-dimensional generalizations of BTZ spacetimes that involves the same hypergeometric functions and also holds for $n=3$. For those two reasons, in this section I simply state which are the radial solutions to be taken into account.

First, let $\Gamma(z)$ be the Euler Gamma function, and let
\begin{equation*}
F(a,b,c;z):= \sum_{s=0}^\infty \frac{(a)_s(b)_s}{\Gamma(c+s)s!}z^s
\end{equation*}
be the hypergeometric function with a branch cut in the sector $|\text{ph}(1-z)|\leq\pi$ of the complex plane. Denote the \textit{standard} hypergeometric solutions by
\begin{subequations}
  \label{eq:solhypergeo BTZ RinderAdS}
  \begin{align}
    &\zeta_{1(0)}(z)=F(a,b;c;z),\label{eq:sol10hypergeo BTZ RinderAdS}\\
    &\zeta_{2(0)}(z)=z^{1-c}F(a-c+1,b-c+1;2-c;z),\label{eq:sol20hypergeo BTZ RinderAdS}\\
    &\zeta_{1(1)}(z)=F(a,b;a+b+1-c;1-z),\label{eq:sol11hypergeo BTZ RinderAdS}\\
    &\zeta_{2(1)}(z)=(1-z)^{c-a-b}F(c-a,c-b;c-a-b+1;1-z).\label{eq:sol21hypergeo BTZ RinderAdS}
  \end{align}
\end{subequations}

We find that the endpoint $z=0$ is limit point and that the most general square-integrable solution is given by
    \begin{equation*}
      R_{0}(z) := z^{\alpha_+} (1-z)^{\beta_+} \{ \zeta_{1(0)}(z)\Theta(-\Imag(\omega)) +  \zeta_{2(0)}(z)\Theta(\Imag(\omega)) \}.
    \end{equation*}
The endpoint $z=1$, however, is limit circle if $\nu\in(0,1)$, and limit point otherwise. In the former case, the solution $R_{\gamma}(z)$ that satisfies the generalized Robin boundary condition at $z=1$ parametrized by $\gamma\in[0,\pi)$ is given by
  		\begin{equation*}
      			  R_{\gamma}(z)=z^{\alpha_+} (1-z)^{\beta_+} \big\{\cos(\gamma) \zeta_{1(1)}(z)  + \sin(\gamma) \zeta_{2(1)}(z)\big\}.
        \end{equation*}

\subsubsection*{Spectral resolution of the radial Green function}
\label{subsec: Spectral resolution of the radial Green function on a static BTZ spacetime and on Rindler-AdS3}

The spectral analysis of the radial Green function described in Section \ref{sec: Physically-sensible states in Schwarzschild-like coordinates}, performed for the BTZ black hole in \cite{bussola2017ground,bussola2018tunnelling}, selects the boundary conditions for which the radial Green function has no poles and yields the radial part of the two-point functions of interest, see Remark \ref{rem: bound states}. Specifically, it gives a threshold,
\begin{equation}
\gamma_c^\ell = \arctan\left(\frac{\Gamma(\nu)}{\Gamma(-\nu)}\frac{\Gamma(1-\beta_+ +\Upsilon_\ell)}{\Gamma(\beta_+ +\Upsilon_\ell)}\frac{\Gamma(1-\beta_+ -\Upsilon_\ell)}{\Gamma(\beta_+ -\Upsilon_\ell)}\right),
\end{equation}
for the existence of bound state modes in the two-point function. That is, for $\gamma\in[0,\gamma_c^\ell)$, the radial part (seeing $z$ as a function of $r$)
\begin{equation}
  \label{eq: phitilde2 BTZ Rindler-AdS}
  \widetilde{\psi}_{2}(r,r') = \frac{1}{2 i \pi \nu}\frac{(\overline{A}B-A\overline{B})}{ |A\sin(\gamma) -B\cos(\gamma)|^2} R_\gamma(r)R_{\gamma}(r'),
\end{equation}
with coefficients
\begin{subequations}
  \label{eq:CoefAeB}
    \begin{align}
      A=\frac{\Gamma(c)\Gamma(c-a-b)}{\Gamma(c-a)\Gamma(c-b)},\quad\text{ and }\quad
      B=\frac{\Gamma(c)\Gamma(a+b-c)}{\Gamma(a)\Gamma(b)}, \label{eq:Coefs A B BTZ Rindler-AdS}
    \end{align}
\end{subequations}
is well-defined. To impose the same boundary condition for all $\ell$-modes, then one should take $\gamma\in [0,\gamma_c]$, where $\gamma_c:=\gamma_c^\infty=\frac{\pi}{2}.$

\subsubsection*{Two-point functions}
\label{subsec: Two-point functions on a static BTZ spacetime and on Rindler-AdS3}

With Equation \eqref{eq: phitilde2 BTZ Rindler-AdS} and Theorems \ref{thm: 2 point ground state schd coord} and \ref{thm: 2 point KMS state schd coord}, we obtain the two-point functions for physically-sensible states. The results reported in \cite{bussola2017ground,bussola2018tunnelling} for a BTZ spacetime are straightforwardly generalized to include Rindler-AdS$_3$ spacetime simply by changing the behaviour of the angular component of the two-point function accordingly. The two-point function for a ground state reads
    \begin{equation}
      \label{eq:2 point ground state BTZ RindlerAdS3}
      \psi_{2,\infty}(x,x')=\lim_{\varepsilon\rightarrow 0^+}\int_{\sigma(\Delta_1)}d\eta_\ell\int_{\mathbb{R}} d\omega\,\Theta(\omega) e^{-i\omega (t-t'-i\varepsilon)}   \widetilde{\psi}_{2}(r,r')  e^{i\ell(\theta-\theta')},
    \end{equation}
while for a KMS state at inverse-temperature $\beta$ it is given by
    \begin{align}
      \label{eq:2 point KMS state BTZ RindlerAdS3}
      \psi_{2,\beta}(x,x')=\lim_{\varepsilon\rightarrow 0^+} \int_{\sigma(\Delta_1)} d\eta_\ell \int_{\mathbb{R}} d\omega \,\Theta(\omega) \left[ \frac{e^{-i\omega (t-t'-i\varepsilon)}}{1 - e^{-\beta\omega}}+ \frac{e^{+i \omega (t-t'+i\varepsilon)}}{{e^{\beta\omega}-1}} \right]   \widetilde{\psi}_{2}(r,r')  e^{i\ell(\theta-\theta')}
    .%
    \end{align}
On BTZ spacetime, the integral with respect to the measure $d\eta_\ell$ associated with the spectrum $\sigma(\Delta_1)$ of the Laplacian reduces to a sum over $\ell\in\mathbb{Z}$, while on Rindler-AdS$_3$ it corresponds to an integral over $\ell\in\mathbb{R}$ with respect to the standard Lebesgue measure. Moreover, note that the Hartle-Hawking state is given by Equation \eqref{eq:2 point KMS state BTZ RindlerAdS3} with $\beta=2\pi/r_h$.

\subsection{The transition rate of an Unruh-DeWitt detector}
\label{subsec: The transition rate on a static BTZ spacetime and on Rindler-AdS3}

Let us now probe the quantum states given in the previous section. Consider an Unruh-DeWitt detector modelled by a two-level system with energy gap $\Omega$ as outlined in Section \ref{sec: Probing quantum states with particle detectors}. For simplicity, assume the detector is following a static trajectory, hence at fixed spatial position $(r,\theta)$, of constant (supercritical) proper acceleration
\begin{equation*}
  a = \frac{r}{L\sqrt{r^2-r_h^2}}>\frac{1}{L},
\end{equation*}
given by a complete integral line of the Killing field $\partial_t$. It is worth mentioning that the temperature observed via the GEMS method \cite{deser1997accelerated,deser1999mapping} coincides with the local Hawking temperature:
\begin{equation*}
  T_{GEMS}=\frac{1}{2\pi}\sqrt{a^2 - \frac{1}{L^2}}\equiv T_H.
\end{equation*}

Note that the temperature is only well-defined for supercritically accelerated observers, that is $T_{GEMS}>0$ if and only $a>\frac{1}{L}$. In this scenario, the transition rate for when the detector is coupled to physically-sensible states is given by Theorem \ref{thm: Transition rate for the physically-sensible states}. Substituting Equations \eqref{eq:2 point ground state BTZ RindlerAdS3} and \eqref{eq:2 point KMS state BTZ RindlerAdS3} in Equations \eqref{eq: from thm transition rate ground Schd coord} and \eqref{eq: from thm transition rate KMS Schd coord}, respectively, we obtain the following expressions. When coupled to the ground state the transition rate is
    		\begin{align}
    			\label{eq:transition rate ground state BTZ}
          \dot{\mathcal{F}}_\infty =  \Theta( -\Omega )   \int_{\sigma(\Delta_1)} d\eta_\ell \,  \frac{2}{ \nu}\frac{\Imag(\overline{A}B)}{ |A\sin(\gamma) -B\cos(\gamma)|^2} R_\gamma(r)^2 \Big|_{\omega = -\sqrt{f(r)} \,\Omega},
    		\end{align}
while for a KMS state at inverse-temperature $\beta$, it reads
    		\begin{align}
    			\label{eq:transition rate KMS state BTZ}
          \dot{\mathcal{F}}_{\beta} = \frac{\text{ sign}(\Omega)}{e^{\beta\sqrt{f(r)}\Omega}-1} \int_{\sigma(\Delta_1)} d\eta_\ell \,  \frac{2}{\nu}\frac{\Imag(\overline{A}B)}{ |A\sin(\gamma) -B\cos(\gamma)|^2} R_\gamma(r)^2 \big|_{\omega =\sqrt{f(r)} |\Omega| }.
    		\end{align}

For a massless conformally coupled Klein-Gordon field, the behaviour of the $\ell=0$ contribution to the transition rates given by Equations \eqref{eq:transition rate ground state BTZ} and \eqref{eq:transition rate KMS state BTZ} and with respect to the local Hawking temperature $T_H$ is illustrated in Figure \ref{fig:groundandkms1}. Figure \ref{fig:unddfdfderst1} illustrates their behavior with respect to the truncation order  $\ell_{max}$. Note that this variable temperature is not the temperature of the Boulware-like ground state, but rather it is the one that would be detected for the KMS state at the corresponding locus. In addition, for aesthetics, they are normalized with respect to the maximum value of the quantity plotted.

\begin{figure}[H]
\centering
    \includegraphics[width=.45\textwidth]{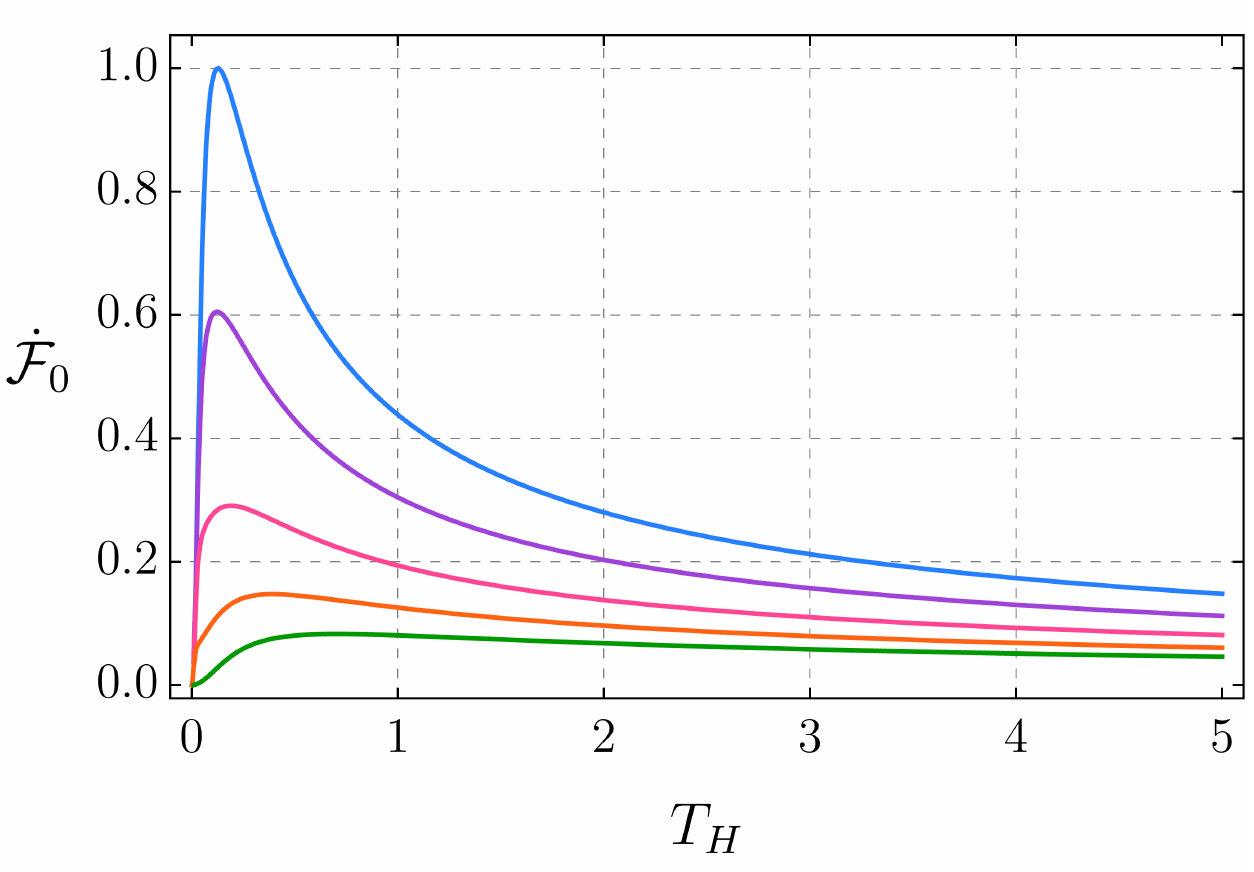}\hspace{.5cm}%
    \includegraphics[width=.45\textwidth]{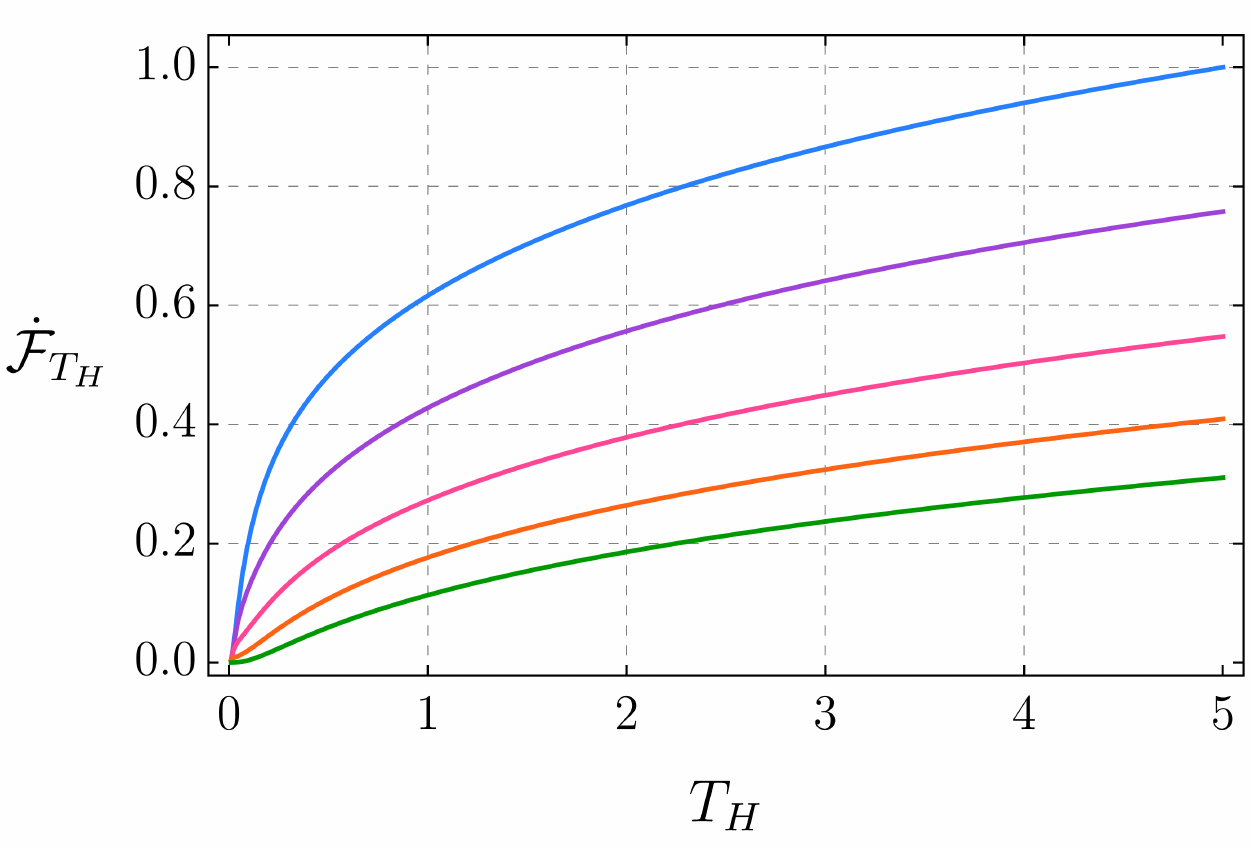}
\caption{ \label{fig:groundandkms1}$\ell=0$ contribution to the transition rate as a function of the Hawking temperature for $r_h=1$, $\Omega=-0.1$ and different boundary conditions; from top to bottom, respectively, $\gamma=(0.50,0.47,0.40,0.25,0)\pi$. On the left, for the ground state; on the right, for the Hartle-Hawking state.}
\end{figure}

Recall the anti-correlation effects defined in Section \ref{sec: Unruh, Hawking, anti-Unruh, anti-Hawking effects}. We see that the $\ell=0$ term of the transition rate for the ground state does decrease as we approach the horizon, which corresponds to $T_H\rightarrow \infty$, for all boundary conditions. That is, the anti-Hawking effect is manifest. However, performing the sum over $\ell$ may drastically change said behavior. The numerical analyses performed, which are available at \cite{git_unruh_deWitt_BTZ} in a Mathematica notebook and in a Jupyter notebook, confirmed that to be the case, depending on the value of $r_h$. Specifically, we verify that for $r_h=1$ the anti-Hawking effect seems to be cancelled out by the $\ell>0$ contributions for Dirichlet boundary condition, while it remains clearly visible for Neumann boundary condition. At the same time, for the thermal state, we find that the transition rate is a monotonically increasing function of $T_H$ in all scenarios, i.e. for all boundary conditions, for small or large values of $r_h$ and also after summing in $\ell$. The counterpart of this analysis on Rindler-AdS$_3$ simply requires that we integrate in $\ell$ instead of performing a discrete sum. In this case, we obtain analogous results for the occurrence of the anti-Unruh effect.

          				\begin{figure}[H]
          				\centering
          						\includegraphics[width=0.45\textwidth]{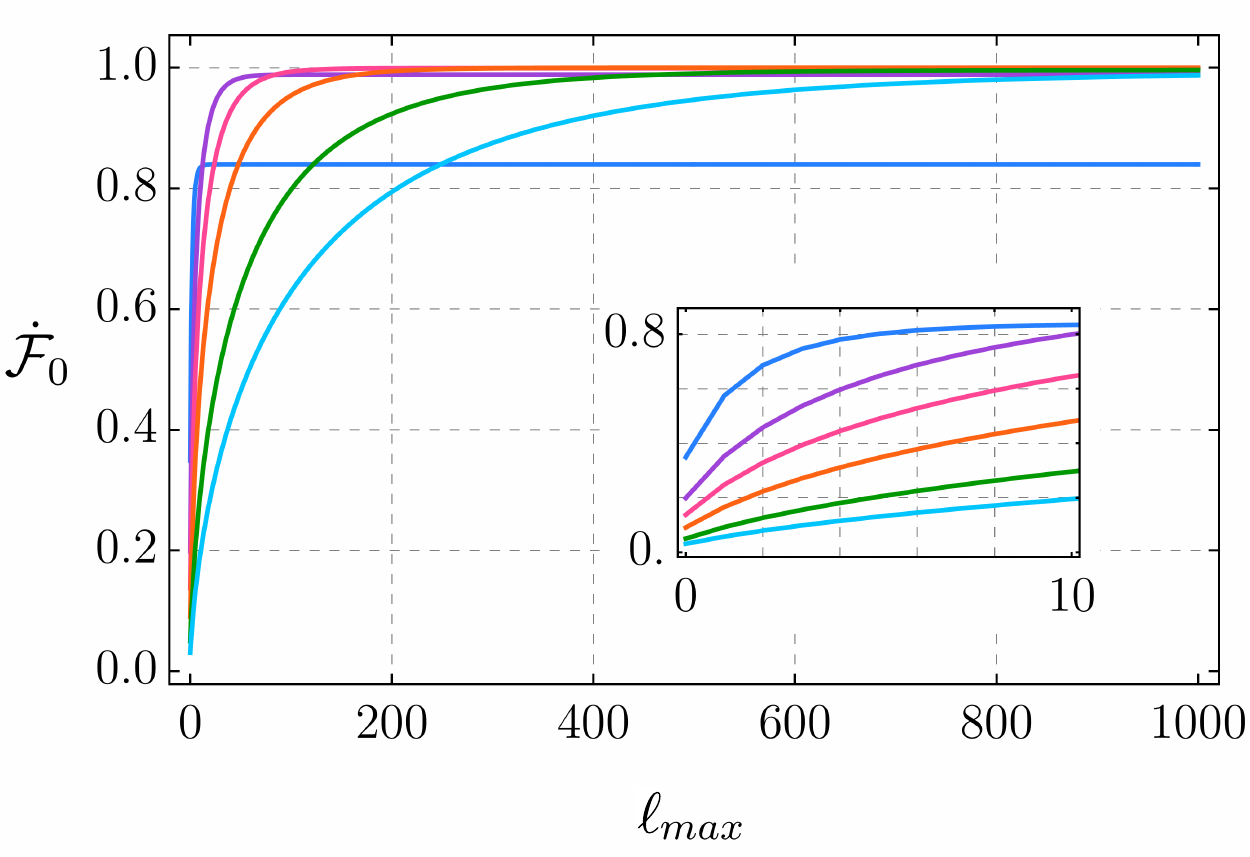}\hspace{.5cm}%
                      \includegraphics[width=.45\textwidth]{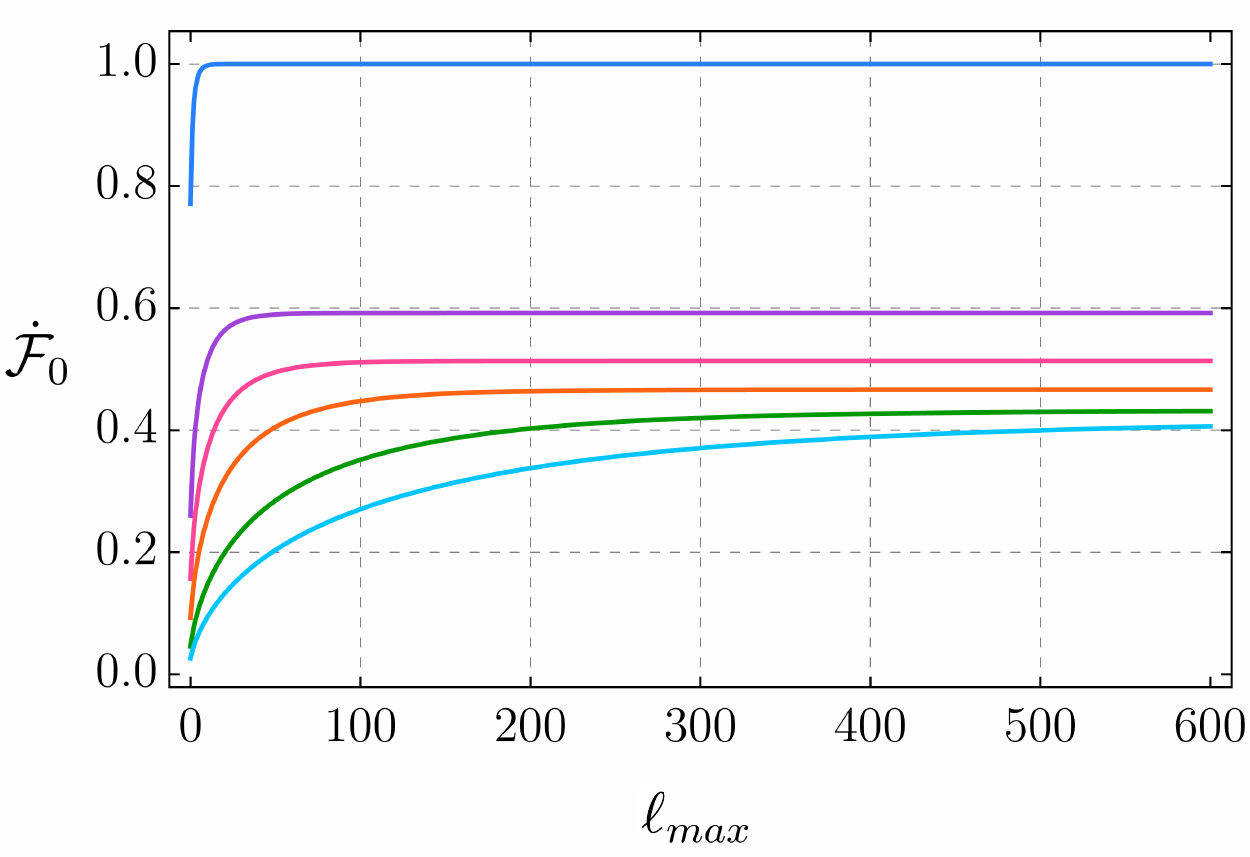}
          				\caption{ Transition rate summed over all natural numbers up to $\ell_{max}$ at several temperatures for $r_h=1$, $\Omega=-0.1$ and Dirichlet boundary condition, from top to bottom in the zoom plot $T_H=(1,5,10,20,50,100)$. On the left, for the ground state: it does look like the transition rate converges to a monotonically increasing function of temperature. On the right, for the KMS state: the curves are manifestly separated, and curves of smaller temperatures remain higher than curves of higher temperatures.}
          				\label{fig:unddfdfderst1}
          				\end{figure}

%


\section{On massless hyperbolic black holes}
\label{sec: On massless hyperbolic black holes}

In this section I apply the quantum field theoretical framework of Chapter \ref{chap: Quantum Field Theory on Static Spacetimes} to the study of the anti-correlation effects on massless hyperbolic black holes. I construct both ground and thermal states for arbitrary effective masses, general Robin boundary conditions at AdS conformal infinity and for spacetimes of any dimension $n\geq3$, generalizing previous works with specific parameters \cite{MannAndAbdalla,morley2021vacuum}. First, in Section \ref{subsec: Ground and thermal states on massless hyperbolic black holes}, I construct these physically-sensible states. Then, in Section \ref{subsec: The transition rate on massless hyperbolic black holes}, I give explicit expressions for the transition rate of a static Unruh-DeWitt detector.  For the case of massless conformally coupled scalar fields, I summarize the numerical analysis performed to study the anti-Hawking effect. The results I present generalize those in the three dimensional case \cite{Henderson2019uqo,deSouzaCampos2020ddx} described in the previous section and that were published in \cite{deSouzaCampos2020bnj}.

Topological black holes are described in Section \ref{ex: chapter 1 toplogical black holes}. The case of $M=0$ and negative, constant sectional curvature corresponds to massless hyperbolic black holes. Let $L$ be the AdS radius, $n\geq 3$ the number of spacetime dimensions, and $d\Xi_{n-2}^2$ be the unit metric on the $(n-2)$-dimensional hyperbolic space. In Schwarzschild-like coordinates, we take $t\in \mathbb{R}$, $r>r_h>0$, $\theta \in\mathbb{R}$, $\varphi_1 \in [0,2\pi), \text{ for }n\geq 4$ and $\varphi_2,...,\varphi_{n-3}\in [0,\pi), \text{ for } n>4$. For $ f(r) := \frac{r^2- r_h^2}{L^2}$, their line element reads
\begin{equation*}
		ds^2=-f(r)dt^2+f(r)^{-1}dr^2+r^2d\Xi_{n-2}^2.
	\end{equation*}
The Killing horizon at $r=r_h$ yields the Hawking temperature as per Equation \eqref{eq:local hawking temperature btz}.

\subsection{Ground and thermal states}
\label{subsec: Ground and thermal states on massless hyperbolic black holes}

\subsubsection*{The radial equation}
\label{subsec: The radial equation on massless hyperbolic black holes}
Given the symmetries of the underlying spacetime, let $Y_\ell(\theta)$ be the eigenfunctions of the Laplacian on the $(n-2)$-hyperbolic space with eigenvalues $\lambda_\ell=-\left( \ell^2 + \left(\frac{n-3}{2}\right)^2 \right)$ and let us assume that a solution of the Klein-Gordon equation can be written as
\begin{equation*}
  \Psi(t,r,\underline{\theta})=e^{-i\omega t}R(r)Y_\ell(\underline{\theta}).
\end{equation*}

Applying the coordinate transformation given by Equation \eqref{eq: coord transformation BTZ}, the radial equation, as per Equation \eqref{eq: the radial equation Shd-like coord} here written in the new coordinate $z$, reads
\begin{equation}
  \label{eq: radial eq z coord massless top bh}
\Big\{4z(1-z)\frac{\partial^2  }{\partial z^2}  + \left( 4+(2n-10)z \right) \frac{\partial  }{\partial z}   +  \frac{L^2 }{r_h^2}\lambda_\ell + \frac{L^4\omega^2}{r_h^2 z} -\frac{m_{\text{eff}}^2 }{1-z}\Big\} R(z) =0.
\end{equation}
Analogously to the previous section, the radial solutions behave asymptotically as per Equation \eqref{eq: radial sol asymp btz rindlerads}, though with exponents given by
		\begin{align}
			\label{eq:auxiliary parameters alpha and beta plus and minus massless hyperbolic bh}
				&\alpha_\pm= \pm \frac{iL^2\omega}{2r_h} \quad \text{ and }\quad \beta_\pm = \frac{1}{2} \left( \frac{(n-1)}{2} \pm \nu\right).
		\end{align}
Thus, considering the ansatz
\begin{equation*}
  R(z) =z^{\alpha_+} (1-z)^{\beta_+} \zeta(z),
\end{equation*}
we find that $\zeta(z)$ satisfies the Gauss hypergeometric equation \eqref{eq:hypergeom eq BTZ}, but with parameters generalized to $n$ dimensions
\begin{subequations}
  \label{eq:hypergeom param BTZ 2}
  \begin{align}
    a &:= \alpha_+ + \beta_+ + \Upsilon_\ell	-\frac{n-3}{4}, \label{eq:a n}\\
    b &:=  \alpha_+ + \beta_+ -	\Upsilon_\ell -\frac{n-3}{4}, \label{eq:b n}\\
    c &:= 1 + 2\alpha_+. \label{eq:c n}
  \end{align}
with auxiliary parameters
\begin{align}
	\label{eq:auxiliary parameter nu upsilon massless top bh}
  &\nu :=  \sqrt{\frac{(n-1)^2}{4} + L^2 m_{\text{eff}}^2 }>0 \quad \text{ and }\quad \Upsilon_\ell :=  \sqrt{  \left(\frac{n-3}{4}\right)^2 + \frac{L^2}{4 r_h^2} \lambda_\ell }.
\end{align}
\end{subequations}
The restriction $\nu>0$ corresponds to the Brei\-ten\-lohn\-er-Freed\-man bound \cite{Breitenlohner1982jf} on the effective mass given by $L^2m_{\text{eff}}^2 >-\frac{(n-1)^2}{4}$. Moreover, it is interesting to note that in the special case $L=r_h$, the auxiliary parameter $\Upsilon_\ell$ does not depend on the spacetime dimension $n$ and it reads simply $i\ell/2$. In fact, in this case, none of the parameters $a$, $b$ and $c$ depend on $n$.

\subsubsection*{The radial solutions}
\label{subsec: The radial solutions on massless hyperbolic black holes}

Casting the radial equation \eqref{eq: radial eq z coord massless top bh} as a Sturm-Liouville problem, as defined in Equation \eqref{eq: def P, Q S of sturm liouville op} with eigenvalue $\omega^2$ and Sturm-Liouville operator
\begin{align*}
  L_{\omega^2}:&=-\frac{1}{\frac{L^4(1-z)^{\frac{1-n}{2}}}{4r_h^2 z}}\left(\frac{d}{dz}\left( z(1-z)^{\frac{3-n}{2}} \frac{d}{dz} \right)  -\frac{(1-z)^{\frac{1-n}{2}}}{4} \left(\frac{L^2\lambda_\ell}{r_h^2} + \frac{m_{\text{eff}}^2 }{1-z}\right)  \right),
\end{align*}
we obtain that the suitable space of solutions is that of square-integrable functions with respect to the measure $\frac{(1-z)^{\frac{1-n}{2}}}{z}$. In a neighbourhood of each singular endpoint $z_0\in\{0,1\}$, there are different convenient choices of bases that, in addition, depend on the integrality of the parameters $a$, $b$ and $c$ of the hypergeometric equation. Let $\{\zeta_{1(z_0)},\zeta_{2(z_0)}\}$ denote a convenient basis for the solutions of the hypergeometric equation close to $z_0$. Accordingly, let us denote $\{R_{1(z_0)},R_{2(z_0)}\}$ a basis for the radial solutions close to $z_0$, which is written in terms of the hypergeometric functions as $R_{i(z_0)}(z)=z^{\alpha_+}(1-z)^{\beta_+} \zeta_{i(z_0)}(z)$, for $i\in\{1,2\}$.

If neither $c$ nor $c-a-b$ are integers, convenient bases are given by the standard hypergeometric functions in Equation \eqref{eq:solhypergeo BTZ RinderAdS}. If either $a$, $b$, $c-a$, or $c-b$ is an integer, then we say this is a degenerate case of the hypergeometric differential equation, and one of the hypergeometric series of the standard solutions includes only a finite number of terms, as explained in \cite[Pg69]{higher}. Let us start analysing the square-integrability of the standard solutions at $z=0$. Asymptotically, we have
\begin{align*}
  &\zeta_{1(0)}(z)=F(a,b;c;z) \overset{z\rightarrow0}{\sim} 1, \\
  &\zeta_{2(0)}(z)=z^{1-c}F(a-c+1,b-c+1;2-c;z) \overset{z\rightarrow0}{\sim}z^{1-c},\\
  &s(z)=\frac{(1-z)^{\frac{1-n}{2}}}{z}\overset{z\rightarrow0}{\sim}\frac{1}{z}.
\end{align*}
At zero, $R_{i(0)} (z) \overset{z\rightarrow0}{\sim} z^{\alpha_+} \zeta_{i(0)}(z)$. Thus
\begin{align*}
  \int |R_{1(0)}(z)|^2s(z)dz \sim
                  \begin{cases}
                   \ln(z) ,& \text{ if }\text{Re}(\alpha_+)=0,\\
                   z^{2\text{Re}(\alpha_+)},&\text{ if }\text{Re}(\alpha_+)\neq0.
                  \end{cases}
\end{align*}
Assuming $\omega\in\mathbb{C}$, we have that $\text{Re}(\alpha_+) = -\frac{L^2\text{Im}(\omega)}{2r_h} $. Therefore, $R_{1(0)}(z)$ is square-integrable in a neighbourhood of $z=0$ if and only if $\text{Im}(\omega)<0$. For the other solution, an analogous computation yields that $R_{2(0)}(z)$ is square-integrable in a neighbourhood of $z=0$ if and only if $\text{Im}(\omega)>0$. That is, we find that if $c\notin\mathbb{Z}$, then $z=0$ is a limit point according to Weyl's endpoint classification (see Definition \ref{def: Weyl's endpoint classification}) and the only square-integrable solution is
\begin{align}
\label{eq:solution at 0 Rz massless top bh}
R_{0}(z) =  \Theta(-\Imag(\omega))R_{1(0)}(z) + \Theta(\Imag(\omega)) R_{2(0)}(z).
\end{align}

Let $N\in\{1,2,,3,...\}$, and suppose $a,b\neq 0,1,...,N-1$. On one hand, if $c=1+N$, then we must replace the standard $\zeta_{2(0)}(z)$ by \cite[Pg.564]{handbook}
    \begin{align}
      \label{eq: pwokepke0r90i3}
      \zeta_{2(0)}(z)=&F(a,b;1+N;z)\ln z +\sum\limits_{n=1}^\infty\frac{(a)_n(b)_n}{(1+N)_nn!}z^n\Big[ \psi(a+n)-\psi(a) +\psi(b+n)-\psi(b) \nonumber\\
      &-\psi(N+1+n) +\psi(N+1)-\psi(n+1)+\psi(1)\Big] -\sum\limits_{n=1}^m\frac{(n-1)!(-N)_n}{(1-a)_n(1-b)_n}z^{-n}.
    \end{align}
If $c=1-N$, then we must replace the standard $\zeta_{1(0)}(z)$ by
    \begin{align*}
      &\zeta_{1(0)}(z)=z^NF(a+N,b+N;1+N;z)\ln z +\sum\limits_{n=1}^\infty\frac{(a+N)_n(b+N)_n}{(1+N)_nn!}z^n\Big[ \psi(a+N+n)\nonumber\\
      & \quad \qquad\qquad -\psi(a+N) +\psi(b+N+n)-\psi(b+N)-\psi(N+1+n) +\psi(N+1)-\psi(n+1) \nonumber\\
      & \quad \qquad \quad \qquad \quad \qquad+\psi(1)\Big]-\sum\limits_{n=1}^N\frac{(n-1)!(-N)_n}{(1-a-N)_n(1-b-N)_n}z^{-n}.
    \end{align*}
On the other hand, using Equations \eqref{eq:auxiliary parameters alpha and beta plus and minus massless hyperbolic bh} and \eqref{eq:c n} we obtain
      \begin{align*}
          & c=1\pm N \iff \omega= \mp i \frac{r_h N}{L^2} \in i\, \mathbb{N}^\mp_*.
       \end{align*}
Therefore, for the values of $\omega$ for which we have $c=1+N$, which have a negative imaginary part, we take the standard $\zeta_{1(0)}(z)$ but replace the standard $\zeta_{2(0)}(z)$ with the solution given in Equation \eqref{eq: pwokepke0r90i3}. Since neither the standard $R_{2(0)}(z)$ nor the replaced one are square-integrable for $\Imag(\omega)<0$, the replacement does not affect the classification of the endpoint $z=0$ as limit point. An analogous argument holds when $c=1-N$. Furthermore, if, when $c$ is an integer, also $a$ or $b$ is an integer, then we have another degenerate case \cite[Pgs71-73]{higher}. In any case, the secondary solution with the logarithmic term is never square-integrable.

Now let us check the square-integrability of the standard solutions at $z=1$. Asymptotically, we have
  \begin{align*}
    &\zeta_{1(1)}=F(a,b;a+b+1-c;1-z)  \overset{z\rightarrow1}{\sim}  1,\\
    &\zeta_{2(1)}=(1-z)^{c-a-b}F(c-b,c-a;c-a-b+1;1-z) \overset{z\rightarrow1}{\sim}  (1-z)^{c-a-b},\\
    &s(z)=\frac{(1-z)^{\frac{1-n}{2}}}{z}\overset{z\rightarrow1}{\sim}(1-z)^{\frac{1-n}{2}}.
  \end{align*}
Hence, assuming $a+b-c=\nu\notin\mathbb{Z}$, we have
  \begin{align*}
    \int |R_{j(1)}(z)|^2s(z)dz \sim \begin{cases} (1-z)^{1+\nu}, &\text{ if }j=1, \\ (1-z)^{1-\nu}, &\text{ if }j=2. \end{cases}
  \end{align*}
For $\nu\in(0,1)$, both solutions are square-integrable and $z=1$ is limit circle. For $\nu\geq1$, which correspond to $m_{\text{eff}}^2 \geq0$, only $R_{1(1)}(z)$ is square-integrable and $z=1$ is limit point.

The hypergeometric solutions \eqref{eq:sol11hypergeo BTZ RinderAdS} and \eqref{eq:sol21hypergeo BTZ RinderAdS} are given as the analytic continuation of \eqref{eq:sol10hypergeo BTZ RinderAdS}. They are related by the fundamental formula
    \begin{align}
    \label{eq:fundamentalfkop}
    F(a,b;c;z) =  A_0 F(a,b&,a+b+1-c;1-z)\nonumber\\&+B_0 (1-z)^{c-a-b}F(c-a,c-b;c-a-b+1;1-z),
  \end{align}
where $A_0$ and $B_0$ are constants depending on the parameters $a$, $b$ and $c$. However, these coefficients $A_0$ and $B_0$ have poles at $c=a+b\pm N$, for integer $N$. That is, Equation \eqref{eq:fundamentalfkop} is not well-defined when $c-a-b$ is an integer. Thus, we need different expressions to study the singularity $z=1$, and these can be found in \cite[Pg74]{higher}. For $c=a+b-N$, $N\in\mathbb{N}^*$, the standard solution $F(a,b,c;z)$ can be written as
\begin{align}
  \label{eq:wijefoiwjfoiej8888}
  &F(a,b;a+b-N;z)\nonumber\\ &\qquad =A_1 (1-z)^{-N}\sum\limits_{n=0}^{N-1}\frac{(a-N)_n(b-N)_n}{n!(1-N)_n}(1-z)^n + \nonumber\\
  &\qquad \qquad\qquad +B_1 \sum\limits_{n=0}^{\infty}\frac{(a)_n(b)_n}{n!(n+N)!}(1-z)^n[\ln(1-z)+k_n],
\end{align}
\noindent where $k_{n}=-\psi(n+1)-\psi(n+N+1)+\psi(a+n)+\psi(b+n)$, while $\psi(z)$ is the digamma function. Using Equation \eqref{eq:wijefoiwjfoiej8888}, we can identify a basis that is numerically satisfactory close to the singularity at $z=1$, dubbed $\{\zeta_{1(1)}(z),\zeta_{2(1)}(z) \}$, by taking
    \begin{align}
      &\zeta_{1(1)}(z) = (1-z)^{-N}\sum\limits_{n=0}^{N-1}\frac{(a-N)_n(b-N)_n}{n!(1-N)_n}(1-z)^n ,\label{eq:zeta11degenerate}\\
      &\zeta_{2(1)}(z) = \sum\limits_{n=0}^{\infty}\frac{(a)_n(b)_n}{n!(n+N)!}(1-z)^n[\ln(1-z)+k_n]\label{eq:zeta21degenerate}.
    \end{align}
Equation \eqref{eq:zeta11degenerate} is nothing but the principal solution at $z=1$, $F(a,b,a+b+1-c;z)$, as given in Equation \eqref{eq:solhypergeo BTZ RinderAdS}. However, $\zeta_{2(1)}(z)$ should be replaced by Equation \eqref{eq:zeta21degenerate}.

In the case $c-a-b=-\nu$ is an integer $N$, which is tantamount to taking effective masses such that $L^2 m_{\text{eff}}^2  = N^2 - \frac{(n-1)^2}{4}$ for $N\in\mathbb{Z}$, the secondary solution is given by Equation \eqref{eq:zeta21degenerate} and it is not square-integrable due to the logarithmic term. It follows that for each $N\in\mathbb{Z}$ there is a value of $m_{\text{eff}}$ for which $z=1$ is limit point. However, take into account the Breitenlohner-Freedman bound and the fact that $n\geq 3$. Also, recall that for non-integer $\nu>1$ the endpoint $z=1$ is also limit point. Then, we have that for any $\nu>1$, integer or not, only the principal solution is square-integrable, $z=1$ is limit point, and only Dirichlet boundary condition can be imposed.

For convenience, from now on I focus on the case of $\nu\in(0,1)$, which corresponds to $L^2 m_{\text{eff}}^2\in \left( -\frac{(n-1)^2}{4} ,0 \right)$ and is more interesting to us since it admits a large class of boundary conditions. In this case, the most general square-integrable solution at $z=1$ satisfies a Robin boundary condition parametrized by $\gamma\in [0,\pi)$ and is given by
\begin{equation}
  \label{eq: Rgamma massless top bh}
  R_\gamma(z)=\cos(\gamma) R_{1(1)}(z)+\sin(\gamma) R_{2(1)}(z).
\end{equation}

\subsubsection*{The radial Green function}
\label{subsec: The radial Green function on massless hyperbolic black holes}

With the radial solutions given by Equations \eqref{eq:solution at 0 Rz massless top bh} and \eqref{eq: Rgamma massless top bh} for $\nu\in(0,1)$, the Green function for the radial equation reads
  \begin{equation*}
    \mathcal{G}(z,z')=\frac{1}{\mathcal{N}_\omega}\left\{ \Theta(z'-z)R_0(z)R_\gamma(z')+ \Theta(z-z')R_0(z')R_\gamma(z)\right\}
  \end{equation*}
with normalization constant, defined in Equation \eqref{eq: definition normalization of radial green function}, given by
  \begin{equation}
    \label{eq:pjofej49j0f9j}
    \mathcal{N}_\omega=-z(1-z)^{\frac{3-n}{2}}\mathcal{W}_z[R_0(z),R_\gamma(z)].
  \end{equation}
  To obtain an explicit expression for the normalization it is useful to take into account the well-known connection formulas of the hypergeometric functions. As given by expressions (15.10.17) and (15.10.18) in \cite{NIST_15} (or implemented in a Mathematica notebook in my GitHub page \cite{git_Hypergeometric2F1}), the functions $\zeta_{i(1)}(z)$, $i\in\{1,2\}$, can be written with respect to the basis $\{\zeta_{1(0)}(z),\zeta_{1(0)}(z)\}$:
  \begin{subequations}\label{eq: oiejf0u984jf}
  \begin{align}
    \zeta_{1(1)}(z) = A \zeta_{1(0)}(z) + B \zeta_{2(0)}(z),\\
    \zeta_{2(1)}(z) = C \zeta_{1(0)}(z) + D \zeta_{2(0)}(z),
  \end{align}
  \end{subequations}
  with coefficients $A,B,C$ and $D$ given by
\begin{subequations}
  \label{eq: A B C D massless top bh}
    \begin{align}
      A&=\frac{\Gamma(c)\Gamma(c-a-b)}{\Gamma(c-a)\Gamma(c-b)} \xmapsto{\omega\mapsto\omega^*} C=\frac{\Gamma(2-c)\Gamma(c-a-b)}{\Gamma(1-a)\Gamma(1-b)},\\
      B&=\frac{\Gamma(c)\Gamma(a+b-c)}{\Gamma(a)\Gamma(b)}  \xmapsto{\omega\mapsto\omega^*} D=\frac{\Gamma(2-c)\Gamma(a+b-c)}{\Gamma(a-c+1)\Gamma(b-c+1)} .
    \end{align}
  \end{subequations}
One can check that these coefficients, seen as functions of the frequency, satisfy
    \begin{equation}
      \label{eq:wrgrgpoprgkopko999}
      \overline{A(\omega)}=A(\omega^*)= C(\omega)=\overline{C(\omega^*)} \quad \text{ and }\quad
     \overline{B(\omega)}=B(\omega^*)= D(\omega)=\overline{D(\omega^*)}.
    \end{equation}
Given the Wronskian $$\mathcal{W}_z[\zeta_{1(1)}, \zeta_{2(1)}]=(a+b-c)z^{-c}(1-z)^{c-a-b-1},$$ and using Equations \eqref{eq: oiejf0u984jf} with \eqref{eq:pjofej49j0f9j}, we find, for $\Imag(\omega)<0$
\begin{align}
  \label{eq:wklrnogo9}
\mathcal{N}_\omega= -\nu(A\sin(\gamma)-B\cos(\gamma))=:\mathcal{N_<}.
\end{align}
From Equation \eqref{eq:wklrnogo9} and the properties of the coefficients given in Equation \eqref{eq:wrgrgpoprgkopko999}, we directly obtain the normalization for $\text{Im}(\omega) > 0$:
  \begin{align*}
    \mathcal{N}_\omega=\mathcal{N}_<|_{\omega\mapsto\omega^*}=-\nu(C\sin(\gamma)-D\cos(\gamma))=:\mathcal{N}_>.
  \end{align*}

\subsubsection*{Spectral resolution of the radial Green function}
\label{subsec: Spectral resolution of the radial Green function on massless hyperbolic black holes}

To construct physically-sensible states by invoking Theorems \ref{thm: 2 point ground state schd coord} and \ref{thm: 2 point KMS state schd coord} we need to know where the poles of the radial Green function are. For convenience, let us focus on the case   $\text{Im}(\omega)<0$, so that the solution that is square-integrable at $z=0$ is simply $R_{1(0)}(z)$. In this section, let $N$ be a non-positive integer. For Dirichlet and Neumann boundary conditions, we have, respectively,
\begin{align*}
\gamma=0\Rightarrow \mathcal{N_<}=0&\iff B=0\iff( a\text{ or }b  \text{ is a non-positive integer})\nonumber\\
&\iff \omega_D= + i (-2N + 1 + \nu \pm 2\Upsilon),
\end{align*}
and
\begin{align*}
\gamma=\pi/2\Rightarrow \mathcal{N_<}=0&\iff A=0\iff( c-a\text{ or }c-b  \text{ is a non-positive integer})\nonumber\\
&\iff \omega_N= + i (-2N + 1 - \nu \pm 2\Upsilon).
\end{align*}
For $L\neq r_h$, $\Upsilon$ assumes real values and a detailed analysis regarding the poles should be taken into account. Assuming $\nu\in(0,1)$, $\text{Im}(\omega)<0$, and $L=r_h$ (so $\Upsilon$ is purely imaginary), the frequencies $\omega_D$ and $\omega_N$ have non-negative imaginary part and thus are not poles of the radial Green function.

Now consider $\gamma\in(0,\pi/2)$ and define $\Xi(\omega) \doteq \frac{B}{A}$. That is
\begin{align*}
  \Xi(\omega)%
             =&\frac{\Gamma(\nu)}{\Gamma(-\nu)}\frac{\Gamma(c-a)\Gamma(c-b)}{\Gamma(a)\Gamma(b)}.
\end{align*}
The poles of the radial Green function, in this case, are the frequencies that, for each $\gamma$, solve the equation \begin{equation}
\label{eq: transcendental xi = tan}
\Xi(\omega)=\tan(\gamma).
\end{equation}
That is, whenever $\Xi(\omega)$ is real, there exist a boundary condition $\gamma=\arctan(\Xi(\omega))\in\mathbb{R}$ for which the radial Green function has a pole. Details on how to approach the transcendental Equation \eqref{eq: transcendental xi = tan} can be found in \cite{bussola2017ground,deSouzaCampos2020bnj}. All in all, we find that there are no bound states for $\gamma\in[0,\pi/2]$. In this case, we can use Jordan's lemma to compute the contour integral of the spectral resolution of the radial Green function given in Equation \eqref{eq:spectral resolution radial green function}. This computation is analogous to that on Minkowski spacetime performed in Section \ref{sec: On Minkowski spacetime in spherical coordinates} and given in details in \cite{deSouzaCampos2020bnj}. Hence, here, I simply state the result: the spatial part of the two-point function of physically-sensible states, as defined in Section \ref{sec: Physically-sensible states in Schwarzschild-like coordinates}, is given by
    \begin{align}
    \label{eq:widetilde psi2 z massless top bh}
    \widetilde{\psi}_2(z,z')  = \frac{1}{\pi \nu }\frac{\Imag(\overline{A}B)}{|A\sin(\gamma) -B\cos(\gamma)|^2} R_\gamma(z)R_\gamma(z'),
    \end{align}
with $A$ and $B$ given in Equation \eqref{eq: A B C D massless top bh} and $R_\gamma(z)$ as per Equation \eqref{eq: Rgamma massless top bh}.

\subsubsection*{Two-point functions}
\label{subsec: Two-point functions on massless hyperbolic black holes}
For $ \widetilde{\psi}_{2}(r,r')$ given by Equation \eqref{eq:widetilde psi2 z massless top bh} with $z=z(r)=\frac{r^2-r_h^2}{r^2}$, the two-point functions for the ground and thermal states at inverse-temperature $\beta$ are given, respectively due to Theorems \ref{thm: 2 point ground state schd coord} and \ref{thm: 2 point KMS state schd coord}, by
\begin{align}
\label{eq: 2 point ground state massless top bh}
\psi_{2,\infty}(x,x^\prime)= &  \lim_{\varepsilon\rightarrow 0^+} \int_{\mathbb{R}} d\ell \int_{\mathbb{R}}d\omega \Theta(\omega) e^{-i\omega (t - t'- i\varepsilon)} \widetilde{\psi}_{2}(r,r')Y_\ell(\underline{\theta})\overline{Y_\ell(\underline{\theta}')},
\end{align}
\begin{align}
  \label{eq: 2 point KMS state massless top bh}
\psi_{2,\beta}(x,x^\prime)=\lim_{\varepsilon\rightarrow 0^+}\int_{\mathbb{R}} d\ell \int_{\mathbb{R}} d\omega \Theta(\omega) \left[ \frac{e^{-i\omega (t-t'-i\varepsilon)}}{1 - e^{-\beta\omega}}+ \frac{e^{+i \omega (t-t'+i\varepsilon)}}{{e^{\beta\omega}-1}} \right] \widetilde{\psi}_{2}(r,r') Y_\ell(\underline{\theta})\overline{Y_\ell(\underline{\theta}')}.
\end{align}

\subsection{The transition rate of an Unruh-DeWitt detector}
\label{subsec: The transition rate on massless hyperbolic black holes}
Consider an Unruh-DeWitt detector with energy gap $\Omega$ following a static trajectory of fixed spatial coordinates $(z=z(r),\underline{\theta})$, as in Section \ref{subsec: The transition rate on a static BTZ spacetime and on Rindler-AdS3}. Theorem \ref{thm: Transition rate for the physically-sensible states} together with Equations \eqref{eq: 2 point ground state massless top bh} and \eqref{eq: 2 point KMS state massless top bh} yields the transition rate for, respectively, the ground and thermal states
\begin{align}
  \label{eq:transition rate ground state massless top bh}
  \dot{\mathcal{F}}_\infty =  \Theta( -\Omega )  \int_{\mathbb{R}}  d\ell\,  \frac{2}{\nu }\frac{\Imag(\overline{A}B)}{|A\sin(\gamma) -B\cos(\gamma)|^2} |Y_\ell(\underline{\theta})|^2 R_\gamma(r)^2 \Big|_{\omega = -\sqrt{f(r)} \,\Omega},
\end{align}
and
\begin{align}
  \label{eq:transition rate KMS state massless top bh}
  \dot{\mathcal{F}}_{\beta} = \frac{\text{ sign}(\Omega)}{e^{\beta\sqrt{f(r)}\Omega}-1}  \int_{\mathbb{R}}  d\ell\,   \frac{2}{ \nu }\frac{\Imag(\overline{A}B)}{|A\sin(\gamma) -B\cos(\gamma)|^2}  |Y_\ell(\underline{\theta})|^2  R_\gamma(r)^2 \big|_{\omega =\sqrt{f(r)} |\Omega| }.
\end{align}

For massless conformally coupled scalar fields on three- and four-dimensional massless hyperbolic black holes, we performed a numerical analysis of the transition rate when the detector is coupled with the ground state, or with thermal states. As a function of the energy gap, we verify that its behaviour, given in Figure \ref{fig:n = 3 and 4 transition of Egap KMS}, is analogous to that on Minkowski spacetime, illustrated in Figure \ref{fig:transition rate Minkowski 3 4 5 6 as a function of the energy gap}. They are similar in the sense that in three dimensions the transition rate plateaus for decreasing energy gaps, and in four dimensions it increases without bound. Their contrast lies in the oscillatory behavior observed in the black hole scenarios. Furthermore, considering the anti-correlation effects, the results we obtain in the three dimensional case are compatible with those on a BTZ black hole shown in the previous section: the anti-Unruh effect is manifest for the ground state, but it does not occur for thermal states. In the four-dimensional case, no anti-correlation effect is observed for either state. This spacetime dimension dependence agrees with the one observed on Minkowski spacetime, as per Figure \ref{fig:transition rate Minkowski 3 4 5 6 as function of a}.
\begin{figure}[H]
 \centering
 \includegraphics[width=0.45\textwidth]{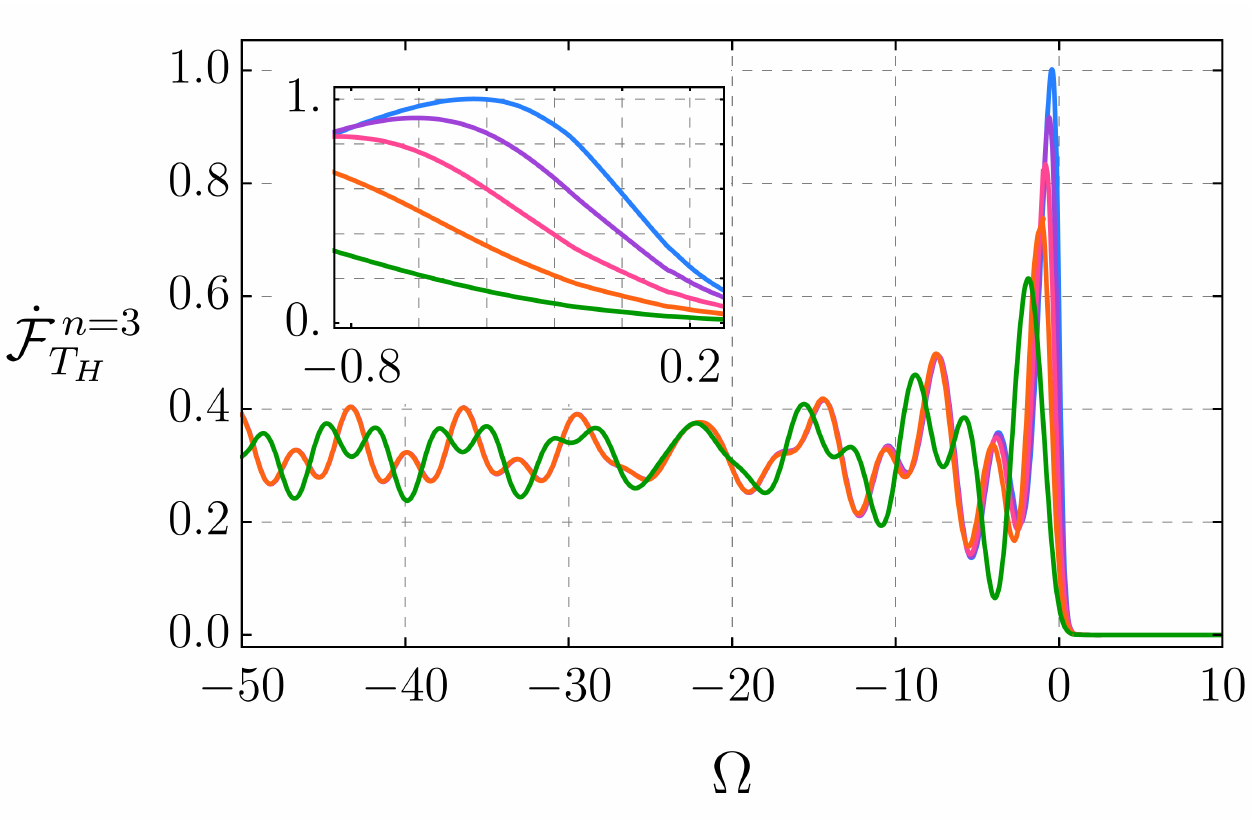}\hspace{.5cm}
 \includegraphics[width=0.45\textwidth]{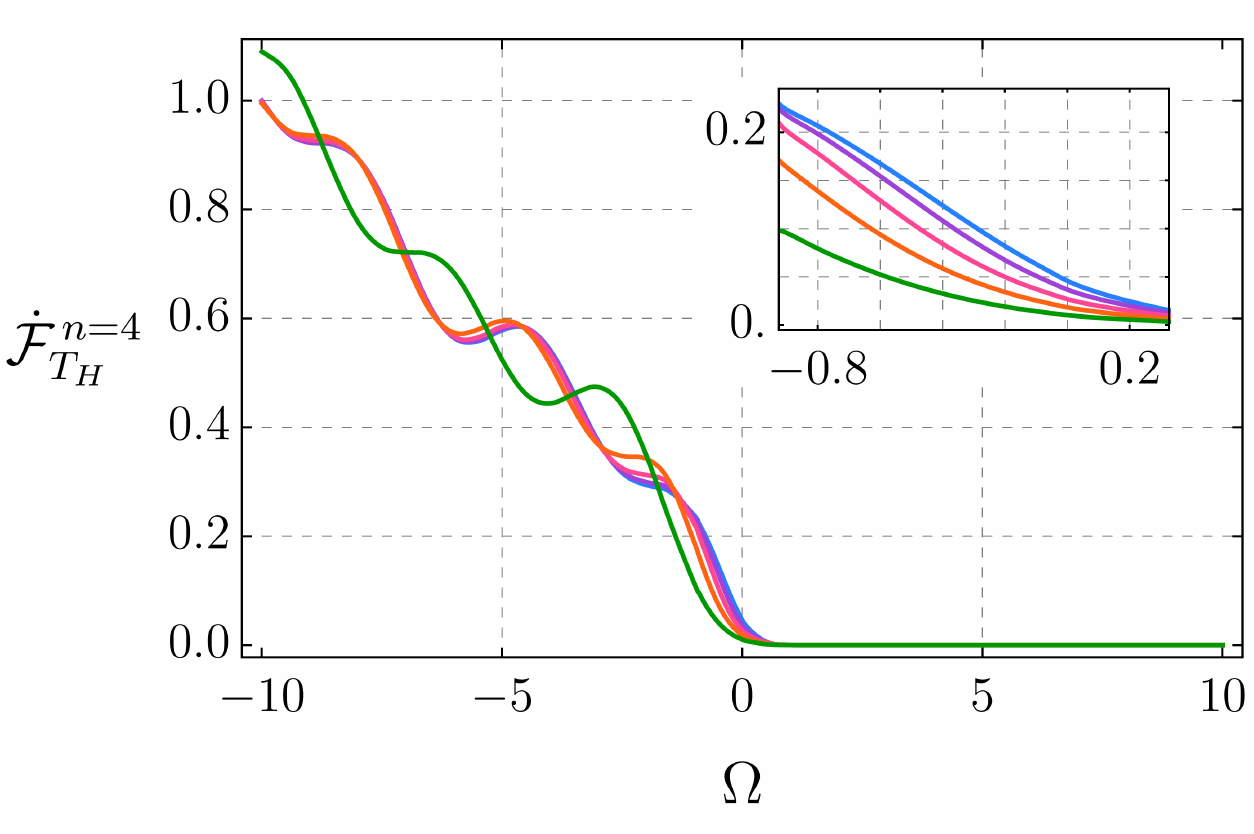}
 \caption{Transition rate as a function of the energy gap for the KMS state at $z_{\mini{D}} = 1/2$, $\theta_{\mini{D}}=\pi^{-1}$ and for different boundary conditions: from top to bottom, $\gamma=(0.50,0.47,0.40,0.25,0)\pi$. On the left, for $n=3$ and with the integration performed up to $\ell=100$; on the right, for $n=4$, at $\varphi_1=0$ and with the summation performed up to $m_1=20$ and integration, up to $\ell=20$.}
 \label{fig:n = 3 and 4 transition of Egap KMS}
 \end{figure}
 \begin{figure}[H]
 \centering
 \includegraphics[width=0.45\textwidth]{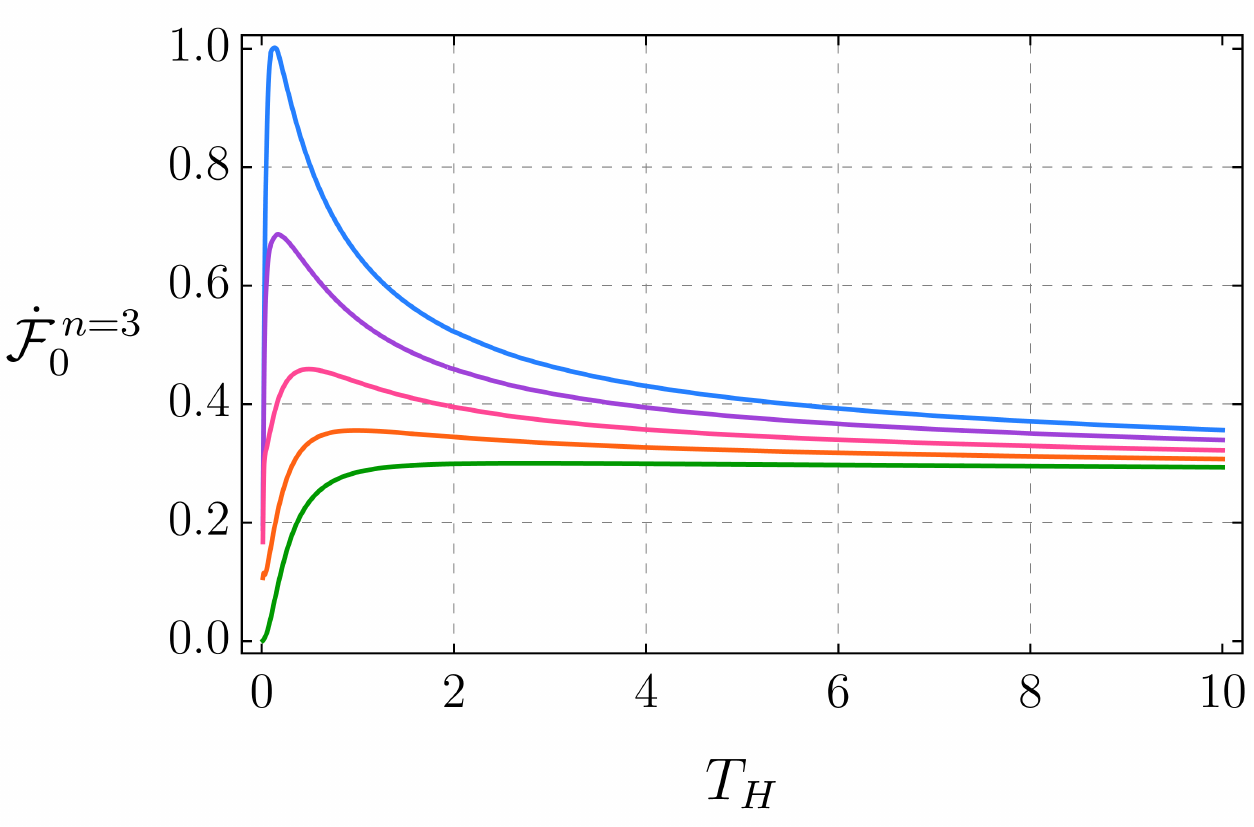}\hspace{.5cm}
 \includegraphics[width=0.45\textwidth]{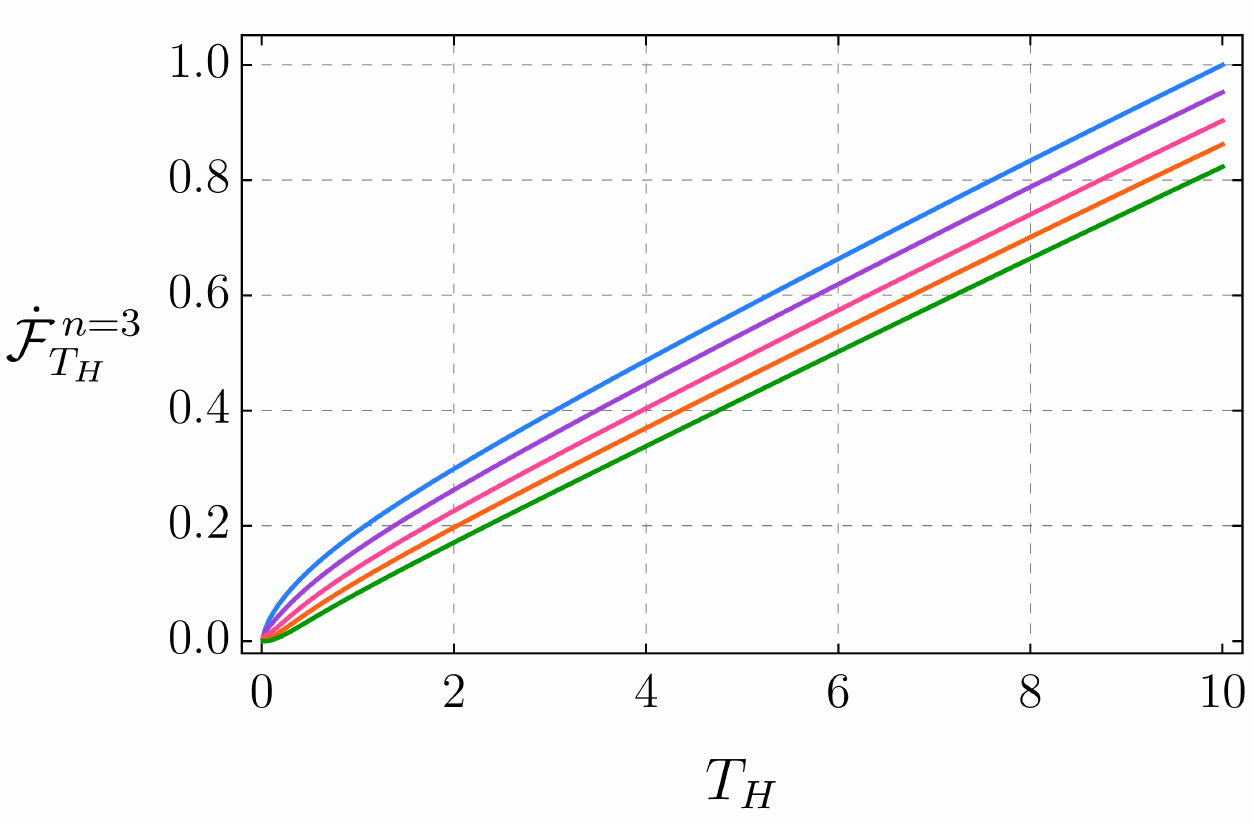}
 \caption{Transition rate, integrated up to $\ell=100$, as a function of the local Hawking temperature on the three-dimensional hyperbolic black hole for $\Omega=-0.1$, $\theta=\pi^{-1}$ and for different boundary conditions. From top to bottom $\gamma=(0.50,0.47,0.40,0.25,0)\pi$. On the left, for the ground state; on the right, for the KMS state.}
 \label{fig:n=3 transition of TH ground and KMS}
 \end{figure}
 \begin{figure}[H]
 \centering
 \includegraphics[width=0.45\textwidth]{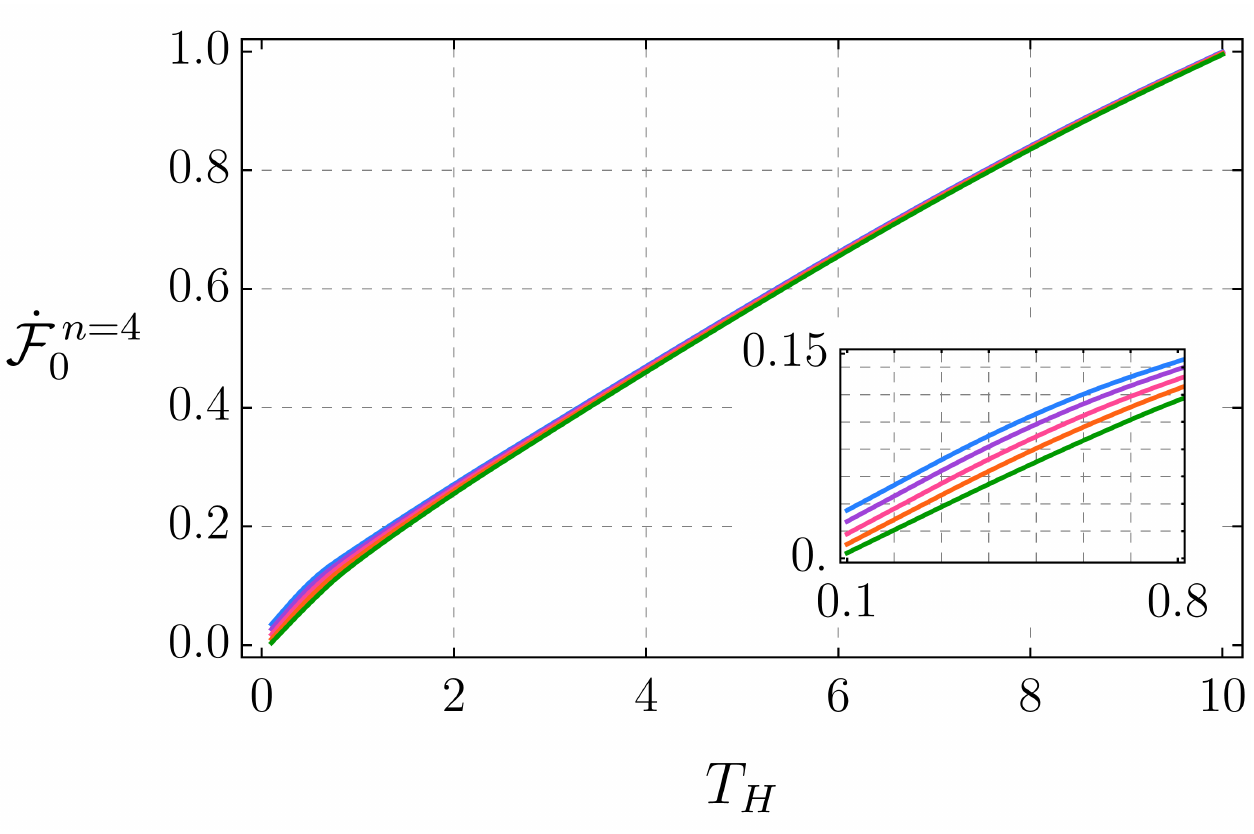}\hspace{.5cm}
 \includegraphics[width=0.45\textwidth]{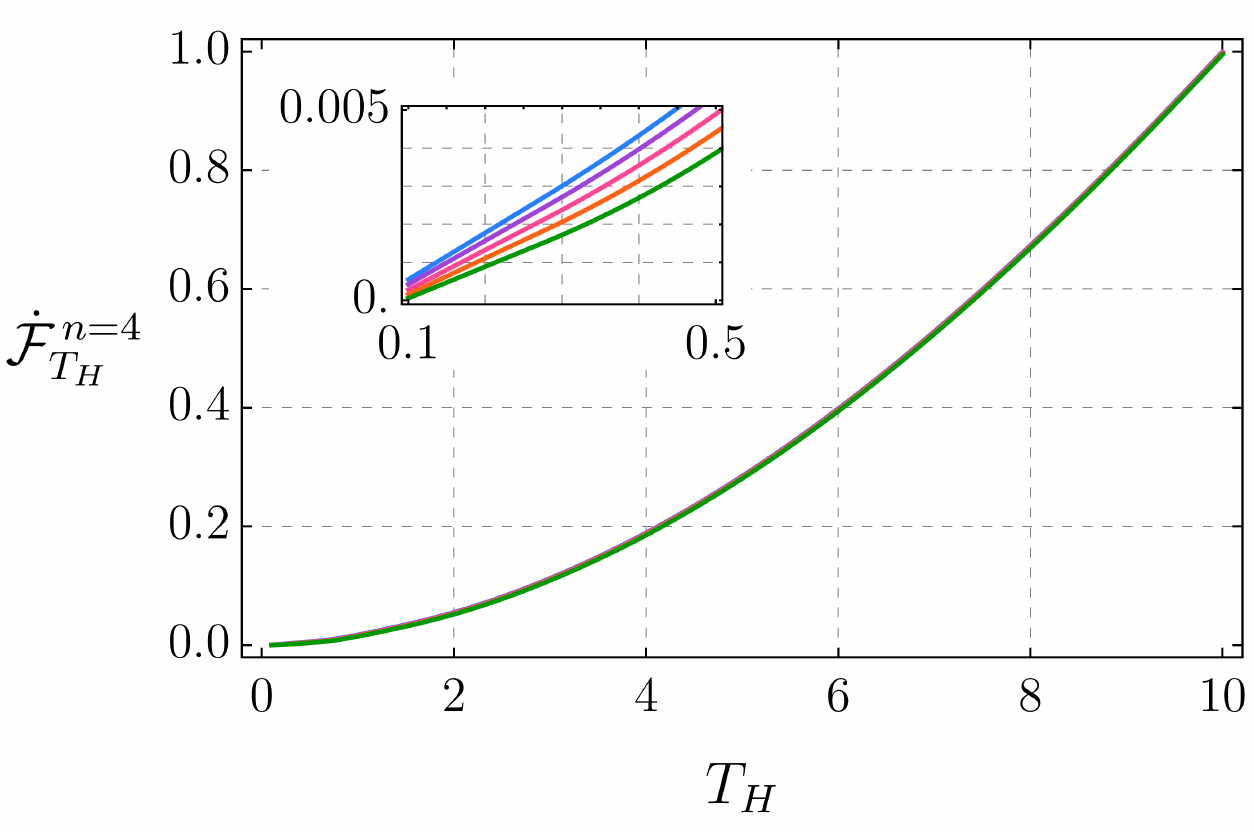}
 \caption{Transition rate as a function of the local Hawking temperature, summed up to $m_1=100$ and integrated up to $\ell=100$, on the four-dimensional hyperbolic black hole for $\Omega=-0.1$, at $\theta_{\mini{D}}=\pi^{-1}$, $\varphi_{1,\mini{D}}=0$ and for different boundary conditions: from top to bottom $\gamma=(0.50,0.47,0.40,0.25,0)\pi$. On the left, for the ground state; on the right, for the KMS state.}
 \label{fig:n=4 transition of TH ground and KMS max 5 and 100}
 \end{figure}
%


\section{On Lifshitz topological black holes}
\label{sec: On Lifshitz topological black holes}

In this section, let us consider a free, scalar, quantum field on four-dimensional Lifshitz topological black holes. In Section \ref{subsec: Ground and thermal states on Lifshitz top bh}, I summarize the construction of the ground and thermal states, analogously to the previous sections. This work has been published in \cite{deSouzaCampos2021role}. In addition, it expands previous analysis \cite{kachru2008gravity}, by considering massive fields, by allowing for more general boundary conditions, and by also setting three types of spacetime interiors.

As given in Section \ref{ex: chapter 1 Lifshitz toplogical black holes}, in Schwarzschild-like coordinates, the line element of Lifshitz topological black holes reads
\begin{align}
  \label{eq:metric lif top bh chap lif top bh}
        ds^2=-f(r)dt^2+\frac{L^{2}}{r^2}f(r)^{-1}dr^2+r^2d\Sigma_{\kappa,2}^2,
\end{align}
where $d\Sigma_{\kappa,2}^2$ is given in Definition \ref{def: Schwarzschild-like coordinates} with $j=\kappa$ and
\begin{align*}
    &f(r):=\frac{r^{2}}{L^2} \left(\frac{r^{2}}{L^2} + \frac{\kappa}{2}\right).
\end{align*}
For $\kappa=0$, $\kappa=-1$ and $\kappa=+1$, Equation \eqref{eq:metric lif top bh chap lif top bh} characterizes, respectively, flat, hyperbolic and spherical Lifshitz topological black holes. In addition, the corresponding Ricci scalar is given by $\mathbf{R}= -\frac{22}{L^2} - \frac{\kappa}{r^2}$.

\subsection{Ground and thermal states}
\label{subsec: Ground and thermal states on Lifshitz top bh}

\subsubsection*{The radial equation}
\label{subsec: The radial equation on Lifshitz top bh}
Let $Y_\kappa(\theta,\varphi)$ be the eigenfunctions \cite{jorgensen1987harmonic,limic1967eigenfunction,deSouzaCampos2020bnj}, with eigenvalues
 \begin{align*}
       \lambda_\kappa = \begin{cases}
                       -(\ell^2 + m^2),\, \ell,m\in\mathbb{R}, &\text{ for }\kappa=0,\\
                       -(\frac{1}{4}+\ell^2),\, \ell\in\mathbb{R}, &\text{ for }\kappa=-1,\\
                       -\ell(\ell+1),\, \ell\in\mathbb{N}_0, &\text{ for
}\kappa=+1,
                   \end{cases}
 \end{align*}
 of the Laplacian
	\begin{equation*}
	\Delta_\kappa = \begin{cases}
	         \frac{1}{\theta^2}\frac{\partial^2}{\partial \varphi^2} + \frac{1}{\theta}\frac{\partial}{\partial\theta}+\frac{\partial^2}{\partial \theta^2}, &\text{ for }\kappa=0; \\
		 \frac{1}{\sinh(\theta)^2}\frac{\partial^2}{\partial \varphi^2} + \frac{1}{\sinh(\theta)}\frac{\partial}{\partial\theta}\left(\sinh(\theta )\frac{\partial}{\partial\theta}\right), &\text{ for }\kappa=-1; \\
	\frac{1}{\sin(\theta)^2}\frac{\partial^2}{\partial \varphi^2} + \frac{1}{\sin(\theta)}\frac{\partial}{\partial\theta}\left(\sin(\theta )\frac{\partial}{\partial\theta}\right), &\text{ for }\kappa=+1.
	\end{cases}
\end{equation*}
The ansatz
  \begin{align*}
  \Psi(t,r,\theta,\varphi) =  e^{-i\omega t } R(r)Y_\ell^m(\theta,\varphi)
\end{align*}
is a solution of the Klein-Gordon Equation \eqref{eq: KG} if $R(r)$ solves the radial equation:
\begin{subequations}
  \label{eq: the radial equation Lifshitz}
\begin{equation}
    R''(r)+ \left[\frac{4 r}{\kappa L^2+2 r^2}+\frac{3}{r}\right]R'(r)+ \left[\frac{4 L^6}{\left(\kappa L^2 r+2 r^3\right)^2} \omega ^2 +
\frac{2 L^2}{\kappa L^2 r^2+2 r^4}\lambda -\frac{2 L^2}{\kappa L^2+2 r^2}
\mu^2\right]R(r)=0,
\end{equation}
where
\begin{equation}
  \label{eq: lambda shited by xi}
  \lambda := \lambda_\kappa + \kappa \xi.
\end{equation}
\end{subequations}
For each value of $\kappa$ we apply a coordinate transformation to write the radial equation in a well-known form:
\begin{align}
\label{eq: coord trans Lifshitz top bh}
    r\mapsto \begin{cases}
                    u = \frac{iL^3\omega}{r^2}\in(0,i\infty) , &\text{ for }\kappa=0,\\
                    s = \frac{2 r^2 - L^2}{2r^2}\in (0,1), &\text{ for }\kappa=-1,\\
                    s = \frac{2 r^2 + L^2}{2r^2}\in(1,\infty), &\text{ for }\kappa=+1.
                  \end{cases}
\end{align}
For $\kappa=0$, Equation \eqref{eq: the radial equation Lifshitz} with Equation \eqref{eq: coord trans Lifshitz top bh} and the ansatz
\begin{align}
  \label{eq: ansatz R(u) k=0}
    &R(u) =  e^{-u/2} \left(\frac{u}{i L^3 \omega}\right)^{\frac{1}{2} \left(2 + \nu \right)} \zeta(u).
\end{align}
yield a confluent hypergeometric equation \cite{NIST_13}
\begin{align*}
  u \zeta''(u) + (b_0 - u)\zeta'(u) -a_0 \zeta(u) =0,
\end{align*}
with
\begin{align*}
    &a_0 = \frac{1+\nu}{2} + \frac{i \lambda }{4 L \omega }\quad \text{ and }\quad b_0 = 1+ \nu .
\end{align*}
For $\kappa=\pm 1$, Equation \eqref{eq: the radial equation Lifshitz} with Equation \eqref{eq: coord trans Lifshitz top bh} and the ansatz
\begin{align*}
    &R(s) =  s^{\kappa i L\omega} [\kappa(s-1)]^{\frac{1}{2} \left(2 + \nu \right)}\zeta(s),
\end{align*}
yield a hypergeometric equation  \cite{NIST_15}
  \begin{align*}
        s(1-s)\zeta''(s) + (c-(a+b+1)s)\zeta'(s) -ab\zeta(s)=0,
  \end{align*}
with parameters
  \begin{align*}
      &a = \frac{1+\nu}{2} +i\kappa L\omega - \Upsilon,\\
      &b = \frac{1+\nu}{2} +i\kappa L\omega + \Upsilon,\\
      &c = 1+ 2 i \kappa L\omega,
  \end{align*}
  where
  \begin{align}
    \label{eq: parameter upsilon lif top bh}
      &\Upsilon =  \kappa\frac{\sqrt{1 -2\kappa\lambda -4L^2\omega^2}}{2}  .
  \end{align}

\subsubsection*{The radial solutions}
\label{subsec: The radial solutions on Lifshitz top bh}

In Sturm-Liouville form, Equation \eqref{eq: the radial equation Lifshitz} reads
\begin{equation*}
        L_{\omega^2}  R(r) :=-\frac{1}{q(r)} \left(\frac{d}{d r} \left(p(r)\frac{d}{d r}\right)+v(r) \right)R(r)= \omega ^2 R(r),
\end{equation*}
with coefficient functions
\begin{align*}
  &q(r) = \frac{4 L^6 r}{\kappa L^2+2 r^2}, \\
  &p(r) = \kappa L^2 r^3+2 r^5, \\
  &v(r) = 2 L^2 r \left(-\mu^2 r^2+\lambda \right).
\end{align*}
After we verify the square-integrability of the radial solutions with respect to the measure $q(r)$ for each case $\kappa\in\{0,-1,+1\}$, we obtain the following results. First, at radial infinity the behaviour is the same in all cases:
\begin{align*}
  &R(r)\sim r^{-2 \pm \nu}   \text{, as  }r\rightarrow \infty,
\end{align*}
where
\begin{align*}
  &\nu := \sqrt{4 + L^2 m_{\text{eff}}^2}>0.
\end{align*}
In addition, $r=\infty$ is limit circle if $\nu \in (0,1)\cup(1,2)$ and the most general square-integrable solution satisfies Robin boundary conditions parametrized by $\gamma\in[0,\pi)$:
\begin{align*}
R_{\gamma}(r)= \cos(\gamma)R_{1(\infty)}(r) + \sin(\gamma)R_{2(\infty)}(r)
\end{align*}
The other endpoint is always limit point, but the most general square-integrable solution is case dependent, as follows.
\noindent For $\kappa=0$, close to $r=0$ we take:
\begin{align*}
    R_0(r) = \Theta(-\Imag(\omega))R_{1(0)}(r) + \Theta(\Imag(\omega))R_{2(0)}(r),
\end{align*}
where the radial solutions are, consistently with Equation \eqref{eq: ansatz R(u) k=0}, given by
\begin{align*}
    &R_{i(u_0)}(u) =  e^{-u/2} \left(\frac{u}{i L^3 \omega}\right)^{\frac{1}{2} \left(2 + \nu \right)} \zeta_{i(u_0)}(u), \quad \text{ for }i\in\{1,2\},
\end{align*}
which are written in terms of the confluent hypergeometric functions $M$ and $U$ \cite{NIST_13}:
\begin{align*}
&\zeta_{1(0)}(u) =  M(a_0,b_0;u),\\
&\zeta_{2(0)}(u)= \left(i L^3\omega \right)^\nu u^{1-b_0}M(a_0-b_0+1,2-b_0;u),\\
&\zeta_{1(\infty)}(u) =  U(a_0,b_0;u),\\
&\zeta_{2(\infty)}(u)=   e^{u}U(b_0-a_0,b_0;-u).
\end{align*}
\noindent For $\kappa=\pm 1$, we have
\begin{align*}
    &R_{i(s_0)}(s) =  s^{\kappa i L\omega} [\kappa(s-1)]^{\frac{1}{2} \left(2 + \nu \right)}\zeta_{i(s_0)}(s), \quad \text{ for }i\in\{1,2\}.
\end{align*}
\noindent For $\kappa=-1$, then close to $s=L/\sqrt{2}$ we take
\begin{align*}
    R_0(r) = \Theta(+\Imag(\omega))R_{1(0)}(r) + \Theta(-\Imag(\omega))R_{2(0)}(r),
\end{align*}
\begin{align*}
&\zeta_{1(0)}(s)=F(a,b;c;s),\\
&\zeta_{2(0)}(s)=s^{1-c}F(a-c+1,b-c+1;2-c;s),\\
&\zeta_{1(1)}(s)=F(a,b;a+b+1-c;1-s),\\
&\zeta_{2(1)}(s)=(1-s)^{c-a-b}F(c-a,c-b;c-a-b+1;1-s).
\end{align*}
\noindent For $\kappa=+1$, then close to $s=0$ we take:
\begin{align}
  \label{eq: square-integrable R(r) at r=0 k=+1}
  &R_{0}(r) := \Theta(\Real(\Upsilon))R_{2(0)}(r).%
\end{align}
\begin{align*}
  &\zeta_{1(1)}(s)=F(a,b;a+b+1-c;1-s),\\
  &\zeta_{2(1)}(s)=(s-1)^{c-a-b}F(c-a,c-b;c-a-b+1;1-s),\\
  &\zeta_{1(\infty)}(s)=s^{-a}F(a,a-c+1;a-b+1;1/s),\\
  &\zeta_{2(\infty)}(s)=s^{-b}F(b,b-c+1;b-a+1;1/s).
\end{align*}

\subsubsection*{The radial Green function}
\label{sec: The radial Green function lif top bh}

For $\nu \in (0, 1)\cup(1,2)$, the radial Green function reads

  \begin{align*}
  \mathcal{G}_{\omega}(r,r') =\frac{1}{ \mathcal{N}_{\omega}}\left( \Theta(r'-r) R_{0}(r)R_{\gamma}(r') + \Theta(r-r')R_{0}(r')R_{\gamma}(r)\right).
\end{align*}
Using the fundamental relation connecting the solutions, for $  i_\kappa = 1 + \delta_{\kappa,\pm 1}$,
  \begin{equation*}
  R_{i_\kappa(0)}(r) = A_\kappa R_{1(\infty)}(r) + B_\kappa R_{2(\infty)}(r),
\end{equation*}
we obtain the normalization $\mathcal{N}_{\omega}$ that is defined in Equation \eqref{eq: definition normalization of radial green function}. For $\Imag(\omega)<0$ and $  c_\kappa:=4^{\delta_{\kappa,0}} L^{4\delta_{\kappa,\pm 1}}\nu $ it is given by
\begin{equation*}
\mathcal{N}_\omega= c_\kappa \{B_\kappa\cos(\gamma) - A_\kappa\sin(\gamma)\},
\end{equation*}
with coefficients given by \cite[(15.10.22)]{NIST_15}
  \begin{align*}
    &A_0:= \frac{\Gamma(1-b_0)}{\Gamma(a_0-b_0+1)},\\
    &B_0:= \frac{\Gamma(b_0-1)}{\Gamma(a_0)}(iL^3\omega)^{-\nu},\\
    &A_{\mini{-1}}:= \frac{\Gamma(2-c)\Gamma(c-a-b)}{\Gamma(1-a)\Gamma(1-b)},\\
    &B_{\mini{-1}}:= \frac{\Gamma(2-c)\Gamma(a+b-c)}{\Gamma(a-c+1)\Gamma(b-c+1)},\\
    &A_{\mini{+1}}:= \frac{\Gamma(b-a+1)\Gamma(c-a-b)}{\Gamma(1-a)\Gamma(c-a)},\\
    &B_{\mini{+1}}:= \frac{\Gamma(b-a+1)\Gamma(a+b-c)}{\Gamma(b)\Gamma(b-c+1)}.
  \end{align*}

\subsubsection*{Spectral resolution of the radial Green function}
\label{sec: Spectral resolution of the radial Green function lif top bh}

The analysis regarding the existence of bound states is rather intricate and fully distilled in \cite{deSouzaCampos2021role}. On one hand, for $\kappa=0$, we obtain results analogous to those on massless hyperbolic black holes as in Section \ref{subsec: Spectral resolution of the radial Green function on massless hyperbolic black holes}: for each $\lambda$-mode, there is a critical boundary condition,
\begin{equation*}
  \gamma^\lambda_{c} := \arctan\left( \left(\frac{-\lambda L^2}{4}\right)^{-\nu}\frac{\Gamma(\nu)}{\Gamma(-\nu)} \right),
\end{equation*}
such that for Robin boundary conditions parametrized by
 \begin{equation*}
 \gamma\in[0,\gamma^\lambda_c] \quad \text{ with }\quad  \gamma^\lambda_c
 \in \begin{cases}
                [\pi/2,\pi) \, &\nu\in(0,1),\\
                 [0,\pi/2) \,     &\nu\in(1,2),
               \end{cases}
 \end{equation*}
the radial Green function has no poles. In addition, note that for all $\lambda$-modes, if $\nu\in(0,1)$, then there are no bound states for $\gamma\in[0,\pi/2]$, as it happened in Sections \ref{subsec: Spectral resolution of the radial Green function on a static BTZ spacetime and on Rindler-AdS3} and \ref{subsec: Spectral resolution of the radial Green function on massless hyperbolic black holes}. On the other hand, for $\kappa=\pm 1$, the existence of bound states for $\nu\in(0,1)\cup(1,2)$ and $\gamma\in[0,\pi)$ depends more specifically on the values of the mass of the field and of the coupling parameter.

Let us consider the case $\kappa=-1$ and define $\lambda_c:=-\frac{\nu(2-\nu)}{2}$, and
\begin{align*}
                \xi_{\pm} :=   \frac{1}{100} \left(9+ 5 L^2 m_0^2 \pm \sqrt{81-10 L^2 m_0^2}\right).
\end{align*}
For $\gamma=0$, there are no bound states, as expected. For $\gamma=\pi/2$, the values of $L$, $\mu_0$ and $\xi$ for which the radial Green function has no poles must satisfy
\begin{align*}
    &\nu\in(0,1) \begin{cases}
    \text{(i) } \,\,L^2 m_0^2 < 8.1, \xi\in \left(\xi_-,\xi_+\right) \text{ and }\lambda\leq\lambda_c; \\
    \text{(ii) }\, L^2 m_0^2 < 8.1 \text{ and }\xi\notin \left(\xi_-,\xi_+\right) ; \\
    \text{(iii) } L^2 m_0^2 \geq 8.1; \end{cases}\\
    &\nu\in(1,2)\begin{cases}
    \text{(iv) } L^2 m_0^2 < 8.1,\, \xi \in \left(\xi_-,\xi_+\right) \text{ and }\lambda\geq\lambda_c; \\
    \text{(v) }\,L^2 m_0^2 = 8.1, \, \xi=\xi_-=\xi_+ \text{ and }\lambda=-\xi.\end{cases}
  \end{align*}
Note that conditions (ii) and (iii) hold true for all $\lambda$-modes, but the others are mode-dependent. For $\gamma\in(0,\pi)\setminus \{\pi/2\}$ and $\lambda\neq\lambda_c$, we find that the radial Green function has no poles if
\begin{equation*}
  \gamma\in \left(0, \gamma_{c}^\lambda \right) \quad \text{ with }\quad
\gamma_{c}^\lambda := \arctan\left( \lim\limits_{\omega\rightarrow 0}\frac{B_{\mini{-1}} }{A_{\mini{-1}}} \right)\in \begin{cases}
                                                                \left(\frac{\pi}{2},\pi\right), &\text{ if }
                                                                  \begin{cases}
                                                                    \lambda<\lambda_c, &\text{ for }\nu\in(0,1)\\
                                                                    \lambda>\lambda_c, &\text{ for }\nu\in(1,2)
                                                                  \end{cases}.\\
                                                                \left(0,\frac{\pi}{2}\right), &\text{ otherwise.}
                                                        \end{cases}
\end{equation*}

Now consider $\kappa=+1$ and define $\omega_c := \frac{\sqrt{| 1-2\lambda| }}{2L}$. First, for simplicity, we impose the restriction $\xi<\frac{1}{2}$. In this case, for $\Upsilon$ given by Equation \eqref{eq: parameter upsilon lif top bh}, we have that $\text{Re}(\Upsilon)>0$ and the radial solution defined in Equation \eqref{eq: square-integrable R(r) at r=0 k=+1} is well-defined for all $\omega\in\mathbb{C}$ such that if $\omega=\text{Re}(\omega)$ then $|\omega|<\omega_c$. For $\nu\in(0,1)\cup(1,2)$, we find that can impose mode-dependent boundary conditions with
\begin{equation*}
\gamma\in[0,\gamma^\lambda_{c}), \text{ with }\gamma^\lambda_{c} := \arctan\left( \lim\limits_{\omega\rightarrow \omega_c}\frac{B_{\mini{-1}} }{A_{\mini{-1}}}  \right)\in[0,\pi).
\end{equation*}
To impose the same boundary condition on all modes, for $\nu\cup(1,2)$ only $\gamma=0$ is admissible. However, for $\nu\in(0,1)$, there are no bound states for $\gamma\in\left[0,\pi/2\right]$.

When the radial Green function has no poles we can perform the contour integration on the complex-plane and obtain the radial component of the two-point functions, as described in general terms in Section \ref{sec: Physically-sensible states in Schwarzschild-like coordinates} and written down in details and with figures in \cite{deSouzaCampos2021role}. All in all, we find that the radial part of the two-point function of a physically-sensible state reads
\begin{equation}
  \label{eq:term lif top bh}
    \widetilde{\psi}_{2}(r,r')=\frac{1}{ \pi c_\kappa}\frac{\Imag\left(\overline{A_\kappa}B_\kappa\right)}{|B_\kappa\cos(\gamma)-A_\kappa\sin(\gamma)|^2} R_{\gamma}(r) R_{\gamma}(r').
\end{equation}

\subsubsection*{Two-point functions}
\label{subsec: Two-point functions on Lifshitz top bh}

Using Equation \eqref{eq:term lif top bh} in Theorems \ref{thm: 2 point ground state schd coord} and \ref{thm: 2 point KMS state schd coord}, we find the two-point functions of the ground state
  \begin{align}
  \label{eq:two-point function ground state lif top bh}
  \psi_2(x,x^\prime)= &  \lim_{\varepsilon\rightarrow 0^+} \int_{\sigma(\Delta_{\kappa})}d\eta_{\kappa}\int_{\mathbb{R}}d\omega \Theta(\omega) e^{-i\omega (t - t'- i\varepsilon)} \widetilde{\psi}_{2}(r,r')Y_{\kappa}(\underline{\theta})\overline{Y_{\kappa}(\underline{\theta}')},
  \end{align}
and of KMS states at inverse-temperature $\beta$
  \begin{align}
    \label{eq:two-point function KMS state lif top bh}
  \psi_2(x,x^\prime)=\lim_{\varepsilon\rightarrow 0^+}\int_{\sigma(\Delta_{\kappa})}d\eta_{\kappa} \int_{\mathbb{R}} d\omega \Theta(\omega) \left[ \frac{e^{-i\omega (t-t'-i\varepsilon)}}{1 - e^{-\beta\omega}}+ \frac{e^{+i \omega (t-t'+i\varepsilon)}}{{e^{\beta\omega}-1}} \right] \widetilde{\psi}_{2}(r,r') Y_{\kappa}(\underline{\theta})\overline{Y_{\kappa}(\underline{\theta}')}.
  \end{align}
\noindent In the expressions above, the integral over the spectrum of the Laplacian with respect to the measure $d\eta_{\kappa}$, with some abuse of notation, is given by
\begin{align*}
    \int_{\sigma(\Delta_{\kappa})}d\eta_{\kappa} \equiv
            \begin{cases}
                    \int\limits_{\mathbb{R}}d\ell\int\limits_{\mathbb{R}}dm, &\text{ for }\kappa=0,\\[10pt]
                    \int\limits_{0}^\infty d\ell\sum\limits_{m=0}^\infty , &\text{ for }\kappa=-1,\\[10pt]
                    \sum\limits_{\ell=0}^\infty\sum\limits_{m=-\ell}^\ell, &\text{ for }\kappa=+1.
            \end{cases}
\end{align*}

\subsection{The transition rate of an Unruh-DeWitt detector}
\label{subsec: The transition rate on Lif top bh}

Using Equations \eqref{eq:two-point function ground state lif top bh} and \eqref{eq:two-point function KMS state lif top bh} in Theorem \ref{thm: Transition rate for the physically-sensible states}, we obtain that the transition rate for a static Unruh-DeWitt detector on Lifshitz topological black holes for the ground and thermal states is given by, respectively,
\begin{align}
  \label{eq:transition rate ground state massless top bh}
  \dot{\mathcal{F}}_\infty =  \Theta( -\Omega ) \int_{\sigma(\Delta_{\kappa})}d\eta_{\kappa}\, \frac{1}{ \pi c_\kappa}\frac{\Imag{(\overline{A_\kappa}B_\kappa)}   }{|B_\kappa\cos(\gamma)-A_\kappa\sin(\gamma)|^2} |Y_\kappa(\underline{\theta})|^2  R_\gamma(r)^2 \Big|_{\omega = -\sqrt{f(r)} \,\Omega},
\end{align}
and
\begin{align}
  \label{eq:transition rate KMS state massless top bh}
  \dot{\mathcal{F}}_{\beta} = \frac{\text{ sign}(\Omega)}{e^{\beta\sqrt{f(r)}\Omega}-1}  \int_{\sigma(\Delta_{\kappa})}d\eta_{\kappa} \,  \frac{1}{ \pi c_\kappa}\frac{\Imag{(\overline{A_\kappa}B_\kappa)}   }{|B_\kappa\cos(\gamma)-A_\kappa\sin(\gamma)|^2}  |Y_\kappa(\underline{\theta})|^2  R_\gamma(r)^2 \big|_{\omega =\sqrt{f(r)} |\Omega| }.
\end{align}

\section{On a global monopole}
\label{sec: On a global monopole}
Within the framework described in Chapter \ref{chap: Quantum Field Theory on Static Spacetimes}, in Section \ref{subsec: Ground and thermal states}, I construct physically-sensible states for a real, massive, arbitrarily coupled Klein-Gordon field on global monopoles, which are described in Section \ref{ex: Chapter 1 Global monopoles}. Free, scalar fields have been studied on global monopoles before. Ground and thermal states for Dirichlet boundary condition are well-known \cite{mazzitelli1991vacuum,carvalho2001vacuum}, and recently ground states for Robin boundary conditions have been constructed \cite{pitelli2009quantum,pitelli2018gmonopole}. For Dirichlet boundary condition, thermal effects have also been taken into account \cite{carvalho2001vacuum}. Here, I generalize these previous works by constructing thermal states compatible with Robin boundary conditions. Then, in Section \ref{sec: Transition rate global monopoles}, by employing the Unruh-DeWitt detector approach, I study thermal effects and their relation with these different boundary conditions admissible at the naked singularity. To better understand the behaviour of the detector, I also compute the thermal contributions to the expectation value of the field squared and the energy density of the thermal states renormalized by the ground state, analogously to the computation in Section \ref{sec: On Minkowski spacetime in spherical coordinates} on Minkowski spacetime. The analysis I report here has been published in \cite{deSouzaCampos2021awm}.

In Schwarzschild-like coordinates, the line element of a global monopole is given in Equation \eqref{eq: metric Global monopole}. However, for consistency with the notation used in \cite{deSouzaCampos2021awm}, I apply the coordinate transformations $t\mapsto t/\alpha $ and $r\mapsto \alpha r$. That is, here, I consider the line element
\begin{align*}
    ds^2 =-dt^2 + dr^2 + \alpha^2 r^2 d\theta^2 + \alpha^2 r^2 \sin^2\theta d\varphi^2.
\end{align*}
In these coordinates, the hypersurface $\theta=\frac{\pi}{2}$ is a cone with a deficit angle of $2\pi(1-\alpha)$, where $\alpha\in(0,1)$. The hypersurface $r\rightarrow0$ is a timelike, naked singularity of curvature type since both the Ricci and the Kretschmann scalars diverge. The latter are given by, respectively
\begin{equation*}
  \mathbf{R} = \frac{2(1-\alpha ^2)}{\alpha ^2 r^2}\quad \text{ and }\quad \mathbf{K} = \mathbf{R}^2.
\end{equation*}

\subsection{Ground and thermal states}
\label{subsec: Ground and thermal states}

Let $Y_\ell^m(\theta,\varphi)$ be the spherical harmonics with eigenvalues $-\ell(\ell+1)$. If
\begin{align*}
\Psi_{\omega,\ell}(t,r,\theta,\varphi) =  e^{-i\omega t } R(r)Y_\ell^m(\theta,\varphi),
\end{align*}
satisfies the Klein-Gordon equation, then the function $R(r)$ solves the Bessel equation
\begin{equation*}
    R''(r)+ \frac{2}{r}R'(r)+ \left(p^2 + \frac{\lambda_{\ell,\xi,\alpha}}{r^2} \right)R(r)=0,
\end{equation*}
with
\begin{align*}
  &p^2 := \omega^2 - m_0^2, \quad \text{ and }\quad\lambda_{\ell,\xi,\alpha} :=- \frac{\ell(\ell+1)+ 2\xi(1-\alpha^2)}{\alpha^2}.
\end{align*}
Let $j_\nu$ and $y_\nu$ be the spherical Bessel functions of first and second kind \cite{NIST_10} and define
\begin{align*}
&R_{1}(pr)= j_{\nu}(pr), \quad \text{ and }\quad
R_{2}(pr) = p\, y_{\nu}(pr),
\end{align*}
where, assuming $\xi\geq 0$,
\begin{align*}
& \nu := \frac{-1 + \sqrt{1 - 4\lambda_{\ell,\xi,\alpha}}}{2}\geq0.
\end{align*}
The most general solution that is square-integrable at the naked singularity with respect to the measure $r^2dr$ satisfies a Robin boundary condition parametrized by $\gamma$ and is given by
\begin{align*}
R_{\gamma_{\nu}}(pr):= \cos(\gamma_{\nu})R_{1}(pr) - \sin(\gamma_{\nu})R_{2}(pr),
\end{align*}
where the auxiliary quantity $\gamma_{\nu}$ is defined as
\begin{align*}
\gamma_{\nu}:=
\begin{cases}
\gamma\in[0,\pi), \quad &\text{ if }\nu<\frac{1}{2},\\
0, \quad &\text{ if }\nu>\frac{1}{2}.
\end{cases}
\end{align*}
Note that $\nu<\frac{1}{2}$ only for $\ell=0$ and only if $\xi\in\left[0,\frac{3}{8}\frac{\alpha^2}{1-\alpha^2} \right)$. Moreover, it holds that for $\gamma\in\left[0,\frac{\pi}{2}\right]$ the radial Green function, defined as per Equation \eqref{eq: green function radial equation}, has no poles. These results can be verified by performing the same analysis of Section \ref{sec: Spectral resolution of the radial Green function Mink 4D} on Minkowski spacetime and by simply adding the $\alpha$-terms in the right places. In the end, the radial part of the two-point functions is given by
\begin{equation}\label{eq: radial part global monopole}
  \widetilde{\psi}_{2}(r,r') = \Theta(\omega-m_0)\frac{p}{\pi \alpha^2}\frac{R_{\gamma_\nu}(pr) R_{\gamma_\nu}(pr')}{\cos(\gamma_\nu)^2+p^2 \sin(\gamma_\nu)^2} .
\end{equation}
Analogously to Equations \eqref{eq: 2 point ground state mink 4 summed in m} and \eqref{eq: 2 point kms state mink 4 summed in m}, with $\widetilde{\psi}_2(r,r')$ given by Equation \eqref{eq: radial part global monopole}, we find that the two-point functions read, for the ground and thermal states on global monopoles respectively,
\begin{equation}\label{eq: 2 point ground state global monopole}
 \psi_{2,\infty}(x,x') = \lim_{\varepsilon \to 0^{+}} \sum_{\ell=0}^\infty \int_{0}^{\infty} d\omega\, e^{-i\omega(t-t'-i\varepsilon)}\widetilde{\psi}_2(r,r')\frac{2\ell+1}{4\pi} P_\ell(\Phi),
\end{equation}
and
\begin{equation}\label{eq: 2 point kms state global monopole}
 \psi_{2,\beta}(x,x') = \lim_{\varepsilon \to 0^{+}} \sum_{\ell=0}^\infty \int_{0}^{\infty} d\omega   \left[ \frac{e^{-i\omega (t-t'-i\varepsilon)}}{1 - e^{-\beta\omega}}+ \frac{e^{+i \omega (t-t'+i\varepsilon)}}{{e^{\beta\omega}-1}} \right]\widetilde{\psi}_2(r,r')\frac{2\ell+1}{4\pi} P_\ell(\Phi),
\end{equation}
where $P_\ell(\Phi)$ is the Legendre function and $\Phi\equiv \cos(\theta)\cos(\theta')+\sin(\theta)\sin(\theta')\cos(\varphi-\varphi')$.

\subsection{Transition rate of an Unruh-DeWitt detector}
\label{sec: Transition rate global monopoles}
As per Theorem \ref{thm: Transition rate for the physically-sensible states}, we can directly compute the transition rate of a static Unruh-DeWitt detector of energy gap $\Omega$ using Equations \eqref{eq: 2 point ground state global monopole} and \eqref{eq: 2 point kms state global monopole}. For a Klein-Gordon field of mass $m_0\geq 0$, we find for the ground state
\begin{equation}
  \label{eq: transition rate ground state}
\dot{\mathcal{F}}_{\infty,\gamma} (r)= \frac{\Theta(-\Omega-m_0)\sqrt{\Omega^2-m_0^2}}{2\pi \alpha^2}\sum_{\ell=0}^\infty    \frac{2\ell+1}{\cos(\gamma_\nu)^2+ (\Omega^2-m_0^2) \sin(\gamma_\nu)^2}   \left[R_{\gamma_\nu}\left(\sqrt{\Omega^2-m_0^2}r\right)\right]^2
\end{equation}
and for a thermal state at inverse-temperature $\beta$
\begin{equation}
  \label{eq: transition rate KMS state global monopole}
\dot{\mathcal{F}}_{\beta,\gamma } (r)= \frac{\text{sign}(\Omega)}{e^{\beta\Omega}-1}\left[\dot{\mathcal{F}}_{\infty} (r)\big|_{\Omega\mapsto -|\Omega|}\right].
\end{equation}
In the particular case of $m_0=0$, Equation \eqref{eq: transition rate KMS state global monopole} simplifies to
\begin{equation}
  \label{eq: transition rate KMS state massless field}
\dot{\mathcal{F}}_{\beta,\gamma}(r)= \frac{1}{2\pi \alpha^2}  \frac{\Omega}{e^{\beta\Omega}-1}   \sum_{\ell=0}^\infty \frac{(2\ell+1)\left[R_{\gamma_\nu}\left(|\Omega|r\right)\right]^2}{\cos(\gamma_\nu)^2+ \Omega^2 \sin(\gamma_\nu)^2}.
\end{equation}
In the limit $\alpha\rightarrow 1$, the expressions above match the expected results on Minkowski spacetime, as computed in Section \ref{sec: On Minkowski spacetime in spherical coordinates}.

Note that since there is no horizon, even though we can construct thermal states neither Equation \eqref{eq: def global hawking temperature} nor Equation \eqref{eq: def local hawking temperature} makes sense on global monopoles, i.e. there is no preferred temperature a priori, no Hawking temperature. Accordingly, the focus in this section is not on anti-correlation effects, but rather on the consequences of imposing different boundary conditions at the naked singularity on the detector transition rates. Thence, to auxiliate in the interpretation of the behaviour of the detector with respect to its distance from the naked singularity for different boundary conditions, we also compute the thermal fluctuations and the energy density, as exemplified on Minkowski spacetime in Section \ref{sec: On Minkowski spacetime in spherical coordinates}.

The thermal contribution to the expectation value of the field squared in the ground state, which we call thermal fluctuations, is given by
\begin{equation}
  \label{eq:erjopfopek4004040404}
\Delta\mathcal{G}_{\beta,\gamma}(x,x'):=  \psi_{2,\beta}(x,x') -  \psi_{2,\infty}(x,x'),
\end{equation}
with the two-point functions defined in Equations \eqref{eq: 2 point ground state global monopole} and \eqref{eq: 2 point kms state global monopole}. Considering the regularized thermal state defined by Equation \eqref{eq:erjopfopek4004040404}, we define the energy density as
\begin{equation*}
  E_{\beta,\gamma}(r)   :=  \braket{:T_{00}(r):}_{\beta,\gamma},
\end{equation*}
where the energy-momentum tensor, together with Equation \eqref{eq: differential operator stress energy tensor moretti}, is defined as
\begin{equation*}
    \braket{:T_{\mu\nu}(r):}_{\beta,\gamma} = \lim\limits_{x'\rightarrow x}\left\{\mathcal{D}_{\mu\nu}(x,x')\left[\Delta\mathcal{G}_{\beta,\gamma}(x,x')\right]\right\}.
\end{equation*}

For the three quantities of interest here---transition rate, thermal fluctuations and energy density---seen as a function of $\alpha$, let us define their counterpart on Minkowski spacetime for the sake of comparison:
\begin{equation*}
\text{quantity}^{\text{Mink}}:= \lim\limits_{\alpha\rightarrow 1} \text{quantity}(\alpha).
\end{equation*}

In the following, I summarize the numerical analysis performed for the massless, minimally and conformally coupled field. This analysis is available at \cite{git_global_monopole} and detailed in \cite{deSouzaCampos2021awm}. All plots follow the color code defined in Figure \ref{fig: legends alpha and gamma}.
\begin{figure}[H]
  \begin{center}
    \includegraphics[width=0.12\textwidth]{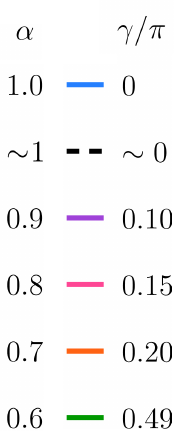}
  \end{center}
  \caption{Legends for the following plots.}
  \label{fig: legends alpha and gamma}
\end{figure}

For both $\xi=0$ and $\xi=\frac{1}{6}$, we find that the transition rate behaves similarly when the detector is far away from the singularity, i.e. in the limit $r\rightarrow \infty$, it approximates, but oscillating around, its value on Minkowski spacetime:
\begin{equation*}
\lim\limits_{r\rightarrow \infty} \dot{\mathcal{F}}_{\beta,\gamma}(r) \sim  \dot{\mathcal{F}}_{\beta,\gamma}^{\text{Mink}} (r) .
\end{equation*}
However, the behavior of the three quantities as we approach the global monopole is dramatically different in each case, as shown next.

\subsubsection*{Minimal coupling}
\label{subsec: Minimal coupling global monopole}

For $m_0=0$ and $\xi=0$, we find that as we approach the global monopole the angular deficit affects the $\ell=0$ terms of the transition rate, of the thermal fluctuations and of the energy density by amplifying them by the same factor, i.e. the three quantities are such that
\begin{equation}
  \label{eq: quantity vs quantity mink}
 \text{quantity}(\alpha) \rightarrow  \frac{1}{\alpha^2}  \text{quantity}^{\text{Mink}}, \quad\text{ as }r\rightarrow 0 .
\end{equation}
This behavior is visible in Figure \ref{fig: transition rate KMS several alphas}. Equation \eqref{eq: quantity vs quantity mink} also holds for the transition rate and the thermal fluctuations if we sum over $\ell$. As a matter of fact, directly from Equation \eqref{eq: transition rate KMS state massless field}, it follows that as $r\rightarrow 0$ only $\ell=0$ contributes to the transition rate. With respect to the boundary condition chosen, we verify that for $\gamma>0$, the three quantities diverge at the naked singularity and remain finite there if and only if $\gamma=0$.

\begin{figure}[H]
  \centering
     \includegraphics[align=c,width=\textwidth]{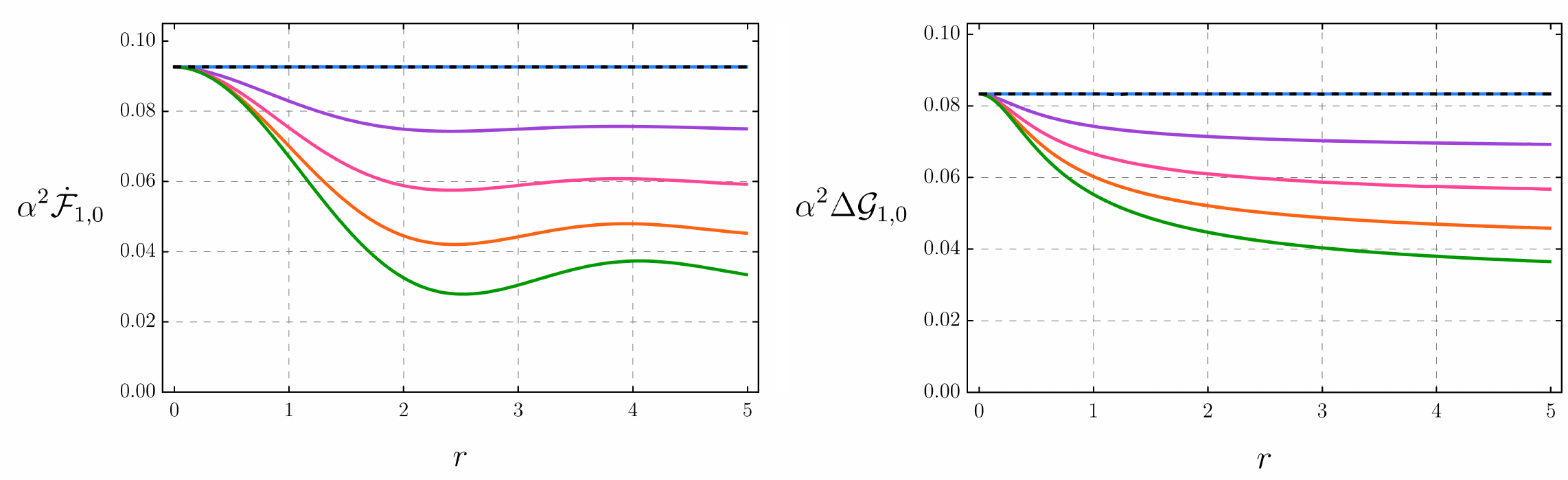}
      \caption{For the thermal state with $m_0=0$, $\xi=0$, $\beta=1$, $\Omega=1$, $\gamma=0$, and, from top to bottom with respect to the apex, $\alpha\in\{0.6,0.7,0.8,0.9,0.99999,1.0\}$. On the left: the transition rate multiplied by $\alpha^2$ with $\ell_{\text{max}}=10$. On the right: thermal fluctuations multiplied by $\alpha^2$ with $\ell_{\text{max}}=50$.}
  \label{fig: transition rate KMS several alphas}
\end{figure}

   \begin{figure}[H]
     \centering
    \includegraphics[align=c,width=\textwidth]{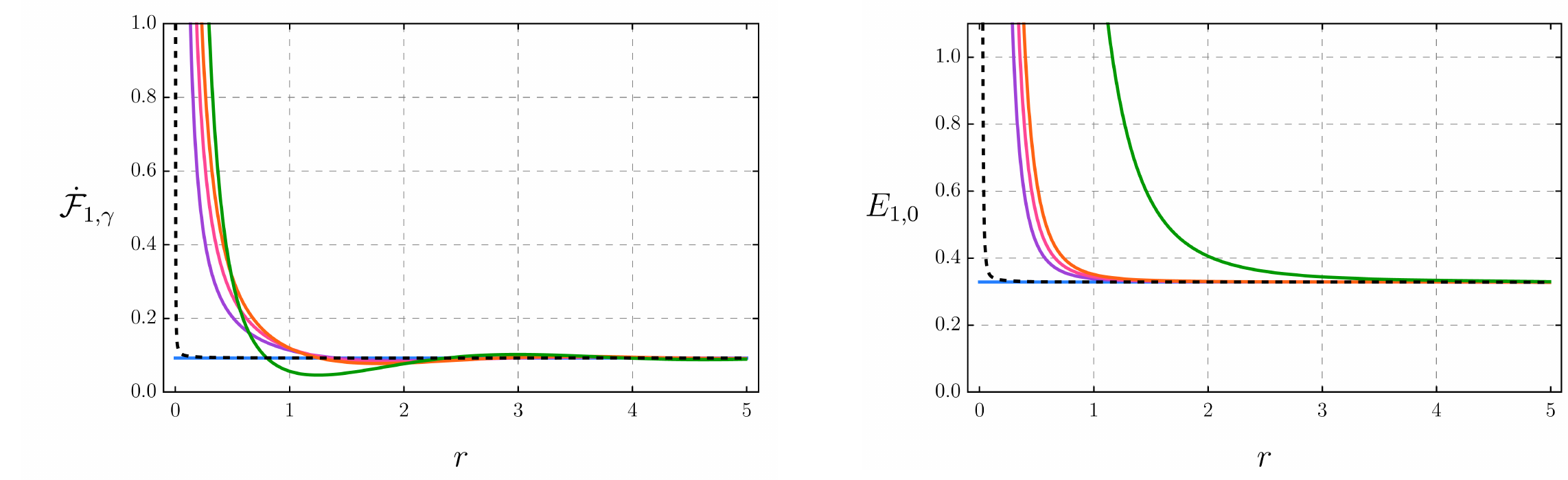}
     \caption{For the thermal state with $m_0=0$, $\xi=0$, $\beta=1$, $\Omega=1$, $\alpha=0.99999$ and several values of $\gamma$. On the left: the transition rate with $\ell_{\text{max}}=10$. On the right: energy density with $\ell_{\text{max}}=50$.}
     \label{fig: transition rate and energy density KMS several gammas}
       \end{figure}

\subsubsection*{Conformal coupling}
\label{subsec: Conformal coupling global monopole}
  \enlargethispage{2\baselineskip}
For $m_0=0$, $\xi= \frac{1}{6}$, $\alpha\in(0,1)$, and $\gamma\geq 0$, the transition rate given by Eq.~\eqref{eq: transition rate KMS state massless field} vanishes as $r\rightarrow 0$, as shown in Figure \ref{fig: transition rate KMS several alphas conformal}. This is in sharp contrast with the minimal coupling case of Figure \ref{fig: transition rate KMS several alphas}, but it is consistent with the behaviour of the thermal fluctuations illustrated on the left in Figure \ref{fig: app conformal thermal fluctuations and energy density l0}. However, it is noteworthy that the transition rate of excitations and the thermal fluctuations vanish at the naked singularity even though the energy density diverges there, as shown on the left in Figure \ref{fig: app conformal thermal fluctuations and energy density l0}. In addition, this divergence occurs for all boundary conditions, even for the Dirichlet one.

\begin{figure}[H]
  \centering
\includegraphics[align=c,width=\textwidth]{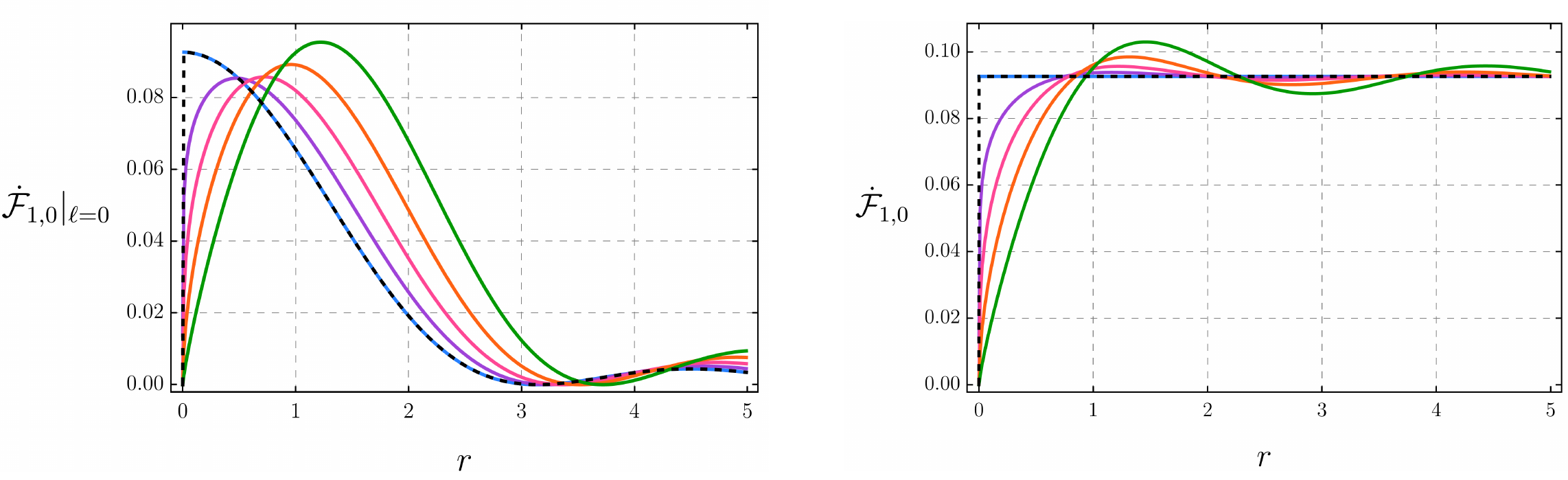}
  \caption{The transition rate for the thermal state with $m_0=0$, $\xi=\frac{1}{6}$, $\beta=1$, $\Omega=1$, $\gamma=0$ and several values of $\alpha$. On the left: $\ell_{\text{max}}=0$. On the right: $\ell_{\text{max}}=10$.}
  \label{fig: transition rate KMS several alphas conformal}
\end{figure}

\begin{figure}[H]
  \centering
   \includegraphics[width=\textwidth]{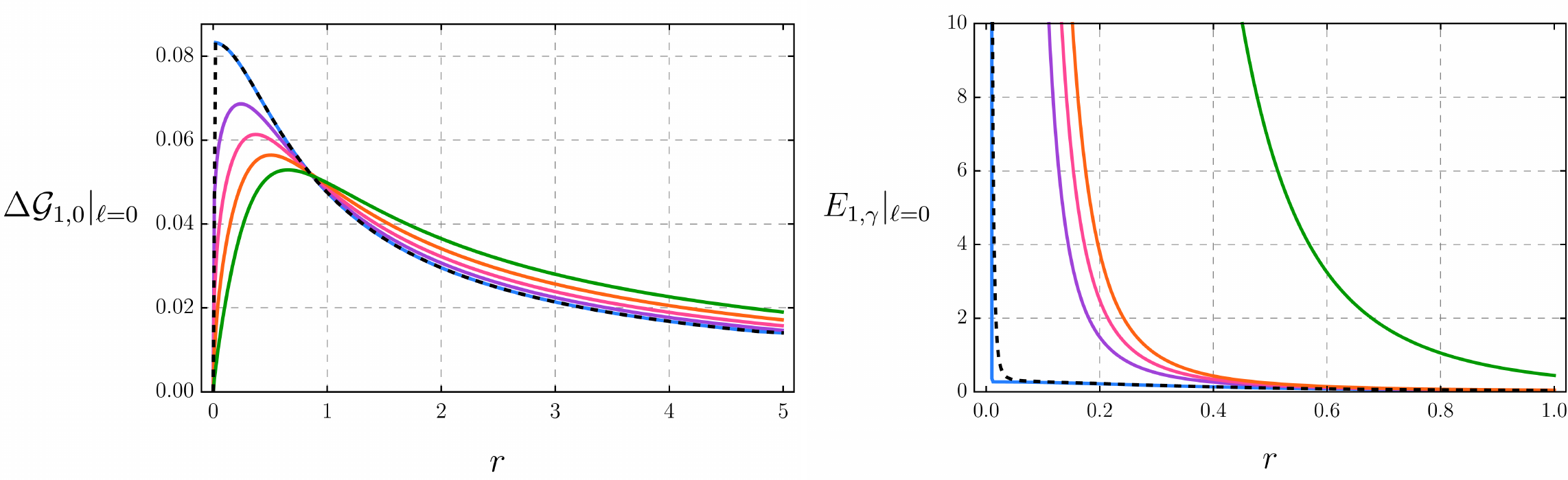}
  \caption{For $m_0=0$, $\xi=\frac{1}{6}$, $\beta=1$, $\gamma=0$, and $\ell_{\text{max}}=0$. On the left: thermal fluctuations for several $\alpha$'s. On the right: energy density for several $\gamma$'s.}
  \label{fig: app conformal thermal fluctuations and energy density l0}
\end{figure}


\pagestyle{myPhDpagestyle}
\chapter*[Conclusions]{Conclusions}
\addcontentsline{toc}{chapter}{Conclusions}

\vspace{-.5cm}

With the overarching aim of probing the interaction between quantum physics and general relativity, I considered physical phenomena within quantum field theory on curved spacetimes. First, in Chapter \ref{chap: The Spacetimes}, I gave key definitions from general relativity to clarify the geometric features of the underlying backgrounds. Second, in Chapter \ref{chap: Quantum Field Theory on Static Spacetimes}, I outlined the steps for establishing a quantum field theoretical framework for free Klein-Gordon fields on static spacetimes. In particular, I delineated a procedure to follow to obtain physically-sensible ground and thermal states and the respective transition rate for static Unruh-DeWitt detectors on spacetimes that admits Schwarzschild-like coordinates. Then, in Chapter \ref{chap: Applications}, I employed the above framework on some spacetimes of interest. Next, I give a summary of the results obtained per spacetime considered. Subsequently, I contemplate associated follow up works and open questions.

\subsubsection*{Summary of the results obtained per spacetime considered}
\begin{itemize}
    \item[$\star$] \underline{On a static BTZ black hole and its universal covering Rindler-AdS$_3$:}
  \item[] In Section \ref{sec: On a static BTZ spacetime and on Rindler-AdS3} we reviewed the construction of physically-sensible states on static BTZ black holes and we computed the transition rate of a static Unruh-DeWitt detector coupled either to ground or thermal states. In particular, by analysing the behavior of the transition rate with respect to the Hawking temperature, we studied the occurrence of the anti-Hawking effect and its relation with the mass of the black hole, the boundary condition chosen at AdS infinity and the state of the quantum field. The results obtained considering the ground state comport and generalize the ones from \cite{Henderson2019uqo}. That is, we obtained that for sufficiently small masses the $\ell=0$ contribution dictates the behaviour of the transition rate such that the anti-Hawking effect is manifest for all Robin boundary conditions parametrized by $\gamma\in[0,\pi/2]$. For larger masses, the numerical analysis does indicate that the effect observed for the $\ell=0$ term may be cancelled for Dirichlet boundary condition by performing the $\ell$ sum, but would still occur for Neumann boundary condition. In addition, we found that the anti-Hawking effect is absent for a KMS state on a static BTZ spacetime---and so is the anti-Unruh effect on its universal covering Rindler-AdS$_3$. These no-show results call upon further investigation on the relation between the effects, the KMS condition and the choice of boundary condition. The results reported here have been published in \cite{deSouzaCampos2020ddx}.\\
  \item[$\star$] \underline{On massless hyperbolic black holes:}
  \item[] In Section \ref{sec: On massless hyperbolic black holes}, we generalized the results from the previous item by constructing physically-sensible states on massless hyperbolic black holes, which can be seen as higher-dimensional generalizations of static BTZ black holes. We gave explicitly expressions for the two-point functions of the ground and thermal states and for the transition rate of a static Unruh-DeWitt detector. For the particular case of massless, conformally coupled fields on three- and four-dimensional spacetimes, we numerically analysed the transition rate and studied the manifestation of the anti-Hawking effect. The main result we found is that no anti-correlation is observed when the underlying spacetime is four-dimensional. On one hand, this is consistent with what we noticed to be the case on Minkowski spacetime concerning the anti-Unruh effect. On the other hand, the confirmation of this spacetime dimension dependence brings into question the statistical significance of these anti-correlation effects. The results reported here have been published in \cite{deSouzaCampos2020bnj}.
  \item[$\star$] \underline{On Lifshitz topological black holes:}
  \item[] In Section \ref{sec: On Lifshitz topological black holes}, we established a free, quantum field theoretical framework on flat, hyperbolic and spherical topological black holes within Lifshitz gravity. Particularly interesting is the fact that this side-by-side construction on these three spacetimes allowed us to compare the effects of the three different types of singularities: coordinate, curvature but hidden and curvature but naked. We found that in all three scenarios ground and thermal states for Klein-Gordon fields are well-defined for a large set of effective masses and Robin boundary conditions. Specifically, for $m_{\text{eff}}\in\left(-\frac{4}{L^2},-\frac{3}{L^2}\right)$, all Robin boundary conditions parametrized by $\gamma\in\left[0, \frac{\pi}{2}\right]$ are admissible at Lifshitz infinity. In addition, we noticed that we can also impose mode-dependent boundary conditions for $m_{\text{eff}}\in\left(-\frac{4}{L^2},-\frac{3}{L^2}\right)\cup\left(-\frac{3}{L^2},0\right) $. The results reported here were published in \cite{deSouzaCampos2021role}.
  \item[$\star$] \underline{On global monopoles:}
  \item[] In Section \ref{sec: On a global monopole}, we constructed ground and thermal states for a Klein-Gordon theory on global monopoles with Robin boundary conditions. We studied the transition rate of a static Unruh-DeWitt detector, its dependence on the boundary condition chosen at the naked singularity, and its interplay with the thermal contributions to the ground state fluctuations and energy density. We found that for a massless, minimally coupled Klein-Gordon field, the spontaneous emission rate of a detector interacting with a thermal state, as well as the thermal fluctuations and the thermal contribution to the energy density are finite everywhere, including at $r\rightarrow 0$, if and only if we impose Dirichlet boundary condition at the global monopole. Moreover, in this case we find that the three quantities compare with their counterparts on Minkowski spacetime in the same manner, and consistently with considering an angular deficit due to a cosmic string, i.e.
  \begin{equation*}
  \frac{\dot{\mathcal{F}}_{\beta,\gamma}(0^+)}{\dot{\mathcal{F}}_{\beta,\gamma}^{\text{Mink}} (0^+)}=  \frac{\Delta\mathcal{G}_{\beta,\gamma}(0^+)}{\Delta\mathcal{G}_{\beta,\gamma}^{\text{Mink}} (0^+)}=\frac{E_{\beta,\gamma}(0^+) }{E_{\beta,\gamma}^{\text{Mink}} (0^+)}\bigg|_{\ell=0}=\frac{1}{\alpha^2}.
\end{equation*}
  To the contrary, for massless, conformally coupled fields, we verified that for all boundary conditions, including the Dirichlet one, the energy density for the renormalized thermal state diverges. Markedly, regardless of this fact, since the thermal fluctuations vanish at the singularity, so does the transition rate. The results reported here were published in \cite{deSouzaCampos2021awm}.
\end{itemize}

\subsubsection*{Follow up work and open questions}
To ponder on follow up work and open questions regarding the results presented in this thesis and summarized above, I bear in mind the limitations of the framework employed together with the main focuses that were pointed out in the Introduction section respectively at pages \pageref{page Limitations of the framework} and \pageref{page main focuses}.
  \begin{enumerate}
  \item \underline{Anti-Unruh and anti-Hawking effects}
  \item[] On a static BTZ black hole, the anti-Hawking effect was first shown to occur for Dirichlet, Neumann and transparent boundary conditions in \cite{Henderson2019uqo}. Here, we confirmed that it also occurs for other Robin boundary conditions \cite{deSouzaCampos2020ddx}. In addition, by generalizing this analysis to higher dimensions, we found out that such effects are absent on all four-dimensional massless hyperbolic black holes \cite{deSouzaCampos2020bnj}. Recently, two noteworthy results have been obtained. First, the effect is intensified if one considers rotating BTZ black holes instead \cite{robbins2021anti}. Second, the effect is absent on other four-dimensional spacetimes \cite{conroy2021response}. Having in mind that the latter is consistent with a detector measuring the Unruh effect on a four-dimensional Minkowski spacetime, these results give rise to the following question: do anti-correlation effects occur on any four- or higher-dimensional spacetime? By studying such effects on particular four- or higher-dimensional spacetimes one after another one might be able to answer that---if such a procedure eventually halts. A good place to start in this regard would be to check if the anti-Hawking effect is manifest on four-dimensional hyperbolic black holes with non vanishing mass, though a more sophisticated numerical work would be needed in this case. Nevertheless, to pursue a better understanding of the statistical significance of the anti-correlation effects, one could attempt to look at the underlying mathematical structure that gives rise to this spacetime dependence. For example, the response of a detector on an $n$-dimensional Minkowski spacetime is bosonic-like or fermionic-like depending on the parity of $n$ \cite{takagi1986vacuum}. In \cite{Unruh1986tc}, a closer look at the origin of this dependence clarified that, for positive energy gap, there is not a relevant qualitatively distintion between the response as a function of the energy gap for different spacetime dimensions (as one can see in Figure \ref{fig:transition rate Minkowski 3 4 5 6 as a function of the energy gap}, or in \cite[Fig.1]{Unruh1986tc}).
  However, as noted here and clear in Figure \ref{fig:transition rate Minkowski 3 4 5 6 as a function of the energy gap}, for negative energy gap the response behaves rather differently for $n=3$. In \cite{Ottewill1987tm} it was shown that the response, besides depending on the ``number of particles'', depends also on a quantity that can be interpreted as a local density of states that relates with the underlying geometry. Their focus was on the simplicity of the results for $n=4$ in comparison with $n>4$, but an analogous analysis could be pursued to understand instead the contrast between $n<4$ with $n\geq 4$.
  \item \underline{Lorentz violation}
  \item[] Having constructed physically-sensible two-point functions on Lifshitz topological black holes, the foundations for the employment of the Unruh-DeWitt detector approach is set. In fact, the transition rates for a static detector coupled to massive, arbitrarily coupled fields admitting certain Robin boundary conditions in the ground or thermal states were explicitly obtained in Section \ref{sec: On Lifshitz topological black holes}. The next steps would be to study those expressions, numerically, in several situations, which is motivated by the fact that we know Unruh-DeWitt detectors are sensible to the topology of the underlying spacetime \cite{MartinMartinez2015qwa} and to Lorentz violating dispersion relations in field theory \cite{Husain2015tna}. The effects of choosing different boundary conditions and of changing the topology of the underlying spacetime have been studied for the vacuum polarization on topological black holes of Einstein gravity \cite{morley2021vacuum}. An analogous study could be performed on Lifshitz topological black holes. Such analysis should focus both on understanding the difference between the boundary conditions for each spacetime, and on the effects of changing just the underlying topology for each boundary condition. In addition, by generalizing the construction we performed on hyperbolic black holes to flat and spherical massless topological black holes of Einstein gravity, it would be interesting to compare the behaviour of the detector in these scenarios with those on Lifshitz topological black holes, to understand the consequences of introducing the Lorentz violation in the perspective of the detector.
  \newpage
  \item \underline{Naked singularities}
  \item[] We studied thermal effects on global monopoles with Robin boundary conditions. We verified that both the choice of boundary condition and the type of coupling change drastically the behavior of the transition rate, the thermal fluctuations and the energy density. We verified that these physical observables diverge at the naked singularity for many sets of the parameters. In these situations, it would be interesting to investigate whether backreaction effects arise and hide the naked singularity, as it happens on BTZ black holes \cite{Casals2016odj}. However, since we did not notice any particular behaviour for any specific value of the temperature, the question of what is the role of temperature on horizonless spacetimes remains most puzzling. Insights in this regard might emerge in a situation when we can compare a naked singularity and a hidden singularity in a shared context. Most stimulating would be to study a dynamical scenario that ends in a naked singularity and approach it as Hawking originally looked into black hole collapse \cite{hawking1975particle}. %
\end{enumerate}

\cleardoublepage                                                                                                                %
\appendix
\pagestyle{empty}

\chapter*{Conventions}
\addcontentsline{toc}{chapter}{Conventions}
\markboth{Conventions}{Conventions}
\footnotesize
\begin{tabular}{cp{.8\textwidth}}
    $\hslash=c=G=1$ & natural units\\
    \multicolumn{2}{l}{$g^{\mu\nu}x_\mu x_\nu<0 \Rightarrow x \text{ is timelike}$}\\
    \multicolumn{2}{l}{$\overline{\text{something}}\equiv\text{something}^\ast\equiv\text{complex conjugate of something}$}\\
    \multicolumn{2}{l}{Fourier transform convention: non-unitary with angular frequency $\hat{f}(\omega):=\frac{1}{2\pi}\int_{-\infty}^{\infty} e^{-i \omega x}f(x)dx $} \\
    $\mathbb{N}_0$&$\{0,1,2,...\}$\\
    $\mathbb{N}^*$&$\{1,2,...\}$\\
    $\mathbb{N}$ &  natural numbers with or without the zero\\
  $\mathcal{M}$ & $n$-dimensional static spacetime\\
  $g$ & Lorentzian metric tensor of $\mathcal{M}$ \\
  $\nabla$ & Levi-Civita connection of $g$ \\
  $\Gamma _{\mu\nu}^{\hspace{.3cm}\kappa }=\frac{1}{2} g^{\kappa \sigma } \left( \partial_\mu\partial g_{\sigma \nu }-\partial_\sigma g_{\mu \nu }+\partial_\nu g_{\sigma \mu }\right)$ & Christoffel symbols \\
  $R_{\mu\nu\alpha}^{\hspace{.5cm}\beta} =  - \partial_\mu \Gamma _{\nu \alpha }^{\hspace{.3cm}\beta }+\partial_\nu\Gamma _{\mu \alpha }^{\hspace{.3cm}\beta }  + \Gamma _{\nu \lambda }^{\hspace{.3cm}\beta } \Gamma _{\mu \alpha }^{\hspace{.3cm}\lambda } -\Gamma _{\mu \lambda }^{\hspace{.3cm}\beta } \Gamma _{\nu \alpha }^{\hspace{.3cm}\lambda } $ & Riemann tensor\\
  $R_{\mu\nu} = R_{\mu\sigma\nu}^{\hspace{.5cm}\sigma}$ & Ricci tensor\\
  $\mathbf{R}=g^{\mu\nu}R_{\mu\nu}$ & Ricci scalar \\
  $G_{\mu\nu} = R_{\mu\nu} - \frac{1}{2} \mathbf{R} g_{\mu\nu}$ & Einstein tensor\\
  greek letters & indices in ${1,...,n}$, effective mass or auxiliary parameter  \\
  $m_0$ & mass of scalar field \\
  $m_{\text{eff}}$ & effective mass of scalar field coupled to curvature \\
  $\text{id}_{C_{0}^{\infty}(\mathcal{M},\mathbb{R})}$ & $\text{id}_{C_{0}^{\infty}(\mathcal{M},\mathbb{R})} f = f, \,\forall f\in C_0^\infty(\mathcal{M},\mathbb{R}) $\\
  expression$(0^+)$ & standard weak-limit for distributions $\lim\limits_{\varepsilon\rightarrow 0^+} $ expression$(\varepsilon)$\\
  $\delta(x,x')$ & Dirac delta\\
  $\delta_{a,b}=\begin{cases} 1, & \text{ if }a=b \\ 0, & \text{ if }a\neq b\end{cases}$ & Kronecker delta \\[.5cm]
  $\Theta(x)=\begin{cases} 1, & \text{ if }x>0 \\ 0, & \text{ if }x<0\end{cases}$ & Heavise step function, not defined for $x=0$ \\
  $W_x[u(x),v(x)]:= \frac{d u}{dx} v - u \frac{dv}{dx}$ & the Wronskian of differentiable functions $u$ and $v$
\end{tabular}
\cleardoublepage
\normalsize

\pagestyle{myPhDpagestyle}
\bibliographystyle{ieeetr}
\bibliography{bib}

\end{document}